\documentclass[aos]{imsart}

\RequirePackage{amsthm,amsmath,amsfonts,amssymb}
\RequirePackage[numbers,sort&compress]{natbib}
\RequirePackage[colorlinks,citecolor=blue,urlcolor=blue]{hyperref}
\RequirePackage{graphicx}%% uncomment this for including figures

\startlocaldefs
\theoremstyle{plain}

\newtheorem{theorem}{Theorem}[section]
\newtheorem{lemma}[theorem]{Lemma}
\theoremstyle{remark}

\ifx\remark\undefined
\newtheorem{remark}[theorem]{Remark}
\fi
\ifx\assumption\undefined
\newtheorem{assumption}[theorem]{Assumption}
\fi
\ifx\proposition\undefined
\newtheorem{proposition}[theorem]{Proposition}
\fi
\endlocaldefs

\input{smile.tex}
\numberwithin{equation}{section}
\begin{document}

\begin{frontmatter}
\title{StarTrek: Combinatorial Variable Selection with False Discovery Rate Control}
%\title{A sample article title with some additional note\thanksref{t1}}
\runtitle{StarTrek}
%\thankstext{T1}{A sample of additional note to the title.}

\begin{aug}
%%%%%%%%%%%%%%%%%%%%%%%%%%%%%%%%%%%%%%%%%%%%%%%
%% Only one address is permitted per author. %%
%% Only division, organization and e-mail is %%
%% included in the address.                  %%
%% Additional information can be included in %%
\author[A]{\fnms{Lu}~\snm{Zhang}\ead[label=e1]{lu\_zhang@g.harvard.edu}}
\and
\author[B]{\fnms{Junwei}~\snm{Lu}\ead[label=e2]{junweilu@hsph.harvard.edu}}
% \and
% \author[B]{\fnms{Third}~\snm{Author}\ead[label=e3]{third@somewhere.com}}
%%%%%%%%%%%%%%%%%%%%%%%%%%%%%%%%%%%%%%%%%%%%%%
%% Addresses                                %%
%%%%%%%%%%%%%%%%%%%%%%%%%%%%%%%%%%%%%%%%%%%%%%
\address[A]{Department of Statistics,
Harvard University\printead[presep={,\ }]{e1}}

\address[B]{Department of Biostatistics,
Harvard University\printead[presep={,\ }]{e2}}
\end{aug}

\begin{abstract}
Variable selection on the large-scale networks has been extensively studied in the literature. While most of the existing methods are limited to the local functionals especially the graph edges, this paper focuses on selecting the discrete hub structures of the networks. Specifically, we propose an inferential method, called StarTrek filter, to select the hub nodes with degrees larger than a certain thresholding level in the high dimensional graphical models and control the false discovery rate (FDR). Discovering hub nodes in the networks is challenging: there is no straightforward statistic for testing the degree of a node due to the combinatorial structures; complicated dependence in the multiple testing problem is hard to characterize and control. In methodology, the StarTrek filter overcomes this by constructing p-values based on the maximum test statistics via the Gaussian multiplier bootstrap. In theory, we show that the StarTrek filter can control the FDR by providing accurate bounds on the approximation errors of the quantile estimation and addressing the dependence structures among the maximal statistics.

To this end, we establish novel Cram\'er-type comparison bounds for the high dimensional Gaussian random vectors. Comparing to the Gaussian comparison bound via the Kolmogorov distance established by \cite{chernozhukov2014anti}, our Cram\'er-type comparison bounds establish the relative difference between the distribution functions of two high dimensional Gaussian random vectors, which is essential in the theoretical analysis of FDR control. Moreover, the StarTrek filter can be applied to general statistical models for FDR control of discovering discrete structures such as simultaneously testing the sparsity levels of multiple high dimensional linear models. We illustrate the validity of the StarTrek filter in a series of numerical experiments and apply it to the genotype-tissue expression dataset to discover central regulator genes.
\end{abstract}



\end{frontmatter}

%%%%%%%%%%%%%%%%%%%%%%%%%%%%%%%%%%%%%%%%%%%%%%
%%%% Main text entry area:
\section{Introduction}
\label{section:introduction}
Graphical models are widely used for real-world problems in a broad range of fields, including social science, economics, genetics, and computational neuroscience \cite{newman2002random,luscombe2004genomic, rubinov2010complex}. Scientists and practitioners aim to understand the underlying network structure behind large-scale datasets. For a high-dimensional random vector $\bX =(\bX_1, \cdots, \bX_d)  \in \RR^d$, we let $\cG = (\cV, \cE)$ be an undirected graph, which encodes the conditional dependence structure among $\bX$. Specifically, each component of $\bX$ corresponds to some vertex in $\cV=\{1,2\cdots, d\}$, and $(j,k) \notin \cE$ if and only if $\bX_j$ and $\bX_k$ are conditionally independent given the rest of variables. \rev{We denote the associated weight matrix by $\bTheta$ with $\bTheta_{jk}$ being the weight on the edge between $j$ and $k$.} Many existing works in the literature seek to learn the structure of $\cG$ via estimating the weight matrix $\bTheta$. For example, \cite{meinshausen2006high, yuan2007model, friedman2008sparse, rothman2008sparse, peng2009partial, lam2009sparsistency, ravikumar2011high, cai2011constrained, shen2012likelihood} focus on estimating the precision matrix in a Gaussian graphical model. Further, there is also a line of work developing methodology and theory to assess the uncertainty of edge estimation, i.e., constructing hypothesis tests and confidence intervals on the network edges, see \cite{cai2013optimal, gu2015local, ren2015asymptotic, cai2016inference, jankova2017honest, yang2018semiparametric, feng2019high, ding2020estimation}. Recently, simultaneously testing multiple hypotheses on edges of the graphical models has received increasing attention \cite{liu2013ggmfdr, cai2013two, xia2015testing, xia2018multiple,li2019ggm, eisenach2020high}.
%------------------------------------------------------------------------------------%

Most of the aforementioned works formulate the testing problems based on continuous parameters and local properties. For example, \cite{liu2013ggmfdr} proposes a method to select edges in Gaussian graphical models with asymptotic FDR control guarantees. Testing the existence of edges concerns the local structure of the graph. Under certain modeling assumptions, its null hypothesis can be translated into a single point in the continuous parameter space, for example, $\bTheta_{jk}=0$ where $\bTheta$ is the precision matrix or the general weight matrix. However, for many scientific questions involving network structures, we need to detect and infer discrete and combinatorial signals in the networks, which does not follow from single edge testing.  For example, in the study of social networks, it is interesting to discover active and impactful users, usually called ``hub users," as they are connected to many other nodes in the social network \cite{ilyas2011distributed,lee2019discovering}. In gene co-expression network analysis, identifying central regulators/hub genes \cite{yuan2017co,liu2019bioinformatics,liu2019identification} is known to be extremely useful to the study of progression and prognosis of certain cancers and can support the treatment in the future. In neuroscience, researchers are interested in identifying the cerebral areas which are intensively connected to other regions \cite{shaw2008neurodevelopmental,van2013network,power2013evidence} during certain cognitive processes. The discovery of such central/hub areas can provide scientists better understanding of the mechanisms of human cognition.

%------------------------------------------------------------------------------------%

Motivated by these applications in various areas, in this paper, we consider the hub node selection problem from the network models. In specific, given a graph $\cG = (\cV, \cE)$, where $\cV$ is the vertex set and $\cE \subseteq \cV \times \cV$ is the edge set, we consider multiple hypotheses on whether the degree of some node $j \in \cV$ exceeds a given threshold $k_\tau$:
\begin{equation}\nonumber
H_{0j}: \text{degree of node } j < k_{\tau} \text{  v.s.  }  H_{1j}:  \text{degree of node } j \ge k_{\tau},
\end{equation}
\rev{based on i.i.d. samples $\bX_1, \cdots \bX_n \stackrel{i.i.d.}{\sim}  \bX  \in \RR^d$.}
Throughout the paper, these nodes with large degrees will be called hub nodes. For each $j\in [d]$, let $\psi_j = 1$ if $H_{0j}$ is rejected and $\psi_j = 0$ otherwise. When selecting hub nodes, we would like to control the false discovery rate, as defined below:
\[
 {\rm FDR} = \EE{\frac{\sum_{j \in \cH_0} \psi_j}{\max\big\{\sum_{j=1}^d \psi_j, 1\big\}}},
\]
where $\cH_0 = \{j \mid  \text{degree of node } j < k_{\tau} \}$. Remark the hypotheses $H_{0j}, j \in [d]$ are not based on continuous parameters. They instead involve the degrees of the nodes, which are intrinsically discrete/combinatorial functionals. To the best of our knowledge, there is no existing literature studying such combinatorial variable selection problems. The most relevant work turns out to be \cite{lu2017adaptive}, which proposes a general framework for inference about graph invariants/combinatorial quantities on undirected graphical models. However, they study single hypothesis testing and have to decide which subgraph to be tested before running the procedure.

The combinatorial variable selection problems bring many new challenges. First, most of the existing work focus on testing continuous parameters \cite{liu2013ggmfdr, javanmard2013nearly,javanmard2014confidence,javanmard2014hypothesis,belloni2014inference,van2014asymptotically,xia2015testing, xia2018multiple, javanmard2019false, sur2019modern,zhao2020asymptotic}. For discrete functionals, it is more difficult to construct appropriate test statistics and estimate its quantile accurately, especially in high dimensions. 
%Besides, how to develop computationally efficient selection procedures is also nontrivial due to the combinatorial feature of the problem. 
Second, many multiple testing procedures rely on an independence assumption (or certain dependence assumptions) on the null p-values \cite{benjamini1995controlling,benjamini2001control,benjamini2010discovering}. However, the single hypothesis here is about the global property of the graph, which means that any reasonable test statistic has to involve the whole graph. Therefore, complicated dependence structures exist inevitably, which presents another layer of difficulty for controlling the false discoveries. Now we summarize the motivating question for this paper: how to develop a combinatorial selection procedure to discover nodes with large degrees on a graph with FDR control guarantees?
%\jlmargin{}{more ref}
This paper introduces the StarTrek filter to select hub nodes. The filter is based on the maximum statistics, whose quantiles are approximated by the Gaussian multiplier bootstrap procedure. 
%For a given subset $E$ of $\cV \times \cV$, denote the maximum test statistic by $T_{E}:= \max_{(j,k)\in E}|\tTheta_{jk}|$, where ${\tTheta} \in \RR^{d\times d}$ is a generic estimator of the weight matrix $\bTheta$.
Briefly speaking, the Gaussian multiplier bootstrap procedure estimates the distribution of a given maximum statistic of general random vectors with unknown covariance matrices by the distribution of the maximum of a sum of the conditional Gaussian random vectors.   The validity of high dimensional testing problems, such as family-wise error rate (FWER) control, relies on the non-asymptotic bounds of the Kolmogorov distance between the true distribution of the maximum statistics and the Gaussian multiplier bootstrap approximation, which is established in \cite{chernozhukov2013gaussian}. However, in order to control the FDR in the context of combinatorial variable selection, a more refined characterization of the quantile approximation errors is required. In specific, we need the so called Cram\'er-type comparison bounds 
    quantifying the accuracy of the p-values  in order to control the FDR in the  simultaneous testing procedures \cite{chang2016cramer}. In our context, consider two centered Gaussian random vectors $U,V\in \RR^{d}$ with different covariance matrices $\bSigma^U$, $\bSigma^V$ and denote the $\ell_{\infty}$ norms of $U,V$ by $\maxnorm{U},\maxnorm{V}$ respectively, then the Cram\'er-type comparison bounds aim to control the relative error $\left|\frac{\mathbb{P}(\maxnorm{U} > t)}{\mathbb{P}(\maxnorm{V} > t)}-1\right|$ for certain range of $t$. Comparing to the Kolmogorov distance  $\sup_{t\in \RR}\left|{\mathbb{P}(\maxnorm{U} > t)}-{\mathbb{P}(\maxnorm{V} > t)}\right|$ \cite{chernozhukov2015comparison},  the Cram\'er-type comparison bound leads to 
the relative error between two cumulative density functions, which is necessary to guarantee the FDR control. In specific, we show in this paper a novel Cram\'er-type Gaussian comparison bound 
  \begin{equation}\label{eq:intro_ccb_max}
      \sup_{0\le t \le C_0\sqrt{\log d}}\left|\frac{\mathbb{P}(\maxnorm{U} > t)}{\mathbb{P}(\maxnorm{V} > t)}-1\right|=  O\rbr{ \min\Big\{(\log d)^{5/2}\maxdiff^{1/2}, \frac{\zerodiff   \log d  }{\discon}\Big\}},
  \end{equation}
  for some constant $C_0>0$, where $\maxdiff:= ||\bSigma^U-\bSigma^V||_{\max}$ is the entrywise maximum norm difference between the two covariance matrices,  $\zerodiff:= ||\bSigma^U-\bSigma^V||_{0}$ with $\nbr{\cdot}_{0}$ being the entrywise $\ell_0$-norm of the matrix, and $\discon$ is the number of connected subgraphs in the graph whose edge set $\cE = \{(j,k): \bSigma^U_{jk}\neq 0  \text{ or }  \bSigma^V_{jk}\neq 0 \}$. This comparison bound in \eqref{eq:intro_ccb_max} characterizes the relative errors between Gaussian maxima via two types of rates: the $\ell_\infty$-norm  $\Delta_\infty$ and the $\ell_0$-norm $\Delta_0$. This implies a new insight that the Cram\'{e}r type bound between two Gaussian maxima is small as long as  either their covariance matrices are uniformly close or only sparse entries of the two   covariance matrices differ. As far as we know, the second type of rate in \eqref{eq:intro_ccb_max} has not been developed even in Kolmogorov distance results of high dimensional Gaussian maxima. In the study of FDR control, we need both types of rates: the $\Delta_\infty$ rate is used to show that the Gaussian multiplier bootstrap procedure is an accurate  approximation for the maximum statistic quantiles and the $\Delta_0$ rate is used to quantify the complicated dependence structure of the p-values for the single tests on the degree of graph nodes. In order to prove the Cram\'{e}r-type comparison bound in \eqref{eq:intro_ccb_max}, we develop two novel theoretic techniques to prove the two types of rates separately. For the $\Delta_\infty$ rate,  we reformulate the Slepian's interpolation \cite{slepian1962one} into an ordinary differential inequality such that the relative error can be controlled via the Gr{\"o}nwall's inequality \cite{gronwall1919note}.  
  To control the $\Delta_0$ rate, the anti-concentration inequality of Gaussian maxima developed in \cite{chernozhukov2015comparison} is no longer sufficient, we establish a new type of anti-concentration inequality for the derivatives of the soft-max of high dimensional Gaussian vectors.  The existing works on the Cram\'{e}r type comparison bounds such as \cite{liu2010cramer,liu2014phase,chang2016cramer} does not cover the high dimensional maximum statistics. Therefore, their techniques can not be directly extended to our case. To the best of our knowledge, it is the first time in our paper to prove the Cram\'er-type Gaussian comparison bounds \eqref{eq:intro_ccb_max} for high dimensional Gaussian maxima.

In summary,  our paper makes the following major contributions. First, we develop a novel StarTrek filter to select combinatorial statistical signals: the hub nodes with the FDR control.  This procedure involves maximum statistic and Gaussian multiplier bootstrap for quantile estimation. 
Second, in theory, the proposed method  is shown to be valid for  many different models with the network structures. In this paper, we provide two examples,  the Gaussian graphical model and the bipartite network in the multiple linear models. Third, we prove a new  Cram\'er-type Gaussian comparison bound with two types of rates: the maximum norm difference and $\ell_0$ norm difference. These results are quite generic and has its own significance in the probability theory.

\subsection{Related work}
%------------------------------------------------------------------------------------%

Canonical approaches to FDR control and multiple testing \cite{benjamini1995controlling,benjamini2001control,benjamini2010discovering} require that valid p-values are available, and they only allow for certain forms of dependence between these p-values. However, obtaining asymptotic p-values with sufficient accuracy is generally non-trivial for high dimensional hypothesis testing problems concerning continuous parameters \cite{javanmard2013nearly,javanmard2014confidence,javanmard2014hypothesis,belloni2014inference,van2014asymptotically,sur2019modern,zhao2020asymptotic}, not even to mention discrete/combinatorial functionals.

Recently, there is a line of work conducting variable selection without needing to act on a set of valid p-values, including \cite{barber2015controlling,barber2019knockoff,panning2019knockoff,xing2019controlling,dai2020false,dai2020scale}. These approaches take advantage of the symmetry of the null test statistics and establish FDR control guarantee. As their single hypothesis is often formulated as conditional independence testing, it is challenging to apply those techniques to select discrete signals for the problem studied in this paper.

Another line of work develops multiple testing procedures based on asymptotic p-values for specific high dimensional models \cite{liu2013ggmfdr,liu2014hypothesis,javanmard2019false,xia2015testing,xia2018multiple,liu2020integrative}. Among them, \cite{liu2013ggmfdr} studies the edge selection problem on Gaussian graphical models, which turns out to be the most relevant work to our paper. However, their single hypothesis is about the local property of the graph. Our problem of discovering nodes with large degrees concerns the global property of the whole network, therefore requiring far more work. 

There exists some recent work inferring combinatorial functionals. For example, the method proposed in \cite{ke2020estimation} provides a confidence interval for the number of spiked eigenvalues in a covariance matrix. \cite{jin2020estimating} focuses on estimating the number of communities in a network and yields confidence lower bounds. \cite{neykov2019combinatorial,lu2017adaptive} propose a general framework for conducting inference on graph invariants/combinatorial quantities, such as the maximum degree, the negative number of connected subgraphs, and the size of the longest chain of a given graph. \cite{shen2020combinatorial} develops methods for testing the general community combinatorial properties of the stochastic block model. Regarding the hypothesis testing problem, all these works only deal with a single hypothesis and establish asymptotic type-I error rate control. While simultaneously testing those combinatorial hypotheses is also very interesting and naturally arises from many practical problems.

\subsection{Outline}
In Section \ref{sec:method}, we set up the general testing framework and introduce the StarTrek filter for selecting hub nodes. In Section \ref{sec:cramer_theory}, we present our core probabilistic tools: Cram\'er-type Gaussian comparison bounds in terms of maximum norm difference and $\ell_0$ norm difference. To offer a relatively simpler illustration of our generic theoretical results, we first consider the hub selection problem on a bipartite network (multitask regression with linear models). Specifically, the input of the general StarTrek filter is chosen to be the estimators and quantile estimates described in Section \ref{sec:bipartite_selection}. Applying the probabilistic results under this model, we establish FDR control guarantees under certain conditions. Then we move to the Gaussian graphical model in Section \ref{sec:hub_selection}. In Section \ref{sec:simul}, we demonstrate StarTrek's performance through empirical simulations and a real data application.

\subsection{Notations}
% Orlicz norms are defined as $\norm{X}_{\psi_{\alpha}}=\inf\{c>0: \EE{\psi_{\alpha}(|X|/c)}\le 1 \}$ with $\psi_{\alpha}(x):=\exp(x^\alpha) - 1$, for $\alpha \ge 1$. 
Let $\phi(x),\Phi(x)$ be the probability density function (PDF) and the cumulative distribution function (CDF) respectively of the standard Gaussian distribution and denote $\bar{\Phi}(x) = 1 - \Phi(x)$. Let $\mathbf{1}_{d}$ be the vector of ones of dimension $d$. We use $\Indrbr{\cdot}$ to denote the indicator function of a set and $|\cdot|$ to denote the cardinality of a set. For two sets $A$ and $B$, denote their symmetric difference by $A \ominus B$, i.e., $A \ominus B = (A\setminus B) \cup (B\setminus A)$; let $A \times B$ be the Cartesian product. For two positive sequences $\{x_n\}_{n=1}^{\infty}$ and $\{y_n\}_{n=1}^{\infty}$, we say $x_n = O\rbr{y_n}$ if $x_n\le C y_n$ holds for any $n$ with some large enough $C>0$. And we say $x_n = o\rbr{y_n}$ if $x_n/y_n  \rightarrow 0$ as $n\rightarrow \infty$. For a sequence of random variables $\{X_n\}_{n=1}^\infty$ and a scalar $a$, we say $X_n  \le a + \smallop$ if for all $\epsilon > 0$, $\lim_{n \rightarrow \infty} \PP{ X_n - a > \epsilon } = 0$. Given a random variable $Z$, we define its $\psi_{\ell}$-norm for $\ell \ge 1$ as $\|Z\|_{\psi_{\ell}} = \sup_{p \ge 1} p^{-1/\ell} (\EEE|Z|^p)^{1/p}$.
Let $[d]$ denote the set $\{1,\dots,d\}$. The $\ell_{\infty}$ norm and the $\ell_{1}$ norm on $\RR^d$ are denoted by $\maxnorm{\cdot}$ and $\norm{\cdot}_1$ respectively. For a random vector $X$, let $\maxnorm{X}$ be its $\ell_{\infty}$ norm. For a matrix $\Ab \in \RR^{d_1\times d_2}$, we denote its minimal and maximal eigenvalues by $\lambda_{\min}(\Ab), \lambda_{\max}(\Ab)$ respectively, the elementwise max norm by $\nbr{\Ab}_{\max} = \max_{i\in [d_1],j\in [d_2]}|\Ab_{ij}|$ and the elementwise $\ell_0$ norm by $\nbr{\Ab}_{0} = \sum_{i\in [d_1],j\in [d_2]}\Indrbr{\Ab_{ij} \ne 0}$. Throughout this paper, $C, C',C'', C_0, C_1, C_2,\dots$ are used as generic constants whose values may vary across different places.

\section{Methodology}\label{sec:method}
Before introducing our method, we set up the problem with more details. Specifically, we consider a graph $\cG = (\cV_1, \cV_2, \cE)$ with the node sets $\cV_1,\cV_2$ and the edge set $\cE$. Let $d_1=|\cV_1|$, $d_2=|\cV_2|$ and denote its weight matrix by $\bTheta \in \RR^{d_1\times d_2}$. In the undirected graph where $\cV_1 = \cV_2:=\cV$, $\bTheta$ is a square matrix and its element $\bTheta_{jk}$ is nonzero when there is an edge between node $j$ and node $k$, zero when there is no edge. 
%Those diagonal entries $\{\bTheta_{jj}\}_{j=1}^d$ are set to be zero.
In a bipartite graph where $\cV_1 \ne \cV_2$, elements of $\bTheta$ describe the existence of an edge between node j in $\cV_{1}$ and node $k$ in $\cV_{2}$. Without loss of generality, we focus on one of the node sets and denote it by $\cV$ with $|\cV|:=d$. We would like to select those nodes among $\cV$ whose degree exceeds a certain threshold $k_{\tau}$, \rev{based on the $n$ data samples $\bX_1, \cdots \bX_n \stackrel{i.i.d.}{\sim}  \bX  \in \RR^d $}. And the selection problem is equivalent to simultaneously testing $d$ hypotheses:
\begin{equation} \label{eq:problem_setup}
H_{0j}: \text{degree of node } j < k_{\tau} \text{  v.s.  }  H_{1j}:  \text{degree of node } j \ge k_{\tau},
\end{equation}
for $j \in [d]$. Let $\psi_j = 1$ if $H_{0j}$ is rejected and $\psi_j = 0$ otherwise, then for some multiple testing procedure with output $\{\psi_j\}_{j\in[d]}$, the false discovery proportion (FDP) and FDR can be defined as below:
\[
   {\rm FDP} = \frac{\sum_{j \in \cH_0}^d \psi_j}{\maxof{1}{\sum_{j=1}^d \psi_j} },\quad {\rm FDR}  := \EEE[{\rm FDP} ],
\]
where $\cH_0 = \{j \mid  \text{degree of node } j < k_{\tau} \}$. \rev{Given the data $\bX_1, \cdots \bX_n$ from the graphical model}, we aim to propose a multiple testing procedure such that the FDP or FDR can be controlled at a given level $0 < q < 1$.

We illustrate the above general setup in two specific examples. In multitask regression with linear models, we are working with the bipartite graph case, then the weight matrix $\bTheta$ corresponds to the parameter matrix whose row represents the linear coefficients for one given response variable. Given a threshold $k_{\tau}$, we want to select those rows (response variables) with $\ell_0$ norm being at least $k_{\tau}$. In the context of Gaussian graphical models where $\cV_1 = \cV_2$, $\bTheta$ represents the precision matrix, and we want to select those hub nodes i.e., whose degree is larger than or equal to $k_{\tau}$.

\subsection{StarTrek filter}\label{sec:startrek}
Letting $\bTheta_{j}$ be the $j$-th row of $\bTheta$ and $\bTheta_{j,-j}$ be the vector $\bTheta_{j}$ excluding its $j$-th element, we formulate the testing problem for each single node as below,
% The above problem is equivalent to the following
\[
H_{0j}: \jdeg  < k_{\tau} \text{ v.s. } H_{1j}:  \jdeg \ge k_{\tau}.
\]
To test the above hypothesis, we need some estimator of the weight matrix $\bTheta$. 
%or some matrix which has the same $\ell_{0}$ entrywise norm as $\bTheta$.
In Gaussian graphical model, it is natural to use the estimator of a precision matrix. In the bipartite graph (multiple response model), estimated parameter matrix will suffice. Denote this generic estimator by $\tTheta$ (without causing confusion in notation), the maximum test statistic over a given subset $E$ of $\cV \times \cV$ will be
\begin{equation}\label{eq:T_E}
  T_{E}:= \max_{(j,k)\in E}\sqrt{n}\abr{\tTheta_{jk}  }  
\end{equation}
%\lzmargin{}{We should use $|\cdot|$ in the definition of $T_E$}
and its quantile is defined as
$
{c} (\alpha,E) = \inf \left\{ t\in \RR \; | \; \PPP \left( T_E \le t  \right) \ge 1-\alpha    \right\}
$, which is often unknown. Assume it can be estimated by $\hat{c} (\alpha,E)$ from some procedure such as the Gaussian multiplier bootstrap, a generic method called skip-down procedure can be used, which was originally proposed in \cite{lu2017adaptive} for testing a family of monotone graph invariants. When applied to the specific degree testing problem, it leads to the Algorithm \ref{algo:skipdown}. 
% we can come up with the following test for single node hypothesis with nominal level $\alpha$.
%\begin{aligned}
% \cU(M,s, r_0) &= \Big\{\bTheta \in \RR^{d \times d} \,\big|\,  \lambda_{\min}(\bTheta) \ge 1/r_0, \lambda_{\max}(\bTheta) \le r_0,  \max_{j \in [d]} \|\bTheta_{j}\|_{0} \le s, \|\bTheta \|_1 \le M \Big\}.
%\end{aligned}

\begin{algorithm}[htp]
\caption{Skip-down Method in \cite{lu2017adaptive} (for testing the degree of node $j$)}
\begin{algorithmic}\label{algo:skipdown}
\STATE \textbf{Input:} $\{ \tTheta_e \}_{e \in \cV \times \cV}$, significance level $\alpha$.
%, and approximate quantiles $\hat c(\alpha, E)$ for any set $E \subset \{j\}\times \cV$.
%，estimated quantile of the maximum statistic $\hat c(\alpha, E)$ for some subset $E$.
\STATE Initialize $t = 0, E_0 = \{(j,k) : k \in [d], k \neq j\}$.
\REPEAT
\STATE $t \gets t+1$;
 \STATE  Select the rejected edges $\cR \gets  \{(j,k) \in E_{t-1} \mid  \sqrt{n}| \tTheta_{jk} | > \hat{c} (\alpha,E_{t-1}) \}$;
 \STATE $E_t \gets E_{t-1}  \backslash \cR$;
 \UNTIL{$|E_t^{c}| \ge k_\tau$ or $\cR = \emptyset$}
%  \STATE Compute the percentile $\hat p_j$ of  $\max_{k \in E^c_{t-1}}  \sqrt{n}|\hat \bTheta_{jk}^{\text{d}}|$, i.e., find $\hat p_j$ such that
%  \[
%   \max_{k \in E_{t-1}}  \sqrt{n}|\hat \bTheta_{jk}^{\text{d}}| = \hat c (\hat p_j, E_{t-1}).
%   \]
\STATE \textbf{Output:} $\psi_{j,\alpha} = 1$ if $|E_t^c| \ge k$ and $\psi_{j,\alpha} = 0$ otherwise.
\end{algorithmic}
\end{algorithm}

To conduct the node selection over the whole graph, we need to determine an appropriate threshold $\hat{\alpha}$ then reject $H_{0j}$ if $\psi_{j,\hat{\alpha}} =1$. A desirable choice of $\hat{\alpha}$ should be able to discover as many as hub nodes with the FDR remaining controlled under the nominal level $q$. For example, if the \rev{BHq procedure} \cite{benjamini1995controlling} is considered, $\hat{\alpha}$ can be defined as follows:
\begin{equation}\label{eq:BHq_alpha}
\hat{\alpha} = \sup \left\{ \alpha\in (0,1) :  \frac{ \alpha d }{\maxof{1}{\sum_{j\in [d]}{ \psi_{j,\alpha}}} } \le q  \right\}.
\end{equation}
The above range of $\alpha$ is $(0,1)$, it will be very computationally expensive if we do an exhaustive search since for each $\alpha$, we have to recompute the quantiles $\hat{c} (\alpha,E)$ for a lot of sets $E$. 

We overcome the computational difficulty and propose a efficient procedure called 
% Thanks to the nice formulation of the single node test in Algorithm \ref{algo:stepdown}, our 
StarTrek filter, which is presented in Algorithm \ref{algo:startrek}.
% which proceeds in a quite computationally efficient manner, and is presented in Algorithm \ref{algo:startrek}.
% \lzmargin{list the skip down algorithm and say we have to tune $\alpha$, mention it is computationally slow
% }{combine algorithm 1 and algorithm 2}
% \begin{algorithm}[htp]
% \caption{Single node testing}
% \begin{algorithmic}\label{algo:stepdown}
% \STATE \textbf{Input:} $\{\hat{\bTheta}\}_{e \in \cV \times \cV}$, confidence level $\alpha$, $\psi_{j,{\alpha}} =0$ for node $j$.
% \STATE Consider a permutation $\pi$ such that $\abr{\hat{\bTheta}_{j, \pi(k)}}$ is monotonely decreasing, and compute $\hat{c}(\alpha, E)$ where $E=\{\pi(k): k_{\tau} \le k \le d\}$.
% \STATE \textbf{Output:} $\psi_{j,{\alpha}} =1$ if $|\hat \bTheta_{j, \pi(k_{\tau})}| \ge \hat{c}(\alpha, E)$.
% \end{algorithmic}
% \end{algorithm}
\begin{algorithm}[htp]
\caption{StarTrek Filter}
\begin{algorithmic}\label{algo:startrek}
\STATE \textbf{Input:} $\{\tTheta_e \}_{e \in \cV \times \cV}$, nominal FDR level $q$.
%estimated quantile of the maximum statistic $\hat c(\alpha, E)$ for some subset $E$.
% \FOR {\texttt{<some condition>}}
%         \State \texttt{<do stuff>}
% \ENDFOR
\FOR {$j \in [d]$}
% \STATE Consider a map $\pi_{j}:[d-1] \rightarrow \{k\in [d]:k\ne j\}$ such that 
% % $\{|\hat \bTheta_{j, \pi(k)}|\}_{k \in [d-1]}$ 
% \[
% |\hat \bTheta_{j, \pi(1)}| \ge |\hat \bTheta_{j, \pi(2)}| \ge \ldots \ge |\hat \bTheta_{j, \pi(d-1)}|
% \]

% % is monotonely decreasing, 
% and compute $\alpha_j = \hat{c}^{-1}(|\hat\bTheta_{j, \pi(k_{\tau})}|,E_j)$ where $E_j=\{\pi(k): k_{\tau} \le k \le d-1\}$.
\STATE We order the elements in $\{|\tTheta_{j\ell}|: \ell \neq j\}$ as
% $\{|\hat \bTheta_{j, \pi(k)}|\}_{k \in [d-1]}$ 
$
|\tTheta_{j, (1)}| \ge |\tTheta_{j, (2)}| \ge \ldots \ge |\tTheta_{j, (d-1)}|,
$ where $|\tTheta_{j, (\ell)}|$ is the $\ell$th largest entry.
 Compute $\alpha_j = \max_{1 \le s \le  k_{\tau}} \hat{c}^{-1}(\sqrt{n}|\tTheta_{j, (s)}|,E^{(s)}_j)$ where $E^{(s)}_j:=\{ (j,\ell):\ell \neq j, |\tTheta_{j\ell}| \le |\tTheta_{j, (s)}|\}$.
\ENDFOR
\STATE Order $\alpha_j$ as $\alpha_{(1)} \le \alpha_{(2)} \le \dots \le \alpha_{(d)} $ and set $\alpha_{(0)}=0$, let $j_{\max} = \max\{0\le j\le d:\alpha_{(j)} \le {qj}/{d}\}$.
% {Find a permutation $\pi:[d]\rightarrow [d]$ such that $\alpha_{\pi(j)}$ is monotone increasing, denote $\alpha_{\pi(0)} = 0 $ and let $j_{\max} = \max\{0\le j\le d:\alpha_{\pi(j)} \le {qj}/{d}\}$.}{use () instead of $\pi$}
\STATE \textbf{Output: $S =\{j: \alpha_j \le \alpha_{(j_{\max})}\}$} if $j_{\max}>0$; $S =\emptyset$ otherwise.
\end{algorithmic}
\end{algorithm}
%\lzmargin{remark it is a t-stat}{}
Remark it only involves estimating $k_{\tau}$ different quantiles of some maximum statistics per node, which is more efficient than the Skip-down procedure \cite{lu2017adaptive} in terms of computation. 
% \rev{Algorithms 1 and 2 can be shown to be equivalent, see the proof in Section A in the supplementary material. Intuitively, 
% they are equivalent as the two algorithms just change the order of choosing the confidence level and conducting variable selection. Algorithm 1 first conducts variable selection at each confidence level and then choose the best confidence level. In comparison, Algorithm 2 changes the order by first assigning p-values to each variable and then conducting the variable selection. We show the numerical comparison between two algorithms in Section \ref{sec:speed_simul}. }
\rev{We shall note that Algorithm \ref{algo:startrek} is equivalent to running the BHq procedure with Algorithm \ref{algo:skipdown}: rejecting $H_{0j}$ if $\psi_{j,\hat{\alpha}} =1$, $j \in [d]$, where $\hat{\alpha}$ is defined by \eqref{eq:BHq_alpha} and the test $\psi_{j,\alpha}$ is defined by Algorithm \ref{algo:skipdown}; see the proof at the beginning of Appendix \ref{app:pf:fdr} in the supplementary material. Without causing confusion, we will refer to the BHq procedure with Algorithm \ref{algo:skipdown} simply by Algorithm \ref{algo:skipdown}. Intuitively, Algorithm \ref{algo:startrek} directly acts the hub node selection by conducting the BHq adjustment to the p-values of all $d$ nodes but Algorithm \ref{algo:skipdown} has to specify a significance level first to test the degree for each node then search for an appropriate significance level for all nodes. We conduct numerical comparison for the two methods in Section \ref{sec:speed_simul}.}

\section{Cram\'{e}r-type comparison bounds for Gaussian maxima}\label{sec:cramer_theory}
%%%%%%%%%%%%%%%%%%%%%%%%%%%%%%%%%%%%%%%%%%%%%%%%%%%%%%%%%%%%%%%%%%%%%%%%%%%%%%%%%%%%%%%%%%%%%%%%%%%%%%%%%%%%%%%%%%%%%%%%%%%%%%%%%%%%%%%%%%%%%%%%%%%%%%%%%%%%%%%%%%%%%%%%%%%%%%%%%%%%%%%%%%%%%%%%%%%%%%%%%%%%%%%%%% cramer type approximation notations
%%%%%%%%%%%%%%%%%%%%%%%%%%%%%%%%%%%%%%%%%%%%%%%%%%%%%%%%%%%%%%%%%%%%%%%%%%%%%%%%%%%%%%%%
In this section, we present the theoretic results on the Cram\'{e}r-type comparison bounds for Gaussian maxima.
Let $U,V\in \RR^{d}$ be two centered Gaussian random vectors with different covariance matrices $\bSigma^U=(\sigma_{jk}^U)_{1\le j,k\le d},\bSigma^V=(\sigma_{jk}^V)_{1\le j,k\le d}$.
% (random vectors with one replicate will not be bolded.)
%Let $\bV_1,\cdots,\bV_n \in \RR^{p}$ denote another centered i.i.d. Gaussian random vectors with covariance matrix $\bSigma^V$. $\bSigma^Z$ and $\bSigma^V$ are in general different, w
Recall that the maximal difference of the covariance matrices is $\maxdiff:= ||\bSigma^U-\bSigma^V||_{\max}$ and the elementwise $\ell_0$ norm difference of the covariance matrices is denoted by $\zerodiff:= \nbr{\bSigma^U-\bSigma^V}_{0} = \sum_{j,k\in [d]}\Indrbr{\sigma^U_{jk}\ne \sigma^{\rev{V}}_{jk}}$. %, where $\nbr{\cdot}_{0}$ is the elementwise $\ell_0$ norm of the matrix. 
%thus we have $\zerodiff=\sum_{j,k\in [d]}\Indrbr{\sigma^U_{jk}\ne \sigma^U_{jk}}$. 
The Gaussian maxima of $U$ and $V$ are denoted as $\maxnorm{U}$ and $\maxnorm{V}$.
%Define the maxima of plain gaussian random vectors as follows.
% \begin{equation}\label{eq:gau_maxima}
% \maxnorm{U}=\max_{1\le j \le p}|U_j|, \maxnorm{V}=\max_{1\le j \le p}|V_j|,
% \end{equation}
Now we present a Cram\'{e}r-type comparison bound (CCB) between Gaussian maxima in terms of the maximum norm difference $\maxdiff$.
\begin{theorem}[CCB with maximum norm difference]\label{thm:ccb_max}
Suppose $(\log d)^{5}\maxdiff = O(1)$, then we have
% \lzmargin{add scaling condition}{present a simple bound}
\begin{equation}\label{eq:ccb_max}
    \sup_{0\le t \le C_0\sqrt{\log d}}\left|\frac{\mathbb{P}(\maxnorm{U} > t)}{\mathbb{P}(\maxnorm{V} > t)}-1\right| = O \left( (\log d)^{5/2}\maxdiff^{1/2} \right),
    %M_1(\log d)^{3/2} A(\maxdiff)e^{M_1(\log d)^{3/2} A(\maxdiff)}
    %[A(\maxdiff)+1]e^{M_1(\log d)^{3/2} A(\maxdiff)}-1 %:= \pi(\maxdiff)
\end{equation}
for some constant $C_0>0$.
% and the constant $C_1$ only depends on
% $A(\maxdiff)=K_1\maxdiff^{1/2} \exp{(K_2\maxdiff^{1/2})}$ with $K_1$ and $K_2$ only depending on
%$\max_{j}(\bSigma^U)_{jj}$, $\min_{j}(\bSigma^V)_{jj}$, $\max_{j}(\bSigma^V)_{jj}$ 
% the median of Gaussian maxima $\maxnorm{U},\maxnorm{V}$, and 
% the variance terms $\min_{1\le j\le d}\{\sigma^U_{jj},\sigma^V_{jj}\},\max_{1\le j\le d}\{\sigma^U_{jj},\sigma^V_{jj}\}$.
% and $M_1$ being some universal constant.
% And $M_1$ is a universal constant, which does dependend
% comes from the smooth approximation of $l_{\infty}$ norm, which has nothing to do with the distributional assumptions.
\end{theorem}
\begin{remark}\label{rk:thm:ccb_max}
{\rm  We can actually prove a more general form (see Theorem \ref{thm:ccb_max_general} in the appendix) of the upper bound on the above term, without the assumption on $\maxdiff$. In fact, we bound the right hand side of \eqref{eq:ccb_max} as 
$
M_3(\log d)^{3/2} A(\maxdiff)e^{M_3(\log d)^{3/2} A(\maxdiff)},
$ 
where $A(\maxdiff)=M_1 
\log d \maxdiff^{1/2} \exp{(M_2 \log^2 d \maxdiff^{1/2})}$ with the constants $M_1, M_2$ only depending on the variance terms $\min_{1\le j\le d}\{\sigma^U_{jj},\sigma^V_{jj}\},\max_{1\le j\le d}\{\sigma^U_{jj},\sigma^V_{jj}\}$ and $M_3$ being a universal constant.}
%$\max_{j}(\bSigma^U)_{jj}$, $\min_{j}(\bSigma^V)_{jj}$, $\max_{j}(\bSigma^V)_{jj}$ 
%the median of Gaussian maxima $\maxnorm{U},\maxnorm{V}$, and the variance terms $\min_{1\le j\le d}\{\sigma^U_{jj},\sigma^V_{jj}\},\max_{1\le j\le d}\{\sigma^U_{jj},\sigma^V_{jj}\}$.
% \lzmargin{remark add the complex form}{}
%The median of the Gaussian maxima $\maxnorm{U}, \maxnorm{V}$ can be bounded as $O(\sqrt{\log d})$, and the worst rates contributed by $K_1,K_2$ are $O(\log d)$ and $O({\log^2 d})$ respectively, which have only a logarithmic dependence on the dimension $d$.
\end{remark}
\rev{When applying Theorem \ref{thm:ccb_max} to Gaussian multiplier bootstrap, $\maxdiff$ actually controls the maximum differences between the true covariance matrix and the empirical covariance matrix, where $\maxdiff = O_P(\sqrt{\log d/n})$.
The proof of the theorem can be found in Appendix \ref{app:pf:thm:ccbmax}. Compared with the proof of Kolmogorov distance results in \cite{chernozhukov2013gaussian,chernozhukov2014anti}, the key innovation in our proof of the Cram\'{e}r-type Gaussian comparison bounds is a contraction mapping inequality. In specific, denote the Slepian interpolation $W(s)= \sqrt{s}U + \sqrt{1-s}V, s\in[0,1]$ and the tail probability of maxima $Q_t(s)=\mathbb{P}(||W(s)||_{\infty}>t)$. Our proof shows that $R_t(s) = Q_t(s)/Q_t(0)-1$ has the following key inequality:
\[
 |R_{t}(s)|\le AB \int_{0}^{s} |R_{t}(\mu)| d\mu + AB\cdot s + A,
\]
where $AB$ and $A$ are only depending on $\Delta_{\infty}$. By Gr\"{o}nwall's inequality \cite{gronwall1919note}, we then derive the bound on $R_t(1)$ explicitly in terms of $A$ and $B$, which finally lead to the desired Cram\'{e}r-type comparison bound in \eqref{eq:ccb_max}.}

\rev{The above theorem is a key ingredient for deriving Cram\'{e}r-type deviation results for the Gaussian multiplier bootstrap procedure. However, in certain situations especially in the applications of graphical models, comparison bounds in terms of maximum norm difference may not be appropriate. There exist cases where the covariance matrices of two Gaussian random vectors are not uniformly closed to each other, but have lots of identical entries. Namely, $\maxdiff$ is not negligible but $\zerodiff$ is small. To this end, we develop a different version of the Cram\'{e}r-type comparison bound as below. }
% {talk about why dependence is need to be used}
%This result can be utilized to quantify the covariance between maxima statistics, which is crucial for the hub node selection problem.
\begin{theorem}
[CCB with elementwise $\ell_{0}$-norm difference]\label{thm:ccb_sparse_unitvar}
%Assume $U$ and $V$ have unit variances i.e., $\sigma^U_{jj}=\sigma^V_{jj}=1, j \in [d]$. Suppose there exists a disjoint partition of nodes $\cup_{\ell=1}^{\discon}\cC_\ell = [d]$ such that $\sigma^U_{jk}= \sigma^V_{jk} = 0$ if $j \in \cC_{\ell}$ and $k\in \cC_{\ell'}$ for some $\ell' \neq \ell$.
Assume the Gaussian random vectors $U$ and $V$ have unit variances, i.e., $\sigma^U_{jj}=\sigma^V_{jj}=1, j \in [d]$ and 
%there exists some $\sigma_0<1$ such that $|\sigma^V_{jk} |\le  \sigma_0$ for any $j\ne k$ and $ |\{(j,k): j\ne k, |\sigma^U_{jk} |>  \sigma_0 \} | \le b_0$ for some constant $b_0$.
there exists some $\sigma_0<1$ such that $|\sigma^V_{jk} |\le  \sigma_0, |\sigma^U_{jk} | \le \sigma_0$ for any $j\ne k$. 
%$\lambda_{\min}(\bSigma^U)\ge 1/b_0>0,\lambda_{\min}(\bSigma^V) \ge 1/b_0>0$ for some constant $b_0>0$. 
Suppose there exists a disjoint $\discon$-partition of nodes $\cup_{\ell=1}^{\discon}\cC_\ell = [d]$ such that  $\sigma^U_{jk}=\sigma^V_{jk}=0$ when $j \in \cC_{\ell}$ and  $k \in \cC_{\ell'}$ for some $\ell \neq \ell'$. We have
\begin{equation}\label{eq:ccb_sparse_unitvar}
  \sup_{0\le t \le C_0 \sqrt{\log d}} \left|\frac{\mathbb{P}(\maxnorm{U} > t)}{\mathbb{P}(\maxnorm{V} > t)}-1\right| = O\left(\frac{ \zerodiff \log d  }{\discon}\right),
\end{equation}
for some constant $C_0 > 0$.
%\jlmargin{}{simplify it}
% whenever $t$ satisfies $0\le t \le K\sqrt{\log d}$.
% where $t$ is chosen such that $\mathbb{P}(T_{U} > t), \mathbb{P}(T_{V} > t)\ge \frac{1}{d}$
\end{theorem}
%\jlmargin{}{How to define a joint connectivity?}
When applying the above result to our multiple degree testing problem, specifically the covariance of maximum test statistics for pairs of non-hub nodes, we can show $\zerodiff =O(1)$. In Theorem \ref{thm:ccb_sparse_unitvar}, the quantity $\discon$ represents the number of connected subgraphs shared by the coviarance matrix networks of $U$ and $V$. 
We refer to Theorem \ref{thm:ccb_sparse} in the appendix for a generalized definition of $\discon$ to strengthen the results in \eqref{eq:ccb_sparse_unitvar}.
The  $\discon$ in the denominator of the right hand side of Cram\'er-type comparison bound in \eqref{eq:ccb_sparse_unitvar} is necessary: it is possible that even if $\Delta_0$ is small, when $\discon$ is large, the Cam\'er-type  Gaussian comparison bound is not converging to zero. For example, consider Gaussian vectors with unit variances
$U=(X_1, X_2, Z, \ldots,Z) \in \RR^d$, 
$V = (Y_1, Y_2, Z, \ldots,Z)\in \RR^d$, where  $\text{corr}(X_1, X_2) = 0.9, \text{corr}(Y_1, Y_2) = 0$ and $(X_1, X_2) \independent Z$, $(Y_1, Y_2) \independent Z$. For this case, the Cam\'er-type Gaussian comparison bound
\[
\sup_{0\le t \le C_0 \sqrt{\log d}} \left|\frac{\mathbb{P}(\maxnorm{U} > t)}{\mathbb{P}(\maxnorm{V} > t)}-1\right| =  \sup_{0\le t \le C_0 \sqrt{\log d}} \left|\frac{\mathbb{P}(\max\{|X_1|,|X_2|, |Z|\} > t)}{\mathbb{P}(\max\{|Y_1|,|Y_2|,|Z| \} > t)}-1\right|
\]
is not converging to zero as $d$ goes to infinity even if the corresponding $\Delta_0$ is 1 but $\discon=2$.
%\jlmargin{}{Why it depends on independency: if U and V differs in one entry but they are very dependent.}

%\fbox{phase transition}
%Compared with Theorem \ref{thm:ccb_max}, the above theorem provides a sharper comparison bound for large $\discon$ and small $\zerodiff$. The two theorems together describe a interesting phase transition phenomenon, i.e., the dependence on $\bSigma^U-\bSigma^V$ of the Cram\'{e}r-type comparison bound exhibits a difference behavior in the regime of large $\discon$ and small $\zerodiff$ versus the regime of small $\maxdiff$.
% \lzmargin{phase transition}{contrast it with thm 3.1}
% \lzmargin{}{a sharper bound when $p$ is large}

\rev{The proof of Theorem \ref{thm:ccb_sparse_unitvar} can be found in Appendix \ref{app:pf:thm:ccb_sparse_unitvar}. Our main technical innovation is to establish a new type of anti-concentration bound for ``derivatives" of Gaussian maxima. Different from the anti-concentration inequalities in  \cite{chernozhukov2014anti} bounding the maxima of Slepian interpolation $\EEE[\Indrbr{t-\epsilon \le ||W(s)||_{\infty}\le t +\epsilon} ]$, we are able to further bound its derivatives:
\begin{equation}\label{eq:varphi_anti}
\EEE[|\partial_j \partial_k \varphi(W(s))|\cdot \Indrbr{t-\epsilon \le ||W(s)||_{\infty}\le t +\epsilon} ] \lesssim\frac{ \PP{\maxnorm{V}> t} (\log d)^2}{\epsilon \beta \discon},
\end{equation}
where $\varphi$ is the some smooth approximation of the maxima with the parameter $\beta$ measuring the level of approximation.
 The above anti-concentration bound is non-uniform and has only a logarithm dependence on the dimension $d$. It provides a relatively sharp characterization when $t$ is large and the graph is not highly connected (i.e., $\discon$ is large). }

\section{Discovering hub responses in multitask regression}
\label{sec:bipartite_selection}
The theoretical results presented in Section \ref{sec:cramer_theory} will be the cornerstone for establishing FDR control of the multiple testing problem described in Section \ref{sec:method}. As seen previously, the testing problem \eqref{eq:problem_setup} is set up in a quite general way: $\bTheta$ is a weight matrix, and we would like to select rows whose $\ell_0$ norm exceeds some threshold. 
% Note that the StarTrek procedure proposed in Algorithm \ref{algo:startrek} and the Cram\'{e}r-type comparison bounds for Gaussian maxima are also applicable to other interesting examples in addition to the Gaussian graphical models. 
This section considers the specific application to multitask/multiple response regression, which turns out to be less involved. We take advantage of it and demonstrate how to utilize the probabilistic tools in Section \ref{sec:cramer_theory}. After that, the theoretical results on FDR control for the Gaussian graphical models are presented and discussed in Section \ref{sec:hub_selection}.
% in the sense that $\bTheta$ can be some weight matrix. In the multi-task regression setting, we might be interested in 
% \lzmargin{mention bipartite, for simplicity of illustration}{
% multi-task regression; cite javamard paper}

In multitask regression problem, multiple response variables are regressed on a common set of predictors. We can view this example as a bipartite graph $\cG = (\cV_1, \cV_2, \cE), |\cV_1|=d_1, |\cV_2|=d_2$, where $\cV_1$ contains the response variables and $\cV_2$ represents the common set of predictors. Each entry of the weight matrix $\bTheta$ indicates whether a given predictor is non-null or not for a given response variable. In the case of parametric model, $\bTheta \in \RR^{d_1 \times d_2}$ corresponds to the parameter matrix. One might be interested in identifying shared sparsity patterns across different response variables. It can be solved by selecting a set of predictors being non-null for all response variables \cite{obozinski2006multi,dai2016knockoff}. This section problem is column-wise in the sense that we want to select columns of $\bTheta$, denoted by $\bTheta_{\cdot j}$, such that $\norm{\bTheta_{\cdot j}}_0 = d_1$. It is also interesting to consider a row-wise selection problem formalized in \eqref{eq:problem_setup}. Under the multitask regression setup, we would like to select response variables with at least a certain amount of non-null predictors. We will call this type of response variables hub responses throughout the section. This has practical applications in real-world problems such as the gene-disease network.

Consider the multitask regression problem with linear models, we have $n$ i.i.d. pairs of the response vector and the predictor vector, denoted by $(\bY_1,\bX_1), (\bY_2,\bX_2), \dots, (\bY_n,\bX_n)$, where
$\bY_i \in \RR^{d_1},\bX_i \in \RR^{d_2}$ satisfy the following relationship,
\begin{eqnarray}\label{eq:mul_linear_model}
\bY_i = \bTheta \bX_i + \bE_i,  \text{ where } \bE_i\sim \cN(0,\Db_{d_1\times d_1}  ) \text{ and }\bX_i \independent \bE_i,
\end{eqnarray}
where $\bTheta  \in \RR^{d_1 \times  d_2}$ is the parameter matrix and $\Db$ is a $d_1$ by $d_1$ diagonal matrix whose diagonal elements $\sigma_j^2$ is the noise variance for response variable $\bY^{(j)}$. Let $\bX$ be the design matrix with
rows $\bX_1^{\top},\dots, \bX_n^{\top}$, shared by different response variables, and assume the noise variables are independent conditional on the design matrix $\bX$.
% As for matrix form, we let $\bY = (\bY_1,\cdots,\bY_n)^{\top}$ and denote by $\bX$ the design matrix with
% rows $\bX_1^{\top},\dots, \bX_n^{\top}$, t
% hen we have $\bY \in \RR^{n \times d_1},\bX \in   \RR^{n \times d_2}$
% \begin{eqnarray}\label{eq:NoisyModel}
% \bY = \bX\bTheta+ \bW, \quad \text{where } \text{vec}({\bW})\sim
% \normal_{n\times d_1}(0,\Ib_{n\times n}\bigotimes \sigma^2 \Ib_{d_1\times d_2})\, .
% \end{eqnarray}
Let $s = \max_{j\in [d_1]}\norm{\bTheta_{j}}_0$ be the sparsity level of the parameter matrix $\bTheta$, we want to select columns of the parameter matrix which has at least $k_{\tau}$ nonzero entries, i.e., select nodes with large degree among $[d_1]$ in the bipartite graph $\cG = (\cV_1, \cV_2, \cE)$.

As mentioned in Section \ref{sec:method}, some estimator of the parameter matrix is needed to conduct hypothesis testing. Debiased Lasso is widely used for parameter estimation and statistical inference in high dimensional linear models \cite{javanmard2014confidence,javanmard2014hypothesis}. For each response variable $\bY^{(j)}, j \in [d_1]$, we compute the debiased Lasso estimator, denoted by $\tdTheta_{j}$ as 
\begin{align}\label{eq:dlasso}
\tdTheta_j = \hat \bTheta_j + \frac{1}{n}\, \Mb \bX^{\top}(\bY^{(j)} - \bX \hat  \bTheta_j), 
\text{ where  }\hat\bTheta_j= \arg\min_{\beta\in\RR^{d_2}}
\Big\{\frac{1}{2n}\|\bY^{(j)}-\bX\beta
\|^2_2+\lambda\|\beta\|_1\Big\}\, .  %\label{eq:lasso}
\end{align}
% where 
% \begin{eqnarray}
% \hat \beta^u = \hat \beta^n(\lambda) + \frac{1}{n}\, M \bX^{\top}(\bY^{(j)} - \bX \hat \beta^n(\lambda)) \label{eq:dlasso}
% \end{eqnarray}
Note the above $\Mb$ is defined as $\Mb = (m_1,\dots,m_{d_2})^{\top}$ where 
\begin{align}
\label{eq:optimization}
m_i = &\argmin_m  m^{\top} \hSigma m, \quad \text{s.t. }  \|\hSigma m - e_i \|_{\infty} \le \mu\,,
\end{align}
% \end{eqnarray}
and here $\hSigma = (\bX^{\top} \bX)/n$. 
%The above $\hat \beta^u$ can be used as the estimate of the $j$-th row of the parameter matrix, i.e., $\hat \bTheta_{j} =\hat \beta^u$. 

Then the debiased estimator of the parameter matrix, defined by $\tdTheta := (\tdTheta_{1}, \cdots, \tdTheta_{d_1})^\top$, will be used the input $\{\tTheta_e\}_{e \in \cV_1 \times \cV_2}$ of Algorithm \ref{algo:startrek}. In addition, we also need to compute the quantile of the maximum statistics. There exist many work studying the asymptotic distribution of the debiased Lasso estimator. Among them, the results in \cite{javanmard2014confidence} (when translated into our multitask regression setup) imply, for each response variable $\bY^{(j)}, j \in [d_1]$,
\begin{equation} \label{eq:dlasso_normal}
    \sqrt{n} (\tdTheta_j - \bTheta_j) = Z + \lerr, \quad Z |\bX \sim \cN(0,\sigma_j^2 M\hSigma M^{\top}),
\end{equation}
under proper assumptions. Additionally with a natural probabilistic model of the design matrix, the bias term can be showed to be $\maxnorm{\lerr} =O(\frac{s \log d_2}{\sqrt{n}}) $ with high probability. As discussed in \cite{javanmard2014confidence}, the asymptotic normality result can be used for deriving confidence intervals and statistical hypothesis tests. As the noise variance $\sigma_j$ is unknown, the scaled Lasso is used for its estimation \cite{javanmard2014confidence,sun2012scaled}, given by the following joint optimization problem,
\begin{equation}\label{eq:scaled_lasso}
    \{\hat\bTheta_j,\hat\sigma_j \} = \arg \min_{\beta \in \RR^{d_2}, \sigma > 0}\Big\{\frac{1}{2\sigma n}\|\bY^{(j)}-\bX\beta
    \|^2_2+ \frac{\sigma}{2} +
    \lambda\|\beta\|_1\Big\}.
\end{equation}
 Regarding our testing problem, intuitively we can use the quantile of the Gaussian maxima of $ \cN(0,\hat \sigma_j^2 M\hSigma M^{\top})$ to approximate the quantile of maximum statistic $T_{E} =\underset{(j,k)\in E}{\max}\sqrt{n} |\tdTheta_{jk}|$ for some given subset $E$. Specifically, let $Z_j\mid \bX, \bY^{(j)}\sim  \cN(0,\hat \sigma_j^2 M\hSigma M^{\top})$ where $Z_j \in \RR^{d_2}$ and consider the subset $E \subset \{j\}\times \cV_2$, we approximate the quantile of $T_E$ by the following
 
\begin{equation}\label{eq:dlasso_chat}
T^{\cN}_{E} := \underset{(j,k)\in E}{\max} \; | Z_{jk} |,\quad \hat{c} (\alpha,E) = \inf \left\{ t\in \RR : \PPP_{Z} \left( T^{\cN}_{E} \le t  \right) \ge 1-\alpha    \right\}.
\end{equation}
% denote the quantile estimate by $\hat c (\alpha, E)$. 
 Indeed, under proper scaling conditions, we can show that, i.e., as ${n,d \rightarrow \infty}$,
\begin{equation}\label{eq:dlasso_quantile_valid}
\sup_{\alpha \in (0,1)} 
    \left|\PPP
    \left( \max_{(j,k) \in E}  \sqrt{n} |\tdTheta_{jk} -{\bTheta}_{jk}|> \hat{c}(\alpha, E) 
    \right) - \alpha
    \right|\rightarrow 0.
\end{equation}
The above result is based on two ingredients: the asymptotic normality result and the control of the bias term $\lerr$. Below we list the required assumptions for those two ingredients, i.e., \eqref{eq:dlasso_normal} and $\maxnorm{\lerr} =O(\frac{s \log d_2}{\sqrt{n}})$. 
\begin{assumption}[Debiased Lasso with random designs]\label{asp:dlasso}
%  Below we list some conditions on the design matrix for de-baised Lasso and some additional scaling conditions for FDP control,
The following assumptions are from the ones of Theorems 7 and 8 in \cite{javanmard2014confidence}.
 \begin{itemize}
     \item Let $\bSigma = \EE{\bX_1 \bX_1^\top}\in\RR^{d_2\times d_2}$ 
    be such that $\sigma_{\min}(\bSigma) \ge C_{\min} > 0$, and
    $\sigma_{\max}(\bSigma) \le C_{\max} < \infty$,
    and $\max_{j\in [d_2]}\bSigma_{jj}\le 1$.
    Assume $\bX\bSigma^{-1/2}$ have independent subgaussian  rows, with
    zero mean and subgaussian norm $\|\bSigma^{-1/2} \bX_i\|_{\psi_2} = \kappa$, for some
    constant $ \kappa \in(0, \infty)$.
    \item  $\mu =a\sqrt{(\log d_2)/n}$, and $n\ge \max(\nu_0s\log (d_2/s), \nu_1\log d_2)$,
    $\nu_1 = \max(1600\kappa^4, a/4)$,
    and 
    $\lambda = \sigma\sqrt{(c^2\log
    d_2)/n}$.
    % \lzmargin{change p}{}
 \end{itemize}
 \end{assumption}
% Now we denote the quantile estimate of $T_{E}$ as $\hat{c} (\alpha,E) = \inf \left\{ t\in \RR : \PPP \left( \maxnorm{Z} \le t  \right) \ge 1-\alpha \right\}$, where $Z \sim \cN(0, )$, specifically we denote $\hat c (\alpha, E)$ for some given subset $E \in [d_2]$. 
% As $\sigma$ is unknown, we will derive a consistent estimator using the procedure in.

% The above assumptions are required \cite{javanmard2014confidence} for establishing asymptotic normality and controlling the bias term $\Delta$ at a certain rate. 
Remark that there may exist other ways of obtaining a consistent estimator of $\bTheta$ and sufficiently accurate quantile estimates under different assumptions. Since it is not the main focus of this paper, we will not elaborate on it. As mentioned before, the Kolmogorov type result in \eqref{eq:dlasso_quantile_valid} can be immediately applied to the global testing problem to guarantee FWER control. {However, it is not sufficient for FDR control of the multiple testing problem in this paper. And this is when the Cram\'{e}r-type comparison bound for Gaussian maxima established in Section \ref{sec:cramer_theory} play its role.} In addition, signal strength condition is needed.
% In the following, we will illustrate how the general techniques are exploited in our multiple testing problem for the multi-task regression model. 
% Before presenting our key FDR control result in Theorem \ref{thm:fdr_linear}, 
% {Add however here. We need stronger signal strength condition in FDR.} 
%{move however to after the Assumption 4.1} 
%  \lzmargin{only include the fisrt two, mention this assumption is needed for debiased lasso, use for quantile}{}
% \lzmargin{present the asymptotic normality result in debiased lasso, and talk about how to compute the quantile }{}
% \lzmargin{replace $A,B$ by $\theta$}{}
Recall that $\cH_0 = \{j\in [d_1]: \norm{\bTheta_{j}}_0 < k_\tau\}$ with $d_0 = |\cH_0|$, we consider the following rows of $\bTheta$,
\begin{equation}\label{eq:strong_Y_set}
    \cB:=\{j\in \cH_{0}^{c}:\forall k \in \text{supp}(\bTheta_j),|\bTheta_{jk}|>c\sqrt{{\log d_2}/{n}}\},
\end{equation}
and define the proportion of such rows as $\rho = |\cB|/d_1$.
In the context of multitask regression, $\rho$ measures the proportion of hub response variables whose non-null parameter coefficients all exceed certain thresholds, thus characterizes the overall signal strength. \rev{As mentioned at the beginning, the application of the StarTrek filter in this section is less involved than that in Gaussian graphical models. The major simplification comes from how the multitask regression problem with linear models gets set up: the response vector $\bY_i \in \RR^{d_1}$ follows $\bY_i = \bTheta \bX_i + \bE_i$, where
$\bE_i\sim \cN(0,\Db_{d_1\times d_1}  ) \text{ and }\bX_i \independent \bE_i$; thus those $d_1$ response variables are independent, conditional on the covariate $\bX_i$. Such conditional independence carries over to the statistics of testing each response (node) in our proposed method. In Gaussian graphical models, the test statistics for each node unavoidably have complicated dependence structures, due to the nature of this combinatorial selection problem and how the nodes are connected to each other. To this end, we will introduce some quantity (see \eqref{eq:dep_term_set}) to measure the dependence level of the graph and characterize how such dependence affects the validity of our StarTrek filter (see \eqref{eq:cond2} in Assumption \ref{asp:tradeoff_fdp}).} Below we present our result on FDP/FDR control under appropriate assumptions.
%\jlmargin{}{specify or standardize $\tilde \Theta^d$}
%$$
%    \frac{\log d_2}{\epsilon^2}\rbr{
%    \frac{ 1}{d_0 \rho}+ \frac{s(\log d_2)^{2} }{n^{1/2}} + 
%     \frac{(\log d_2)^2}{(n \rho)^{1/5}} + 
%    + \frac{1}{\rho d_2}}
%$$ 
\begin{theorem}[FDP/FDR control]\label{thm:fdr_linear}
Under Assumption \ref{asp:dlasso} and the scaling condition  
    $
%    \frac{ \log d_2 }{\min{(d_0, d_2)} \rho}
\frac{d_2\log d_2 + d_0}{d_0 d_2 \rho} +   \frac{s \log^2 d_2 }{n^{1/2}} + 
     \frac{\log^2 d_2}{(n \rho)^{1/5}} 
%+ \frac{1}{n^{1/5}s \log d_2 \rho}     
%    + \frac{s^{1/4}(\log d_2)^{15/4}}{n^{1/5}}
     = o(1)$, 
%    for any $\delta>0$,
% Under similar scaling condition on $(n,d)$ as in Proposition \ref{prop:approx_cramer_gmb} $(\log ed)^3 (\log (ed+n))^{56/3}/n + s   (\log d)^{5/2}/\sqrt{n}= o(1)$ and the condition listed below
%, where $b_{0}$ is some constant which does not depend on $n$ or $d$, 
if we implement the StarTrek procedure in Algorithm \ref{algo:startrek} with $\bTheta$ estimated by \eqref{eq:dlasso} and the quantiles approximated by \eqref{eq:dlasso_chat},  as $(n,d_1, d_2)\rightarrow \infty$, we have
\begin{equation}\label{eq:fdp_linear}
    %\PP{ \mathrm{FDP} \le q\frac{d_{0}}{d_1}} ~~~ \text{in probability}, 
     \mathrm{FDP} \le q\frac{d_0}{d_1} + \smallop \quad \text{and}
    \quad   \lim_{(n,d_1, d_2)\rightarrow \infty}\mathrm{FDR} \le q\frac{d_{0}}{d_1}.
\end{equation}
% Under a weaker condition, we have the FDR control result
% \begin{equation}\label{eq:fdp_linear}
    %  \lim_{(n,d)\rightarrow \infty}\mathrm{FDR} \le q\frac{d_{0}}{d_1}.
% \end{equation}
\end{theorem}
The proof of Theorem \ref{thm:fdr_linear} can be found in Appendix \ref{pf:thm:fdr_linear}.
%\lzmargin{remove $\delta$}{}
Note that signal strength conditions which require some entries of parameter matrix $\bTheta$ have magnitudes exceeding $c\sqrt{\log d_2/n}$ are usually assumed in existing work studying FDR control problem for high dimensional models \cite{liu2013ggmfdr,liu2014phase,liu2014hypothesis,xia2015testing,xia2018multiple,javanmard2019false}. 

\section{Discovering hub nodes in Gaussian graphical models}\label{sec:hub_selection}
% \lzmargin{recall 2.2. 2.3; debiased step and quantile estiamtion and startrek filter for node selection }{}
This section focuses on the hub node selection problem on Gaussian graphical models \rev{where $\bX_1,\ldots, \bX_n  \stackrel{\mathrm{i.i.d.}}{\sim} N_d (0,\bSigma)$. Let the weight matrix be the precision matrix $\bTheta = \bSigma^{-1}$. Given some estimator $\hat{\bTheta}$ (e.g., the graphical Lasso (GLasso) estimator \cite{friedman2008sparse} or the CLIME estimator \cite{cai2011constrained}),
% which will have the same $\ell_{0}$ elementwise norm as the adjacency matrix $\bTheta$. 
we consider the following one-step estimator $\{\dTheta_{e}\}_{e\in \cV \times \cV}$:
\begin{equation}
\label{eq:ggm_tdTheta}
\dTheta_{jk} := \hat{\bTheta}_{jk} - \frac{\hat{\bTheta}_{j}^{\top} \left( \hat{\bSigma} \hat{\bTheta}_k - \eb_k\right)}{\hat{\bTheta}_j^{\top} \hat{\bSigma}_j},
\quad
\tdTheta_{jk}:={\dTheta_{jk}}/{\sqrt{\dTheta_{jj}\dTheta_{kk}}}.
\end{equation}
where $\eb_k$ denotes the $k$th canonical basis in $\RR^d$.
then use its standardized version $\{\tdTheta_{e}\}_{e\in \cV \times \cV}$ as the input of Algorithm \ref{algo:startrek}.
}
% Recall in Section \ref{sec:method}, we first compute the one-step estimator $\{\dTheta_{e}\}_{e\in \cV \times \cV}$ in \eqref{eq:one_step} then take its standardized version $\{\tdTheta_{e}\}_{e\in \cV \times \cV}$ as the input of Algorithm \ref{algo:startrek} i.e.,
% \begin{equation}
% \label{eq:ggm_tdTheta}
% \dTheta_{jk} := \hat{\bTheta}_{jk} - \frac{\hat{\bTheta}_{j}^{\top} \left( \hat{\bSigma} \hat{\bTheta}_k - \eb_k\right)}{\hat{\bTheta}_j^{\top} \hat{\bSigma}_j},
% \quad
% \tdTheta_{jk}:={\dTheta_{jk}}/{\sqrt{\dTheta_{jj}\dTheta_{kk}}}.
% \end{equation}
Our StarTrek filter selects nodes with large degrees based on the maximum statistics $T_{E}= \underset{(j,k)\in E}{\max} \; \sqrt{n} | \tdTheta_{jk}|$ over certain subset $E$. The quantiles are approximated using the Gaussian multiplier bootstrap \cite{chernozhukov2013gaussian}: given the Gaussian multipliers $\xi_i \stackrel{\mathrm{i.i.d.}}{\sim} N(0,1)$, we compute
\begin{equation}\label{eq:ggm_chat}
\hat{c} (\alpha,E) = \inf \left\{ t\in \RR : \PPP_\xi \left( T^{\cB}_{E} \le t  \right) \ge 1-\alpha    \right\},
\end{equation}
where $T^{\cB}_{E} := \underset{(j,k)\in E}{\max} \; \frac{1}{\sqrt{n~ {\hat{\bTheta}_{jj}\hat{\bTheta}_{kk}}}}   \big|\sum_{i=1}^n  {\hat{\bTheta}^{\top}_j \left( \bX_i \bX_i^{\top} \hat{\bTheta}_k - \eb_k  \right)} \xi_i \big|$.
% \begin{equation}\label{eq:gmb_quantile}
% T^{\cB}_{E} := \underset{(j,k)\in E}{\max} \; \frac{1}{\sqrt{n~ {\hat{\bTheta}_{jj}\hat{\bTheta}_{kk}}}}   \bigg|\sum_{i=1}^n  {\hat{\bTheta}^{\top}_j \left( \bX_i \bX_i^{\top} \hat{\bTheta}_k - \eb_k  \right)} \xi_i \bigg|,
% \end{equation}
% and obtain
% \begin{equation}\label{eq:ggm_chat}
% \hat{c} (\alpha,E) = \inf \left\{ t\in \RR : \PPP_\xi \left( T^{\cB}_{E} \le t  \right) \ge 1-\alpha    \right\}.
% \end{equation}

 \cite{chernozhukov2013gaussian} shows that the above quantile approximation is accurate enough for FWER control in modern high dimensional simultaneous testing problems. Their results are based on the control of the non-asymptotic bounds in a Kolmogorov distance sense. \cite{lu2017adaptive} also takes advantage of this result to test single hypothesis of graph properties or derive confidence bounds on graph invariants.

However, in order to conduct combinatorial variable selection with FDR control guarantees, we need more refined studies about the accuracy of the quantile approximation. This is due to the ratio nature of the definition of FDR, as explained in Section \ref{sec:startrek}. Compared with the results in \cite{chernozhukov2013gaussian}, we provide a Cram\'er-type control on the approximation errors of the Gaussian multiplier bootstrap procedure. This is built on the probabilistic tools in Section \ref{sec:cramer_theory}, in particular, the Cram\'er-type Gaussian comparison bound with max norm difference in Theorem \ref{thm:ccb_max}.
%\lzmargin{recall the method, recall the multiple testing and how to estimate the quantile}{}
% \lzmargin{add some plots for connectivity and $S$}{}
% In this section, we will study the FDR control problem for the undirected Gaussian graphical models, where $\cG = (\cV, \cE)$ with $|\cV|=d$, and the precision matrix $\bTheta$. 
%The previous section has provided a nice illustration of the usage of the probabilistic tools in Section \ref{sec:cramer_theory} and relevant techniques for multiple testing problem.
% \lzmargin{to show FDR, we need cramer type bound due to the ratio form, emphasize ratio of FDR; comparison bound with max norm is for what; l0 norm is for what; characterize dependency via the difference between dependent case and independent case via sparse pattern, this is why we use l0 norm}{}
Due to the dependence structure behind the hub selection problem in Graphical models, we also have to utilize Theorem \ref{thm:ccb_sparse_unitvar}. In a bit more detail, computing the maximum test statistic for testing node actually involves the whole graph, resulting complicated dependence among the test statistics. The non-differentiability of the maximum function makes it very difficult to track this dependence. Also note that, this type of difficulty can not be easily circumvented by alternative methods, due to the discrete nature of the combinatorial inference problem. However, we figure out that the Cram\'er-type Gaussian comparison bound with $\ell_0$ norm difference plays an important role in handling this challenge.

In general, the sparsity/density of the graph is closed related to the dependence level of multiple testing problem on graphical models. For example, \cite{liu2013ggmfdr,xia2015testing,xia2018multiple} make certain assumptions on the sparsity level and control the dependence of test statistics when testing multiple hypotheses on graphical models/networks. For the hub node selection problem in this paper, a new quantity is introduced, and we will explain why it is suitable. % \lzmargin{due to the dependence structure we have to characterize the sparisity/density, i.e., S, how node tests interact }{cite Weidong Liu Ying Xia}
%  Now we will focus on a more complicated model since there exist complicated dependence structures between test statistic of different nodes. And the discrete nature of the combinatorial inference problem makes it even more challenging. Following the same notations, l
Recall that we define the set of non-hub response variables in Section \ref{sec:bipartite_selection}. Similarly, the set of non-hub nodes is denoted by $\cH_0 = \{j\in [d]: \norm{\bTheta_{j}}_0 < k_\tau\}$ with $d_{0}=|\cH_{0}|$. Now we consider the following set, 
% \begin{equation}\label{eq:strong_hub_set}
%     \cB:=\{j\in \cH_{0}^{c}:\forall k \in \text{supp}(\bTheta_j),|\bTheta_{jk}|>c\sqrt{{\log d}/{n}}\}.
% \end{equation}
% Signal strength condition is nearly necessary for controlling false discovery proportion, the signal strength threshold $\sqrt{\log d/n}$ is of the same rate as previous work on FDP control for edge testing problem \cite{liu2013ggmfdr}.
% \lzmargin{change the condition on $\cB$ in a ratio form}{connection with Weiding Liu's work}
% where $\cH_{0}^{c}$ is the set of ``big hubs'' and $E_{j}^{\ast}=\{k:\bTheta_{jk}\neq 0\}$ for $j\in \cH_{0}^{c}$. 
\begin{equation}\label{eq:dep_term_set}
%\nonumber
S = \{(j_1,j_2,k_1,k_2): j_1,j_2\in \cH_0, j_1\ne j_2,k_1 \ne k_2 , \bTheta_{j_1 j_2} = \bTheta_{j_1 k_1} = \bTheta_{j_2 k_2} = 0, \bTheta_{j_1 k_2} \ne 0, \bTheta_{j_2 k_1}\ne 0\}.
\end{equation}
%\begin{figure}[htbp]
%    \centering
%    \includegraphics[width = 0.3\linewidth]{S_demo.pdf}
%    \caption{Caption}
%    \label{fig:my_label}
%\end{figure}
% As mentioned before, in our multiple testing problem, each test of a single node will involve computing a statistic based on the whole graph, hence results in complicated dependence structures between test statistic/p-values of different nodes.
% \begin{figure}[htbp]
%     \centering
%     \includegraphics[width = 0.65\linewidth]{plots/demo1.pdf}
%     \caption{A graphical demonstration of the definition of $S$ via some examples of a $4$-vertex graph}
%     \label{fig:demo1}
% \end{figure}
%% 10 15 24 51
\begin{figure}[htbp]
    \centering
    \includegraphics[width = 0.75\linewidth]{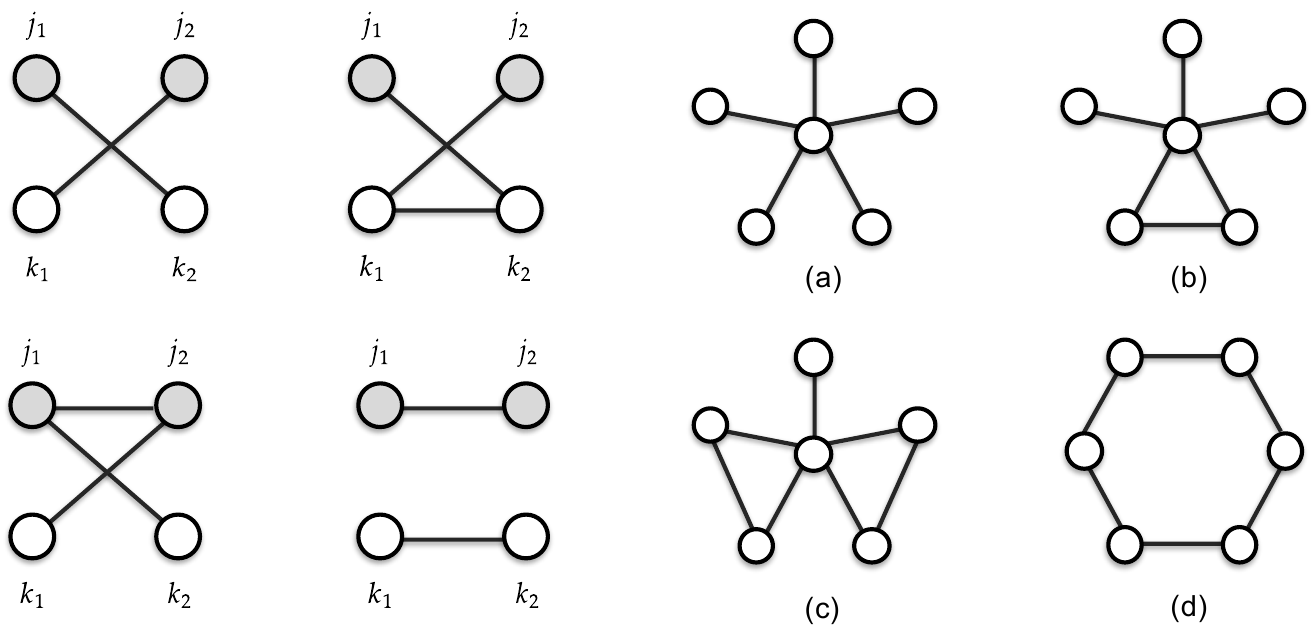}
    \caption{Left panel: a graphical demonstration of the definition of $S$ via four examples of a $4$-vertex graph; Right panel: four different graph patterns with $6$ vertices. Calculating $|S|$ yields $10,15,24,51$ for (a),(b),(c),(d) respectively.}
    \label{fig:demo}
\end{figure}
Remark that in the above definition, $k_1$ can be the same as $j_2$ and $k_2$ can be the same as $j_1$. 
% Then $(j_1,j_2,k_1,k_2)$ form a graph with two or three vertices. As mentioned before, the dependence among test statistic (of non-hub nodes) will be disturbing factors for detecting true hub nodes. In the above definition, each time we look at a pair of two non-hub nodes $j_1,j_2$ which are not connected by an edge. Due to the upper bound on the degree, they only have a limited number of neighbours. Therefore, the number of pairs of $k_1,k_2$ such that $\bTheta_{j_1k_1}\bTheta_{j_1k_2}\ne 0$ or $\bTheta_{j_2k_1}\bTheta_{j_2k_2}\ne 0 $ is also small, To this end, we look at $(k_1k_2)$ such that $\bTheta_{j_1k_1} = \bTheta_{j_2k_2}=0 $ but $\bTheta_{j_1k_2}, \bTheta_{j_1k_2}\ne0 $. Based on this intuition, we look at all pairs of non-hub nodes and examining the edges in subgraph with nodes $(j_1,j_2,k_1,k_2)$. 
% We define the set $S$ as \eqref{eq:dep_term_set} and use $|S|$ as our measure of the dependence level of the graph. 
If there exists a large number of nodes which are neither connected to $j_1$ nor $j_2$, we then do not need to worry much about the dependence between the test statistics for non-hub nodes. Therefore, $|S|$ actually measures the dependence level via checking how a pair of non-hub nodes interact through other nodes. \cite{liu2013ggmfdr,cai2013two} also examine the connection structures in the $4$-vertex graph and control the dependence level by carefully bounding the number of the $4$-vertex graphs with different numbers of edges. 
% \begin{figure}[!htbp]
%     \centering
%     \includegraphics[width = 0.55\linewidth]{plots/demo2.pdf}
%     \caption{A graphical demonstration of the definition of $S$ via a graph with $2$ hub nodes out of $6$ nodes.}
%     \label{fig:demo2}
% \end{figure}

We provide a graphical demonstration of $S$ and show how $|S|$ looks like in certain types of graph patterns via some simple examples. 
%For the given non-hub nodes $j_1,j_2$, we denote $N_{0j_1} = \{k_1: \bTheta_{j_1k_1}=0\}$, $N_{0j_2} = \{k_2: \bTheta_{j_2k_2}=0\}$. 
Though the definition of $S$ does not exclude the possibility of $(j_1,j_2,k_1,k_2)$ being a graph with $2$ or $3$ vertices, we only draw $4$-vertex graph in Figure \ref{fig:demo} for convenience. In the left panel of Figure \ref{fig:demo}, we consider four different cases of the $4$-vertex graph. The upper two belong to the set $S$, while the lower \rev{two} do not.
% Look at the graph with 6 vertices in Figure \ref{fig:demo2}, if we specify $k_\tau $ to be $3$, then there are two hubs nodes, one with $4$ edges, the other with $4$ edges. The four nodes in the rectangle with red dash border belongs to $S$ if we set the lower two nodes as $j_1,j_2$ and the upper two nodes as $k_2,k_1$. Regarding the four nodes in the rectangle with blue dash border, there are only two non-hub nodes. Both of them are connected to the other hub-node, thus the $4$-vertex graph does not belongs to $S$.
In the right panel, we consider four graphs which all have $6$ vertices. They have different graph patterns. For example, (a) clearly has a hub structure. All of the non-hub nodes are only connected to the hub node. While in (d), the edges are evenly distributed and each node are connected to its two nearest neighbours. For each graph, we count the value of $|S|$ and obtain $10,15,24,51$ respectively, which show a increasing trend of $|S|$. This sort of matches our intuition that it is relatively easier to discover hub nodes on graph (a) compared with graph (d). See more evidence in the empirical results of Section \ref{sec:simul}.
% \lzmargin{add some plots of for demonstrating S }{}
%We actually ses $|S|$ is directly related to the term $\zerodiff$ when applying the Cram\'er-type comparison bound in Theorem \ref{thm:ccb_sparse_unitvar} to prove FDR control.
% \lzmargin{
% how to explain S, connect to Ying Xia's dependence assumption;
% related to the dependency of test;
% how $j_1$ and $j_2$ interact through other nodes;
% S is related to $\Delta_0$; add some plots}{ss}
% ZZ
% \lzmargin{delete assmp 5.2; under general graphical model assumption and mild scaling condition -> FDR;
% assmp 5.1 -> FDP
% }{}

In addition to $|S|$, we also characterize the dependence level via the connectivity of the graph, specifically let $p$ be the number of connected components. And similarly as in Section \ref{sec:bipartite_selection}, we define $\rho$ to measure the signal strength, i.e., $\rho = |\cB|/d$, where $\cB:=\{j\in \cH_{0}^{c}:\forall k \in \text{supp}(\bTheta_j),|\bTheta_{jk}|>c\sqrt{{\log d}/{n}}\}$. In the following, we list our assumptions needed for FDR control.
\begin{assumption}\label{asp:tradeoff_fdp}
{Suppose that $\bTheta \in \cU(M,s, r_0)$ and the following conditions hold:}
%Suppose that $\bTheta \in \cU(M,s, r_0)$ and $|\cB|>0$, and the following two trade-off conditions are satisfied
% Below we list some tail conditions on $\bX$ and some trade-off conditions between the connectivity, dependence and signal strength of the graphical model associated with the Gaussian vector $\bX$.
%Suppose that $\bTheta \in \cU(M,s, r_0)$, and the following conditions hold
\begin{enumerate}[(i)]
    % \lzmargin{\item Suppose there exists a constant $0\le K_d < \infty$ such that $\max_{1\le j\le d}|\bX_j|_{\psi_1} \le K_d$.}{delete it}
    % \item 
    % Signal strength \& scaling \& edge sparsity:
    \item Signal strength and scaling condition.
    \begin{eqnarray}\label{eq:cond1}
    \frac{\log d }{\rho}\rbr{
    \frac{(\log d)^{19/6}}{n^{1/6}} + \frac{ (\log d)^{11/6}}{\rho^{1/3} n^{1/6}} + 
    % \frac{(\log d)^{{13}/{8}}}{n^{1/8}} + 
    \frac{{s(\log d)^{3}}}{{n}^{1/2}}
    } = o(1).
    \end{eqnarray}
    % \lzmargin{do not use sqrt}{}
    \item Dependency and connectivity condition.
%    \& connectivity \& star sparsity: 
    % for any $\delta>0$.
    \begin{eqnarray}\label{eq:cond2}
   \frac{\log d}{\rho d_0} + \frac{({\log d})^{2}|S|}{\rho d_0^2 p} = o(1).
    \end{eqnarray}
\end{enumerate}
\end{assumption}
In the above assumption, \eqref{eq:cond1} places conditions on the signal strength and scaling. The first and the second term come from the Cram\'er-type large deviation bounds in the high dimensional CLT setting \cite{arun2018cram} and the Cram\'er-type Gaussian comparison bound established in Theorem \ref{thm:ccb_max}. And the third term comes from the fact that the relevant test statistics arise as maxima of approximate averages instead of the exact averages and thus the approximation error needs to be controlled. See similar discussions about this in \cite{chernozhukov2013gaussian}. Remark that the signal strength condition is mild here, due to similar reasons as the discussion in Section \ref{sec:bipartite_selection}. Regarding \eqref{eq:cond2}, there is a trade-off between the dependence level and connectivity level of the topological structure. $|S|/d_0^2$ characterizes how the test statistics of non-hub nodes are correlated to each other in average. $p$ by definition describes the level of connectivity. Due to the condition \eqref{eq:cond2}, larger signal strength generally makes the hub selection problem easier. And when $|S|/d_0^2$ is small, the graph is allowed to be more connected. When there exist more sub-graphs, we allow higher correlations between the non-hub nodes. Note that the cardinality of $S$ is directly related to the $\ell_0$ norm covariance matrix difference term $\Delta_0$, and arises from the application of Theorem \ref{thm:ccb_sparse_unitvar}. 
%\lzmargin{scaling condition on the topological structures}{}
%\lzmargin{comment on assumption (i,
%explain the two terms in (5.2), cite upenn group
%comment on how the rates come from
%mention cramer type comparison bound, mention $\zerodiff$ and how it relates to $|S|$, mention $p$}{}
%\lzmargin{tradeoff between $|S|/d0^2$ and $p$; $|S|/d0^2$ characterizes in average how correlated (test statistics of null nodes); $\rho$ is the signal strength condition, larger $\rho$ makes the problem easier; when $|S|/d0^2$ is small, then the graph can be more connected.}{}
% \begin{assumption}[Trade-off assumption for FDR control]
% \label{asp:tradeoff_fdr}
% \lzmargin{scaling condition}{}
% \end{assumption}
In the following, we present our core theoretical result on FDP/FDR control for hub selection using the StarTrek filter on Gaussian graphical models.
%\jlmargin{}{specify what is $\tilde \Theta^d$}
\begin{theorem}[FDP/FDR control]\label{thm:fdr_hub}
Under Assumption \ref{asp:tradeoff_fdp},
% Under similar scaling condition on $(n,d)$ as in Proposition \ref{prop:approx_cramer_gmb} $(\log ed)^3 (\log (ed+n))^{56/3}/n + s   (\log d)^{5/2}/\sqrt{n}= o(1)$ and the condition listed below
%, where $b_{0}$ is some constant which does not depend on $n$ or $d$,
the StarTrek procedure in Algorithm \ref{algo:startrek} with \eqref{eq:ggm_tdTheta} as input and the quantiles approximated by \eqref{eq:ggm_chat} satisfies:
%for the StarTrek procedure described in Section \ref{sec:startrek}, we have the following error rate control result, i.e., 
% $\lim_{(n,d)\rightarrow \infty} \mathrm{FDR} \le  q$,
% \begin{equation}\label{eq:fdp_hub}
% \end{equation}
as $(n,d)\rightarrow \infty$,
\begin{equation}\label{eq:fdp_hub}
    % \lim_{(n,d)\rightarrow \infty} 
    %\mathrm{FDP} \le q\frac{d_{0}}{d} ~~~ \text{in probability and}
    \mathrm{FDP} \le q\frac{d_0}{d} + \smallop \quad \text{and}
    \quad \lim_{(n,d)\rightarrow \infty}\mathrm{FDR} \le q\frac{d_{0}}{d}.
\end{equation} 
% Under weaker condition, i.e., Assumption \ref{asp:tradeoff_fdr},
% We list the following two conditions below:
% \begin{itemize}\label{cond:sparse_strength}
%     \item Hub sparsity condition: $d_{0}\ge \rho d $ for some $\rho > 0$. 
%     \item Signal strength condition: $|\cB| = \omega_d$, where $\omega_d \rightarrow \infty$ and satisfies $\omega_d  h_d \rightarrow \infty$ for some sequence $h_d = o (1/\sqrt{\log d})$
%     % \item Graph connectivity condition: the number of connected components on the associated graph is $p$.
%     \item Scaling conditions: $\frac{d(\log d)^{19/6}}{n^{1/6}} = o(1)$.
%     \item Connectivity and edge sparsity trade-off conditions: $\frac{|S|\log d}{d_0 p} = o(1)$.
% \end{itemize}
% where $S$ represents the following set:
% $$
% \{(j_1,j_2,k_1,k_2): j_1\ne j_2,k_1 \ne k_2 , \bTheta_{j_1 j_2} = \bTheta_{j_1 k_1} = \bTheta_{j_2 k_2} = 0, \bTheta_{j_1 k_2} \ne 0, \bTheta_{j_2 k_1}\ne 0\}
% $$
% and $p$ is the number of connected components on the associated graph.
% it suffices to show the term $\frac{|S|\log d}{d_0 p} = o(1)$, which actually holds under the trade-off condition.
\end{theorem}
% \lzmargin{write down the FDR control proof}{}
% \lzmargin{first FDR then FDP }{}
% \lzmargin{
% Note that the sparsity condition is a natural condition. Since otherwise $d_{0} = o(d)$, then rejecting all the hypothesis still guarantee $\mathrm{FDR}\le \frac{d_0}{d}\rightarrow 0$ with power always being $1$. 

% And this sparsity condition does not correspond to the usual sparsity condition in high dimensional linear models or Gaussian graphical models. We claim the signal strength condition is almost necessary and not so strong. For example, one choice of $w_{d}$ sequence can be $(\sqrt{\log d})^{\delta}$ with some $\delta >1$.
% }{delete 
% }

% \lzmargin{
% proof idea:
% - valid p-value\\
% - why in multiple testing, need more accurate bounds for p-value\\
% - p-valid dependency \\
% - repeat some points in sec 3 \\
% }{}

% \lzmargin{check $log (d) (q+o(1))$ for FDR control}{}
%\lzmargin{%contrast FDR and FDP, different assumption, 
%mention the importance of FDP control}{}
%\lzmargin{unlike the multiple selection for continuous signals cite weidong and xia ying, the dependence between the test statistic can not be solved by the existing tools/results, but it can be quantified as the $\maxdiff$ term, $\maxdiff$ term is related to $|S|$}{}
%\lzmargin{in order to establish FDP control, we have to characterize the dependence between test statistic}{}
The proof can be found in Appendix \ref{app:pf:thm:fdr_hub}. Remark that control of the FDR does not prohibit the FDP from varying. Therefore our result on FDP provides a stronger guarantee on controlling the false discoveries. See clear empirical evidence in Section \ref{sec:synthetic}. To the best of our knowledge, the proposed StarTrek filter in Section \ref{sec:method} and the above  FDP/FDR control result are the first Algorithm and theoretical guarantee for the problem of simultaneously selecting hub nodes. Existing work like \cite{liu2013ggmfdr, liu2014hypothesis, xia2015testing,xia2018multiple,javanmard2019false} focus on the discovery of continuous signals and their tools are not applicable to the problem here.

\section{StarTrek for general graphical models}
\label{sec:gen_fdr}
\rev{Sections \ref{sec:bipartite_selection} and \ref{sec:hub_selection}  apply the StarTrek filter to two concrete examples: Gaussian graphical models and multitask regression and provide FDR results with explicit assumptions. In this section, we will discuss how to generalize the results in Theorem \ref{thm:fdr_hub} to the general graphical models. }

Recall that in Section \ref{sec:startrek}, we denote $\bTheta$ as the weight matrix of the general graphical models, i.e.,  $\bTheta_{e} \neq 0$ if and only if $e \in \cE$ where $\cE$ is the edge set, and $\tilde \bTheta$ is the generic estimator of $\bTheta$. Since Algorithm \ref{algo:startrek} requires the quantile of the maximal statistics in \eqref{eq:T_E}, we need the estimator $\tilde \bTheta$ to be asymptotically Gaussian. In specific, for each $e \in \cV \times \cV$,  there exist $n$ i.i.d. mean zero random vectors $\bY_1(e),\cdots,\bY_n(e)$ such that $\tilde \bTheta_e - \bTheta_e = n^{-1}\sum_{i=1}^n \bY_i(e) + o_P(1/\sqrt{n})$. We can then estimate the quantile of $T_E = \max_{e \in E} \sqrt{n}|\tilde\bTheta_e - \bTheta_e|$ by some Gaussian multiplier bootstrap statistic $T^{\cB}_{E}$, i.e.,
\begin{equation}
\label{eq:ggm_gen}
\hat{c} (\alpha,E) = \inf \left\{ t\in \RR : \PPP_\xi \left( T^{\cB}_{E} \le t  \right) \ge 1-\alpha    \right\},
\end{equation}
%e.g., \eqref{eq:ggm_chat} for the Gaussian graphical model. 
In specific, similar to the assumptions in Theorem 3.2 of \cite{chernozhukov2013gaussian}, we need the following general assumptions on $\tilde \bTheta$.

\begin{assumption}\label{asmp:zeta12} 
There exist $\zeta_1 \ge 0$ and $\zeta_2  \ge 0$ such that for any edge set $E \subseteq \cV\times \cV$, we have
 \begin{align}\nonumber
 % \label{eq:TY_zetas}
    &\mathbb{P}(|T_E-T_{0E}|>\zeta_{1})<\zeta_{2}, \text{ where } T_{0E}= \max_{e\in E}\Big| \frac{1}{\sqrt{n}}\sum_{i=1}^{n}\bY_{i} (e) \Big| ,\\ \nonumber
 \label{eq:TW_zetas}&\mathbb{P}(\mathbb{P}_{\cqt}(|T_E^{\cB}-T_{0E}^{\cB}|>\zeta_{1})>\zeta_{2})<\zeta_{2}, \text{ where } T_{0E}^{\cB} = \max_{e\in E} \Big| \frac{1}{\sqrt{n}}\sum_{i=1}^{n}\bY_{i} (e) \xi_i  \Big|,
\end{align}
where the i.i.d. mean zero random vectors $\bY_1(e),\cdots,\bY_n(e) \sim \bY(e)$ satisfy $\min_{e \in \cV \times \cV} \mathbb{E}[ \bY^2(e)]>c$, $\max_{e \in \cV \times \cV}  \norm{\bY(e)}_{\psi_1} \le C$ and $\max_{e \in \cV \times \cV}  \norm{\bY^2(e)}_{\psi_1} \le C\sqrt{n/\log d}$ for some positive constants $c$ and $C$, the multiplier variables $\xi_{i}\stackrel{i.i.d.}{\sim}\cN(0,1)$ are independent from the data and $\mathbb{P}_{\cqt}$ is the measure only on $\xi_1, \ldots, \xi_n$.
\end{assumption}
For the multi-task regression problem, we verify the above assumption in Lemma~\ref{lem:linear_can_approx} in the supplementary material. As for the Gaussian graphical models, the assumption is verified in the proof of Lemma \ref{lem:approx_quantile} in the supplementary material. The assumption has also been validated under other graphical models. For example, \cite{yu2020simultaneous} proved the assumption holds for the exponential family pairwise graphical models, which include non-negative Gaussian, conditionally specified mixed graphical models, exponential square-root graphical model, etc.

%The proof of Theorem \ref{thm:fdr_hub} (Appendix \ref{pf:thm:fdr_hub} in the supplementary material) is not inherently tied to the normality assumptions in the data generating models and there are only two parts which are specific to the Gaussian graphical models: (1) deriving the rates $\zeta_1$ and $\zeta_2$ which are plugged into the Cram\'{e}r-type deviation bounds in \eqref{eq:TTZ_accuracy}; (2) bounding the covariance $\cov(W_j(\alpha), W_k(\alpha))$ for $j\ne k, j, k \in \cH_0$ where $W_j(\alpha)= \Indrbr{\max_{e \in N_{0j} } \sqrt{n}|\tilde\bTheta_e| \ge \hat c(\alpha, N_{0j}) }$  and $N_{0j}=\{(j,\ell) : \bTheta_{j\ell}=0\}$. Both parts can be conducted similarly for non-Gaussian graphical models and would impose model-specific scaling conditions. We shall mention that the Cram\'{e}r-type deviation results in Appendix \ref{sec:gmb_theory} in the supplementary material are general and do not rely on any particular distributional assumptions about the data: 
% we can derive the following Cram\'{e}r-type deviation bound  under some scaling conditions: 
%\begin{equation}
% \label{eq:TTZ_accuracy}
% \sup_{\alpha \in [\alpha_{L},1]}
% \left|
% \frac{\mathbb{P}(T_E > \hat{c}(\alpha, E))}{\alpha} - 1 
% \right| 
% =   O\rbr{\frac{(\log d)^{19/6}}{n^{1/6}} + \frac{ (\log d)^{11/6}}{n^{1/6}\alpha_L^{1/3}} + \zeta_1 \log d + \frac{\zeta_2}{\alpha_L} },
% \end{equation}
% under Assumption \ref{asmp:zeta12}, where $\alpha_L = \Omega(1/d)$, $\zeta_1, \zeta_2$ are defined in \eqref{eq:TY_zetas}, \eqref{eq:TW_zetas}. 

Similar to \eqref{eq:dep_term_set}, for the general case, we also need to define a dependency set $S$  as
\[
%\nonumber
S=\{(j_1,j_2,k_1,k_2): j_1, j_2\in \cH_0, j_1 \ne j_2, \bTheta_{j_1k_1} =\bTheta_{j_2k_2}=0,
\cov(\bY((j_1,k_1)),\bY((j_2,k_2))) \neq 0\}.
\]
We also impose the cardinality of $S$ similar to Assumption \ref{asp:tradeoff_fdp} for general graphical models. 
\vspace{-0.2cm}
\begin{assumption}\label{asp:tradeoff_fdp_gen}
{Suppose that $\bTheta \in \cU(M,s, r_0)$ and the following scaling 
 condition holds:}
\begin{eqnarray} \nonumber 
%    \label{eq:cond1_gen}
    \frac{\log d }{\rho}\rbr{
    \frac{(\log d)^{19/6}}{n^{1/6}} + \frac{ (\log d)^{11/6}}{\rho^{1/3} n^{1/6}} + 
   \zeta_1 \log d + \frac{\zeta_2}{\rho}
    } + \zeta_2 d^4  +   \frac{\log d}{\rho d_0} + \frac{({\log d})^{2}|S|}{\rho d_0^2} = o(1),
%    \\ \nonumber 
%       \frac{\log d}{\rho d_0} + \frac{({\log d})^{2}|S|}{\rho d_0^2} = o(1).
    \end{eqnarray}
 where the definitions of $d_0$ and $\rho$ are the same as in Section \ref{sec:hub_selection}.   
%\begin{enumerate}[(i)]
%    \item Signal strength and scaling condition.
%    \begin{eqnarray} \nonumber 
%%    \label{eq:cond1_gen}
%    \frac{\log d }{\rho}\rbr{
%    \frac{(\log d)^{19/6}}{n^{1/6}} + \frac{ (\log d)^{11/6}}{\rho^{1/3} n^{1/6}} + 
%   \zeta_1 \log d + \frac{\zeta_2}{\rho}
%    } + \zeta_2 d^4 = o(1).
%    \end{eqnarray}
%    \item Dependency and connectivity condition.
%    \begin{eqnarray}  \nonumber
%%    \label{eq:cond2_gen}
%   \frac{\log d}{\rho d_0} + \frac{({\log d})^{2}|S|}{\rho d_0^2} = o(1).
%    \end{eqnarray}
%\end{enumerate}
%where the definitions of $d_0$ and $\rho$ are the same as in Section \ref{sec:hub_selection}. 
\end{assumption}
%Note Assumption \ref{asp:cov_form} is 
%Assumption \ref{asp:tradeoff_fdp_gen} are general scaling assumption. 
%Now we state our general FDR control result as below.
%\vspace{-0.2cm}

We then have the following theorem on the FDR control for the general graphical models.
\begin{theorem}\label{thm:fdr_gen}
Under Assumptions \ref{asmp:zeta12} and \ref{asp:tradeoff_fdp_gen}, the StarTrek procedure in Algorithm \ref{algo:startrek} with the generic estimator $\tTheta$ and the approximated quantiles \eqref{eq:ggm_gen} satisfies:
as $(n,d)\rightarrow \infty$,
\begin{equation}\nonumber
%\label{eq:fdp_gen} 
    \mathrm{FDP} \le q{d_0}/{d} + \smallop \quad \text{and}
    \quad \lim_{(n,d)\rightarrow \infty}\mathrm{FDP} \le q{d_{0}}/{d}.
\end{equation} 
\end{theorem}
\vspace{-0.3cm}
The proof can be found in Appendix \ref{app:pf:thm:fdr_gen}.

\section{Numerical results}
\label{sec:simul}
\rev{In this section, we conduct simulation studies to complement the main theoretical claims of the paper and demonstrate the empirical performance of our method. Section \ref{sec:mult_simul} presents numerical results of applying the StarTrek filter to the multitask regression problem in Section \ref{sec:bipartite_selection}. Section \ref{sec:synthetic} focuses on the Gaussian graphical models studied in Section \ref{sec:hub_selection}. In Section \ref{sec:speed_simul}, we numerically compare our method to the grid search based on the skip-down method and demonstrate the computational advantages of the StarTrek filter. We also study the power performance of our approach against three competitor testing methods.}

\subsection{Simulations for multitask regression}
\rev{
Section \ref{sec:bipartite_selection} considers the application of the StarTrek filter to the multitask regression problem and provides theoretical results on FDP/FDR control. Here we conduct some simulation studies. The synthetic datasets are generated from the multitask regression model described in \eqref{eq:mul_linear_model}. We sample the covariates from a Gaussian autoregressive model of order 1 (AR(1)) and choose the noise variance to be $1$ for all responses. Now we describe how to generate the parameter matrix. First, the number of non-zero coefficients for each row is independently uniformly sampled from the integers between 0 and 20. Then the locations of non-zero coefficients are independently uniformly drawn from among the covariates. Finally, the values of non-zero coefficients are taking uniform random signs and identical magnitudes of 1. Throughout the simulated examples, we fix the number of responses and the number of covariates ($d_1=d_2 =300$) and vary the sample size $n$ and the autocorrelation coefficient of the AR(1) design. We also run the selection procedure under two choices of the nominal FDR level i.e., $q \in \{ 0.1, 0.2\}$. Given the sparsity level of each response is uniformly distributed over integers between 0 and 20, we choose the threshold $k_\tau$ for determining hub responses to be 19, which is roughly the upper 10\% quantile of the sparsity level's distribution. To run the StarTrek filter, we exactly follow the procedures described in Section \ref{sec:bipartite_selection} to calculate the test statistics and the approximated quantiles. The involving estimation steps are based on \cite{javanmard2014confidence,sun2012scaled}.

Table \ref{tb:multFDR} shows that the FDRs of the StarTrek filter are all well controlled below the nominal levels for different sample sizes and autocorrelation coefficient $\bar\sigma$. From Table \ref{tb:multPower}, we find that the power of the proposed method increases as the sample size grows and  decreases as the covariates become more dependent (i.e., with higher autocorrelations).
}
\vspace{-0.1cm}
\label{sec:mult_simul}
\begin{table}[hptb]
\addtolength{\tabcolsep}{-4pt}
\begin{center}
    \caption{\rev{{Empirical FDR for the multitask regression problem. We set $d_1 = d_2 = 300$ and $n = 150, 200, 250$ and $300$. The autocorrelaiton $\bar\sigma$ varies from $0.3$ to $0.7$. }}}
     \vspace{0.05cm}
    \begin{tabular}{
    @{\hspace{0.01em}}c
    @{\hspace{1em}}c
    @{\hspace{1em}}c
    @{\hspace{1em}}c
    @{\hspace{1em}}c
    @{\hspace{1em}}|c
    @{\hspace{1em}}  c
    @{\hspace{1em}}  c
    @{\hspace{1em}}  c
    @{\hspace{1em}}  c
    @{\hspace{1.2em}} }
    \toprule
  &\multicolumn{4}{c}{$q=0.1$}&\multicolumn{4}{c}{$q=0.2$} \\
      \hline
 n &150 &200 &250 &300
        &150 &200 &250 &300 \\
    \hline
$\bar\sigma=0.3$ & 0.0656 & 0.0416 & 0.0162 & 0.0184 & 0.0991 & 0.0676 & 0.0269 & 0.0355 \\ 
  $\bar\sigma=0.4$ & 0.0638 & 0.0355 & 0.0144 & 0.0158 & 0.1006 & 0.0577 & 0.0252 & 0.0300 \\ 
  $\bar\sigma=0.5$ & 0.0554 & 0.0376 & 0.0177 & 0.0179 & 0.0827 & 0.0532 & 0.0253 & 0.0349 \\ 
  $\bar\sigma=0.6$ & 0.0525 & 0.0316 & 0.0144 & 0.0155 & 0.0762 & 0.0516 & 0.0257 & 0.0270 \\ 
  $\bar\sigma=0.7$ & 0.0406 & 0.0454 & 0.0233 & 0.0224 & 0.0557 & 0.0662 & 0.0464 & 0.0385 \\ 
    \toprule
    \end{tabular}
    \label{tb:multFDR}
  \end{center}
     \vspace{-20pt} 
\end{table}
\vspace{-0.3cm}
\begin{table}[hptb]
\addtolength{\tabcolsep}{-4pt}
\begin{center}
    \caption{\rev{{Empirical power for the multitask regression. We set $d_1 = d_2 = 300$ and $n = 150, 200, 250$ and $300$. The autocorrelaiton $\bar\sigma$ is chosen from $0.3$ to $0.7$. }}}
    \vspace{0.15cm}
    \begin{tabular}{
    @{\hspace{1.2em}}c
    @{\hspace{1.2em}}c
    @{\hspace{1em}}c
    @{\hspace{1em}}c
    @{\hspace{1em}}c
    @{\hspace{1em}}|c
    @{\hspace{1em}}  c
    @{\hspace{1em}}  c
    @{\hspace{1em}}  c
    @{\hspace{1em}}  c
    @{\hspace{1.2em}} }
    \toprule
  &\multicolumn{4}{c}{$q=0.1$}&\multicolumn{4}{c}{$q=0.2$} \\
    \hline
  n  &150 &200 &250 &300
        &150 &200 &250 &300 \\
    \hline
$\bar\sigma=0.3$ & 0.8902 & 1.0000 & 1.0000 & 1.0000 & 0.9206 & 1.0000 & 1.0000 & 1.0000 \\ 
  $\bar\sigma=0.4$ & 0.8090 & 1.0000 & 1.0000 & 1.0000 & 0.8481 & 1.0000 & 1.0000 & 1.0000 \\ 
  $\bar\sigma=0.5$ & 0.6563 & 0.9912 & 1.0000 & 1.0000 & 0.7081 & 0.9953 & 1.0000 & 1.0000 \\ 
  $\bar\sigma=0.6$ & 0.4058 & 0.9549 & 0.9976 & 1.0000 & 0.4590 & 0.9622 & 0.9990 & 1.0000 \\ 
  $\bar\sigma=0.7$ & 0.1215 & 0.7119 & 0.9678 & 0.9965 & 0.1578 & 0.7621 & 0.9778 & 0.9995 \\ 
    \toprule
    \end{tabular}
    \label{tb:multPower}
  \end{center}
     \vspace{-20pt} 
\end{table}

\subsection{\rev{Simulations for Gaussian graphical models}}
\label{sec:synthetic}
\rev{In this section, we provide simulations results for Section \ref{sec:hub_selection}.} The synthetic datasets are generated from Gaussian graphical models. The corresponding precision matrices are specified based on four different types of graphs. Given the number of nodes $d$ and the number of connected components $p$, we will randomly assign those nodes into $p$ groups. Within each group (sub-graph), the way of assigning edges for different graph types will be explained below in detail. After determinning the adjacency matrix of the graph, we follow \cite{zhao2012huge} to construct the precision matrix, more specifically, we set the off-diagonal elements to be of value $v$ which control the magnitude of partial correlations and is closely related to the signal strength. In order to ensure positive-definiteness, we add some value $v$ together with the absolute value of the minimal eigenvalues to the diagonal terms. In the following simulations, $v$ and $u$ are set to be $0.4$ and $0.1$ respectively. Now we explain how to determine the edges within each group (sub-graph) for four different graph patterns.
\begin{itemize}
    \item \textbf{Hub graph.} We randomly pick one node as the hub node of the sub-graph, then the rest of the nodes are made to connect with this hub node. There is no edge between the non-hub nodes.
    \item \textbf{Random graph.} This is the Erd\"os-R\'enyi random graph. There is an edge between each pair of nodes with certain probability independently. In the following simulations, we will set this probability to be $0.15$ unless stated otherwise.
    \item \textbf{Scale-free graph.} In this type of graphs, the degree distribution follows a power law. We construct it by the Barab\'asi-Albert algorithm: starting with two connected nodes, then adding each new node to be connected with only one node in the existing graph; and the probability is proportional to the degree of the each node in the existing graph. The number of the edges will be the same as the number of nodes.
    \item \textbf{K-nearest-neighbor (knn) graph.} For a given number of $k$, we add edges such that each node is connected to another $k$ nodes. In our simulations, $k$ is sampled from $
    \{1,2,3,4\}$ with probability mass $\{0.4,0.3,0.2,0.1\}$.
\end{itemize}
See a visual demonstration of the above four different graph patterns in Appendix \ref{app:graph_pattern}. Throughout the simulated examples, we fix the number of nodes $d$ to be $300$ and vary other quantities such as sample size $n$ or the number of connected components $p$. To estimate the precision matrix, we run the graphical Lasso algorithm with 5-fold cross-validation. Then we obtain the standardized debiased estimator as described in \eqref{eq:ggm_tdTheta}. To obtain the quantile estimates, we use the Gaussian multiplier bootstrap with {4000} bootstrap samples. The threshold $k_\tau$ for determining hub nodes is set to be $3$. And all results (of FDR and power) are averaged over 64 independent replicates. 

As we can see from Table \ref{tb:FDR}, the FDRs of StarTrek filter for different types of graph are well controlled below the nominal levels. In hub graph, the FDRs are relatively small but the power is still pretty good. Similar phenomenon for multiple edge testing problem is observed \cite{liu2013ggmfdr}. In the context of node testing, it is also unsurprising. These empirical results actually match our demonstration about $|S|$ in Figure \ref{fig:demo}: hub graphs have a relatively weaker dependence structure (smaller $S$ values) and make it is easier to discover true hub nodes without making many errors.
% as the hub graph has a simple structure, which makes it easy to select true hub nodes without making too many errors. This also matches the fact that the value of $|S|$ is relatively small in the graph with hub patterns, as showed in Figure \ref{fig:demo}. 
 \vspace{-0.15in}
  \hspace{-0.2cm}
\begin{table}[hptb]
% \small
\addtolength{\tabcolsep}{-4pt}
\begin{center}
    \caption{\rev{{Empirical FDR for Gaussian graphical models. }}}
    \begin{tabular}{
    @{\hspace{1.2em}}c 
    @{\hspace{1.2em}}c
    @{\hspace{1em}}c
    @{\hspace{1em}}c|
    @{\hspace{1em}}c
    @{\hspace{1em}}c
    @{\hspace{1em}} c
    @{\hspace{1em}} c
    @{\hspace{1em}} c 
    @{\hspace{1.2em}} }
    \toprule
    $d= 300~$ &\multicolumn{3}{c}{$q=0.1$}&\multicolumn{3}{c}{$q=0.2$} \\
    \hline
    $n$ &200 &300 &400
        &200 &300 &400\\
    \hline & \multicolumn{5}{c}{$\quad\quad p=20$}&\multicolumn{1}{c}{} \\
    \hline
% hub & 0.0000 & 0.0000 & 0.0007 & 0.0000 & 0.0000 & 0.0029 \\ 
%   random & 0.0255 & 0.0383 & 0.0467 & 0.0521 & 0.0770 & 0.0833 \\ 
%   scale-free & 0.0093 & 0.0211 & 0.0282 & 0.0352 & 0.0486 & 0.0581 \\ 
%   knn & 0.0101 & 0.0296 & 0.0370 & 0.0228 & 0.0620 & 0.0769 \\ 
hub & 0.0000 & 0.0007 & 0.0000 & 0.0018 & 0.0016 & 0.0015 \\ 
  random & 0.0186 & 0.0329 & 0.0438 & 0.0438 & 0.0727 & 0.0851 \\ 
  scale-free & 0.0091 & 0.0243 & 0.0259 & 0.0265 & 0.0480 & 0.0579 \\ 
  knn & 0.0103 & 0.0288 & 0.0345 & 0.0275 & 0.0648 & 0.0736\\ 
    
    \hline & \multicolumn{5}{c}{$\quad\quad p=30$}&\multicolumn{1}{c}{} \\
    \hline

% hub & 0.0013 & 0.0000 & 0.0016 & 0.0027 & 0.0054 & 0.0036 \\ 
%   random & 0.0347 & 0.0359 & 0.0568 & 0.0725 & 0.0753 & 0.0963 \\ 
%   scale-free & 0.0215 & 0.0335 & 0.0317 & 0.0521 & 0.0624 & 0.0584 \\ 
%   knn & 0.0297 & 0.0420 & 0.0563 & 0.0504 & 0.0857 & 0.1030 \\   
hub & 0.0012 & 0.0017 & 0.0000 & 0.0031 & 0.0039 & 0.0036 \\ 
  random & 0.0464 & 0.0498 & 0.0478 & 0.0874 & 0.0969 & 0.0911 \\ 
  scale-free & 0.0205 & 0.0326 & 0.0271 & 0.0414 & 0.0602 & 0.0580 \\ 
  knn & 0.0216 & 0.0475 & 0.0431 & 0.0551 & 0.0909 & 0.0883 \\ 
    \toprule
    \end{tabular}
    \label{tb:FDR}
  \end{center}
     \vspace{-20pt} 
\end{table}

The power performance of the StarTrek filter is showed in Table \ref{tb:Power}. As the sample size grows, we see the power is increasing for all four different types of graphs. When $p$ is larger, there are more hub nodes in general due to the way of constructing the graphs, and we find the power is higher. Among different types of graphs, the power in hub graph and scale-free graph is higher than that in random and knn graph since the latter two are relatively denser and have more complicated topological structures.
\begin{table}[htbp]
% \small
\addtolength{\tabcolsep}{-4pt}
\begin{center}
    \caption{\rev{{Empirical power for Gaussian graphical models.}}}
    \begin{tabular}{
    @{\hspace{1.2em}}c 
    @{\hspace{1.2em}}c
    @{\hspace{1em}}c
    @{\hspace{1em}}c|
    @{\hspace{1em}}c
    @{\hspace{1em}}c
    @{\hspace{1em}} c
    @{\hspace{1em}} c
    @{\hspace{1em}} c 
    @{\hspace{1.2em}} }
    \toprule
    $d= 300~$ &\multicolumn{3}{c}{$q=0.1$}&\multicolumn{3}{c}{$q=0.2$} \\
    \hline
    $n$ &200 &300 &400
        &200 &300 &400\\
    \hline & \multicolumn{5}{c}{$\quad\quad p=20$}&\multicolumn{1}{c}{} \\
    \hline
    
% hub & 0.7109 & 0.9453 & 0.9898 & 0.7805 & 0.9648 & 0.9938 \\ 
%   random & 0.3343 & 0.7815 & 0.9408 & 0.4520 & 0.8514 & 0.9604 \\ 
%   scale-free & 0.4524 & 0.8145 & 0.9363 & 0.5281 & 0.8614 & 0.9568 \\ 
%   knn & 0.0905 & 0.5306 & 0.8067 & 0.1634 & 0.6511 & 0.8630 \\ 
hub & 0.6789 & 0.9406 & 0.9812 & 0.7727 & 0.9609 & 0.9867 \\ 
  random & 0.3445 & 0.7734 & 0.9390 & 0.4637 & 0.8413 & 0.9592 \\ 
  scale-free & 0.4799 & 0.8050 & 0.9347 & 0.5549 & 0.8479 & 0.9545 \\ 
  knn & 0.1337 & 0.5689 & 0.8381 & 0.2254 & 0.6913 & 0.8920 \\   
    \hline & \multicolumn{5}{c}{$\quad\quad p=30$}&\multicolumn{1}{c}{} \\
    \hline

% hub & 0.6848 & 0.9244 & 0.9706 & 0.7588 & 0.9459 & 0.9784 \\ 
%   random & 0.4882 & 0.8863 & 0.9790 & 0.5770 & 0.9225 & 0.9870 \\ 
%   scale-free & 0.6472 & 0.9047 & 0.9810 & 0.7197 & 0.9331 & 0.9870 \\ 
%   knn & 0.2409 & 0.6841 & 0.8922 & 0.3298 & 0.7706 & 0.9241 \\ 

hub & 0.6861 & 0.9242 & 0.9736 & 0.7497 & 0.9405 & 0.9810 \\ 
  random & 0.5136 & 0.8728 & 0.9741 & 0.6027 & 0.9085 & 0.9842 \\ 
  scale-free & 0.6296 & 0.8975 & 0.9778 & 0.7060 & 0.9230 & 0.9842 \\ 
  knn & 0.2442 & 0.7036 & 0.8990 & 0.3396 & 0.7799 & 0.9335 \\ 

    \toprule
    \end{tabular}
    \label{tb:Power}
  \end{center}
     \vspace{-20pt} 
\end{table}

  \hspace{-0.2cm}
\begin{figure}[t]
    \centering
\begin{center}
\begin{tabular}{ c c }
$p=20$ & $p=30$\\
 \hspace{-0.4cm}
\includegraphics[width = 0.5\linewidth]{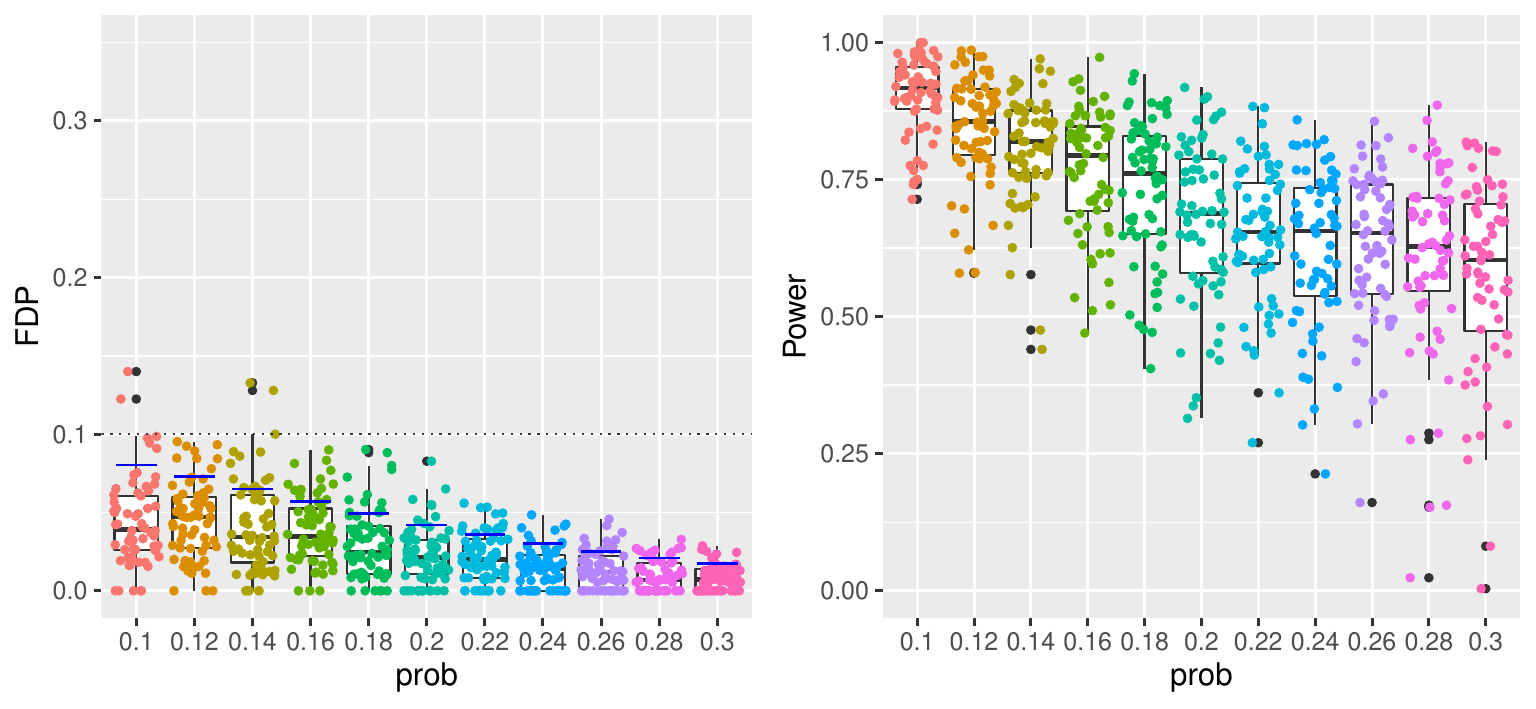}
    &  \hspace{-0.5cm}  \includegraphics[width = 0.5\linewidth]{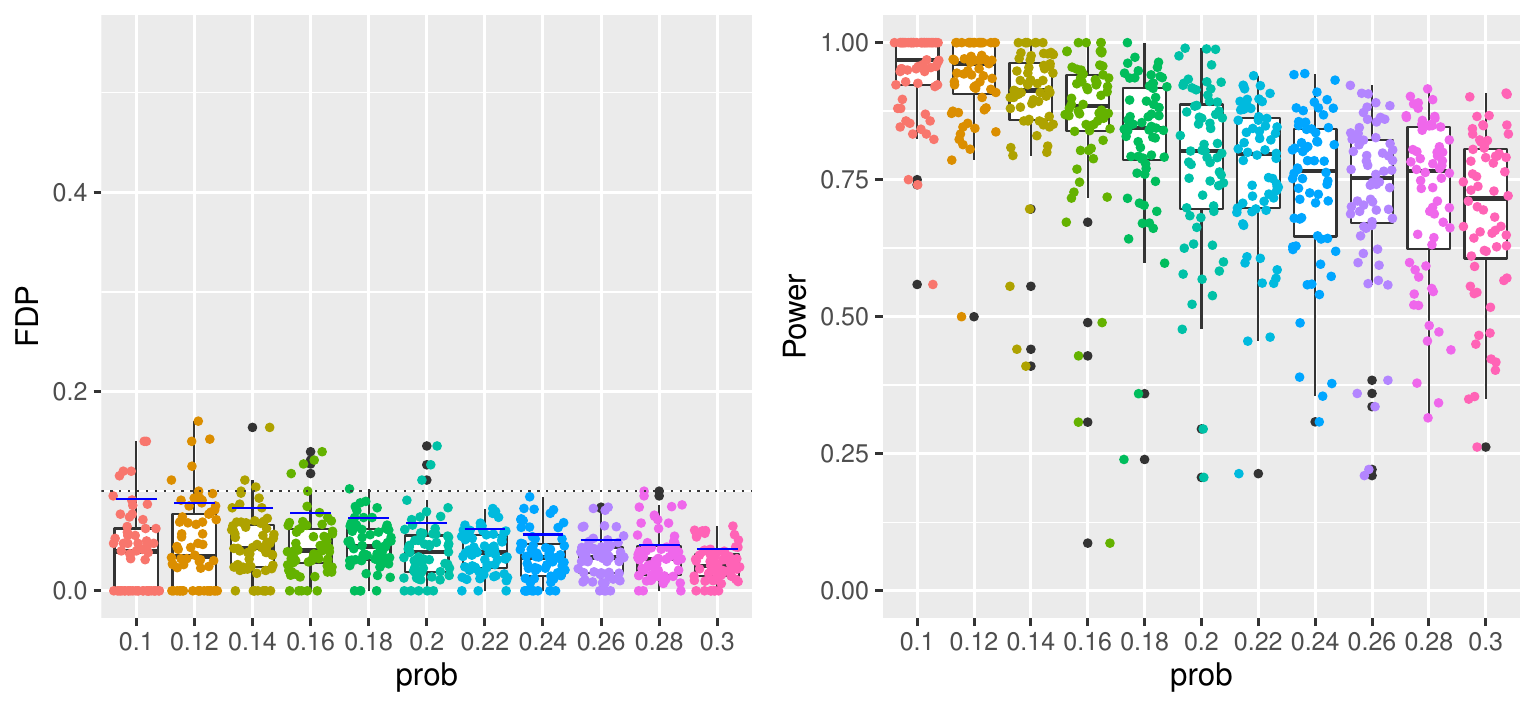}
\end{tabular}
\end{center}    
\vspace{-10pt}
        \caption{\rev{FDP and power plots for the StarTrek filter in the random graph. The connecting probability is varied on the x-axis. The number of samples $n$ is chosen to be $300$ and the number of connected components $p$ equals 20 and 30 . The nominal FDR level is set to be $q=0.1$; the short blue solid lines correspond to  $q{d_0}/{d}$, calculated by averaging over the $64$ replicates. For both panels, the box plots are plotted with the black points representing the outliers. Colored points are jittered around, demonstrating how the FDP and power distribute.}}
    \label{fig:p20q0.1}
\end{figure}
In \rev{Figure \ref{fig:p20q0.1}}, we demonstrate the performance of our method in the random graph with different parameters. Specifically, we vary the connecting probability changing from $0.1$ to $0.3$ in the x-axis. In those plots, we see the FDRs are all well controlled below the nominal level $q=0.1$. As the connecting probability of the random graph grows, the graph gets denser, resulting more hub nodes. Thus we can see the height of the short blue solids lines (representing $q d_0/d$) is decreasing. Based on our results in Theorem \ref{thm:fdr_hub}, the target level of FDP/FDR control is $q d_0/d$. This is why we find the mean and median of each box-plot is getting smaller as the connecting probability increases (hence $d_0$ decreases). 
%corresponding to Figures 9 and 10 can be found Appendix G.4.

The box-plots and the jittering points show that our StarTrek procedure not only controls the FDR but also prohibit it from varying too much, as implied by the theoretical results on FDP control in Section \ref{sec:hub_selection}. Regarding the power plots, we see that the power is smaller when the graph is denser since the hub selection problem becomes more difficult with more disturbing factors. Plots with nominal FDR level $q=0.2$ are deferred to Appendix \ref{app:fdp_power_plots}.

\subsection{\rev{Comparison with the grid search based on the Skip-down method}}
\label{sec:speed_simul}
\begin{figure}
    \centering
    \includegraphics[width = 1\linewidth]{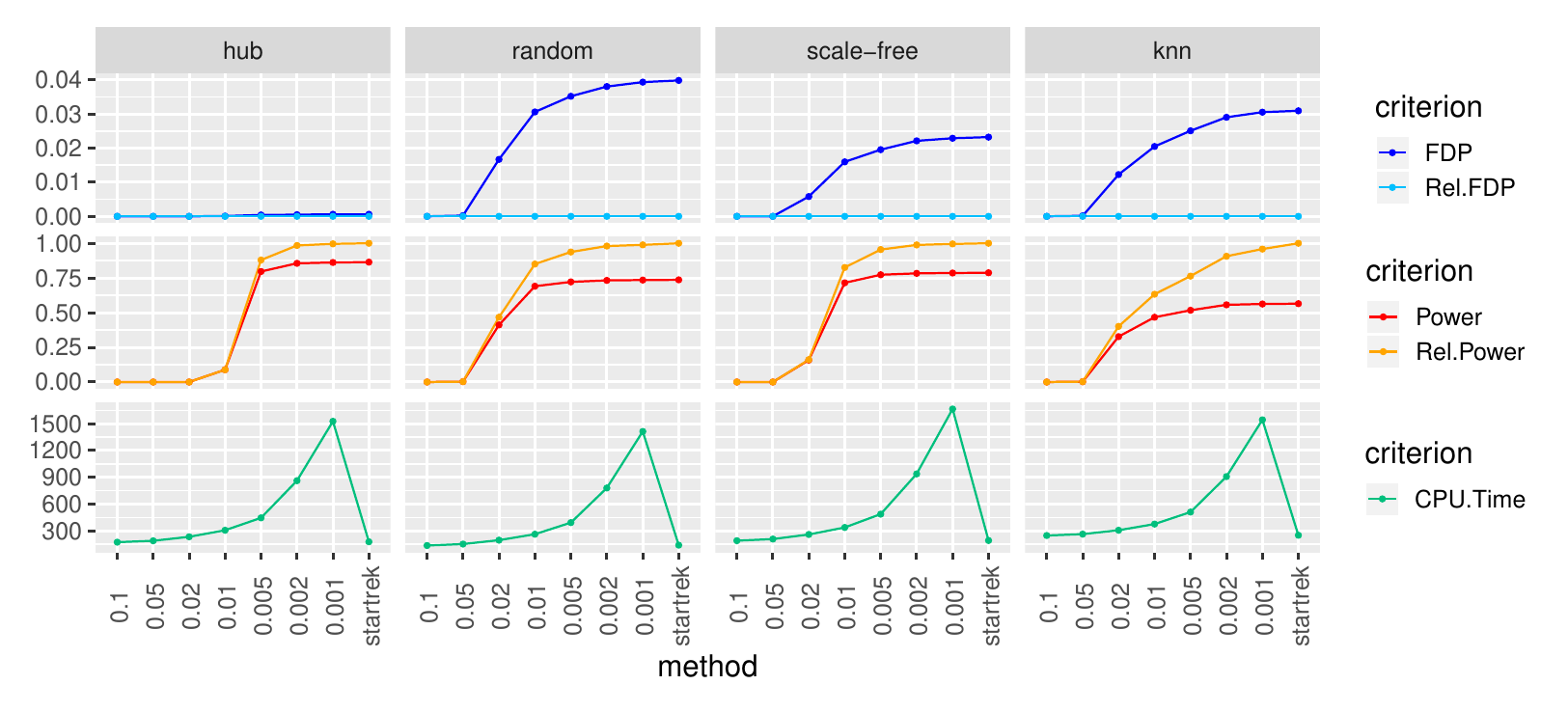}
    \vspace{-20pt}
    \caption{\rev{Empirical comparison of the StarTrek filter with the grid search based on the Skip-down method. On the x-axis, we have the grid method with 7 granularity levels and our proposed StarTrek filter. The top panel compares FDR and relative FDR; the middle panel compares power and relative power; the bottom panel compares the CPU time (in seconds). Four columns correspond to four different graph patterns. For each graph pattern, we average those criterion quantities over 6 different settings (i.e., $3$ choices of $n$ and $2$ choices of $p$) where 64 independent replicates are simulated for each setting. The FDR nominal level $q$ is $0.1$.}}
    \label{fig:speedq0.1}
\end{figure}
% \begin{figure}
%     \centering
%     \includegraphics[width = 1.1\linewidth]{plots/simul/Speed_allgraph_qlevel0.2.pdf}
%     \caption{Caption}
%     \label{fig:speedq0.2}
% \end{figure}
\rev{
In this section, we empirically compare the performance of two methods: our StarTrek filter (Algorithm \ref{algo:startrek}) and the BHq procedure with Algorithm 1 (referred simply as Algorithm \ref{algo:skipdown} without causing confusion).
% the grid search based on Algorithm \ref{algo:skipdown} and our StarTrek filter (Algorithm \ref{algo:startrek}). Without causing confusion, we refer to the grid search based on Algorithm \ref{algo:skipdown} simply by Algorithm \ref{algo:skipdown}. 
To implement Algorithm \ref{algo:skipdown}, we estimate $\hat{\alpha}$ in \eqref{eq:BHq_alpha} by a grid-search of the suprema on a evenly-spaced grid over $(0, q)$ with $7$ spacing sizes: $0.1, 0.05, 0.02, 0.01$, $0.005, 0.002, 0.001$. Note smaller spacing sizes correspond to higher granularity levels. We can see that the computation complexity of Algorithm \ref{algo:startrek} is $O(d k_\tau + d^2\log d)$ and the time complexity is $O(d^2 k_\tau/g)$ for Algorithm \ref{algo:skipdown} where $g$ is the grid spacing size. In additional to the computational differences, we shall note that the selected hub node set from the grid search method must be a subset of that from the StarTrek filter, and as the grid becomes sufficiently granular (i.e., the gird spacing size becomes sufficiently small), the selected hub node sets from the two methods will be the same. To illustrate such points in our empirical comparison, we will additionally compute a relative version of FDR and power. Specifically, we compute the FDR and power for both methods but treating the selected hub node set from the StarTrek filter as the true hub node set. We follow Section \ref{sec:synthetic} to generate the synthetic data and consider exactly the same settings in Table \ref{tb:FDR}. The results are then visualized in Figure \ref{fig:speedq0.1}. First, we see that the relative FDR is always 0 and the relative power approaches $1$ as the gird spacing size decreases to $0$, which illustrates the equivalence of two algorithms. In terms of power and computational performances, we find that StarTrek filter achieves higher power than Algorithm \ref{algo:skipdown} when the granularity level is coarse.  By making the grid sufficiently granular, the grid search method can attain comparable power but cost much longer computational time than the StarTrek filter, hence demonstrate the superiority of our proposed method.
}

\subsection{Comparison with other testing procedures}
\label{sec:comp_simul}
\begin{figure}
    \centering
    \includegraphics[width = 1\linewidth]{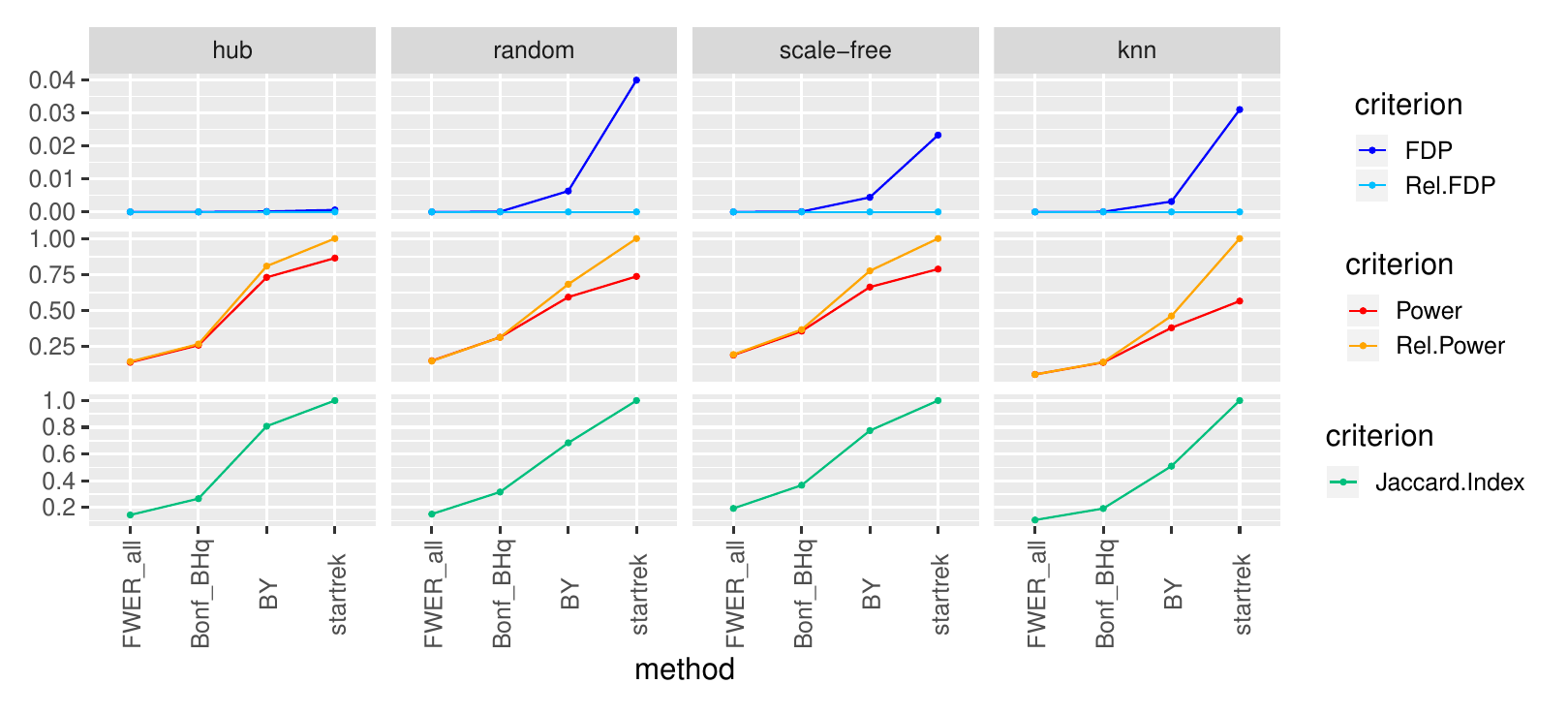}
    \vspace{-20pt}
    \caption{\rev{Empirical comparison of the StarTrek filter with the three competitor methods. On the x-axis, we have Method 1 (FWER\_all), Method 2 (Bonf\_BHq), Method 3 (BY) and our proposed StarTrek filter. The top panel compares FDR and relative FDR; the middle panel compares power and relative power; the bottom panel compares the Jaccard index (with respect to the set of selected hub nodes by the StarTrek filter). Four columns correspond to four different graph patterns. For each graph pattern, we average those criterion quantities over 6 different settings (i.e., $3$ choices of $n$ and $2$ choices of $p$) where 64 independent replicates are simulated for each setting. The FDR nominal level $q$ is $0.1$.}}
    \label{fig:compq0.1}
\end{figure}

This section compares the performance of the StarTrek filter against some other testing procedures. Three competitor methods are considered. 
\begin{enumerate}
    \item Method 1 computes the p-values with respect to testing $\bTheta_{jk}=0$ for all the $(d^2-d)/2$ pairs of $(j,k)$. Then it adopts the canonical FWER control method to select the significant edges and count the selected edges for each row/column to determine whether each node is selected to be a hub node.
    \item Method 2 computes all the p-values as in Method 1, but changes the way of applying FWER control adjustment. For each node $j$, it applies the Bonferroni procedure to the $d-1$ p-values corresponding to the $j$-th column of the precision matrix, resulting the node p-value. Then the BHq procedure is further applied to these node p-values to select the hub nodes.
    \item Method 3 utilizes the node p-values computed from Algorithm 2 but applies the BY procedure \cite{benjamini2001control} instead of the BHq procedure \cite{benjamini1995controlling}. 
\end{enumerate}
All the three competitor methods are more conservative than the StarTrek filter since they are either only based on continuous edge testing procedures or not adapting to the complex dependence structures. We follow Section \ref{sec:synthetic} to generate the synthetic data and consider exactly the same settings in Table \ref{tb:FDR} which involve $3$ choices of the sample size $n$, $2$ choices of $p$ and $4$ different types of graph patterns. 

%\vspace{-0.1in}
In Figure \ref{fig:compq0.1}, we visualize the performances of  StarTrek filter against the above three competitor methods in terms of FDR and power. To understand how the set of selected hub nodes produced from each competitor method is similar/different to that from the StarTrek filter, we also calculate a relative version of FDR and power similarly as in Section \ref{sec:speed_simul} and the Jaccard index \cite{jaccard1901distribution}. We find that our proposed StarTrek filter is less conservative and more powerful than all the three competitor methods, among which Method 1 is the most conservative method and Method 3 has the most similar selected hub node set to the StarTrek filter.

\subsection{Application to gene expression data}
\label{sec:data}
\begin{figure*}
   \centering
\begin{tabular}{m{1em}ccccc}
\rot{{$\quad\quad\quad\quad\quad\quad\quad$Male}}&
\includegraphics[width=3cm]{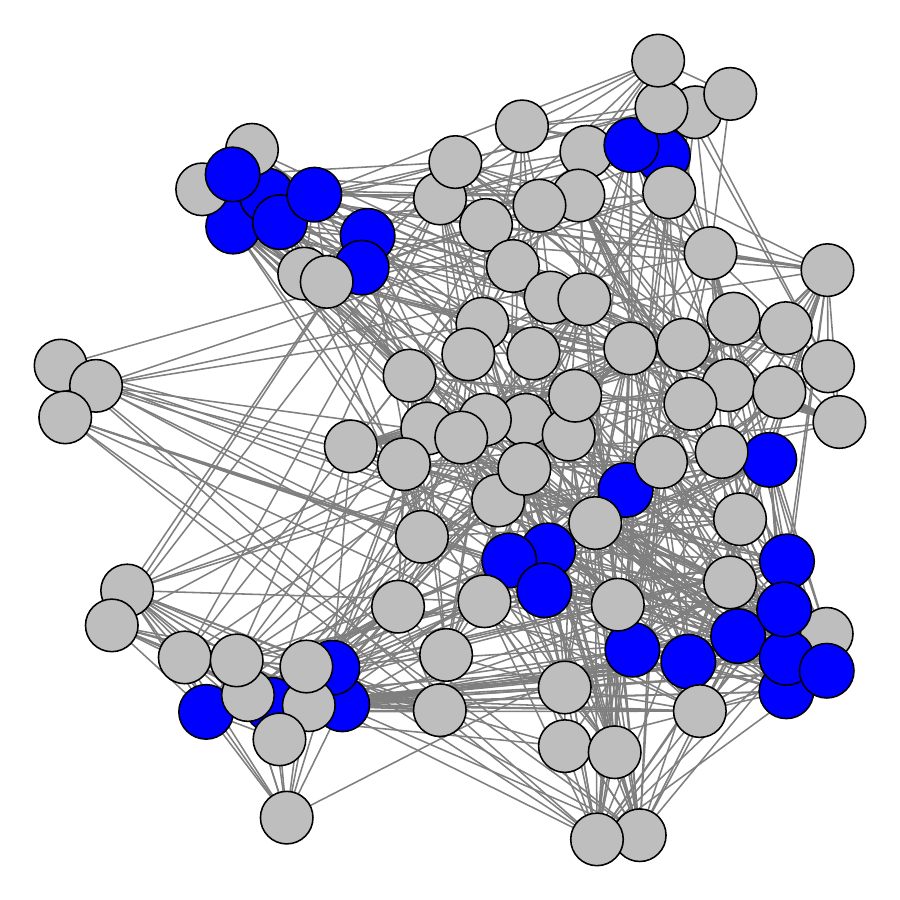}&
\includegraphics[width=3cm]{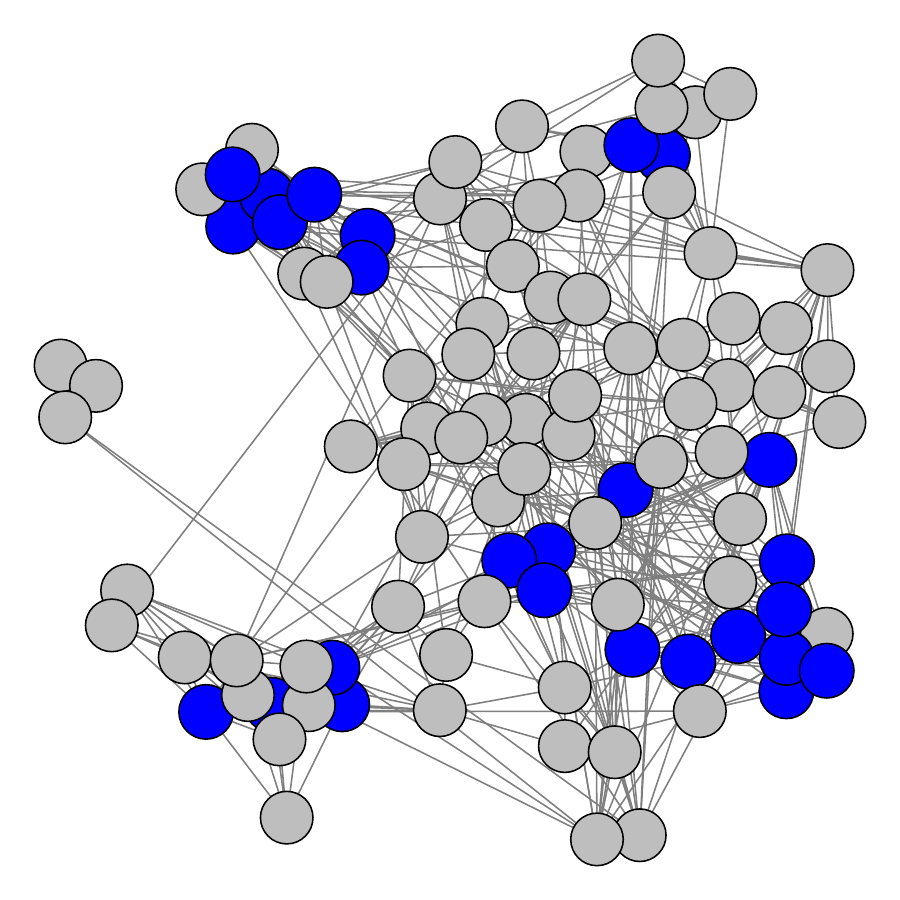}&
\includegraphics[width=3cm]{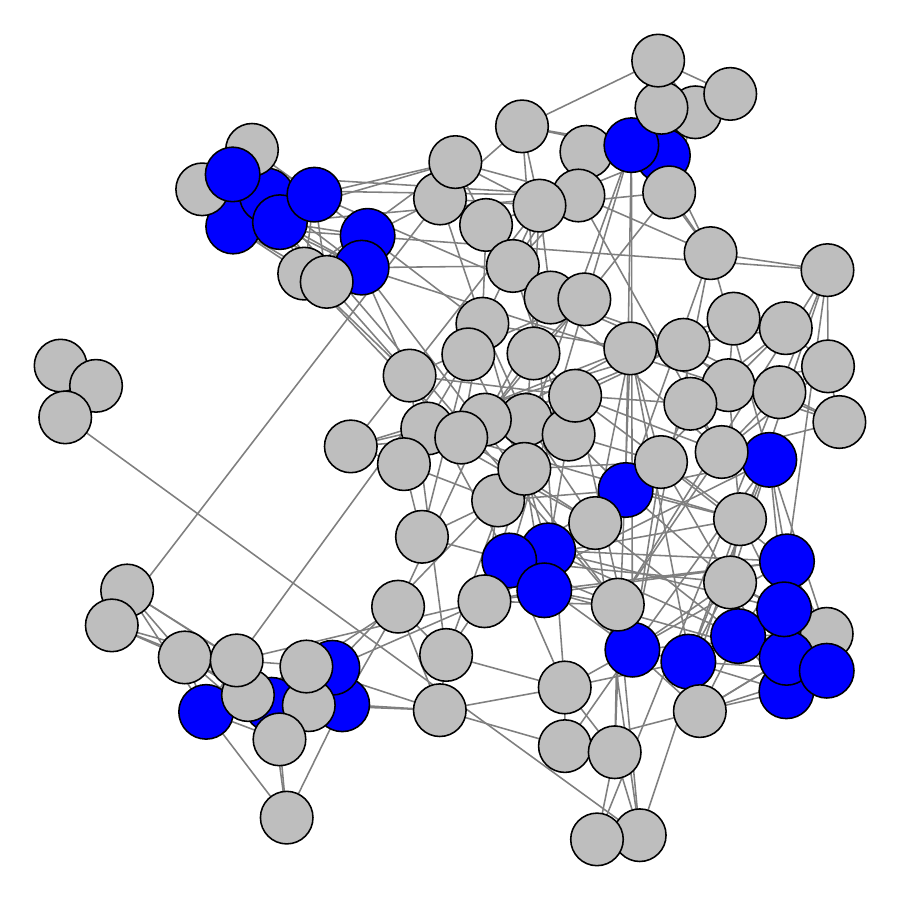}&
\includegraphics[width=3cm]{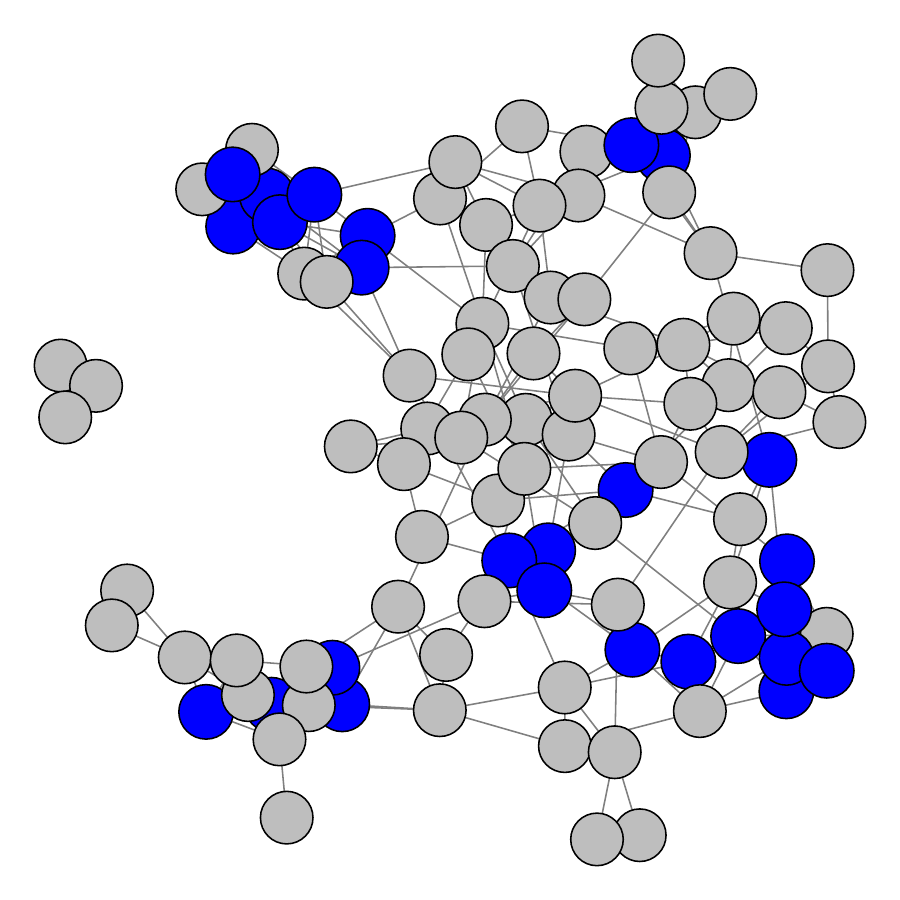}\\[-40pt]
\rot{{$\quad\quad\quad\quad\quad\quad$Female}}&
\includegraphics[width=3cm]{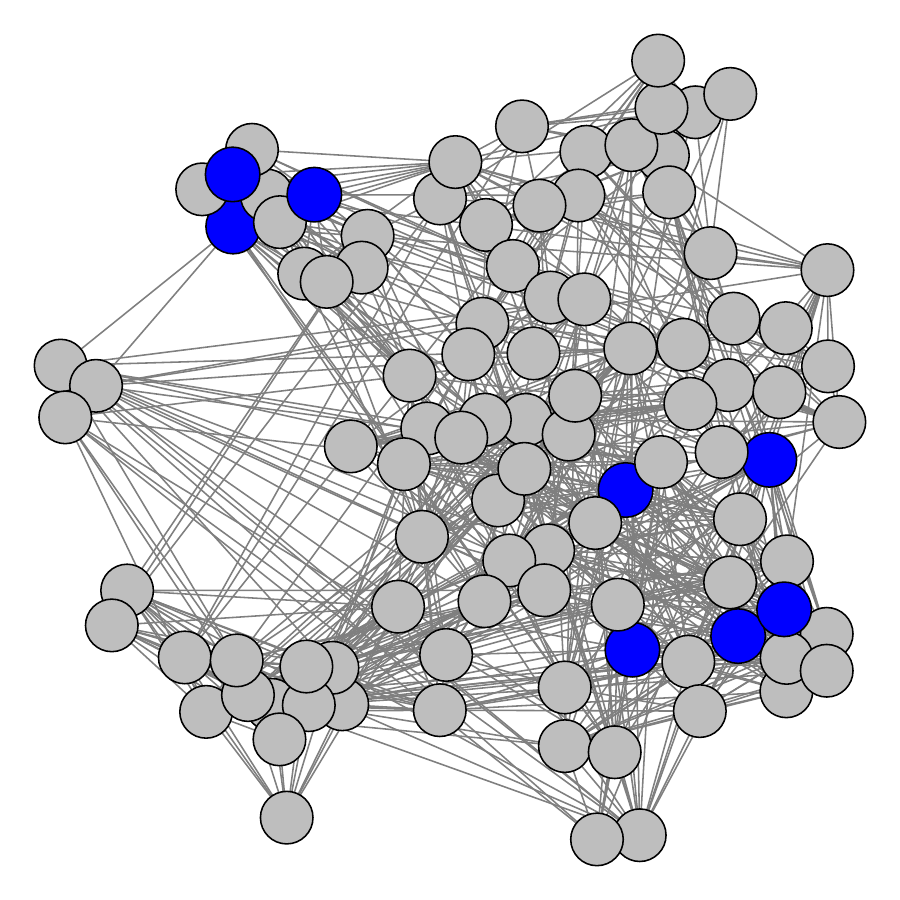}&
\includegraphics[width=3cm]{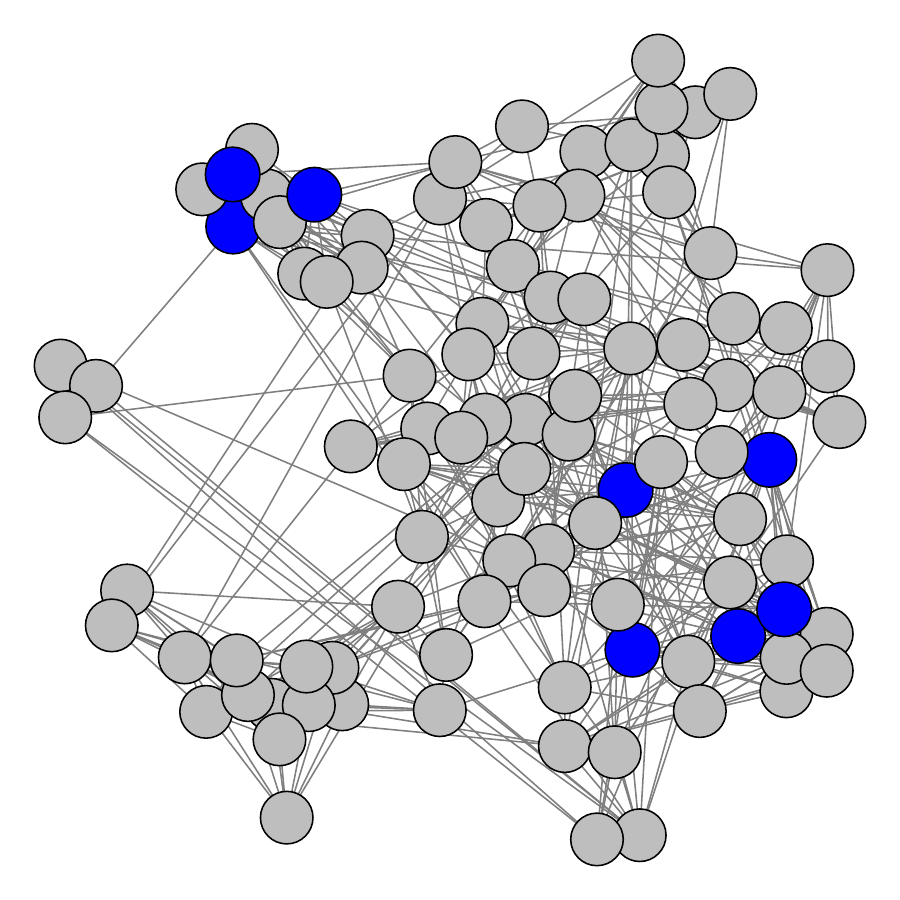}&
\includegraphics[width=3cm]{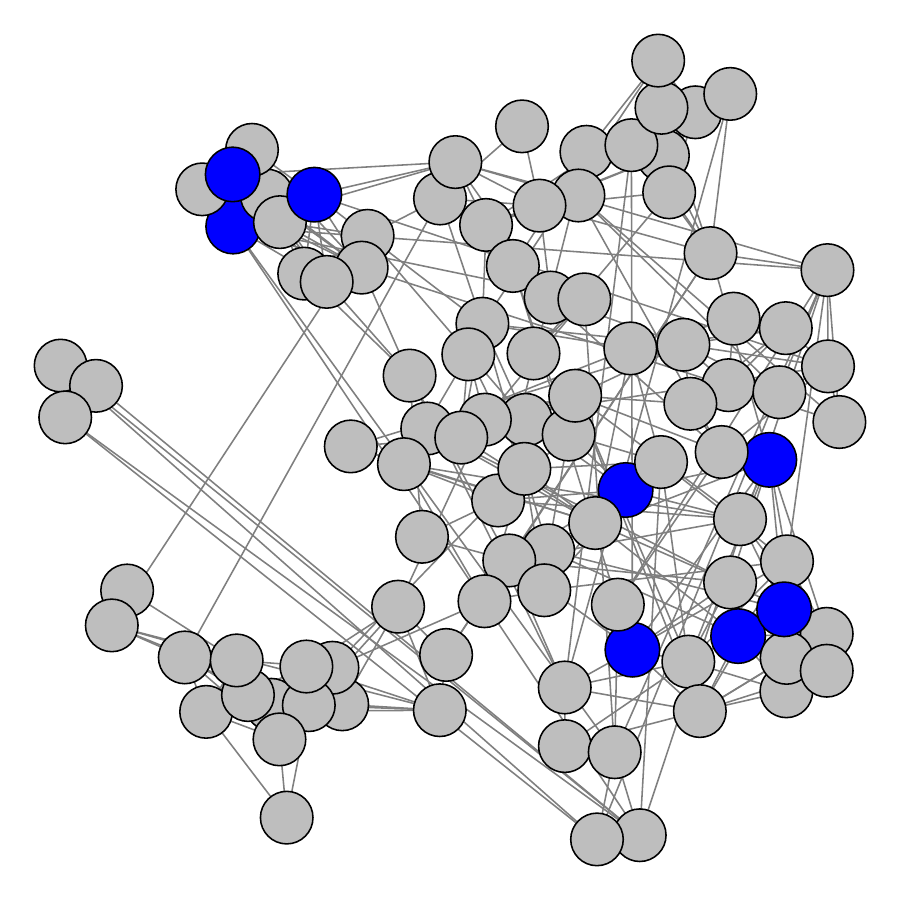}&
\includegraphics[width=3cm]{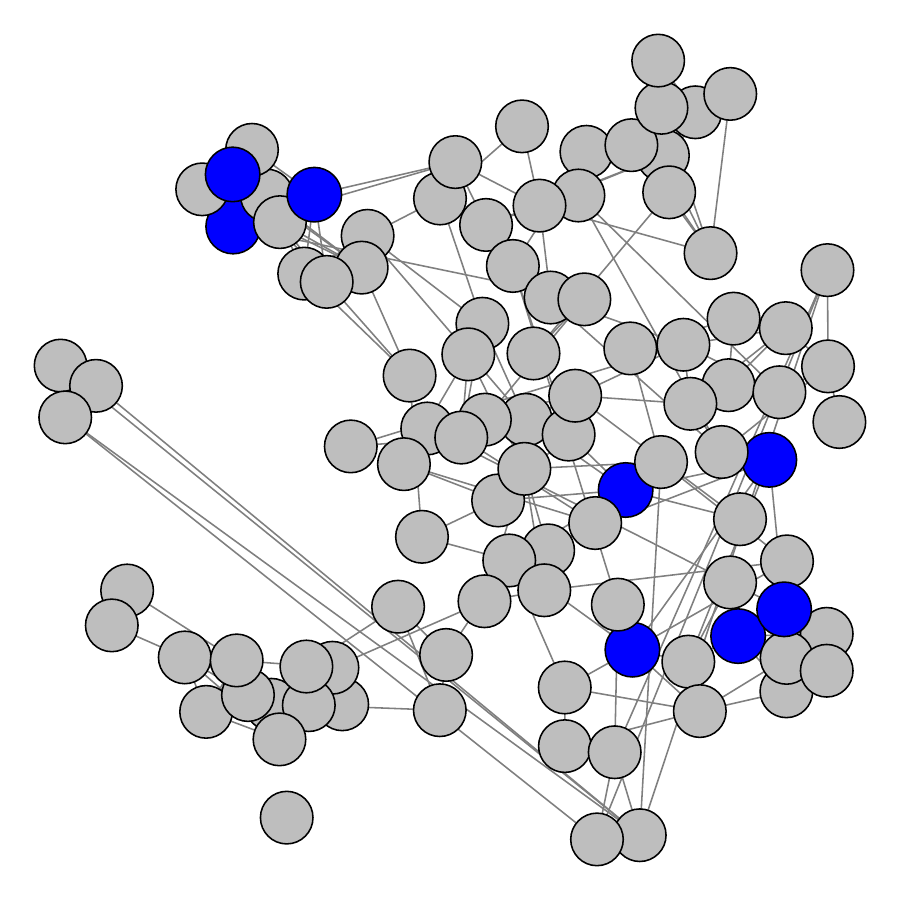}\\
\end{tabular}
\vskip-40pt
    \caption{The above graphs are based the estimated precision matrices (the left two plots). The adjacency matrices of the other six plots are based on the standardized estimated precision matrices but thresholded at $0.025, 0.05, 0.075$ respectively. Blue vertices represent the selected hub genes. }
    \label{fig:breast_gr}
\end{figure*}
\vskip-5pt
We also apply our method to the Genotype-Tissue Expression (GTEx) data studied in \cite{lonsdale2013genotype}. Beginning with a 2.5-year pilot phase, the GTEx project establishes a great database and associated tissue bank for studying the relationship between certain genetic variations and gene expressions in human tissues. The original dataset involves 54 non-diseased tissue sites across 549 research subjects. Here we only focus on analyzing the breast mammary tissues. It is of great interest to identify hub genes over the gene expression network.
% \vspace{}
% \begin{figure}
%     \centering
%     \includegraphics[width = 1.1\linewidth]{plots/simul/Comp_allgraph_qlevel0.2.pdf}
%     \caption{Caption}
%     \label{fig:compq0.2}
% \end{figure}
% \vspace{-0.1in}
% \begin{figure}[!htbp]
%     \centering
%     \includegraphics[width = 0.87\linewidth]{plots/real/breast_igraph.pdf}
%     \caption{The above graphs are based the estimated precision matrices (the left two plots). The adjacency matrices of the other six plots are based on the standardized estimated precision matrices but thresholded at $0.04, 0.06, 0.08$ respectively. Blue vertices represent the selected hub genes. }
%     \label{fig:breast_gr}
% \end{figure}

%\vspace{-0.1in}
First we calculate the variances of the gene expression data and focus on the top $100$ genes in the following analysis. The data involves $n=291$ samples for male individuals and $n=168$ samples for female individuals. The original count data is log-transformed and scaled. We then obtain the estimator of the precision matrix by the Graphical Lasso with 2-fold cross-validation. As for the hub node criterion, we set $k_\tau$ as the 50\% quantile of the node degrees in the estimated precision matrix. We run StarTrek filter with $2000$ bootstrap samples and nominal FDR level $q=0.1$ to select hub genes.
% \vskip-8pt
% for both the male and female datasets.
%The selected hub genes are showed in Figure \ref{fig:breast_pval}, we plot the sorted p-values ($\alpha_j,j\in [d]$ in Algorithm \ref{algo:startrek}) and the rejection line
%We also apply our method to the Genotype-Tissue Expression (GTEx) data studied in Lonsdale et al. (2013). The GTEx project began with a 2.5-year pilot phase to study tissue-specific gene expression and regulation. The data was collected from 54 non-diseased tissue sites across 549 research subjects. We focus mainly on analyzing the breast mammary tissues.
%\vspace{-0.3in}
\begin{figure}[!htbp]
    \centering
    \includegraphics[width = 0.96 \linewidth]{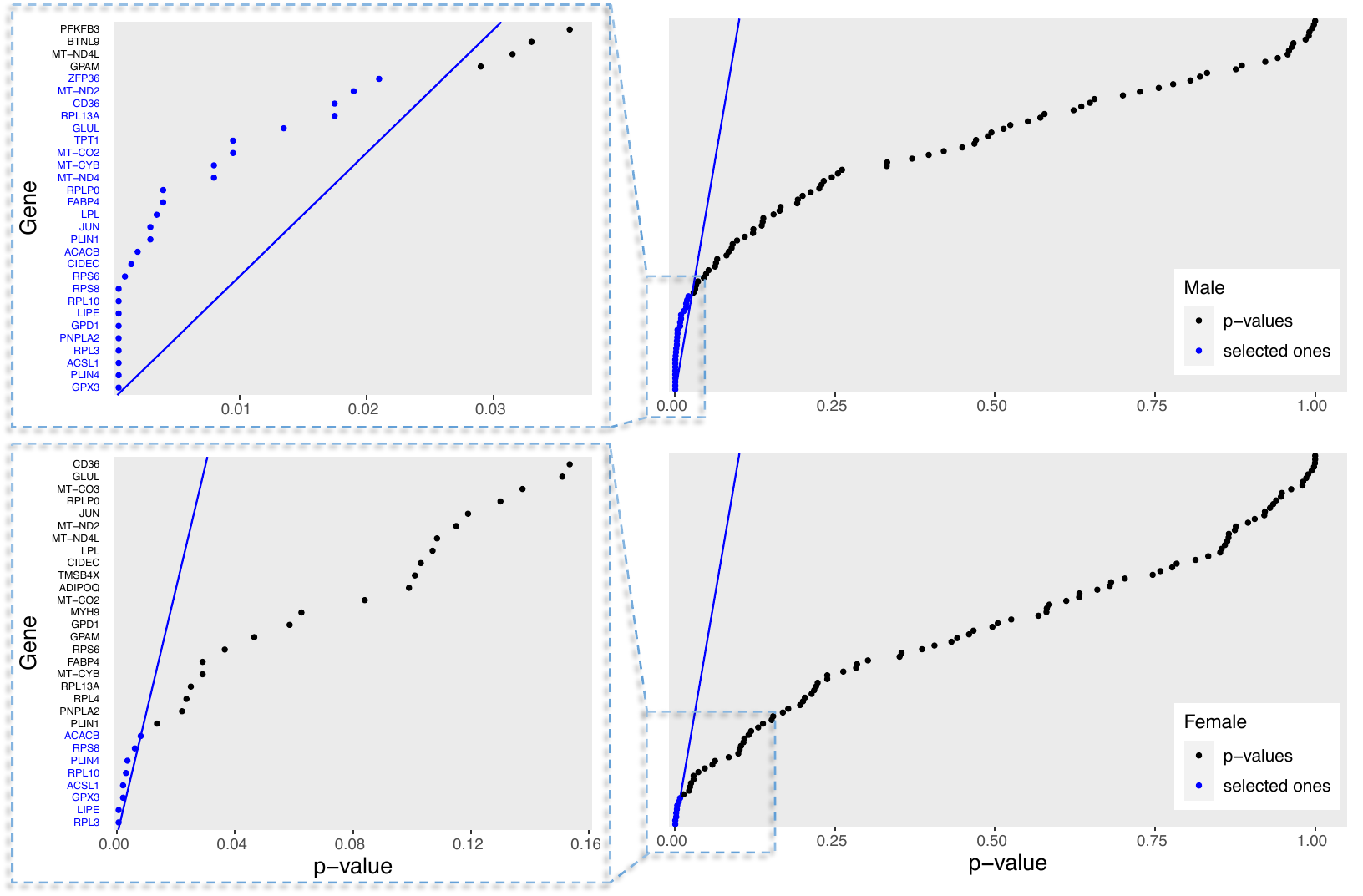}
    \caption{Plots of the sorted p-values ($\alpha_j,j\in [d]$) in Algorithm \ref{algo:startrek}. Those blue points correspond to selected hub genes. The blue line is the rejection line of the BHq procedure. The coordinates of the plots are flipped. We abbreviate the names of the 100 genes and only show selected ones with blue colored text. 
    % The upper panel and the lower panel are based on male and female data respectively.
    }
    \label{fig:breast_pval}
\end{figure}
%\hspace{-0.2cm}
% \vskip-15pt

Figure \ref{fig:breast_gr} shows that the selected hub genes by the StarTrek filter also have large degrees on the estimated gene networks (based on the estimated precision matrices). In Figure \ref{fig:breast_pval}, the results for male and female dataset agree with each other except that the number of selected hub genes using female dataset is smaller due to a much smaller sample size. The selected hub genes are found to play an important role in breast-related molecular processes, either as central regulators or their abnormal expressions are considered as the causes of breast cancer initiation and progression, see relevant literature in genetic research such as \cite{hellwig2016epsin,blein2015targeted,chen2016systematic,li2019association,lou2020overexpression,mohamed2014promoter,bai2019screening,sirois2019unique,marino2020upregulation,malvia2019study}. 
Therefore, our proposed method for selecting hub nodes can be applied to the hub gene identification problem. It may improve our understanding of the mechanisms of breast cancer and provide valuable prognosis and treatment signature. 

\section{Discussions}
%\label{sec:discuss}
In this paper, we have proposed a novel method to select the hub nodes in the graph with degrees larger than a certain thresholding level. To show the validity of the method, we prove Cram\'er-type Gaussian comparison bounds with two types of covariance matrix differences and Cram\'er-type deviation results of the Gaussian multiplier bootstrap procedure. The extension of our results to other bootstrap methods is interesting for future research. In specific, \cite{deng2020beyond} generalizes the Kolmogorov distance results of the Gaussian multiplier bootstrap \cite{chernozhukov2013gaussian} to the wild bootstrap and empirical bootstrap by proposing new comparison bounds and anti-concentration inequalities. Their techniques have the potential to be extended to the Cram\'er-type deviation bounds in the future. Moreover, \cite{chernozhuokov2022improved} showed a faster rate of Kolmogorov distance consistency of the Gaussian multiplier bootstrap and it could be extended to Cram\'er-type deviation bounds to improve the rates in our paper as well.

 % \begin{acks}[Acknowledgments]
 % The authors are grateful for the support of NSF DMS1916211, NIH R35 CA220523, NIH R01 ES32418, NIH U01CA209414.
 % \end{acks}  

\begin{acks}[Acknowledgments]
The authors are grateful for the support of NSF DMS1916211, NIH R35 CA220523, NIH R01 ES32418, NIH U01CA209414.
\end{acks}

\bibliographystyle{imsart-number} % Style BST file (imsart-number.bst or imsart-nameyear.bst)
\bibliography{reference}       % Bibliography file (usually '*.bib')

\appendix 

\newpage
% \setcounter{page}{1}

%\begin{center}
%\textit{\large Supplementary material to}
%\end{center}
%\begin{center}
%\title{\large StarTrek: Combinatorial Variable Selection with False Discovery Rate Control}
%\end{center}

%
%\author{Lu Zhang{\small${^1}$} \and Junwei Lu{\small${^2}$}}
%
%\vspace{2ex}
%\author{%
%\small
%    {\small${^1}$}Department of Statistics, Harvard University\\[0.5ex]
%    {\small${^2}$}Department of Biostatistics, Harvard University\\[2ex]%
%%    \today
%}

%\setcounter{page}{1}

\begin{supplement}
 \stitle{}
\sdescription{}
This document contains the supplementary material to the paper
  ``StarTrek:  Combinatorial Variable Selection with False Discovery Rate Control". Appendix \ref{app:pf:fdr} presents the proofs of the FDR control results. In Appendix \ref{app:pf:cramer_theory}, we provide the proofs of two types of Cram\'{e}r-type comparison bounds for Gaussian maxima. Appendix \ref{sec:gmb_theory} proves the Cram\'{e}r-type deviation bounds for the Gaussian multiplier bootstrap. In Appendix \ref{app:node_test_validity}, we establish the validity and a power result of our test on the degree of a single node. Appendix \ref{app:plots_tables} contains some plots and tables deferred from the main paper.
\section{Proofs for FDR control}
\label{app:pf:fdr}
In this section, we aim to prove Theorem \ref{thm:fdr_hub}.
% Throughout the proof for the Gaussian graphical models, we denote the standardized version of the one-step estimator \eqref{eq:one_step} by $\tdTheta_e$. 
%\jlmargin{}{change back to $\tilde \Theta^{d}$}
%\jlmargin{}{fix one notation for indicator $\Indrbr{\cdot}$}
In order to prove the theorem, we need Lemma \ref{lem:single_test} which is about the test of single node degree. Remark that this lemma proves the asymptotic validity of the test in Algorithm \ref{algo:skipdown} and provides a power analysis. The signal strength condition is only required for the power analysis part. To see why Lemma \ref{lem:single_test} is useful for establishing FDR control for our StarTrek procedure in Algorithm \ref{algo:startrek}, we notice the following equivalence:
\begin{equation} \label{eq:equiv}
\{ \psi_{j,\alpha}  = 1\}  = \{\alpha_j \le \alpha \},  
\end{equation}
where $\alpha$ is a given type-I error level, $\psi_{j,\alpha}$ is the test described in Algorithm \ref{algo:skipdown}, and $\alpha_j$ is defined in Algorithm \ref{algo:startrek}. \rev{Remark that $\{ \psi_{j,\alpha}  = 1\}  = \{\alpha_j \le \alpha \}$ implies the StarTrek Filter (Algorithm \ref{algo:startrek}) is essentially equivalent to the procedure described in \eqref{eq:BHq_alpha} where the test $ \psi_{j,\alpha}$ is calculated from the skip-down method (Algorithm \ref{algo:skipdown}).

To see why the equivalence in \eqref{eq:equiv} holds, we first} show $ \{\alpha_j \le \alpha \} \subset \{ \psi_{j,\alpha}  = 1\}$. Note
\begin{align}\nonumber
\{\alpha_j \le \alpha \} ~=& \bigcap_{1\le s \le  k_{\tau}}\{ \hat{c}^{-1}(\sqrt{n}|\tTheta_{j, (s)}|,E^{(s)}_j) \le \alpha\}\\
~=& \bigcap_{1\le s \le  k_{\tau}}\{ \sqrt{n}|\tTheta_{j, (s)}| \ge \hat{c}(\alpha, E^{(s)}_j) \} \label{eq:reject_set},
\end{align}
 where $E^{(s)}_j:=\{ (j,\ell):\ell \neq j, |\tTheta_{j\ell}| \le |\tTheta_{j, (s)}|\}$. The first equality is due to the definition of $\alpha_j$ and the second equality holds by the definition of $\hat{c}^{-1}$. Examining \eqref{eq:reject_set}, we immediately know $\sqrt{n}|\tTheta_{j, (1)}| \ge \hat{c}(\alpha, E^{(1)}_j)$ (here $E^{(1)}_j = E_0 = \{(k,j): k\in [d], k\ne j \}$), thus the edge corresponding to $\tTheta_{j, (1)}$ will be rejected in the first iteration of Algorithm \ref{algo:skipdown}. Regarding the edge corresponding to $\tTheta_{j, (2)}$, if $\sqrt{n}|\tTheta_{j, (2)}| \ge \hat{c}(\alpha, E^{(1)}_j)$, then it will be rejected in the first iteration, too. Otherwise, Algorithm \ref{algo:skipdown} enters the second iteration. Since \eqref{eq:reject_set} implies $\sqrt{n}|\tTheta_{j, (2)}| \ge \hat{c}(\alpha, E^{(2)}_j)$, we know the edge corresponding to $\tTheta_{j, (2)}$ must be rejected in the second iteration of Algorithm \ref{algo:skipdown}. Following this kind of argument, we are able to show that \eqref{eq:reject_set} implies that all those edges corresponding to $\{\tTheta_{j, (s)}, 1\le s \le  k_{\tau}\}$ will be rejected according to Algorithm \ref{algo:skipdown}. Since the number of rejected edges is at least $k_\tau$, we have $\psi_{j,\alpha}  = 1$. Second, we show $\{ \psi_{j,\alpha}  = 1\}  \subset \{\alpha_j \le \alpha \}$. If $\psi_{j,\alpha}  = 1\}$, we know the edges corresponding to $\{\tTheta_{j, (s)}, 1\le s \le  k_{\tau}\}$ will be rejected, which immediately imply $\sqrt{n}|\tTheta_{j, (1)}| \ge \hat{c}(\alpha, E^{(1)}_j)$. Regarding the edge corresponding to $\tTheta_{j, (2)}$, it must get rejected in the first two iterations of Algorithm \ref{algo:skipdown}. In either cases, we always have $\sqrt{n}|\tTheta_{j, (2)}| \ge \hat{c}(\alpha, E^{(2)}_j)$ due to $ E^{(2)}_j \subset E^{(1)}_j$ and the fact that $ \hat{c}(\alpha, E) \le  \hat{c}(\alpha, E') $ when $E \subset E'$. Finally, we establish \eqref{eq:equiv}.

%\lzmargin{explain the equivalence and why we need the lemma.}{}
\rev{
The validity of the single node test in Algorithm \ref{algo:skipdown} is mainly based on the following result which characterizes the accuracy of the approximate quantiles.
% Before going to the formal statements of Lemma \ref{lem:single_test}, we also present a key result below:
\begin{lemma}\label{lem:quantile} Let $\cU(M,s, r_0)$ denote the parameter space of precision matrices i.e., 
\begin{equation}\nonumber
\begin{aligned}
 \cU(M,s, r_0) &= \Big\{\bTheta \in \RR^{d \times d} \,\big|\,  \lambda_{\min}(\bTheta) \ge 1/r_0, \lambda_{\max}(\bTheta) \le r_0,  \max_{j \in [d]} \|\bTheta_{j}\|_{0} \le s, \|\bTheta \|_1 \le M \Big\}.
\end{aligned}
\end{equation}
Suppose $\bTheta \in \cU(M,s, r_0)$ and $(\log (dn))^7/n + s^2 (\log dn)^{4}/ {n} = o(1)$, for any edge set $E \subseteq \cV \times \cV$, we have
\begin{equation}\label{eq:quantile-valid2}
    \lim_{(n,d)\rightarrow \infty} \sup_{\bTheta \in \cU(M,s, r_0)}  \sup_{\alpha \in (0,1)} 
    \left|\PPP
    \left( \max_{e \in E}  \sqrt{n} |\tdTheta_{e}-\sTheta_e|> \hat{c}(\alpha, E) 
    \right) - \alpha
    \right|=0.
\end{equation}
where $\tdTheta$ defined in \eqref{eq:ggm_tdTheta} is the standardized version of the one-step estimator and $\sTheta$ denotes the
standardized true precision matrix $(\bTheta_{jk}/\sqrt{\bTheta_{jj} \bTheta_{kk}})_{j,k \in [d]}$.
\end{lemma}
}
% As mentioned in Section \ref{sec:startrek}, such a result
%In the above expression, the standardized precision matrix $(\bTheta_{jk}/\sqrt{\bTheta_{jj} \bTheta_{kk}})_{j,k \in [d]}$ is denoted by $\sTheta$. 
% Note that $\cU(M,s, r_0)$ denotes the parameter space of precision matrices and is defined as below:
\begin{lemma}\label{lem:single_test}
 Under the same conditions as Lemma \ref{lem:quantile}, given some $ 1\le j  \le d$, we have the following results.
%  \jlmargin{}{Bold $\Theta$?}
%  {\color{orange} Note that this signal strength condition is only needed for power results, not for type I error control.
%  }

\begin{enumerate}[(i)]
    \item \acc{Additionally, suppose for any $|\bTheta_{jk}|>0$, we also have $|\bTheta_{jk}| \ge c\sqrt{\log d/ n}$ for some constant $c>0$.} Under the alternative hypothesis $H_{1j}:  \jdeg \ge k_{\tau}$, we then have for any $\alpha \in (0,1)$,
    % If $\|\bTheta_j\|_{0} \ge k_{\tau}$, i.e., the alternative hypothesis $H_{1j}:\text{ degree of node } j \ge k$ is true, we have
    
 $$
   \lim_{(n,d)\rightarrow \infty} \PPP(\psi_{j,\alpha} = 1)= 1.
   %\text{ and } \lim_{(n,d)\rightarrow \infty} \PPP( \alpha_j \le 1/d) = 1;
 $$
    \item 
    % \st{Additionally, suppose for any $|\bTheta_{jk}|>0$, we also have $|\bTheta_{jk}| \ge c\sqrt{\log d/ n}$ for some constant $c>0$.} 
    Under the null hypothesis $H_{0j}:  \jdeg < k_{\tau} $, we have for any $u \in (0,1)$,
% $\|\bTheta_j\|_{0} \le k_{\tau}-1$, i.e., the null hypothesis $H_{0j}: \text{ degrees of node } j \le k$ is true, we have for any $u \in (0,1)$,
 $$
   \lim_{(n,d)\rightarrow \infty} \PP{\psi_{j,\alpha} = 1}\le \alpha. 
   %\text{ and } \lim_{(n,d)\rightarrow \infty} \PP{ \alpha_j \le u} = u.
 $$

\end{enumerate}
\end{lemma}
The proof of the above two lemmas are deferred to \acc{Appendix} \ref{sec:proof:lem:single_test}. 
% And the standardized precision matrix $(\bTheta_{jk}/\sqrt{\bTheta_{jj} \bTheta_{kk}})_{j,k \in [d]}$ is still denoted by $\bTheta$ without causing confusion. 
The maximum statistic used in our testing procedure takes the form of 
$T_{E}= \max_{(j,k) \in E}  \sqrt{n}|\tdTheta_{jk}| $. 
In our key proof procedure, we deal with the case where $E = \{(j,k):\bTheta_{jk}=0\}$. Since some of the results hold for general $E$, we will work with the general notations. Specifically, through out Appendices \ref{app:pf:thm:fdr_hub} and \ref{app:lems:fdr_hub}, we introduce the following notations: in order to approximate 
 \begin{equation}\label{eq:TE}
   T_E := \max_{(j,k) \in E}  \sqrt{n}\left|(  \dTheta_{jk}/\sqrt{\dTheta_{jj} \dTheta_{kk}} - {\bTheta_{jk}}/{\sqrt{\bTheta_{jj}\bTheta_{jk}}})\right|
 \end{equation}
 by the multiplier bootstrap process
 \begin{equation}\label{eq:TEcB}
   T_{E}^{\cB} :=\max_{(j,k) \in E}  \frac{1}{\sqrt{n~ \hat{\bTheta}_{jj}\hat{\bTheta}_{kk} }} \bigg| \sum_{i=1}^n \hat{\bTheta}_{j}^{\top} (\bX_i \bX_i^{\top} \hat{\bTheta}_{k}- \eb_k)\xi_i \bigg|,
 \end{equation}
we define two intermediate processes:
\begin{align}\label{eq:brTE}
  \breve{T}_E &:=  \max_{(j,k) \in E} \bigg| \frac{1}{\sqrt{n~ {\bTheta}_{jj}{\bTheta}_{kk}}}  \sum_{i=1}^n  {\bTheta}_{j}^{\top} (\bX_i \bX_i^{\top} {\bTheta}_{k}- \eb_k) \bigg|,
 \\ \label{eq:brTEcB}
 \breve{T}_{E}^{\cB} &:=  \max_{(j,k) \in E}\bigg| \frac{1}{\sqrt{n~ \bTheta_{jj}\bTheta_{kk}}}  \sum_{i=1}^n  {\bTheta}_{j}^{\top} (\bX_i \bX_i^{\top} {\bTheta}_{k}- \eb_k)\xi_i \bigg|.
\end{align}

\subsection{Proof of Theorem \ref{thm:fdr_hub}}
\label{app:pf:thm:fdr_hub}
\begin{proof}[Proof of Theorem \ref{thm:fdr_hub}]
\label{pf:thm:fdr_hub}
Given some $j\in\cH_{0}$, denote $N_{0j}=\{(j,k):\bTheta_{jk}=0\}$. 
By  the first part of Lemma \ref{lem:single_test}, 
%and Markov's inequality
we have \acc{$\forall j \in \cB$,}
\begin{equation}
 \label{eq:fact-e1}
 \acc{\PP{\psi_{j,\alpha} = 1} > 1- 3/d^2,}
%\frac{\sum_{j\in \cB}{ \psi_{j,\alpha}}  }{|\cB|} \rightarrow 1 \quad \text{in probability},
\end{equation}
when $\alpha = \Omega (1/d)$, where $\cB:=\{j\in \cH_{0}^{c}:\forall k \in \text{supp}(\bTheta_j),|\bTheta_{jk}|>c\sqrt{{\log d}/{n}}\}$. 
%%%% directly applying Markov, and we have the numerator is less than or equal to |\cB|, thus E|... - 1| = 1 - E(...)
%Then by the definition of $\hat{\alpha}$ in the StarTrek filter Algorithm \ref{algo:startrek}, we have
Note that we have
%together with the definition of signal strength condition \ref{cond:sparse_strength}, we have 
 \begin{equation}
 \label{eq:alpha_interval}
     \mathbb{P}\left(\frac{q|\cB|}{d}\le \hat{\alpha} \le 1\right) \ge \PP{ \frac{ {q|\cB|}/{d}\cdot d }{\sum_{j\in [d]} \psi_{j,{q|\cB|}/{d}}}\le q  } = \PP{ \frac{ |\cB| }{\sum_{j\in [d]} \psi_{j,{q|\cB|}/{d}}}\le 1  }  \acc{\ge 1 - 3/d },
     %% with eps = 1/2 or some othe choices
 \end{equation}
where the first inequality is by \eqref{eq:BHq_alpha} and \acc{the last inequality is due to  $q\frac{|\cB|}{d} = \Omega(1/d)$, \eqref{eq:fact-e1} and the the union bound}. Rewrite the FDP (with $\hat \alpha$) as
\begin{align} \nonumber
   \mathrm{FDP}(\hat \alpha) 
    :=~ \frac{\sum_{j\in \cH_0}{ \psi_{j,\hat \alpha} } }{\max \left\{   1, \sum_{j\in [d]} \psi_{j,\hat \alpha}  \right\}} %\\ \nonumber
  =~  \frac{ \hat \alpha d }{\max \left\{   1, \sum_{j\in [d]} \psi_{j,\hat \alpha}   \right\}}\cdot 
  {\frac{\sum_{j\in \cH_0} \psi_{j,\hat \alpha} }{ d_0 \hat \alpha  }}\cdot \frac{d_0}{d},
\end{align}
and notice that
$$
 \frac{ \hat \alpha d }{\max \left\{   1, \sum_{j\in [d]} \psi_{j,\hat \alpha}  \right\}}\cdot \frac{d_0}{d} \le \frac{q d_0}{d} \le q.
$$
Then it suffices to control the $\mathrm{FDP}(\hat \alpha)$ by dealing  with ${\rbr{\sum_{j\in \cH_0}\psi_{j,\hat \alpha}}}/{ d_0 \hat \alpha  }$. By \eqref{eq:alpha_interval}, the FDP control problem is now reduced to showing 
%\jlmargin{}{what does in prob mean?} 
\begin{equation} \nonumber
% \label{eq:fdp_term1}
    \sup_{\alpha\in [\alpha_L,1]}
    \frac{\sum_{j\in\cH_{0}}\psi_{j,\alpha}}{d_{0}\alpha} \le 1 +\smallop,
    %~~~ \text{in probability}.
\end{equation}
where $\alpha_L =q|\cB|/d$,
%\begin{equation} \nonumber
%% \label{eq:fdp_term1}
%    1 \ge  \sup_{\alpha\in [\alpha_L,1]}
%    \frac{\sum_{j\in\cH_{0}}\psi_{j,\alpha}}{d_{0}\alpha} - \smallop.
%    %~~~ \text{in probability}.
%\end{equation}
By \eqref{eq:monotone} in the proof of the second part of Lemma \ref{lem:single_test}, $\psi_{j,\alpha}=1$ implies that $\max_{e \in N_{0j} } \sqrt{n}|\tdTheta_e -\sTheta_e| \ge \hat c(\alpha, N_{0j})$, where $N_{0j}=\{(j,k) : \bTheta_{jk}=0\}  =\{(j,k) : \sTheta_{jk}=0\}$. Therefore, we have
\[
\frac{\sum_{j\in\cH_{0}}\psi_{j,\alpha}}{d_{0}\alpha} \le \frac{\sum_{j\in \cH_{0}} \Indrbr{\max_{e \in N_{0j} } \sqrt{n}|\tdTheta_e -\sTheta_e| \ge \hat c(\alpha, N_{0j}) } }{d_{0}\alpha}.
\]
Hence it suffices to prove that 
\begin{equation}\label{eq:fdp_term2}
    \sup_{\alpha \in [\alpha_L,1]} 
    \left |
    \frac{\sum_{j\in \cH_{0}}\Indrbr{\max_{e \in N_{0j} } \sqrt{n}|\tdTheta_e -\sTheta_e| \ge \hat c(\alpha, N_{0j}) }}{d_{0}\alpha} - 1 \right |\rightarrow 0 ~~~ \text{in probability}.
\end{equation}
% where $\tilde{\bTheta}_{jk} = {{\bTheta}_{jk}}/{\sqrt{{\bTheta}_{jj}{\bTheta}_{kk}}}$. 
In order to prove \eqref{eq:fdp_term2}, we construct a discrete grid of the interval $[\alpha_L,1]$. The number of grid points is denoted by $\lambda_d$ and will be decided later. First, we let $t_1 :=  \hat c(1, N_{0j}) = 0$, $t_{\lambda_d} := \hat c(\alpha_L, N_{0j})$. Here $\hat c(\alpha_L, N_{0j}) = \inf \left\{ t\in \RR : \PPP_\xi \left( T^{\cB}_{N_{0j}} \le t  \right) \ge 1-\alpha   \right\}$ is the quantile based on the Gaussian multiplier bootstrap process and depends on the data $\bX$. Note that the involving random vectors in the Gaussian multiplier bootstrap process are Gaussian conditioning on the data $\bX$ and have bounded variances with probability growing to $1$. Since $\alpha_L = \Omega(1/d)$, then by the maximal inequalities for sub-Gaussian random variables (Lemma 5.2 in \cite{van2014probability}), we have $t_{\lambda_d} = O(\sqrt{\log d})$ with probability growing to $1$. Second, note there exists $h_d$ such that $h_d t_{\lambda_d} = o(1)$ and $ t_{\lambda_d}/h_d = O(\log d)$. Based on such $h_d$, we construct equally spaced sequences $\{t_m\}_{m=1}^{\lambda_d}$ over the range $[t_1, t_{\lambda_d}] = [0, t_{\lambda_d}]$ with $t_{m} - t_{m-1} = h_d$. Then by setting $\alpha_m$ such that $t_{m} = \hat c(\alpha_m, N_{0j})$, we obtain a discrete grid $\{\alpha_m\}_{m=1}^{\lambda_d}$ of the interval $[\alpha_L,1]$. For such $\alpha_m, 1\le m \le \lambda_d$, we have
%, conditioning on the data $\bX$,
%In order to prove \eqref{eq:fdp_term2}, we construct a discrete grid of the interval $[\alpha_L,1]$: $1 = \alpha_{1} \ge \alpha_2 \ge \ldots \ge \alpha_{\lambda_{d}}=\alpha_L$ such that the quantiles of these grid points:  $t_{m} = \hat c(\alpha_m, N_{0j}) = q(\alpha_{m};T^{\cB}_{ N_{0j}})$, recalling that the Gaussian multiplier bootstrap process $T^{\cB}_{ N_{0j}}$ is defined in \eqref{eq:TEcB}, satisfy
%\[
%   \text{$t_{m}-t_{m-1}=h_{d}$ for some sequence $h_d = o(1/\sqrt{\log d})$, for all $1 \le m \le \lambda_d$. }
%\]
%Here the length of the grids $\lambda_d$ needs to be specified later to guarantee that the condition above is satisfied.
%Notice \jlmargin{$t_{\lambda_d} =O(\sqrt{\log d})$ since $\alpha_L = \Omega(1/d)$}{Why?}, we can choose $h_d$ decreasing (to zero) just a little bit faster than $1/\sqrt{\log d}$ such that $ h_d t_{\lambda_d} = o(1)$. \jlmargin{}{Still unclear how to construct $\alpha_m$}
%\jlmargin{Notice, $t_{\lambda_d} =O(\sqrt{\log d})$ \fbox{why} We can choose $h_d$ growing just a little slower than $1/\sqrt{\log d}$  such that $t_{\lambda_d}/h_d = O(\log d)$.}{}
% (Theorem $2.1$ in \cite{arun2018cram} $2$ or $3$)
%By applying \jlmargin{anti-concentration bound for Gaussian maxima}{cite} to $T^{\cB}_{ N_{0j}}$, 
%We have 
\begin{align} \nonumber
  \max_{1\le m \le \lambda_d} \left|\frac{\alpha_{m-1}}{\alpha_{m}}-1\right| 
  &= \max_{1\le m \le \lambda_d} \left|\frac{\PP{T^{\cB}_{ N_{0j}} > t_{m-1}}}{\PP{T^{\cB}_{ N_{0j}} > t_{m}}} - 1 \right| \\ \label{eq:alphas_ratio_bound}
  &\le \max_{1\le m \le \lambda_d} C''(t_m-t_{m-1})(t_m+1) \exp(C' (t_m-t_{m-1})(t_m+1) ) = o(1)
\end{align}
with probability growing to 1, where the first equality holds by the definition of $\alpha_m$, the first inequality holds due to part $2$ and $3$ of Theorem 2.1 in \cite{arun2018cram} (by first choosing $r-\epsilon, r+\epsilon$ in part $3$ to be $t_{m-1}, t_m$ respectively then letting $r-\epsilon, r$ in part $2$ to be $t_{m-1},t_{m}$ respectively). And the right hand side of the inequality is $o(1)$ since {$(t_{m}-t_{m-1})t_m \le h_d t_{\lambda_d} =o(1)$} with probability growing to 1.
%\jlmargin{where the first equality is because $t_m$ is the $\alpha_m$-quantile of $T_N$, the first inequality is due to part 3 of Theorem 2.1, cite{XX}, and the last inequality is due to $h_d t_m = XX = o(1)$}{}.
%uniformly over $m =1,\cdots,\lambda_d$. 

Denote  
%\begin{equation}
   $ I_j(\alpha) = \Indrbr{\max_{e \in N_{0j} } \sqrt{n}|\tdTheta_e -\sTheta_e| \ge \hat c(\alpha, N_{0j})}$.
%\end{equation}
Then given $\alpha_{m}\le \alpha \le \alpha_{m-1}$, for $m=1,\cdots,\lambda_d$, we have
\begin{eqnarray} \label{eq:squeeze}
       \frac{\sum_{j\in \cH_{0}} I_j(\alpha_{m})}{d_{0}\alpha_{m}}\cdot \frac{\alpha_{m}}{\alpha_{m-1}} \le \frac{\sum_{j\in \cH_{0}}  I_j(\alpha) }{d_{0}\alpha} 
       \le \frac{\sum_{j\in \cH_{0}} I_j(\alpha_{m-1})}{d_{0}\alpha_{m-1}}\cdot \frac{\alpha_{m-1}}{\alpha_{m}}.
\end{eqnarray}
 Hence by \eqref{eq:alphas_ratio_bound} and \eqref{eq:squeeze}, showing \eqref{eq:fdp_term2} is reduced to proving
 \begin{equation}\label{eq:seq_def}
     \max_{1\le m \le \lambda_{d}} 
    \left |
    \frac{\sum_{j\in \cH_{0}} I_j(\alpha_{m})}{d_{0}\alpha_{m}} - 1 \right |\rightarrow 0, ~~~ \text{in probability}.
\end{equation}
Then it suffices to show that, for any $\epsilon>0$, 
$$
\PP{
    \max_{1\le m \le \lambda_{d}} 
    \abr{
    \frac{\sum_{j\in \cH_{0}} I_j(\alpha_{m})}{d_{0}\alpha_{m}} - 1 
    }
    \ge \epsilon } \rightarrow 0.
$$
By the union bound argument and Chebyshev's inequality, we have
\begin{align} \nonumber
    &~~~~ 
    ~\PP{
    \max_{1\le m \le \lambda_{d}} 
    \abr{
    \frac{\sum_{j\in \cH_{0}} I_j(\alpha_{m})}{d_{0}\alpha_{m}} - 1 
    }\ge \epsilon}\\ \nonumber
    &\le ~\sum_{m=1}^{\lambda_{d}} 
   \PP{\abr{
    \frac{\sum_{j\in \cH_{0}} I_j(\alpha_{m})}{d_{0}\alpha_{m}} - 1 } \ge \epsilon}\\  \label{eq:seq_moments}
    &\le~  \sum_{m=1}^{\lambda_{d}} \frac{\EE{\sum_{j\in \cH_{0}} I_j(\alpha_{m}) - d_{0}\alpha_{m}}^2}{\epsilon^{2} d_{0}^{2}\alpha_{m}^{2}}\\ %\label{eq:seq_III_12}
    \begin{split}
    & = ~ \underbrace{\sum_{m=1}^{\lambda_{d}} \frac{\sum_{j\in \cH_{0}}\Var{I_j(\alpha_{m}) - d_{0}\alpha_{m}}}{\epsilon^{2} d_{0}^{2}\alpha_{m}^{2}}}_{\mathrm{III}_1} + \underbrace{ 
    \sum_{m=1}^{\lambda_{d}} \frac{\rbr{\EE{\sum_{j\in \cH_{0}} I_j(\alpha_{m}) - d_{0}\alpha_{m}}}^2}{\epsilon^{2} d_{0}^{2}\alpha_{m}^{2}} }_{\mathrm{III}_2} \\ \label{eq:seq_III_3}
    &~~~~~ +\underbrace{ \sum_{m=1}^{\lambda_{d}}\frac{ 
    \sum_{j_1,j_2 \in \cH_0, j_1 \ne j_2} \Cov{I_{j_1}(\alpha_m)}{ I_{j_2}(\alpha_m)}
    }{\epsilon^{2} d_{0}^{2}\alpha_{m}^{2}}}_{\mathrm{III}_3}.
    \end{split}
\end{align}
By Lemma \ref{lem:moments_Ij} and Lemma \ref{lem:crossterm_Ij}, we have 
\begin{eqnarray}\nonumber
\mathrm{III}_{1} + \mathrm{III}_{2}  +\mathrm{III}_{3}&\le& \frac{C't_{\lambda_d} }{\epsilon^2h_d } 
\left( \frac{d}{d_0|\cB|}+
\eta^2(d,n)\right) 
+ \frac{C'''d}{\epsilon^2 |\cB|d_0}  \cdot\frac{t_{\lambda_d}}{h_d  } \cdot
\left(1 + \eta(d,n) d_0+\frac{|S|\log d}{d_0 p}
\right)\\
&\le& \frac{C_1 t_{\lambda_d} \eta^2(d,n)}{\epsilon^2h_d } + \frac{C_2}{\epsilon^2 \rho d_0}  \cdot\frac{t_{\lambda_d}}{h_d  } \cdot \left(1 + \eta(d,n) d_0+\frac{|S|\log d}{d_0 p}
\right),
\label{eq:III123_bound}
\end{eqnarray}
where we substitute $\zeta_1 = {s(\log d)^2}/\sqrt{n}$, $\zeta_2 = 1/d^2$ and $\alpha_L = q|\cB|/d = \Omega  (\rho) $ in $\eta(d,n,\zeta_1,\zeta_2, \alpha_L)$ of Lemma \ref{lem:moments_Ij} and note $|\cB|>0$ then obtain the concise form $\eta(d,n)$ below,
\begin{eqnarray} \nonumber
     \eta(d,n) = \frac{(\log d)^{19/6}}{n^{1/6}} + \frac{ (\log d)^{11/6}}{\rho^{1/3} n^{1/6}} + \frac{{s(\log d)^{3}}}{n^{1/2}} + \frac{1}{d}.
%   \eta(d,n,\zeta_1,\zeta_2, \alpha_L)
    % &=& \eta_1(d,n,\alpha_L)+\zeta_1 q(\alpha_{L}/2;T_{\bZ})+\frac{\zeta_2}{\alpha_L}\\
%     &=& \frac{(\log d)^{8/3}}{n^{1/6}}\cdot \sqrt{\log d} ~\bigvee
% \frac{(\log d)^{{5}/{8}}}{n^{1/8}}\cdot {\log d} + \frac{{s(\log d)^{5/2}}}{\sqrt{n}} +     \frac{1}{d|\cB|}
\end{eqnarray}
%Note that in above terms, $t_{\lambda_d} = q(\alpha_L;T^B_{N_{0j}}) = O\rbr{\sqrt{\log d}}$ for any $j$, $h_d$ is some sequence satisfying $h_d = o(1/\sqrt{\log d})$. 
%As for $\eta(d,n,\zeta_1,\zeta_2, \alpha_L)$, we apply the result in Lemma \ref{lem:quantile}, i.e., 
%$\zeta_1 = {s(\log d)^2}/\sqrt{n}$, $\zeta_2 = 1/d^2$. 
% Simply bound $q(\alpha_{L}/2;T_{\bZ}) = O\rbr{\sqrt{\log d}}$, 
%thus obtain the following (with constants ignored)
%by the definition of $\alpha_L$ and the assumption $|\cB|>0$ (which holds under the assumption on $\rho$ in Theorem \ref{thm:fdr_hub}). 
Recall that $t_{\lambda_d} = q(\alpha_L;T^B_{N_{0j}}) = O\rbr{\sqrt{\log d}}$ with probability growing to 1 and $t_{\lambda_d}/h_d = O(\log d)$.
Under {Assumption \ref{asp:tradeoff_fdp}},
%the stated condition in Theorem \ref{thm:fdr_hub},
we have 
\begin{eqnarray} \nonumber
 \frac{ \log d }{\rho}\rbr{
    \frac{(\log d)^{19/6}}{n^{1/6}} + \frac{ (\log d)^{11/6}}{\rho^{1/3} n^{1/6}} + 
% \frac{(\log d)^{{13}/{8}}}{n^{1/8}} 
    \frac{{s(\log d)^{3}}}{{n}^{1/2}}
   } = o(1),\quad  \frac{\log d}{\rho d_0}  +  \frac{({\log d})^{2}|S|}{\rho d_0^2 p} = o(1), %\frac{ (\log d)^{11/6}}{\rho^{1/3} n^{1/6}} 
\end{eqnarray}
and thus $\mathrm{III}_{1} + \mathrm{III}_{2}  +\mathrm{III}_{3} = o(1)$ with probability growing to 1.
% and
% \begin{eqnarray}
%  \RN{4}_2 =  \frac{(\sqrt{\log d})^{\delta +4}|S|}{\rho d_0^2 p}
% \end{eqnarray}
%Then rearranging the terms in \eqref{eq:III123_bound}, we can derive $ \mathrm{III}_{1} + \mathrm{III}_{2}  +\mathrm{III}_{3} = o(1)$ when $\RN{4}_1 + \RN{4}_2=o(1) $.
% \RN{4}_1 + \RN{4}_2 $ up to some factor of ${C_3}/{\epsilon^2}$, where
Therefore, we have proved $\eqref{eq:fdp_term2}$, and finally establish the FDP control result below,
% $$
%   \lim_{(n,d)\rightarrow \infty}  \mathrm{FDP} \le q\frac{d_{0}}{d} ~~~ \text{in probability}.
% $$
$$
\mathrm{FDP}(\hat \alpha) \le q\frac{d_0}{d} + \smallop.
$$
%$$
%q\frac{d_0}{d} \ge \mathrm{FDP}(\hat \alpha)  - \smallop.
%$$
In order to establish FDR control, it remains to check the uniformly integrability of the random variable sequence in \eqref{eq:seq_def}. Note for a sequence of random variable $R_1,R_2, \cdots$, we have $
\sup_{n} \EE{|R_n|\Indrbr{|R_n| > x}} \le x^{-1} \sup_{n} \EE{R_n^2}$ by Markov's inequality. Then to show the uniform integrability of the random variable sequence $\{R_n \}_{n=1}^\infty$, where $R_n = \max_{1\le m \le \lambda_{d}} 
    \abr{
    \frac{\sum_{j\in \cH_{0}} I_j(\alpha_{m})}{d_{0}\alpha_{m}} - 1 
    }
    $, it suffices to show $\sup_n \EE{R^2_n} < \infty$. Indeed, we have
    \begin{eqnarray*} \nonumber
    &~~&  \sup_{n}\EE{ \rbr{ \max_{1\le m \le \lambda_{d}} 
    \abr{
    \frac{\sum_{j\in \cH_{0}} I_j(\alpha_{m})}{d_{0}\alpha_{m}} - 1 
    } }^2}\\
    &\le & \sup_{n} \sum_{m=1}^{\lambda_{d}} \frac{\EE{\sum_{j\in \cH_{0}} I_j(\alpha_{m}) - d_{0}\alpha_{m}}^2}{ d_{0}^{2}\alpha_{m}^{2}} \\
    & = & \sup_{n} \epsilon^2 ( \mathrm{III}_{1} + \mathrm{III}_{2}  +\mathrm{III}_{3}). 
        \end{eqnarray*} 
Since $ \mathrm{III}_{1} + \mathrm{III}_{2}  +\mathrm{III}_{3} = o(1)$ with probability growing to 1, we immediately have  $\sup_n \EE{R^2_n} < \infty$, thus finally establish the FDR control result:
$$
\lim_{(n,d)\rightarrow \infty}\mathrm{FDR} \le q\frac{d_{0}}{d}.
$$
\end{proof}
% In Theorem \ref{thm:fdr_hub}, we have showed $\mathrm{III}_{1}$ and $\mathrm{III}_{2}$ are negligible. 
% and $V$ which only differs from $U$ by the fact that its upper block covariance matrix is purely zero, we have for case $1$,
% \begin{equation}
% \left|
% \frac{\sum_{j,m \in \cH_{0}, \bTheta_{jm}=0} \cov(G_j(\alpha),G_m(\alpha))}{\epsilon^{2}d_{0}^{2}\alpha^{2}} 
% \right| 
% \le  C_{\bTheta}\frac{k_{\tau}^{2} {\log d}}{p\epsilon^{2}}
% \end{equation}
% where $C_{\bTheta}$ is the constant, only depending on the constant in $\cU(M,s, r_0)$.

% also assume there exists a certain number of independent blocks, whose cardinality is $p$, then we have
% \begin{equation}
%     \left|\frac{\mathbb{P}(T_{U} > t)}{\mathbb{P}(T_{V} > t)}-1\right|\le C_{\bSigma} ???
% \end{equation}
% where $0\le t \le 2\sigma_{\max}\sqrt{\log d}$ (Note that we only need to prove the case where $\frac{\omega_d}{d} \le \alpha \le 1$. When choosing $\alpha = \frac{2}{d}$, we can find the upper bound of $t$ by considering the independence case) and $C_{\bSigma}$

\subsection{Ancillary lemmas for Theorem \ref{thm:fdr_hub}}
\label{app:lems:fdr_hub}
\begin{lemma}\label{lem:moments_Ij}
Recalling the definitions of $\mathrm{III}_{1}, \mathrm{III}_{2}$ in \eqref{eq:seq_III_3}, we have
$$
\mathrm{III}_{1} + \mathrm{III}_{2}\le  \frac{C't_{\lambda_d} }{\epsilon^2h_d } 
\left(
\frac{1}{\rho d_0}+
\eta^2(d,n,\zeta_1,\zeta_2, \alpha_L)  \right),
%\frac{C'}{\epsilon^2}\cdot \frac{q(\alpha_L;T^B)\eta_2(d,n,\alpha_L) }{h_d} +
$$
where $\eta(d,n,\zeta_1,\zeta_2, \alpha_L)= O \big(\frac{(\log d)^{19/6}}{n^{1/6}} + \frac{ (\log d)^{11/6}}{n^{1/6}\alpha_L^{1/3}} + \zeta_1 \log d + \frac{\zeta_2}{\alpha_L}\big)$ with $\zeta_1 = {s(\log d)^2}/\sqrt{n} $, $\zeta_2 = 1/d^2$.
\end{lemma}
\begin{proof}[Proof of Lemma \ref{lem:moments_Ij}]
First note the definitions of $T_{E}, \breve{T}_E, T_{E}^{\cB}$ and $\breve{T}_{E}^{\cB}$ in \eqref{eq:TE}, \eqref{eq:brTE}, \eqref{eq:TEcB} and \eqref{eq:brTEcB} respectively, then we apply Proposition \ref{prop:approx_cramer_gmb} to $T=T_{E}, T_{\bY} = \breve{T}_E, T^\cB = T_{E}^{\cB}, T_{\cW} = \breve{T}_{E}^{\cB}$ with $E=N_{0j}$. And we can find the terms $\zeta_1, \zeta_2$ in \eqref{eq:approx_cond1}, \eqref{eq:approx_cond2} to be ${s(\log d)^2}/\sqrt{n}, 1/d^2$ respectively, due to \eqref{eq:zeta12_rate1} and \eqref{eq:zeta12_rate2} (i.e., the bound on the differences $T_E -T_0$, $T_{E}^{\cB}-T_0^{\cB}$) in the proof of Lemma \ref{lem:quantile}. Thus we have 
\begin{equation}\label{eq:fdp_cgmb}
    % \sup_{\alpha \in [\alpha_{L},1]} 
    \left |
    \frac{\mathbb{P}(\max_{e \in N_{0j} } \sqrt{n}|\tdTheta_e -\sTheta_e| \ge \hat c(\alpha, N_{0j}) )}{\alpha} - 1 \right | = \eta(d,n,\zeta_1,\zeta_2, \alpha_L),
\end{equation}
where $ \sTheta_e= 0, e\in N_{0j} $ and $\eta(d,n,\zeta_1,\zeta_2, \alpha_L)= O\rbr{\frac{(\log d)^{19/6}}{n^{1/6}} + \zeta_1 \log d + \frac{\zeta_2}{\alpha_L}}$ with $\zeta_1 = {s(\log d)^2}/\sqrt{n}$, $\zeta_2 = 1/d^2$. Recalling the definition of $\mathrm{III}_{2}$ in \eqref{eq:seq_III_3}, we have
$$
\mathrm{III}_{2} = \sum_{m=1}^{\lambda_{d}} \frac{\rbr{\EE{\sum_{j\in \cH_{0}} I_j(\alpha_{m}) - d_{0}\alpha_{m}}}^2}{\epsilon^{2} d_{0}^{2}\alpha_{m}^{2}},
$$
where $I_j(\alpha) = \Indrbr{\max_{e \in N_{0j} } \sqrt{n}|\tdTheta_e - \sTheta_e| \ge \hat c(\alpha, N_{0j})}$. Note that $\alpha_m \in [\alpha_L, 1],$ $\forall~ 1\le m \le \lambda_d$, then we arrive at the following bound
\begin{equation}
\label{eq:III2}
\mathrm{III}_{2} \le \frac{\lambda_d}{\epsilon^2}\cdot \eta^2(d,n,\zeta_1,\zeta_2, \alpha_L) \le \frac{t_{\lambda_d}}{\epsilon^2 h_d}\cdot \eta^2(d,n,\zeta_1,\zeta_2, \alpha_L)
\end{equation}
% we can bound $\mathrm{III}_{2}$ by $\frac{\lambda_d}{\epsilon^2}\cdot\eta_2(d,n,\alpha_L)$. Since $\lambda_{d}\sim \frac{t_{\lambda_d}}{h_{d}}\le \frac{q(\alpha_{L};T_{Z})}{h_{d}} = (\sqrt{\log d})^{1+\delta}$ for some $\delta>1$, the term $\mathrm{III}_{2}$ will vanish \cyancom{since the $o(1)$ term above actually has a rate of polynomial decaying in $n$, up to some polynomial factor in $\log d$.}
% By the independence assumption about GMB p-values, together with \eqref{eq:fdp_cgmb}, we obtain
up to some constant, where the first inequality holds by \eqref{eq:fdp_cgmb}. As for the second inequality, we recall the construction of $\{t_m\}_{m=1}^{\lambda_d}$ (over the course of derivations from \eqref{eq:fdp_term2} to \eqref{eq:alphas_ratio_bound}) in the proof of Theorem \ref{thm:fdr_hub} thus note $\alpha_1 = 1, t_1 = 0$ and $t_{\lambda_d} - t_{1} = \sum_{m=2}^{\lambda_d} (t_m - t_{m-1}) = (\lambda_d - 1)h_d$. 
%(see the construction of $\{t_m\}_{m=1}^{\lambda_d}$ in the proof of Theorem \ref{thm:fdr_hub} and the properties (i), (ii), (iii)). 
Regarding the term $\mathrm{III}_{1}$, we have
\begin{eqnarray}\nonumber
    \mathrm{III}_{1} 
    &=& \sum_{m=1}^{\lambda_{d}} \frac{\sum_{j\in \cH_{0}}\Var{I_j(\alpha_{m}) - d_{0}\alpha_{m}}}{\epsilon^{2} d_{0}^{2}\alpha_{m}^{2}}\\
    &=& \sum_{m=1}^{\lambda_{d}} \frac{\sum_{j\in \cH_{0}}\EEE( I_j(\alpha_{m}))(1-\EEE( I_j(\alpha_{m}))}{\epsilon^{2} d_{0}^{2}\alpha_{m}^{2}}
    \le \frac{1}{\epsilon^2 d_{0}}\sum_{m=1}^{\lambda_{d}} \frac{C}{\alpha_{m}}
    % &\le& \frac{1}{\epsilon^2 d_{0}h_{d}} \int_{0}^{t_{\lambda_{d}}}\frac{\eta_2(d,n,\alpha_L) }{\mathbb{P}(T^{\cB} > t )} dt + \frac{1}{\epsilon^2 d_{0}} \frac{\eta_2(d,n,\alpha_L) }{\alpha_{\lambda_{d}}}\\
     \le \frac{C}{\epsilon^2 d_{0}\alpha_{L}} \cdot\frac{t_{\lambda_d}}{h_d  }, \label{eq:III1}
     %\frac{1}{\epsilon^2 d_{0}} \frac{\eta_2(d,n,\alpha_L) }{\alpha_{\lambda_{d}}}\\
     %&\le& \frac{C}{\epsilon^2 \rho} \cdot \frac{q(\alpha_{L};T_{Z})}{w_{d}h_{d}}
\end{eqnarray}
where the first inequality holds due to $\eqref{eq:fdp_cgmb}$ and the second inequality holds since $\alpha_m \ge \alpha_L$ $\forall~ 1\le m \le \lambda_d$ and $t_{\lambda_d} = (\lambda_d - 1)h_d$.
% we have
% \begin{equation}
%      \mathrm{III}_{1}  \le \frac{1}{\epsilon^2 d_{0}}\sum_{m=1}^{\lambda_{d}} \frac{C}{\mathbb{P}(T_{\bZ} > t_{m} )}\le 
%      \frac{1}{\epsilon^2 d_{0}h_{d}} \int_{0}^{t_{\lambda_{d}}}\frac{C}{\mathbb{P}(T^{\cB} > t )} dt + \frac{1}{\epsilon^2 d_{0}} \frac{C}{\alpha_{\lambda_{d}}}
% \end{equation}
Therefore, combining \eqref{eq:III2} with \eqref{eq:III1}, we obtain
\begin{eqnarray}\nonumber
\mathrm{III}_{1}  + \mathrm{III}_{2} 
\le \frac{1}{\epsilon^2} \cdot \frac{t_{\lambda_d}}{h_d  } \left( \frac{C}{d_0 \alpha_L} + \eta^2(d,n,\zeta_1,\zeta_2, \alpha_L) \right)
\le   \frac{C't_{\lambda_d} }{\epsilon^2h_d } \left(
\frac{1}{\rho d_0}+
\eta^2(d,n,\zeta_1,\zeta_2, \alpha_L) 
\right)
\end{eqnarray}
for some constant $C'$, where the second inequality holds by the definition $\alpha_L =q|\cB|/d$ in the proof of Theorem \ref{thm:fdr_hub} and the definition $\rho = |\cB|/d$ in Section \ref{sec:hub_selection}.
\end{proof}

\begin{lemma}\label{lem:crossterm_Ij}
Recalling the definition of $\mathrm{III}_{3}$ in \eqref{eq:seq_III_3}, we have
$$
\mathrm{III}_{3} \le 
\frac{C'''t_{\lambda_d}}{\rho\epsilon^2 d_0 h_d}
\left(1 + \eta(d,n,\zeta_1,\zeta_2, \alpha_L)  d_0+\frac{|S|\log d}{d_0 p}
\right),
%  \frac{C'''d}{\epsilon^2 |\cB|d_0}  \cdot\frac{t_{\lambda_d}}{h_d  } \cdot
% \left(1 + \eta_2(d,n,\alpha_L)d_0+\frac{|S|\log d}{d_0 p}
% \right)
%\frac{C'''d}{\epsilon^2 d_{0}|\cB|}  \cdot\frac{t_{\lambda_d}}{h_d  } \cdot
%\left(1 + \eta_2(d,n,\alpha_L)+\frac{|S|\log d}{d_0 p}\right)
$$
where $\eta(d,n,\zeta_1,\zeta_2, \alpha_L)= O \big(\frac{(\log d)^{19/6}}{n^{1/6}} + \frac{ (\log d)^{11/6}}{n^{1/6}\alpha_L^{1/3}} + \zeta_1 \log d + \frac{\zeta_2}{\alpha_L}\big)$ with $\zeta_1 = {s(\log d)^2}/\sqrt{n}$, $\zeta_2 = 1/d^2$.
\end{lemma}
\begin{proof}[Proof of Lemma \ref{lem:crossterm_Ij}]
Note that $\mathrm{III}_{3}$ in \eqref{eq:seq_III_3} equals
%\jlmargin{}{Add label and cite it in Lemma A.5}
\begin{align}
\label{eq:Ijdef}
\begin{split}
%\nonumber
\mathrm{III}_{3}&=\sum_{m=1}^{\lambda_{d}}\frac{ 
    \sum_{j_1,j_2 \in \cH_0, j_1 \ne j_2} \Cov{I_{j_1}(\alpha_m)}{ I_{j_2}(\alpha_m)}
    }{\epsilon^{2} d_{0}^{2}\alpha_{m}^{2}},\\
&\quad \text{ where $I_j(\alpha) = \Indrbr{\max_{e \in N_{0j} } \sqrt{n}|\tdTheta_e - \sTheta_e| \ge \hat c(\alpha, N_{0j})}$ } 
\end{split}
\end{align}
for $j\in \{j_1,j_2\}$. To quantify the covariance between $I_{j_1}(\alpha_m)$ and $I_{j_2}(\alpha_m)$ for $j_1,j_2 \in \cH_0, j_1\ne j_2$, we define 
\begin{equation}\label{eq:Wjdef}
    W_j(\alpha) = \Indrbr{\max_{e \in N_{0j} } |Z_{e}| \ge  c(\alpha, N_{0j})},
\end{equation}
where $(Z_e)_{e \in E}$ (with $E= N_{0j}$) is a Gaussian random vector and shares the same mean vector and covariance matrix as the term $(\frac{1}{\sqrt{n~ {\bTheta}_{jj}{\bTheta}_{kk}}}  \sum_{i=1}^n  {\bTheta}_{j}^{\top} (\bX_i \bX_i^{\top} {\bTheta}_{k}- \eb_k) )_{(j,k)\in E}$ in $\breve{T}_E$. Here $\breve{T}_E$ (with $E= N_{0j}$) has the explicit form below
$$
 \breve{T}_E =  \max_{(j,k) \in E}  \frac{1}{\sqrt{n~ {\bTheta}_{jj}{\bTheta}_{kk}}} \bigg| \sum_{i=1}^n  {\bTheta}_{j}^{\top} (\bX_i \bX_i^{\top} {\bTheta}_{k}- \eb_k) \bigg|. 
$$
Remark here $\breve{T}_E$ corresponds to the term $T_{\bY}$ in Proposition \ref{prop:approx_cramer_gmb} and $\max_{e \in E}|Z_{e}|$ corresponds to the term $T_{\bZ}$ in Proposition \ref{prop:cramer_gmb}. And $c(\alpha, N_{0j})$ is the corresponding Gaussian maxima quantile $q(\alpha;T_{\bZ})$ (which does not need to be computed). Since $\mathbb{P}(T_{\bZ} > q(\alpha;T_{\bZ}) ) = \alpha$, we immediately have have $\EEE[W_j(\alpha)] = \PP{\max_{e \in N_{0j} } \sqrt{n}|Z_{e}| \ge  c(\alpha, N_{0j})}  = \alpha$. 

Now we replace $I_{j_{1}}(\alpha), I_{j_{2}}(\alpha)$ in $\mathrm{III}_{3}$ by $W_{j_{1}}(\alpha),W_{j_{2}}(\alpha)$ and define $\mathrm{III}'_{3}$ as
\begin{equation}\label{eq:III3'_def}
    \mathrm{III}'_{3} := \sum_{m=1}^{\lambda_{d}}\frac{ 
    \sum_{j_1,j_2 \in \cH_0, j_1 \ne j_2} \Cov{W_{j_1}(\alpha_m)}{ W_{j_2}(\alpha_m)}
    }{\epsilon^{2} d_{0}^{2}\alpha_{m}^{2}}.
\end{equation}
To bound $\mathrm{III}_{3}- \mathrm{III}'_{3}$, we first note $ \cov(I_{j_1}(\alpha), I_{j_2}(\alpha)) 
= \EEE[I_{j_1}(\alpha) I_{j_2}(\alpha)] -  \EEE[I_{j_1}(\alpha)]\EEE[ I_{j_2}(\alpha)]  $ then separately deal with the term $\left| \EE{I_{j_1}(\alpha)I_{j_2}(\alpha)} -  \EE{W_{j_1}(\alpha)W_{j_2}(\alpha)} \right|$ and the term $\left| \EE{I_{j_1}(\alpha)}\EE{I_{j_2}(\alpha)} -  \EE{W_{j_1}(\alpha)}\EE{W_{j_2}(\alpha)} \right|$.

By Lemma \ref{lem:IW_diff_bound}, we have up to some constant factor,
$$
\frac{\left| \EE{I_{j_1}(\alpha)I_{j_2}(\alpha)} -  \EE{W_{j_1}(\alpha)W_{j_2}(\alpha)} \right|   }{\alpha^2} \le  \frac{\eta(d,n,\zeta_1,\zeta_2, \alpha_L)}{\alpha}.
$$
%since $\EE{ W_{j_1}(\alpha) W_{j_2}(\alpha)} =  $. 
Applying the same strategy to the term $\EE{I_{j_1}(\alpha)}\EE{I_{j_2}(\alpha)}$, we obtain
$$
\frac{\left| \EE{I_{j_1}(\alpha)}\EE{I_{j_2}(\alpha)} -  \EE{W_{j_1}(\alpha)}\EE{W_{j_2}(\alpha)} \right|   }{\alpha^2} \le  
\frac{\eta(d,n,\zeta_1,\zeta_2, \alpha_L)}{\alpha}.
%\frac{1}{\alpha}\cdot  O\left(\frac{(\log d)^{19/6}}{n^{1/6}}\right)
$$
% under the scaling conditions, we have $o(1)/\alpha = o(1)$. 
% Now it suffices to deal with the term 
% Then the contribution of the above quantities (approximation bounds) to the term $\mathrm{III}_3$
% $$
% \left(\frac{1}{\epsilon^{2} \rho d_0}\sum_{m=1}^{\lambda_{d}}\frac{1
%     }{ \alpha_{m}}\right) \cdot  O\left(\frac{d(\log d)^{19/6}}{n^{1/6}}\right)
% $$
Combining the above two inequalities, and noting the definition of $\mathrm{III}'_{3}$ in \eqref{eq:III3'_def}, we derive the following bound on the difference between $\mathrm{III}_{3}$ and $\mathrm{III}'_{3}$,
$$
\left|
\mathrm{III}_{3} - \mathrm{III}'_{3}
\right|\le \frac{1}{\epsilon^{2}}\sum_{m=1}^{\lambda_{d}}\frac{\eta(d,n,\zeta_1,\zeta_2, \alpha_L)
    }{ \alpha_{m}}\le \frac{C' t_{\lambda_d} }{\rho \epsilon^2 h_d} \cdot \eta(d,n,\zeta_1,\zeta_2, \alpha_L).
%\cdot  O\left(\frac{d(\log d)^{19/6}}{n^{1/6}}\right)
$$
where the second inequality holds due to the fact $\alpha_m \ge \alpha_L$ $\forall~ 1\le m \le \lambda_d$ and $t_{\lambda_d} = (\lambda_d - 1)h_d$, the definition $\alpha_L =q|\cB|/d$ in the proof of Theorem \ref{thm:fdr_hub}, and the definition $\rho = |\cB|/d$ in Section \ref{sec:hub_selection}.

% Applying the derivations for \eqref{eq:III1} in Theorem \ref{thm:fdr_indep} and assuming the stated scaling conditions, we can show the above term is $o(1)$. 
The above bound on $\left|\mathrm{III}_{3} - \mathrm{III}'_{3} \right|$, when combined with Lemma \ref{lem:crossterm_Wj}, immediately establishes
\begin{eqnarray}\nonumber
\mathrm{III}_{3}
&\le& \frac{C' t_{\lambda_d} }{\rho \epsilon^2 h_d} \cdot \eta(d,n,\zeta_1,\zeta_2, \alpha_L)
+\frac{C''t_{\lambda_d}}{\rho\epsilon^2  d_{0} h_d} 
 \left(1+ C_{\bTheta}\frac{|S|\log d}{d_0 p}
\right) \\ \nonumber
&\le& \frac{C'''t_{\lambda_d}}{\rho\epsilon^2 d_0 h_d}
\left(1 + \eta(d,n,\zeta_1,\zeta_2, \alpha_L) d_0+\frac{|S|\log d}{d_0 p}
\right),
\end{eqnarray}
for some constant $C'''$.
\end{proof}

\begin{lemma}\label{lem:crossterm_Wj}
Recalling the term $\mathrm{III}'_{3}$ from \eqref{eq:III3'_def} in the proof of Lemma \ref{lem:crossterm_Ij}, we have
$$
 \mathrm{III}'_{3} = \sum_{m=1}^{\lambda_{d}}\frac{ 
    \sum_{j_1,j_2 \in \cH_0, j_1 \ne j_2} \cov(W_{j_1}(\alpha_m), W_{j_2}(\alpha_m))
    }{\epsilon^{2} d_{0}^{2}\alpha_{m}^{2}}\le 
    \frac{C''t_{\lambda_d}}{\rho \epsilon^2 d_{0}h_d} 
     \left(1+ C_{\bTheta}\frac{|S|\log d}{d_0 p}
    \right).
% \frac{C''t_{\lambda_d}}{\epsilon^2 \rho d_{0} h_d} \left(1+ C_{\bTheta}\frac{|S|\log d}{d_0 p}\right)
 $$
\end{lemma}
\begin{proof}[Proof of Lemma \ref{lem:crossterm_Wj}]
%We use similar strategy to deal with the term $\Cov{W_{j_1}(\alpha)}{ W_{j_2}(\alpha)}$:
Similarly as in the proof of Lemma \ref{lem:crossterm_Ij}, we define $(Z_e)_{e \in N_{0j_1} \cup N_{0j_2}}$ to be jointly Gaussian such that this $(|N_{0j_1}| + |N_{0j_2}|)$-dimensional Gaussian random vector shares the same mean vector and covariance matrix as $(\frac{1}{\sqrt{n~ {\bTheta}_{jj}{\bTheta}_{kk}}}  \sum_{i=1}^n  {\bTheta}_{j}^{\top} (\bX_i \bX_i^{\top} {\bTheta}_{k}- \eb_k) )_{(j,k)\in N_{0j_1}\cup N_{0j_2}}$. 
%in $\breve{T}_E$ 
Note that the two sub-vectors $(Z_e)_{e \in N_{0j_1}}$ and $(Z_e)_{e \in N_{0j_1}}$ are generally dependent. Then we define $(Z'_e)_{e \in N_{0j_1}}, (Z'_e)_{e \in N_{0j_2}}$ to be two Gaussian random vectors such that 
\begin{equation} %\nonumber
\label{eq:Ze12_def}
(Z'_e)_{e \in N_{0j_1}} \stackrel{d}{=} (Z_e)_{e \in N_{0j_1}},~~  (Z'_e)_{e \in N_{0j_2}} \stackrel{d}{=} (Z_e)_{e \in N_{0j_2}} ~~ \text{ and }  (Z'_e)_{e \in N_{0j_1}} \independent  (Z'_e)_{e \in N_{0j_2}}.
\end{equation}
Recalling the definition of $W_j(\alpha)$ in \eqref{eq:Wjdef}: $W_j(\alpha) = \Indrbr{\max_{e \in N_{0j} } |Z_{e}| \ge  c(\alpha, N_{0j})}$, we thus have the following,
\begin{align}
\mathrm{IV}_{j_1j_2}(\alpha)\label{eq:IV_def}
:=~& \frac{
\left|
\cov(W_{j_1}(\alpha_m), W_{j_2}(\alpha_m))
\right| }{\alpha^2}\\ \nonumber
=~&
\frac{
\left|
\EE{W_{j_1}(\alpha)W_{j_2}(\alpha)}-
\EE{W_{j_1}(\alpha)}\EE{W_{j_2}(\alpha)}
\right| }{\alpha^2}\\ \nonumber
=~& \frac{1}{\alpha^2}
\Big|
\PPP(\max_{e \in N_{0j_1} } |Z_{e}| \ge  c(\alpha, N_{0j_1}),\max_{e \in N_{0j_2} } |Z_{e}| \ge  c(\alpha, N_{0j_2})) - \\ \nonumber
~~~& \PPP(\max_{e \in N_{0j_1} } |Z'_{e}| \ge  c(\alpha, N_{0j_1}),\max_{e \in N_{0j_2} } |Z'_{e}| \ge  c(\alpha, N_{0j_2}))
\Big| \\ \nonumber
=~& \frac{1}{\alpha^2}
\Big|
\PPP(\max_{e \in N_{0j_1} } |Z_{e}| \ge  t,\max_{e \in N_{0j_2} } |Z_{e}| \ge  t )-  \PPP(\max_{e \in N_{0j_1} } |Z'_{e}| \ge  t,\max_{e \in N_{0j_2} } |Z'_{e}| \ge  t)
\Big| \\
=~& \frac{1}{\alpha^2}
\Big|
\PPP(\max_{e \in N_{0j_1}\cup N_{0j_2} } |Z_{e}| \ge  t) - \PPP(\max_{e \in N_{0j_1}\cup N_{0j_2} } |Z'_{e}| \ge  t)
\Big|,
\label{eq:note_CGB}
\end{align}
%where $(Z_e)_{e \in N_{0j_1}},(Z_e)_{e \in N_{0j_2}}$ have the same covariance matrix structures as $(Z'_e)_{e \in N_{0j_1}}, (Z'_e)_{e \in N_{0j_2}}$, but $(Z'_e)_{e \in N_{0j_1}}\independent  (Z'_e)_{e \in N_{0j_2}}$. 
where the third equality follows due to the construction of $(Z_e)_{e \in N_{0j_1}\cup N_{0j_2}}, (Z'_e)_{e \in N_{0j_1}\cup N_{0j_2}}$.  Note that in the fourth equality, we assume $c(\alpha, N_{0j_1}) = c(\alpha, N_{0j_2}):=t $ without loss of generality, since we can rescale one of the maximum statistic by rescaling the Gaussian random vectors. Remark that the scaling will not break down the application of Theorem \ref{thm:ccb_sparse_unitvar}, which will be explained in detail 
later in this proof.
% following derivations since the variances within each block $N_{0j}$ will still be the same. 
 The last inequality holds by \eqref{eq:Ze12_def} and the fact that $ \PP{A \cap B} = \PP{A} + \PP{B} - \PP{A \cup B} $.

 Notice that we can apply the Cram\'{e}r-type Gaussian comparison bound with $\ell_0$ norm to control \eqref{eq:note_CGB}. Specifically, we first figure out the difference between the covariance matrices of $(Z_e)_{e \in N_{0j_1}\cup N_{0j_2}}$ and $(Z'_e)_{e \in N_{0j_1}\cup N_{0j_2}}$. Denote the covariance matrices by $\bSigma^{Z}$ and $\bSigma^{Z'} $ respectively. As these two Gaussian random vectors have two sub-vectors, we write their covariance matrices in a block form
 $$
\bSigma^{Z} = \left(
\begin{array}{cc}
  \bSigma^{Z}_{11}   & \bSigma^{Z}_{12} \\
  \bSigma^{Z}_{21}   & \bSigma^{Z}_{22}
\end{array}
\right),~~\bSigma^{Z'} = \left(
\begin{array}{cc}
  \bSigma^{Z'}_{11}   & \bm{O} \\
  \bm{O}   & \bSigma^{Z'}_{22}
\end{array}
\right).
$$
where $\bSigma^{Z'}$ is block diagonal due to \eqref{eq:Ze12_def}. Note that we also have $\bSigma^{Z}_{11} = \bSigma^{Z'}_{11} $ and $\bSigma^{Z}_{22} = \bSigma^{Z'}_{22}$. Then we have
\begin{equation} \label{eq:bSigmaZ_diff}
\bSigma^{Z} - \bSigma^{Z'} = \left(
\begin{array}{cc}
  \bm{O}  & \bSigma^{Z}_{12} \\
  \bSigma^{Z}_{21}   & \bm{O}
\end{array}
\right).
\end{equation}
%where the block correspond to 
% By the construction of these two Gaussian random vectors, we have 

 Throughout the following proof, we assume $\bTheta_{jj}=1,j\in[d]$ without loss of generality, since the standardized version is considered in $\breve{T}_E$ \eqref{eq:brTE}. Recall that $(Z_e)_{e \in N_{0j_1} \cup  N_{0j_2}}$ shares the same covariance structure as $(Y_e)_{e \in N_{0j_1} \cup  N_{0j_2}}$ where $ Y_{e}$ (with $e = (j,k)$) is defined as
%  write down the involving Gaussian random vectors $(Z_e)_{e \in N_{0j_1} \cup  N_{0j_2}}$ below:
% have the same covariance structures as the following
%Ignoring the approximation issue for now, we can relate $F_j(\alpha)$ to the following quantity
\begin{equation} \nonumber
 Y_{e}:=  \frac{1}{\sqrt{n}}  \sum_{i=1}^n  {\bTheta}_{j}^{\top} (\bX_i \bX_i^{\top} {\bTheta}_{k}- \eb_k). %,~ k  \in N_{0j_1}.
% \quad   \frac{1}{\sqrt{n}}  \sum_{i=1}^n  {\bTheta}_{j_2}^{\top} (\bX_i \bX_i^{\top} {\bTheta}_{k}- \eb_k) ,~ k  \in N_{0j_2}.
\end{equation}
%Remark that in the above expression, we assume $\bTheta_{jj}=1,j\in[d]$ without loss of generality, since the standardized version is considered in $\breve{T}_E$ \eqref{eq:brTE}. 
Then we are ready to calculate the covariance matrix  $\bSigma^{Z}$. Specifically, we compute the entries in each block.
% of 
%$$
%\bSigma^{Z} = \left[
%\begin{array}{cc}
%  \bSigma^{Z}_{11} & \bSigma^{Z}_{12}\\ 
%  \bSigma^{Z}_{21} & \bSigma^{Z}_{22}
%\end{array}
%\right].
%$$
Regarding the block $\bSigma^{Z}_{11} $, for any $k,k' \in N_{0 j_1}$ where $N_{0j_1} = \{k: \bTheta_{j_1k}=0\}$, we have the corresponding $(k,k')$ entry in $\bSigma^{Z}_{11} $ equals
\begin{equation}\label{eq:diag_block_entry}
    \cov({\bTheta}_{j_1}^{\top} (\bX_i \bX_i^{\top} {\bTheta}_{k}- \eb_{k}), {\bTheta}_{j_1}^{\top} (\bX_i \bX_i^{\top} {\bTheta}_{k'}- \eb_{k'})) = \bTheta_{j_1 j_1} \bTheta_{k k'} + \bTheta_{j_1 k}\bTheta_{j_1 k'}  = \bTheta_{k k'},   
    %\bTheta_{j_1 j_1} \bTheta_{k k'}  
\end{equation}
by applying Isserlis' theorem \cite{isserlis1918formula} and noting $\bTheta_{j_1 k} = \bTheta_{j_1 k'} =0 $. Similar results hold for the block $\bSigma^{Z}_{22} $. Regarding the block $\bSigma^{Z}_{12}$, consider $k_1  \in N_{0 j_1}, k_2 \in N_{0 j_2}$, then we have the corresponding $(k_1,k_2)$ entry in the block equals
\begin{equation}\label{eq:offdiag_block_entry}
    \cov({\bTheta}_{j_1}^{\top} (\bX_i \bX_i^{\top} {\bTheta}_{k_1}- \eb_{k_1}), {\bTheta}_{j_2}^{\top} (\bX_i \bX_i^{\top} {\bTheta}_{k_2}- \eb_{k_2})) = \bTheta_{j_1 j_2} \bTheta_{k_1 k_2} + \bTheta_{j_1 k_2}\bTheta_{j_2 k_1}.
\end{equation}
Now we have fully characterized the covariance matrix $\bSigma^Z$ and the covariance matrix difference in \eqref{eq:bSigmaZ_diff} 
%the off-diagonal elements in the covariance matrix of the Gaussian random vector $(Z_e)_{e \in N_{0j_1} \cup  N_{0j_2}}$ 
for any $j_1,j_2 \in \cH_0, j_1 \ne j_2$. Specifically, we have
$\norm{ \bSigma^{Z} - \bSigma^{Z'}}_{0} = \norm{ \bSigma_{12}^{Z}}_{0} = \sum_{k_1  \in N_{0 j_1}, k_2 \in N_{0 j_2}}\Indrbr{\bTheta_{j_1 j_2} \bTheta_{k_1 k_2} + \bTheta_{j_1 k_2}\bTheta_{j_2 k_1} \ne 0} $.
Based on whether $\bTheta_{j_1 j_2}$ is zero or not, we consider the following two cases then handle them separately:
\begin{itemize}
  \item Case 1: $\bTheta_{j_1j_2} = 0$. If $k_1=k_2$, then we have the covariance matrix entry \eqref{eq:offdiag_block_entry} equal zero; If $k_1 \ne k_2$, then \eqref{eq:offdiag_block_entry} is nonzero only if $\bTheta_{j_1 k_2}\ne 0, \bTheta_{j_2 k_1}\ne 0$ (i.e., $k_2 \notin N_{0j_1}$, $k_1 \notin N_{0 j_2}$). By the fact $j_1,j_2 \in \cH_{0}, j_1\ne j_2$ and the definition of $\cH_0 = \{j: \jdeg < k_{\tau}\}$, we have $\#\{(k_1,k_2):k_1\ne k_2, \bTheta_{j_1 k_2}\ne 0, \bTheta_{j_2 k_1}\ne 0\}\le k_{\tau}^{2}$. Hence $\norm{ \bSigma^{Z} - \bSigma^{Z'}}_{0} \le k_\tau^2$.
  % the covariance matrix difference $\bSigma^Z - \bSigma^{Z'}$ has at most $k_{\tau}^{2}$ nonzero elements.
%   in its upper right block $\bSigma^{Z}_{12}$ ( whose dimension is $|N_{0 j_1}|\times |N_{0 j_2}|$).
  \item Case 2: $\bTheta_{j_1j_2} \ne 0$. The covariance matrix entry \eqref{eq:offdiag_block_entry} is nonzero only if $\bTheta_{j_1 k_2}\ne 0, \bTheta_{j_2 k_1}\ne 0$ (i.e., $k_2 \notin N_{0j_1}$, $k_1 \notin N_{0 j_2}$) or $\bTheta_{k_1k_2} \ne 0$.
\end{itemize}

We start from the simpler case, i.e., Case 2 where $\bTheta_{j_1j_2} \ne 0$. Simply, we obtain
$$
\mathrm{IV}_{j_1j_2}(\alpha)= \frac{|\cov(W_{j_1}(\alpha),W_{j_2}(\alpha))|}{\alpha^2} \le \frac{\var(W_{j_1}(\alpha))}{\alpha^2} + \frac{\var(W_{j_2}(\alpha))}{\alpha^2} \le  \frac{C}{\alpha},
$$
for some constant $C$ since $\Var{W_j(\alpha)} = \EE{W_j(\alpha)} (1 - \EE{W_j(\alpha)}) = \alpha (1 - \alpha)$ for $j = j_1, j_2$. 
%For the second inequality, we  due to \eqref{eq:fdp_cgmb}.
For a fixed $j_1$, we also know that $|\{j_2 \in \cH_0:j_2\ne j_1, \bTheta_{j_1j_2}\ne 0 \}|< k_{\tau}$. Then we have
%can bound the relevant terms $\mathrm{IV}_{j_1j_2}(\alpha)$ (see \eqref{eq:IV_def}) under Case $2$ (i.e., $\bTheta_{j_1,j_2} \ne 0$) in $\mathrm{III}'_3$ by 
\begin{equation}
\label{eq:case2_terms}
\sum_{m=1}^{\lambda_d} \sum_{\bTheta_{j_1j_2} \ne 0}\frac{\mathrm{IV}_{j_1j_2}(\alpha_m)}{\epsilon^2 d^2_{0}}  \le 
\sum_{m=1}^{\lambda_d} \frac{ d_0k_{\tau}}{\epsilon^2 d_0^2} \cdot \frac{C}{\alpha_m}\le \frac{1}{\epsilon^2 d_0}  \sum_{m=1}^{\lambda_d} \frac{C'}{\alpha_m}   ,
% \cdot\frac{t_{\lambda_d}}{h_d  }
\end{equation}
where the last inequality holds due to the same derivations for $\mathrm{III}_1$ in the proof of Lemma \ref{lem:moments_Ij}.
%, by the definition of $\alpha_L, \lambda_d$ in Appendix \ref{pf:thm:fdr_hub} and $\rho$ in Section \ref{sec:hub_selection}.
%which will also vanish due to the same reason as \eqref{eq:III1}. 

Regarding Case $1$ where $\bTheta_{j_1j_2} = 0$, we will give a more careful treatment to $\mathrm{IV}_{j_1j_2}(\alpha)$ in \eqref{eq:IV_def}. 
Due to the discussion about Case 1, we have $\norm{ \bSigma^{Z} - \bSigma^{Z'}}_{0} \le k_\tau^2$.
% $\nbr{ \bSigma^{U}_{12} - \bSigma^{V}_{12}}_{0} = \nbr{ \bSigma^{U}_{12}}_{0} < k_{\tau}^2$, thus 
%$ \nbr{ \bSigma^{U} - \bSigma^{V}}_{0} < k_{\tau}^2$. 
This fact will be utilized to derive a nice bound on $\mathrm{III'_3}$. Indeed, we can apply Theorem \ref{thm:ccb_sparse_unitvar} to \eqref{eq:note_CGB} (with $U$ and $V$ chosen to be $Z_e)_{e \in N_{0j_1}\cup N_{0j_2}}$ and $(Z'_e)_{e \in N_{0j_1}\cup N_{0j_2}}$ respectively)
%(using the form in Remark %\ref{rk:thm:ccb_sparse}
%)
%(using the more refined version from \eqref{eq:uneq_entry} instead of the one from \eqref{eq:uneq_entry_max})
 and obtain
\begin{equation}\label{eq:case1_IVj1j2}
% \frac{
% \left|
% \EEE[W_{j_1}(\alpha)W_{j_2}(\alpha)]-
% \EEE[W_{j_1}(\alpha)]\EEE[W_{j_2}(\alpha)]
% \right| }{\alpha^2} 
\mathrm{IV}_{j_1j_2}(\alpha)
\le  \frac{\log d }{\alpha p}\rbr{ \sum_{k_1  \in N_{0 j_1}, k_2 \in N_{0 j_2}, k_1 \ne k_2} \Indrbr{\bTheta_{j_1 k_2}\bTheta_{j_2 k_1} \ne 0}}. 
\end{equation}
when $\bTheta_{j_1j_2} = 0$ (i.e., under Case 1). Recall Theorem \ref{thm:ccb_sparse_unitvar} assumes for Gaussian random vectors $U$ and $V$, there exists a disjoint $\discon$-partition of nodes $\cup_{\ell=1}^{\discon}\cC_\ell = [d]$ such that  $\sigma^U_{jk}=\sigma^V_{jk}=0$ when $j \in \cC_{\ell}$ and  $k \in \cC_{\ell'}$ for some $\ell \neq \ell'$.
%there exist disjoint partitions of nodes $\cup_{\ell=1}^{\discon}\cC^U_\ell = \cup_{\ell=1}^{\discon}\cC^V_\ell= [d], $ such that  $\sigma^U_{jk}$ ($\sigma^V_{jk}$) equals $0$ when $j,k$ belong to different components $\cC^U_\ell$ ($\cC^V_\ell$), and $\forall \ell \in[\discon]$, $\cC^U_{\ell} \cap \cC^V_{\ell} \ne \emptyset$.
 This is the connectivity assumption. Theorem \ref{thm:ccb_sparse_unitvar} also assumes that $U$ and $V$ have unit variances i.e., $\sigma^U_{jj}=\sigma^V_{jj}=1, j \in [d]$ and there exists some $\sigma_0<1$ such that $|\sigma^V_{jk} |\le  \sigma_0$ for any $j\ne k$ and $ |\{(j,k): j\ne k, |\sigma^U_{jk} |>  \sigma_0 \} | \le b_0$ for some constant $b_0$. Under its general version (which is actually proved in Appendix \ref{pf:thm:ccb_sparse_unitvar}), we only need to assume $a_0 \le\sigma^U_{jj}=\sigma^V_{jj} \le a_1,~\forall j \in [d]$, and given any $j \in \cC_\ell$ with some $\ell$, there exists at least one $m  \in \cC_{\ell'}$ such that $\sigma^U_{jj}= \sigma^V_{jj} = \sigma^{U}_{mm}  =  \sigma^{V}_{mm}$ for any $\ell' \ne \ell$. 
% and $ \sigma^U_{jj} = \sigma^U_{kk}$ for $j, k $ belong to the same component. 
From now, we will call it the general variance condition. Accordingly, we assume there exists some $\sigma_0<1$ such that $|\sigma^V_{jk}/\sqrt{\sigma^V_{jj} \sigma^V_{kk}} |\le  \sigma_0$ for any $j\ne k$ and $ |\{(j,k): j\ne k, |\sigma^U_{jk} | \sqrt{\sigma^U_{jj} \sigma^U_{kk}} >  \sigma_0 \} | \le b_0$ for some constant $b_0$. Such condition is referred as the general covariance assumption. 
%$\lambda_{\min}(\bSigma^U)\ge 1/b_0>0,\lambda_{\min}(\bSigma^V) \ge 1/b_0>0$ for some constant $b_0>0$. 
%Theorem \ref{thm:ccb_sparse_unitvar} also assumes there exist disjoint partitions of nodes $\cup_{\ell=1}^{\discon}\cC^U_\ell = \cup_{\ell=1}^{\discon}\cC^V_\ell= [d], $ such that  $\sigma^U_{jk}$ ($\sigma^V_{jk}$) equals $0$ when $j,k$ belong to different components $\cC^U_\ell$ ($\cC^V_\ell$), and $\forall \ell \in[\discon]$, $\cC^U_{\ell} \cap \cC^V_{\ell} \ne \emptyset$. This is the connectivity assumption. 
Below we give the details of applying Theorem \ref{thm:ccb_sparse_unitvar} (with a general version of the variance assumption) by checking those three conditions.
%Below we give details on how we apply Theorem \ref{thm:ccb_sparse} by checking the connectivity assumption and the variance assumption. 

%Specifically, we verify the sufficient condition (SC2) of Assumption \ref{asp:connect_prop}. 
%Note that Theorem \ref{thm:fdr_hub} assumes that the the number of connected components in the associated graph $\cG$ of the 
%We will construct a subset $\cE_0$
%{Regarding the connectivity assumption}, 
%Recall that the 
%Specifically, we verify the connectivity assumption, covariance assumption and the general version of the variance assumption.
We start from the connectivity assumption and the general variance condition. Notice that in Section \ref{sec:hub_selection}, $p$ denotes the number of connected components in the associated graph $\cG$ of $\bX$. Then we know there exist disjoint partitions of nodes $\cup_{\ell=1}^{p}\cC^X_{\ell} = [d]$ such that $\bTheta_{jk} =0$ when $j\in \cC^X_{\ell}, k \in \cC^X_{\ell'}$ for some $\ell \ne \ell'$. We will utilize this fact to examine the covariance matrices of $U:= (Z_e)_{e \in N_{0j_1}\cup N_{0j_2}}$ and $V:=(Z'_e)_{e \in N_{0j_1}\cup N_{0j_2}}$ and show the connectivity assumption holds. 
%  and show that for $\any s \in [0,1]$, the Gaussian random vector $W(s) = $
	Note that for given $j_1,j_2 \in \cH_0, j_1 \ne j_2$, there exist at least $p-2$ components $\cup_{\ell=1}^{p-2}\cC^X_{\ell}$ such that $j_1$ and $j_2$ do not belong to them. Without loss of generality, we write $j_1, j_2 \notin \cup_{\ell=1}^{p-2}\cC^X_\ell$. Thus we have $\cup_{l=1}^{p-2}\cC^X_\ell \subset N_{0j_1} \cap N_{0j_2}$ by definition.
	
	In the following, we will show the number of connected components on the associated graph of the Gaussian random vector $U:= (Z_e)_{e \in N_{0j_1}\cup N_{0j_2}}$ is at least $2(p-2)$ by examining its covariance matrix $\bSigma_{Z}$. First we focus on the covariance entries in the block $\bSigma_{11}^Z$. When $\ell_1, \ell_2 \in [p-2]$ and $\ell_1 \ne \ell_2$, we have for any $k \in \cC^X_{\ell_1}, k' \in \cC^X_{\ell_2}$ (thus $k,k' \in N_{0j_1} \cap N_{0j_2}$), the $(k,k')$ covariance entry \eqref{eq:diag_block_entry} in the block $\bSigma_{11}^Z$  equals
		\begin{equation}\label{eq:zero1}
			\bTheta_{j_1 j_1} \bTheta_{k k'} + \bTheta_{j_1 k}\bTheta_{j_1 k'}  = \bTheta_{j_1 j_1} \bTheta_{k k'} = 0,
		\end{equation}
	where the first equality holds since $k,k' \in N_{0j_1}$, and the second equality holds since $\ell_1 \ne \ell_2$.
%	Based on \ref{eq:diag_block_entry} and \ref{eq:offdiag_block_entry}, we have the $(k, k')$ entry in $\bSigma_{Z}_{11}$ equals 
%
%	Consider any two different  components $ C_{\ell_1}, C_{\ell_2} \subset \cup_{l=1}^{p-2}\cC_\ell \subset N_{0j_1} \cap N_{0j_2}$. 
	Similarly, we have the $(k,k')$ covariance entry in the block $\bSigma_{22}^Z$ also equals to zero. Next we compute the covariance entries in the block $\bSigma_{12}^Z$. For the same $(k, k')$, we know that $k \in N_{0j_1}, k' \in N_{0j_2}$. Thus the corresponding covariance entry \eqref{eq:offdiag_block_entry} equals
	\begin{equation}\label{eq:zero2}
		 \bTheta_{j_1 j_2} \bTheta_{k k'} + \bTheta_{j_1 k'}\bTheta_{j_2 k} = 0,
	\end{equation}
	since we also have $k \in N_{0j_2}, k' \in N_{0j_1}$ and $k \in \cC^X_{\ell_1}, k' \in \cC^X_{\ell_2}$ for some $\ell_1 \ne \ell_2$. Denote the nodes in the associated graph of $\bSigma^Z$ by $\cV_{Z}:=\{(j, k): k \in N_{0j}, j = j_1, j_2\}$. Remark here we use a pair $(j, k)$ to represent a node since there exists some $k \in N_{0j_1} \cap N_{0j_2}$ and we have to distinguish the covariance entries $(j_1, k)$ and $(j_2, k)$. Based on previous calculations, we immediately find $ \cup_{\ell=1}^{2(p-2)}\cC^Z_{\ell} \subset \cV_Z$, where $\cC^Z_{\ell} $ is chosen to be
	\begin{equation}
		\cC^Z_{\ell} = \left\{  
		\begin{array}{ll}
			\{(j_1, k) : k \in \cC^X_{\ell} \} 	&~\text{when}~ 1 \le  \ell \le p-2,  \\
			\{(j_2, k) : k \in \cC^X_{\ell} \} 	&~\text{when}~ p-1 \le  \ell \le 2(p-2).		\end{array}
			\right.
	\end{equation}
	Further, we know they form different components on the associated graph of $\bSigma^Z$. This is due to \eqref{eq:zero1} and \eqref{eq:zero2}.
%	First we compute the covariance entry in the block $\bSigma_{11}^U$. For any $k \in C_{\ell_1}, k' \in C_{\ell_2},\ell_1\ne \ell_2$, we have the covariance matrix entry (which is calculated explicitly in \eqref{eq:diag_block_entry}) equals
%$$
% \bTheta_{j_1 j_1} \bTheta_{k k'} + \bTheta_{j_1 k}\bTheta_{j_1 k'}  = \bTheta_{k k'}  = 0.
%$$
%since and $k, k' \in N_{0j_1}$
%And for any $k_1 \in C_{\ell_1}, k_2 \in C_{\ell_2},\ell_1\ne \ell_2$ (thus $k_1,k_2 \in N_{0j_1}\cap N_{0j_2}$), we have the covariance matrix entry (which is calculated explicitly in \eqref{eq:offdiag_block_entry}) equals
%$$
%\bTheta_{j_1 j_2} \bTheta_{k_1 k_2} + \bTheta_{j_1 k_2}\bTheta_{j_2 k_1} = 0.
%$$
%since and $k_1\in N_{0j_2},k_2 \in N_{0j_1}$
The above results also apply to the Gaussian random vector $V:=(Z'_e)_{e \in N_{0j_1}\cup N_{0j_2}}$ by construction of $Z'_e$, i.e., we have the same subset of nodes $ \cup_{\ell=1}^{2(p-2)}\cC^Z_{\ell} \subset \cV_Z$ from different components on the associated graph of $\bSigma^{Z'}$.
%Remark that $U$ and $V$ share the same partition of different components due to the construction of $V$, thus the connectivity assumption is satisified. 
%It remains to check the covariance assumption and the general variance condition. 

When $k \in \cC_\ell$ for some $\ell \in [p-2]$, the corresponding diagonal entries of the covariance matrices $\bSigma^Z, \bSigma^{Z'}$ equal
$$
\bTheta_{j_1 j_1} \bTheta_{k k} + \bTheta_{j_1 k}\bTheta_{j_1 k} = \bTheta_{j_1 j_1} \bTheta_{k k} = 1 = \bTheta_{j_2 j_2} \bTheta_{k k}, 
$$
where the first equality holds since $ \bTheta_{j_1 k} = 0$ when $k \in \cC_\ell \subset N_{0j_1}$. As for the second equality, we use the fact that $\bTheta_{jj}=1,j\in[d]$. This is because $\breve{T}_E$ in \eqref{eq:brTE} considers the standardized version ${{\bTheta}_{jk}}/{\sqrt{{\bTheta}_{jj}{\bTheta}_{kk}}}$. Remark that the rescaling in Lemma \ref{lem:crossterm_Wj} is performed on one of the two random vectors $(Z'_e)_{e \in N_{0j_1}}, (Z'_e)_{e \in N_{0j_2}}$. Then we have the variances across the $p-2$ components $\cup_{\ell=1}^{p-2}\cC^Z_{\ell}$ are the same. The variances across the other $p-2$ components $\cup_{\ell=p-1}^{2(p-2)}\cC^Z_{\ell}$ are also the same. Finally, we show there exist at least $p-2$ components $\cup_{\ell=1}^{p-2}\cC^Z_{\ell}$ (or $\cup_{\ell=p-1}^{2(p-2)}\cC^Z_{\ell}$) satisfying the requirement in the connectivity assumption and the general variance condition. 
%Hence we immediately prove that there exist at least $p-2$ components $\cup_{\ell=1}^{p-2}$ (or $\cC^Z_{\ell} \cV_Z$) on the associated graph of $W(s) = \sqrt{s} U + \sqrt{1-s} V$ for $\forall s \in [0,1]$. Finally, the sufficient condition SC2 is verified.

%into $p-1$ components, where $\cup_{\ell=1}^{p-2}$ (or $\cC^Z_{\ell} \cV_Z$) form $p-2$ components and the rest of nodes $\cV_Z \setminus  \cup_{\ell=1}^{p-2}$ (or $\cV_Z \setminus  \cup_{\ell=p-1}^{2(p-2)}$) is treated as one component. 
%the variances of each random variable within $U_1$ (or $U_2,V_1,V_2$) will still be the same. Hence the variance assumption is also satisfied.
%Then the contribution of the above quantity to the term $\mathrm{III}_3$ can be quantified as below
Regarding the general covariance condition, we first note
that $\bTheta \in \cU(M,s, r_0)$ which says that $\lambda_{\min}(\bTheta) \ge 1/r_0, \lambda_{\max}(\bTheta) \le r_0$. Thus we have $\max_{j,k \in [d], j\ne k }|\bTheta_{jk}| \le \sigma_0$ for some $\sigma_0 <1$. Below we will examine all the off-diagonal entries of $\bSigma^{Z}$ and $\bSigma^{Z'}$. Regarding the block $\bSigma^{Z}_{11} $, for any $k,k' \in N_{0 j_1}, k\ne k'$ where $N_{0j_1} = \{k: \bTheta_{j_1k}=0\}$, \eqref{eq:diag_block_entry} says that the corresponding $(k,k')$ entry in $\bSigma^{Z}_{11} $ equals $\bTheta_{k k'}$ (here we have $|\bTheta_{k k'}| \le \sigma_0$). Similar results hold for the block $\bSigma^{Z}_{22} $.
%\begin{equation} %\label{eq:diag_block_entry}
%  \bTheta_{k k'},   
%    %\bTheta_{j_1 j_1} \bTheta_{k k'}  
%\end{equation}
%by applying Isserlis' theorem \cite{isserlis1918formula} and noting $\bTheta_{j_1 k} = \bTheta_{j_1 k'} =0 $. Similar results hold for the block $\bSigma^{Z}_{22} $. 
Regarding the block $\bSigma^{Z}_{12}$, consider $k_1  \in N_{0 j_1}, k_2 \in N_{0 j_2}$, then we have the corresponding $(k_1,k_2)$ entry in the block equals $\bTheta_{j_1 k_2}\bTheta_{j_2 k_1}$. This is due to \eqref{eq:offdiag_block_entry} and the fact that $\bTheta_{j_1 j_2} = 0$ under Case 1. Only when $k_2 = j_1, k_1 = j_2$, we have $\bTheta_{j_1 k_2}\bTheta_{j_2 k_1} = 1$. Otherwise, $|\bTheta_{j_1 k_2}\bTheta_{j_2 k_1}| \le \sigma_0^2 <  \sigma_0$ always holds. As for the $\bSigma^{Z'}$, since its  block $\bSigma^{Z'}_{12} = \bm{O}$, we immediately have the absolute values of all its off-diagonal entries is bounded by $\sigma_0$. In summary, we verify the covariance condition of Theorem \ref{thm:ccb_sparse_unitvar} (here $U$ and $V$ are chosen to be $Z_e)_{e \in N_{0j_1}\cup N_{0j_2}}$ and $(Z'_e)_{e \in N_{0j_1}\cup N_{0j_2}}$ respectively).
%\begin{equation}% \label{eq:offdiag_block_entry}
%    \cov({\bTheta}_{j_1}^{\top} (\bX_i \bX_i^{\top} {\bTheta}_{k_1}- \eb_{k_1}), {\bTheta}_{j_2}^{\top} (\bX_i \bX_i^{\top} {\bTheta}_{k_2}- \eb_{k_2})) = \bTheta_{j_1 j_2} \bTheta_{k_1 k_2} + \bTheta_{j_1 k_2}\bTheta_{j_2 k_1}.
%\end{equation} 

Having checked all the three conditions, we now obtain 
%arrive at the following bound for the relevant terms \eqref{eq:IV_def} under Case $1$ (i.e., $\bTheta_{j_1,j_2} \ne 0$) in $\mathrm{III}'_3$ by summing \eqref{eq:case1_IVj1j2} over all such $j_1,j_2 \in \cH_0, j_1 \ne j_2$,
\begin{eqnarray}\nonumber
&~& \sum_{m=1}^{\lambda_d} \sum_{\bTheta_{j_1j_2} = 0}\frac{\mathrm{IV}_{j_1j_2}(\alpha_m)}{\epsilon^2 d^2_{0}}
%\sum_{m=1}^{\lambda_{d}}\left|
%\frac{\sum_{\text{Case 1}} \cov(W_{j_1}(\alpha_m),W_{j_2}(\alpha_m))}{\epsilon^{2}d_{0}^{2}\alpha^{2}} 
%\right| 
\\ \nonumber
&\le& \sum_{m=1}^{\lambda_{d}}\bigg\{\frac{1 }{\epsilon^2 d_0^2 } \cdot \frac{\log d }{ \alpha_m p}\bigg( \sum_{k_1  \in N_{0 j_1}, k_2 \in N_{0 j_2}, k_1 \ne k_2} \Indrbr{\bTheta_{j_1 k_2}\bTheta_{j_2 k_1} \ne 0}\bigg)  \bigg\}\\ 
&\le&  \frac{C_{\bTheta}|S|\log d}{ \epsilon^{2} d_0 p}  \left(\frac{1}{d_0}\sum_{m=1}^{\lambda_d} \frac{C'}{\alpha_m} \right), \label{eq:case1_terms}
%\sum_{(j_1,j_2,k_1,k_2) \in S} \bTheta_{j_1 k_2}\bTheta_{j_2 k_1} 
\end{eqnarray}
where $S$ represents the set
$$
%\nonumber
S = \{(j_1,j_2,k_1,k_2): j_1,j_2\in \cH_0, j_1\ne j_2,k_1 \ne k_2 , \bTheta_{j_1 j_2} = \bTheta_{j_1 k_1} = \bTheta_{j_2 k_2} = 0, \bTheta_{j_1 k_2} \ne 0, \bTheta_{j_2 k_1}\ne 0\}
$$
as defined in Section \ref{sec:hub_selection}, and $C_{\bTheta}$ is some universal constant over $\bTheta \in \cU(M,s, r_0)$. Finally, combining \eqref{eq:case1_terms} with \eqref{eq:case2_terms}, we obtain the following bound on $ \mathrm{III}'_{3}$,
\begin{eqnarray} \nonumber
\label{eq:III3'_bound}
 \mathrm{III}'_{3}
 &\le& \frac{C_{\bTheta}|S|\log d}{\epsilon^{2} d_0 p}  \left(\frac{1}{d_0}\sum_{m=1}^{\lambda_d} \frac{C'}{\alpha_m} \right)+ \frac{1}{\epsilon^2 d_0}  \sum_{m=1}^{\lambda_d} \frac{C'}{\alpha_m} 
%\frac{C''t_{\lambda_d}}{\rho \epsilon^2 d_{0}h_d} 
\\ \nonumber
 &= &  \left(1+ \frac{C_{\bTheta}|S|\log d}{d_0 p}
 \right)\cdot \frac{1}{\epsilon^2 d_0}  \sum_{m=1}^{\lambda_d} \frac{C'}{\alpha_m} 
%  \cdot\frac{t_{\lambda_d}}{h_d  } 
 \\ \nonumber
 &\le& \frac{C''t_{\lambda_d}}{\rho \epsilon^2 d_{0}h_d} 
%  \cdot\frac{t_{\lambda_d}}{h_d  } 
 \left(1+ \frac{C_{\bTheta}|S|\log d}{d_0 p}
 \right),
\end{eqnarray}
where the last inequality holds due to the same derivations for $\mathrm{III}_1$ in the proof of Lemma \ref{lem:moments_Ij}.
% Similar as previous derivation, it suffices to show the term $\frac{|S|\log d}{d_0 p} = o(1)$, which actually holds under the trade-off condition.
\end{proof}

\begin{lemma} \label{lem:IW_diff_bound}
Recall the definitions of $I_{j}(\alpha)$ and $W_{j}(\alpha)$ in \eqref{eq:Ijdef} and  \eqref{eq:Wjdef}, for $j_1, j_2 \in \cH_0, j_1 \ne j_2$, when $\alpha \in [\alpha_L,1]$, we have
  \begin{equation}
  \label{eq:lem_IW_diff_bound}
{\big| \EE{I_{j_1}(\alpha)I_{j_2}(\alpha)} -  \EE{W_{j_1}(\alpha)W_{j_2}(\alpha)} \big|   } \le    {\eta(d,n,\zeta_1,\zeta_2, \alpha_L)} {\alpha}.
\end{equation}
\end{lemma}
\begin{proof}[Proof of Lemma \ref{lem:IW_diff_bound}]
First express $\big| \EE{I_{j_1}(\alpha)I_{j_2}(\alpha)} -  \EE{W_{j_1}(\alpha)W_{j_2}(\alpha)} \big|$ as
\begin{align} \nonumber
   ~& \big| \EE{I_{j_1}(\alpha)I_{j_2}(\alpha)} -  \EE{W_{j_1}(\alpha)W_{j_2}(\alpha)} \big| \\ \nonumber
   =~&  \Big|  \PPP \Big( \max_{e \in N_{0 j_1} } \sqrt{n}|\tdTheta_e | \ge \hat c(\alpha, N_{0j_1}), 
\max_{e \in N_{0 j_2} } \sqrt{n}|\tdTheta_e  | \ge \hat c(\alpha, N_{0j_2})\Big )  \\ \nonumber
 ~~& -  \PPP \Big( \max_{e \in N_{0j_1} } |Z_{e}| \ge  c(\alpha, N_{0j_1}), \max_{e \in N_{0j_1} } 
 |Z_{e}| \ge  c(\alpha, N_{0j_2}) \Big)  \Big|\\ 
 \begin{split}\label{eq:IW_diff_expand}
   =~&  \Big|  \PPP \Big( T_{ N_{0j_1}} \ge \hat c(\alpha, N_{0j_1}), 
T_{ N_{0j_2}} \ge \hat c(\alpha, N_{0j_2}) \Big)  \\
 ~~& -  \PPP \Big( \max_{e \in N_{0j_1} } |Z_{e}| \ge  c(\alpha, N_{0j_1}), \max_{e \in N_{0j_1} } |Z_{e}| \ge  c(\alpha, N_{0j_2}) \Big)  \Big|,
 \end{split}
\end{align}
where the second equality holds by the definition of $T_{E}$ in \eqref{eq:TE} and the definitions of $N_{0j_1},N_{0j_2}$. Now proving the bound in \eqref{eq:lem_IW_diff_bound} is reduced to showing
\begin{align} \label{eq:IW_diff_bound2}
   \begin{split}
  &  \Big|  \PPP \Big( T_{ N_{0j_1}} \ge \hat c(\alpha, N_{0j_1}), 
T_{ N_{0j_2}} \ge \hat c(\alpha, N_{0j_2}) \Big) 
 -  \PPP \Big( \max_{e \in N_{0j_1} } |Z_{e}| \ge  c(\alpha, N_{0j_1}), \max_{e \in N_{0j_1} } |Z_{e}| \ge  c(\alpha, N_{0j_2}) \Big)  \Big| \\
 ~~&\le   {\eta(d,n,\zeta_1,\zeta_2, \alpha_L)} {\alpha}.
 \end{split}
\end{align}
We first relate the notations in the above expression to the notations in Appendix \ref{sec:gmb_theory}: $T_{ N_{0j_1}}, T_{ N_{0j_2}}$ correspond to $T$; $\hat c(\alpha, N_{0j_1}), \hat c(\alpha, N_{0j_2})$ correspond to $q_{\cqt}(\alpha, T^{\cB})$; $ \max_{e \in N_{0j_1} } |Z_{e}|$, $\max_{e \in N_{0j_2} } |Z_{e}|$ correspond to $T_{\bZ}$; $ c(\alpha, N_{0j_1}), c(\alpha, N_{0j_2})$ correspond to $q(\alpha; T_{\bZ})$. In Appendix \ref{sec:gmb_theory}, we prove Propositions \ref{prop:cramer_gmb} and \ref{prop:approx_cramer_gmb}. And the strategy can be used to derive the bound on \eqref{eq:IW_diff_expand}. First, we note that $T_{ N_{0j_1}}, T_{ N_{0j_2}}$ satisfy the conditions of Proposition \ref{prop:approx_cramer_gmb}, i.e., \eqref{eq:approx_cond1} and \eqref{eq:approx_cond2}. This is due to the same derivations as the first parapraph of the proof of Lemma \ref{lem:moments_Ij}. Since the proving strategy is quite similar, we omit the proof of \eqref{eq:IW_diff_bound2} for simplicity. Instead, we prove \eqref{eq:simple_IW_diff_bound}, i.e.,
%it suffices to show \eqref{eq:simple_IW_diff_bound}, i.e.,
%We can apply similar strategy as in the proofs of Propositions \ref{prop:cramer_gmb} and \ref{prop:approx_cramer_gmb} to derive the bound on \eqref{eq:IW_diff_expand}. Note Proposition \ref{prop:approx_cramer_gmb} can be seen as an extended version of Propositions \ref{prop:cramer_gmb}. Though the test statistics and the approximate quantiles in \eqref{eq:lem_IW_diff_bound} are under the same setup as in Proposition \ref{prop:approx_cramer_gmb}, \jlmargin{we omit the proof of \eqref{eq:IW_diff_expand} for simplicity}{rewrite}.
%Following a similar proof of Proposition \ref{prop:approx_cramer_gmb}, we have \jlmargin{what you need here}{} 
%Instead, we consider the same setup as in Propositions \ref{prop:cramer_gmb} and give the proof of the following result: 
when $\alpha \in [\alpha_L,1]$,
\begin{align}
  \begin{split}\label{eq:simple_IW_diff_bound}
\mathrm{D} := ~ & \Big| \PP{T_{\bY_1} \ge  q_{\cqt}(\alpha;T_{\bW_1}), 
T_{\bY_2} \ge q_{\cqt}(\alpha;T_{\bW_2})} 
 -  \PP{ T_{\bZ_1} \ge q(\alpha; T_{\bZ_1}),  T_{\bZ_2} \ge q(\alpha; T_{\bZ_2}) }  \Big| \\
 \le ~&   C \alpha \rbr{ \frac{ (\log d)^{11/6}}{n^{1/6}\alpha_L^{1/3}}
+\frac{   (\log d)^{19/6}}{n^{1/6}}}   ,
 \end{split}
 \end{align}
where $T_{\bY_1}, T_{\bY_2}$ correspond to $\breve{T}_{E}$ with $E=N_{0j_1}, N_{0j_2}$ respectively,  $T_{\bW_1}, T_{\bW_2}$ correspond to $\breve{T}_{E}^\cB$ with $E=N_{0j_1}, N_{0j_2}$ respectively, and $T_{\bZ_1} = \max_{e \in N_{0j_1} } |Z_{e}|, T_{\bZ_2} = \max_{e \in N_{0j_2} } |Z_{e}|$. As for the quantiles, $q_{\cqt}(\alpha;T_{\bW_1}), q_{\cqt}(\alpha;T_{\bW_2})$ are the Gaussian multiplier bootstrap quantiles based on $T_{\bW_1},T_{\bW_2}$. $q(\alpha; T_{\bZ_1}), q(\alpha; T_{\bZ_2})$ are the quantiles of the Gaussian maxima $T_{\bZ_1}, T_{\bZ_2}$. 
%Hence \eqref{eq:simple_IW_diff_bound} is under the same setup as in Proposition \ref{prop:cramer_gmb}.
Denote $A_1 = \{ T_{\bY_1} \ge  q_{\cqt}(\alpha;T_{\bW_1}) \}, A_2 = \{ T_{\bY_2} \ge  q_{\cqt}(\alpha;T_{\bW_2})\}$, $B_1= \{ T_{\bY_1} \ge q(\alpha; T_{\bZ_1})\}$, $B_2= \{ T_{\bY_2} \ge q(\alpha; T_{\bZ_2}) \}$, we have
\begin{align}\nonumber
\mathrm{D}_{12} ~:= ~&   \big| \PP{T_{\bY_1} \ge  q_{\cqt}(\alpha;T_{\bW_1}), 
T_{\bY_2} \ge q_{\cqt}(\alpha;T_{\bW_2})} 
 -  \PP{ T_{\bY_1} \ge q(\alpha; T_{\bZ_1}),  T_{\bY_2} \ge q(\alpha; T_{\bZ_2}) }  \big| \\ \nonumber
 \le ~&  \PP{ (A_1 \cap A_2)  \ominus (B_1 \cap  B_2)} \\ \nonumber
 = ~&  \PP{ (A_1 \cap A_2)  \cap  (B_1^{c} \cup  B_2^c )}  + \PP{ (B_1 \cap B_2)   \cap (A_1^{c} \cup  A_2^c )} \\  \nonumber
  \le ~&  \PP{ A_1 \cap  B_1^{c}  }  +   \PP{  A_2  \cap  B_2^c }  +  \PP{  B_1 \cap  A_1^{c} } + \PP{  B_2   \cap  A_2^c } \\  \nonumber
   = ~&  \PP{ (A_1 \cap  B_1^{c}) \cup (B_1 \cap  A_1^{c})  }  +   \PP{  (A_2  \cap  B_2^c )\cup (B_2   \cap  A_2^c)} \\ 
    = ~&  \PP{ A_1 \ominus B_1  }  +   \PP{ A_2 \ominus B_2  }. \label{eq:D12_2diffs}
\end{align}
%We directly apply the derivations in the proof of Proposition \ref{prop:cramer_gmb} to bound $ \PP{ A_1 \ominus B_1  }$ and $\PP{ A_2 \ominus B_2  }$. 

%Specifically, \eqref{eq:prop1_rate} and \eqref{eq:EE_maxdiff_rate} together imply the bound on $\mathrm{II}$ in the proof of Proposition \ref{prop:cramer_gmb}, which is equivalently expressed as
%\begin{equation}\nonumber
%   \left|
%     {\mathbb{P}(T_{\bY} > q_{\cqt}(\alpha;T_{\bW}) )} - \mathbb{P}(T_{\bY} > q(\alpha;T_{\bZ}) )
%     \right| \le      C' \alpha \rbr{ \frac{ (\log d)^{11/6}}{n^{1/6}\alpha_L^{1/3}}
% +\frac{   (\log d)^{19/6}}{n^{1/6}}}. 
% \end{equation}
By \eqref{eq:prop1_rate} and \eqref{eq:EE_maxdiff_rate}, we can bound \eqref{eq:D12_2diffs} as 
\begin{equation} \label{eq:lem_D12_bound}
\mathrm{D}_{12} \le  \PP{ A_1 \ominus B_1  }  +   \PP{ A_2 \ominus B_2  } \le  ~2C' \alpha \rbr{ \frac{ (\log d)^{11/6}}{n^{1/6}\alpha_L^{1/3}}
+\frac{   (\log d)^{19/6}}{n^{1/6}}}.
\end{equation}
By the triangle inequality, we have the following bound on $\mathrm{D}$,
\begin{align}\nonumber 
  \mathrm{D} = ~& \Big| \PP{T_{\bY_1} \ge  q_{\cqt}(\alpha;T_{\bW_1}), 
T_{\bY_2} \ge q_{\cqt}(\alpha;T_{\bW_2})} 
 -  \PP{ T_{\bZ_1} \ge q(\alpha; T_{\bZ_1}),  T_{\bZ_2} \ge q(\alpha; T_{\bZ_2}) }  \Big|\\ \nonumber
 \le ~& \mathrm{D}_{12} + {
 \Big| \PP{T_{\bY_1} \ge  q(\alpha; T_{\bZ_1}), 
T_{\bY_2} \ge q(\alpha; T_{\bZ_2})} 
 -  \PP{ T_{\bZ_1} \ge q(\alpha; T_{\bZ_1}),  T_{\bZ_2} \ge q(\alpha; T_{\bZ_2}) }  \Big|
 }\\ \nonumber 
 \le ~& \mathrm{D}_{12} + \Big|  \PP{  T_{\bY_1} \ge  q(\alpha; T_{\bZ_1})}  -  \PP{  T_{\bZ_1} \ge  q(\alpha; T_{\bZ_1})}\big| + \big|  \PP{  T_{\bY_2} \ge  q(\alpha; T_{\bZ_2})}  -  \PP{  T_{\bZ_2} \ge  q(\alpha; T_{\bZ_2})}\Big|  \\
&~ + \underbrace{
 \Big| \PP{ \{ T_{\bY_1} \ge  q(\alpha; T_{\bZ_1})\} \cup\{  
T_{\bY_2} \ge q(\alpha; T_{\bZ_2}) \} } 
 -  \PP{ \{ T_{\bZ_1} \ge q(\alpha; T_{\bZ_1})\} \cup \{  T_{\bZ_2} \ge q(\alpha; T_{\bZ_2})\} }  \Big|
 }_{\mathrm{D}'_{12}},
 \label{eq:lem_D_bound} 
\end{align}
where the last inequality holds since $ \PP{A \cap B} = \PP{A} + \PP{B} - \PP{A \cup B} $. For the second term and the third term in \eqref{eq:lem_D_bound}, we can directly apply the results \eqref{eq:I_rate} in Proposition \ref{prop:cramer_gmb} and bound them as 
\begin{align} \label{eq:direct_bounds}
\begin{split} 
~~& \Big|  \PP{  T_{\bY_1} \ge  q(\alpha; T_{\bZ_1})}  -  \PP{  T_{\bZ_1} \ge  q(\alpha; T_{\bZ_1})}\big| + \big|  \PP{  T_{\bY_2} \ge  q(\alpha; T_{\bZ_2})}  -  \PP{  T_{\bZ_2} \ge  q(\alpha; T_{\bZ_2})}\Big| \\  
~~& \le C \alpha \cdot \frac{ (\log d)^{19/6}}{n^{1/6}} 
\end{split}
\end{align}
for some constant $C$. 
Regarding the term $\mathrm{D}'_{12}$, we assume $q(\alpha; T_{\bZ_2}) = q(\alpha; T_{\bZ_2}):=t$ without loss of generality. This is because $q(\alpha; T_{\bZ_1}), q(\alpha; T_{\bZ_2})$ are all deterministic values and we can rescale the random vector inside one of the maximum statistics $T_{\bZ_1}, T_{\bZ_2}$. Now we rewrite $\mathrm{D}'_{12}$ based on $q(\alpha; T_{\bZ_2}) = q(\alpha; T_{\bZ_2})=t$ and derive the following bound:
\begin{eqnarray}\label{eq:D12prline1}
\mathrm{D}'_{12} 
& = & \Big| \PP{ \max \{T_{\bY_1},T_{\bY_2}\} \ge t} 
 -  \PP{ \max \{T_{\bZ_1},T_{\bZ_2}\} \ge t }  \Big| \\ \nonumber 
&\le &   \frac{ C''(\log d)^{19/6}}{n^{1/6}} \cdot \PP{ \max \{T_{\bZ_1},T_{\bZ_2}\} \ge t}\\ \nonumber
&\le &   \frac{ C''(\log d)^{19/6}}{n^{1/6}} \cdot \big( \PP{ T_{\bZ_1} \ge q(\alpha; T_{\bZ_1})}+ \PP{ T_{\bZ_2} \ge q(\alpha; T_{\bZ_2})} \big) \\ 
& = &  2 C'' \alpha \cdot \frac{ (\log d)^{19/6}}{n^{1/6}},\label{eq:lem_D12pr_bound}
\end{eqnarray} 
where the first inequality holds by applying Corollary 5.1 of \cite{arun2018cram} similarly as in the derivation of \eqref{eq:I_rate}. Here we briefly explain why Corollary 5.1 of \cite{arun2018cram} is applicable to \eqref{eq:D12prline1}.
Note that $\max \{T_{\bY_1},T_{\bY_2}\} = {T_{\bY_{12}}}$ is the maximum statistic with respect to the random vectors which concatenate the random vectors involved in $T_{\bY_1}, T_{\bY_2}$. Write $T_{\bY_1}, T_{\bY_2}$ explicitly as
\begin{equation} \nonumber
  T_{\bY_1}:= \left \Vert \frac{1}{\sqrt{n}}\sum_{i=1}^{n}\bY^{(1)}_{i}\right \Vert_{\infty},
\quad T_{\bY_2}:=\left\Vert\frac{1}{\sqrt{n}}\sum_{i=1}^{n}\bY^{(2)}_{i}\right\Vert_{\infty},
\end{equation}
and denote $\bY^{(12)}_i = (\bY^{(1)}_i, \bY^{(2)}_i)$, then $T_{\bY_{12}}$ is defined as 
\begin{equation}\nonumber
  T_{\bY_{12}}:= \left \Vert \frac{1}{\sqrt{n}}\sum_{i=1}^{n}\bY^{(12)}_{i}\right \Vert_{\infty}.
\end{equation}
By the definition of $\bZ_1, \bZ_2$, we have ${\rm Cov}((\bZ_1^\top, \bZ_2^\top)^\top) = {\rm Cov}((\bY_1^\top, \bY_2^\top)^\top)$.
% the Gaussian random vector $(Z_{e})_{e \in N_{0j_1}, N_{0j_2} }$ shares the same covariance matrix structure as the i.i.d. random vectors $\{ \bY^{(12)}_i \}_{i=1}^n$ by the construction of $ \{Z_{e}\}_{e \in N_{0j_1}}, \{Z_e\}_{N_{0j_2}}$ in the proof of Lemma \ref{lem:crossterm_Ij}. 
Hence we can apply Corollary 5.1 of \cite{arun2018cram} to \eqref{eq:D12prline1}.
Now we combine \eqref{eq:lem_D12_bound}, \eqref{eq:lem_D_bound}, \eqref{eq:direct_bounds} with \eqref{eq:lem_D12pr_bound} and obtain the following bound
$$
\mathrm{D} \le  C \alpha \rbr{ \frac{ (\log d)^{11/6}}{n^{1/6}\alpha_L^{1/3}}
+\frac{   (\log d)^{19/6}}{n^{1/6}}}, 
$$
 for some constant $C$, thus \eqref{eq:simple_IW_diff_bound} is established. The above strategy of obtaining \eqref{eq:simple_IW_diff_bound} can be similarly applied to the term in \eqref{eq:IW_diff_expand}, then establishes the bound in \eqref{eq:lem_IW_diff_bound}.

%$T=T_{E}, T_{\bY} = \breve{T}_E, T^\cB = T_{E}^{\cB}, T_{\bW} = \breve{T}_{E}^{\cB}$ with $E=N_{0j}$ 
\end{proof}

% \section{Proofs in Section \ref{sec:bipartite_selection}}
\subsection{Proof of Theorem \ref{thm:fdr_linear}}
\label{app:pf:thm:fdr_linear}
\begin{proof}[Proof of Theorem \ref{thm:fdr_linear}]
\label{pf:thm:fdr_linear}
 Throughout the proof, we condition on the design matrix $\bX$, but without explicitly writing it out in order to simplify the notation. In the context of selecting hub response variables, we recall $\cH_0 = \{j\in [d_1]: \norm{\bTheta_{j}}_0 \ge k_\tau\}$ and $d_0 = |\cH_0|$. For a non-hub response variable $j\in\cH_{0}$, let $N_{0j}$ be the set of its null covariates, i.e., $N_{0j}=\{(j,k):\bTheta_{jk}=0\}$. 
 
 To establish FDR control, we follow the same derivations as in the proof of Theorem \ref{thm:fdr_hub}. Specifically, it suffices to bound
\begin{eqnarray}\nonumber
%     && \max_{1\le m \le \lambda_{d}} 
%     \mathbb{P}(\left |
%     \frac{\sum_{j\in \cH_{0}} I_j(\alpha_{m})}{d_{0}\alpha_{m}} - 1 \right |\ge \epsilon)\\
%     &\le& \sum_{m=1}^{\lambda_{d}} 
%   \mathbb{P}(\left |
%     \frac{\sum_{j\in \cH_{0}} I_j(\alpha_{m})}{d_{0}\alpha_{m}} - 1 \right | \ge \epsilon)\\ 
%     &\le & \sum_{m=1}^{\lambda_{d}} \frac{\EEE[\sum_{j\in \cH_{0}} I_j(\alpha_{m}) - d_{0}\alpha_{m}]^2}{\epsilon^{2} d_{0}^{2}\alpha_{m}^{2}}\\
    & ~ & \sum_{m=1}^{\lambda_{d}} \frac{\mathbb{\var}[\sum_{j\in \cH_{0}} I_j(\alpha_{m}) - d_{0}\alpha_{m}]}{\epsilon^{2} d_{0}^{2}\alpha_{m}^{2}} +  \sum_{m=1}^{\lambda_{d}} \frac{(\EEE[\sum_{j\in \cH_{0}} I_j(\alpha_{m}) - d_{0}\alpha_{m}])^2}{\epsilon^{2} d_{0}^{2}\alpha_{m}^{2}}  \\ \label{eq:linear_cross_term}
    &~& + \sum_{m=1}^{\lambda_{d}}\frac{ 
    \sum_{j_1,j_2 \in \cH_0, j_1 \ne j_2} \cov(I_{j_1}(\alpha_m), I_{j_2}(\alpha_m))
    }{\epsilon^{2} d_{0}^{2}\alpha_{m}^{2}}:= \mathrm{III}_{1} + \mathrm{III}_{2}  +0
\end{eqnarray}
for any $\epsilon>0$. In the above terms, the sequence $\{\alpha_m\}_{m=1}^{\lambda_d}$ is chosen similarly as in the proof of Theorem \ref{thm:fdr_hub} and $I_j(\alpha)$ is defined as
\begin{equation}\nonumber
    I_j(\alpha) = \Indrbr{\max_{e \in N_{0j} } \sqrt{n}|\tdTheta_e| \ge \hat c(\alpha, N_{0j})},
\end{equation}
where $\tdTheta_j$ is the debiased Lasso estimator defined in \eqref{eq:dlasso}. Note that the cross term in \eqref{eq:linear_cross_term} equals zero as $\cov(I_{j_1}(\alpha_m), I_{j_2}(\alpha_m))=0$. This is because $\bY^{(j)},j\in [d_1]$ are conditionally independent given $\Xb$. Therefore it suffices to bound $\mathrm{III}_{1}$ and $\mathrm{III}_{2}$. By applying Lemma \ref{lem:moments_Ij} with the term $\eta(d,n,\zeta_1,\zeta_2, \alpha_L)$ replaced by $ \eta_0(d_1,d_2, n, \zeta_1,\zeta_2, \alpha_L)$ in Lemma \ref{lem:multitask_cramer_approx}, \eqref{eq:linear_cross_term} can be controlled by 
$$
\mathrm{III}_{2}  +\mathrm{III}_{2}\le
\frac{C't_{\lambda_{d_2}} }{\epsilon^2h_{d_2} } \left(
\frac{d_1}{d_0|\cB|}+  \eta_0(d_1,d_2, n, \zeta_1, \zeta_2, \alpha_L)
% \eta_2(d_1, d_2, n, \delta_n,\eta_n, \zeta_{1}, \zeta_{2})
\right),
$$
where $\alpha_L =q|\cB|/d_1$ and $t_{\lambda_{d_2}} ,h_{d_2}$ are similarly defined as in the proof of Theorem \ref{thm:fdr_hub}. According to Lemma \ref{lem:multitask_cramer_approx}, we have the explicit form of $\eta_0(d_1,d_2, n, \zeta_1, \zeta_2, \delta, \alpha_L)$:
\begin{equation} \nonumber
\eta_0(d_1,d_2, n, \zeta_1, \zeta_2,  \delta, \alpha_L)= \zeta_1 \log d_2 + (\log d_2)^{5/2}\delta^{1/2} + \frac{\eta + \zeta_2}{\alpha_L},
%    \eta_0(d_1,d_2, n, \alpha_L),
%    = \eta_1(d_1, n, \delta_{n}, \eta_n)+\zeta_1 \sqrt{\log d_2}+{\zeta_2}/\alpha_L 
    % \eta_2(d_2,n,\alpha_L) = \eta_1(d_2,n,\alpha_L)+\zeta_1 q(\alpha_{L}/2;T_{\bZ})+{\zeta_2}  
\end{equation}
where 
%$\eta_1(d_2,n,\alpha_L)=
%{\pi(\delta_{n})}+ \eta_n$ with $\pi(\maxdiff) = [A(\maxdiff)+1]e^{M_1(\log d)^{3/2} A(\maxdiff)}-1$, and 
 $\zeta_1 = O({s\log d_2}/{\sqrt{n}})$, $\zeta_2 = O( e^{-c_1n}  + d_2^{-\tilde{c_0}\wedge c_2})$, $\delta$ satisfies $\frac{1}{\delta}\sqrt{\frac{s \log d_2}{n }} =O(1) $ and $\eta = e^{-c_1 n} + \frac{1}{d_2} + \frac{1}{n \delta^2 }$. By rearranging, we obtain the following bound on $\mathrm{III}_{2}  +\mathrm{III}_{2}$:
 % can finally bound \eqref{eq:linear_cross_term} as
%\begin{equation}\nonumber
%    \frac{\log d_2}{\epsilon^2}\rbr{
%    \frac{ 1}{d_0 \rho}+ \frac{s(\log d_2)^{3/2} }{n^{1/2}} + \frac{s^{1/4}(\log d_2)^{11/4}}{n^{1/5}}
%    + \frac{1}{\rho}\big(\frac{1}{d_2} + \frac{1}{n^{1/5}s \log d_2}   + {e^{-c_1n}  + d_2^{-\tilde{c_0}\wedge c_2}} \big)
%    }.
%\end{equation}
\begin{equation}\nonumber
    \frac{\log d_2}{\epsilon^2}\rbr{
    \frac{ 1}{d_0 \rho}+ \frac{s(\log d_2)^{2} }{n^{1/2}} + 
    (\log d_2)^{5/2}\delta^{1/2} + \frac{1}{n\delta^2 \rho} + 
     \frac{1}{\rho}\big(\frac{1}{d_2}   + {e^{-c_1n}  + d_2^{-\tilde{c_0}\wedge c_2}} \big)
    }.
\end{equation}
where $\rho = \cB/d_1$. We choose $\delta $ to be $ \frac{1}{(n \rho)^{2/5}\log d_2}$ and have $\delta> \frac{1}{n^{2/5}\log d_2}$ (since $\rho <1$). Thus this choice of $\delta$ satisfies the requirement in Lemma \ref{lem:multitask_cramer_approx}. Finally we have $\eqref{eq:linear_cross_term}$ is bounded as
\begin{equation}\nonumber
    \frac{\log d_2}{\epsilon^2}\rbr{
    \frac{ 1}{d_0 \rho}+ \frac{s(\log d_2)^{2} }{n^{1/2}} + 
     \frac{(\log d_2)^2}{(n \rho)^{1/5}} 
    + \frac{1}{\rho d_2} 
    %\big(\frac{1}{d_2}   + {e^{-c_1n}  + d_2^{-\tilde{c_0}\wedge c_2}} \big)
    }.
\end{equation}
%$$
%\left( \frac{d_2\log d_2 + d_0}{d_0 d_2 \rho} + \frac{1}{n^{1/5}s \log d_2 \rho} \right)
%$$
Under the stated assumption in Theorem \ref{thm:fdr_linear}, the above term is $o(1)$. Thus the FDP control result is established. Due to similar derivations as in Theorem \ref{thm:fdr_hub}, the FDR control result follows. 
\end{proof}

\subsection{Ancillary lemmas for Theorem \ref{thm:fdr_linear}}
To prove FDR control, we will establish a key result, i.e., Lemma \ref{lem:multitask_cramer_approx} in this section. Recall that in Section \ref{sec:bipartite_selection}, we utilize the following result 
  \begin{equation} \nonumber %\label{eq:Z_Delta_terms}
    \sqrt{n} (\tdTheta_j - \bTheta_j) = Z_j + \lerr , \quad Z_j |\bX \sim \cN(0,\sigma_j^2 M\hSigma M^{\top}).
    \end{equation}
and approximate the quantile of the maximum statistics $T_{E} =\max_{(j,k) \in E}\sqrt{n}|\tdTheta_{jk}|$ by $T_{E}^\cN = \max_{(j,k) \in E} |Z_{jk}|$. Lemma \ref{lem:multitask_cramer_approx} basically establishes the Cram\'{e}r deviation bounds for such quantile approximation. Note that this lemma can be seen as a special case of Proposition \ref{prop:cramer_gmb} since the involving random vector $\sqrt{n} (\tdTheta_j - \bTheta_j)$ can be decomposed into a Gaussian random vector plus some error term. Hence we do not need to use the results in \cite{arun2018cram} to handle the case of a general random vector (and quantify Gaussian approximation errors). 

In this section, we will define some notations similar to the theoretical results in Appendix \ref{sec:gmb_theory}. First, we will drop the $j$-th subscript for simplicity. Without loss of generality, we prove relevant results for $E = \{(j,k): k \in [d_2]\} $ and drop the subscript $E$. Note the results hold for any $j\in [d_1]$ and any subset of $\{(j,k): k \in [d_2]\} $. Now we rewrite \eqref{eq:dlasso_normal} using new notations, i.e.,
  \begin{equation} \label{eq:Z_Delta_terms}
    \sqrt{n} (\tdTheta_j - \bTheta_j) = \bZ + \lerr , \quad \bZ |\bX \sim \cN(0,\sigma_j^2 M\hSigma M^{\top}),
    \end{equation}
   and denote its maximum by $T_{\bZ} = \maxnorm{\bZ}$. Intuitively, we can use the quantile of $T_{\bZ}$ to approximate the quantile of $T:= \sqrt{n}\maxnorm {\tdTheta_j - \bTheta_j}$. Since the covariance matrix $\sigma_j^2 M\hSigma M^{\top}$ of the Gaussian random vector $\bZ$ is not completely known, we can not directly compute its quantile (denoted by $q(\alpha;\bZ)$). Instead, we first estimate the unknown parameter $\sigma_j$ by $ \hat\sigma_j$, which is constructed according to \eqref{eq:scaled_lasso}. Then we define $\bW \sim \cN(0,\hat \sigma_j^2 M\hSigma M^{\top})$ (given the data $\bX, \bY^{(j)}$), and denote its maximum by $T_{\bW} = \maxnorm{W}$. We will approximate the unknown quantile of $T$ by the conditional quantile $q_{\xi}(\alpha;T_{\bW})$. Here we use the $\xi$ subscript to emphasize that we are conditioning on the data when defining such quantiles. 

Due to the existence of the term $\lerr$ in \eqref{eq:Z_Delta_terms}, there also exist additional estimation errors when we approximate the quantiles of $T$ by the conditional quantiles $q_{\xi}(\alpha;T_{\bW})$. Lemma \ref{lem:linear_maxdiff_bound} characterizes such approximation errors. As for the difference between the distributions of the two Gaussian random vectors $\bW$ and $\bZ$, Lemma \ref{lem:linear_maxdiff_bound} provides a bound on the maximal difference of their covariance matrices, which is denoted by $\maxdiff$. Finally, Lemma \ref{lem:multitask_cramer_approx} builds on these results and establishes the Cram\'{e}r-type deviation bounds for the quantile approximation of $T$.
\begin{lemma}\label{lem:linear_can_approx}
In the context of multiple linear models, we have
$$
 \mathbb{P}(|T-T_{\bZ}|>\zeta_{1})<\zeta_{2},
$$
where $\zeta_1 = O({s\log d_2}/{\sqrt{n}})$ and $\zeta_2 = O( e^{-c_1n}  + d_2^{-\tilde{c_0}\wedge c_2})$.
\end{lemma}
\begin{proof}[Proof of Lemma \ref{lem:linear_can_approx}]
By Theorem 2.5 in \cite{javanmard2014confidence}, we have 
\begin{align} \nonumber
 \sqrt{n} (\tdTheta_j - \bTheta_j) = \bZ + \lerr, \quad \bZ |\bX \sim \cN(0,\sigma_j^2 M\hSigma M^{\top}),
\end{align}
and
\begin{align} \nonumber
\PPP\left(\| \lerr \|_{\infty} \ge
\Big(\frac{16ac\, \sigma}{C_{\rm min}} \Big)\frac{s\log d_2}{\sqrt{n}}\right)\le  4\,e^{-c_1n}  + 4\,d_2^{-\tilde{c_0}\wedge c_2}\,.
\end{align}
Thus we immediately obtain the following bound on the difference between $T$ and $T_{\bZ}$:
$$
 \mathbb{P}(|T-T_{\bZ}|>\zeta_{1})<\zeta_{2}
$$
where $\zeta_1 = O({s\log d_2}/{\sqrt{n}})$ and $\zeta_2 = O( e^{-c_1n}  + d_2^{-\tilde{c_0}\wedge c_2})$.
\end{proof}

\begin{lemma}\label{lem:linear_maxdiff_bound}
For the the maximal difference term $\maxdiff = \norm{\hat{\sigma}^2 M\hSigma M^{\top} - \sigma^2 M\hSigma M^{\top}}_{\max} $, we have 
\begin{equation}\label{eq:linear_maxdiff_bound}
   	\PP{\maxdiff \ge \delta } \le \eta,
%    \sup_{\bX \in  \cE_n}\Pcmid{\maxdiff \ge \delta_n}{\bX}\le \eta_n,
\end{equation}
where $\delta$ satisfies $\frac{1}{\delta}\sqrt{\frac{s \log d_2}{n }} =O(1) $ and $\eta = O\rbr{e^{-c_1 n} + \frac{1}{d_2} + \frac{1}{n \delta^2 }}$.
% +\frac{1}{n^{1/2}s \log d_2}}$.
\end{lemma}
\begin{proof}[Proof of Lemma \ref{lem:linear_maxdiff_bound}]
To bound $\maxdiff$, we start with the term $|\hat{\sigma}/{\sigma} - 1|$. First we denote
$$
\cE_n= \cE_n(\phi_0, s_0, K) := \Big\{\bX \in \RR^{n \times d_1}: \min_{S:|S| \le s_0} \phi(\hat \bSigma, S) \ge \phi_0, \max_{j\in [d_1]}\bSigma_{jj}\le K, \bSigma = (\bX^\top \bX)/n
\Big\}
$$
similarly as in Theorem 7.(a) of \cite{javanmard2014confidence}, where $\phi(\hat \bSigma, S)$ is the compatibility constant as defined in Definition 1 of \cite{javanmard2014confidence}. Following the proof of Lemma $14$ in  \cite{javanmard2014confidence}, we have 
\begin{align} 
\nonumber
	\PP{ \Big| \frac{\hat{\sigma}}{\sigma} -1 \Big| \ge \epsilon }  
	&~\le   \PP{ \bX \notin \cE_n} + \sup_{\bX \in  \cE_n}\PPP \Big(\Big| \frac{\hat{\sigma}}{\sigma} -1 \Big| \ge \epsilon \, \Big| \, \bX \Big)
	\\ 
    &~ \le  4e^{-c_1 n} + 
    \sup_{\bX \in  \cE_n}\PPP\Big( \frac{\maxnorm{\bX^{\top} \bE}}{n \sigma^\star}\ge {\tilde \lambda}/{4} \, \Big| \,  \bX \Big) + 
     \sup_{\bX \in  \cE_n} \PPP\Big(  \Big| \frac{{\sigma}^\star}{\sigma} -1  \Big| \ge \frac{\epsilon}{10} \, \Big| \,\bX \Big)\label{eq:noise_bound1}
    \end{align}
where $\tilde \lambda =10 \sqrt{(2 \log d_2)/n}$, $\sigma^\star $ is the oracle estimator of $\sigma$ introduced in \cite{sun2012scaled} and $\epsilon$ satisfies $\frac{2 \sqrt{s}\tilde \lambda}{\sigma^\star \phi_0} \le \frac{\epsilon}{2} < a_0$. 
%we let $\epsilon = \epsilon_n := C_{\sigma}\sqrt{\frac{s \log d_2}{n^{4/5}}}$ with some constant $C_{\sigma}$ to satisfy the constraint on $\epsilon$. 

Now we separately bound the last two terms in \eqref{eq:noise_bound1}. The second term in \eqref{eq:noise_bound1} can be bounded by the derivation in the proof of Theorem 2 (ii) \cite{sun2012scaled}, i.e.,
\begin{eqnarray}\nonumber
    \sup_{\bX \in  \cE_n}\PPP\Big(  \frac{\maxnorm{\bX^{\top} \bE}}{n \sigma^\star}\ge {\tilde \lambda}/{4} \, \Big| \, \bX \Big) 
    &\le& d_2  \PPP\Big( |L_k|\ge \sqrt{{2 \log (d_2^{25/4})}/{n}} \, \Big| \, \bX \Big)\\
    \label{eq:boundfromA7}
    &\le& d_2 \cdot \frac{C}{ d_2^{25/4} \sqrt{ \log d_2}} \le \frac{C}{ d_2},
\end{eqnarray}
where $L_k$ is the $k$-th element of $\frac{\bX^{\top} \bE}{n \sigma^\star}$ and $\frac{\sqrt{n-1}L_k}{\sqrt{1-L_k^2}}$ follows the Student's t-distribution with $n-1$ degrees of freedom. Then \eqref{eq:boundfromA7} holds due to equation (A7) in \cite{sun2012scaled} together with the union bound. As for the last term in $\eqref{eq:noise_bound1}$, we note $n(\sigma^\star/\sigma)^2$ follows the $\chi^2_n$ distribution according to \cite{sun2012scaled}. Thus by Markov's inequality, we have
\begin{eqnarray}\label{eq:oracle_sigma_bound}
    \sup_{\bX \in  \cE_n}\PPP\Big( \Big| \frac{{\sigma}^\star}{\sigma} -1  \Big| \ge \frac{\epsilon}{10} \, \Big| \, \bX \Big) 
    &\le& \frac{C' \EE{(n(\sigma^\star/\sigma)^2 -n)^2} }{n^2 \epsilon^2 } \le \frac{2 C'}{n\epsilon^2}. 
    %= \frac{C_1}{n^{1/5}s \log d_2}.
\end{eqnarray}
Now we arrive at the following bound on $\maxdiff$:
\begin{eqnarray} \nonumber
    &~~&  \PPP\Big(\maxdiff \ge (\epsilon^2 + 2\epsilon ) \cdot \sigma^2 \norm{M\hSigma M^{\top}}_{\max} \Big) \\ \nonumber
    &=&  \PPP\Big(\norm{\hat{\sigma}^2 M\hSigma M^{\top} - \sigma^2 M\hSigma M^{\top}}_{\max} \ge {(\epsilon^2 + 2\epsilon ) \cdot \sigma^2 \norm{M\hSigma M^{\top}}_{\max} } \Big) \\ \nonumber
    &\le &  \PPP\Big(\Big| \frac{\hat{\sigma}}{\sigma} -1 \Big| \ge \epsilon \Big) \\ \nonumber    
    &\le &  4e^{-c_1 n} + \frac{C}{ d_2} + \frac{2 C'}{n\epsilon^2},
    % \le 4e^{-c_1 n} + \frac{C}{ d_2}+  \frac{C_1}{n^{1/5}s \log d_2},
\end{eqnarray}
%where $C''$ is some constant such that $C'' \ge \frac{\hat{\sigma}}{\sigma}+1$ with probability greater than . 
where the last inequality comes from combining \eqref{eq:noise_bound1}, \eqref{eq:boundfromA7} with \eqref{eq:oracle_sigma_bound}. Note that the proof of Theorem 16 in \cite{javanmard2014confidence} shows that $\norm{M\hSigma M^{\top}}_{\max} = O(1)$. Hence, we finally establish \eqref{eq:linear_maxdiff_bound} with 
$$
\delta = {\sigma^2 (\epsilon^2 + 2\epsilon )} {\norm{M\hSigma M^{\top}}_{\max}} = C_\sigma \epsilon, \quad \eta = O\rbr{e^{-c_1 n} + \frac{1}{d_2} + \frac{1}{n \delta^2}}.
$$
where $ C_\sigma$ is some constant and $\delta = C_{\sigma}\epsilon $  satisfies $\frac{1}{\delta}\sqrt{\frac{s \log d_2}{n }} =O(1) $ due to the choice of $\epsilon$.
\end{proof}

\begin{lemma} \label{lem:multitask_cramer_approx}
%[CGMB with approximation]\label{prop:approx_cramer_gmb}
%Under the same conditions as in Proposition \ref{prop:cramer_gmb} 
Based on the result about the approximation error between $T$ and $T_{\bZ}$ (Lemma \ref{lem:linear_can_approx}) and the bound on $\maxnorm{\Delta}$ in Lemma \ref{lem:linear_maxdiff_bound}, we have
%\begin{eqnarray} \nonumber
%    \mathbb{P}(|T-T_{\bZ}|>\zeta_{1})<\zeta_{2}
%    % \mathbb{P}(\mathbb{P}_{e}(|T^{\cB}-T_{\bW}|>\zeta_{1})>\zeta_{2})<\zeta_{2}
%\end{eqnarray}
%where $\zeta_1 = O({s\log d_2}/{\sqrt{n}}),\zeta_2 = O( e^{-c_1n}  + d_2^{-\tilde{c_0}\wedge c_2})$, we have
% and satisfy the condition that
% \begin{eqnarray}\label{eq:approx_rate}
%   \zeta_{2} = o\rbr{\alpha_{L}}, ~~~~ \zeta_{1}=o\rbr{\frac{1}{q(\alpha_{L};T_{\bZ}){\log d}}}
% \end{eqnarray}
%Recall that 
\begin{equation}
 \label{eq:linear_approx_cmb}
\sup_{\alpha \in [\alpha_{L},1]}
\left|
\frac{\mathbb{P}(T > q(\alpha;T_{\bW}) )}{\mathbb{P}(T_{\bZ} > q(\alpha;T_{\bZ}) )} - 1 
\right| = O\left(\eta_0(d_1,d_2, n, \zeta_1, \zeta_2, \delta,\alpha_L)
\right),
%= O\left(\eta_2(d_1, d_2, n,\delta_{n}, \eta_n, \zeta_1, \zeta_2)
%\right) 
% o(1), \bluecom{~?~ \max\left\{
% O\left(\frac{(\log d)^{19/6}}{n^{1/6}}\right)
% ,
% O\left(\frac{s(\log d)^{3/2}}{n^{1/2}}\right)
% \right\}}
\end{equation}
where $\eta_0(d_1,d_2, n, \zeta_1, \zeta_2, \delta, \alpha_L):= \zeta_1 \log d_2 + (\log d_2)^{5/2}\delta^{1/2}+ \frac{\eta + \zeta_2}{\alpha_L}$ with $\zeta_1 = O({s\log d_2}/{\sqrt{n}})$, $\zeta_2 = O( e^{-c_1n}  + d_2^{-\tilde{c_0}\wedge c_2})$. Here $\delta$ is a term to be determined and we requite $\frac{1}{\delta}\sqrt{\frac{s \log d_2}{n }} =O(1) $. $\eta$ depends on $\delta$, i.e.,
 $\eta =  e^{-c_1 n} + \frac{1}{d_2} + \frac{1}{n \delta^2}$.
%\eta_2(d_1, d_2, n, \delta_n,\eta_n, \zeta_{1}, \zeta_{2}, \alpha_L)= 
%\eta_1(d_1, n, \delta_{n}, \eta_n)+\zeta_1 \sqrt{\log d_2}+{\zeta_2}/\alpha_L$.
\end{lemma}
\begin{proof}[Proof of Lemma \ref{lem:multitask_cramer_approx}]
First we have $ \mathbb{P}(|T-T_{\bZ}|>\zeta_{1})<\zeta_{2}
$ by Lemma \ref{lem:linear_can_approx}, thus we obtain 
\begin{equation}\nonumber
\left|
\frac{\mathbb{P}(T > q(\alpha;T_{\bW}) )}{\mathbb{P}(T_{\bZ} > q(\alpha;T_{\bZ}) )} - 1 
\right| 
 \le \max\{\mathrm{II}_{1},\mathrm{II}_{2}\}+\frac{2\zeta_{2}}{\alpha}
\end{equation}
for $ \alpha \in [\alpha_{L},1]$, where $\mathrm{II}_{1}$ and $\mathrm{II}_{2}$ are defined as: 
\begin{equation}\nonumber
\mathrm{II}_{1} := \left|
\frac{\mathbb{P}(T_{\bZ} > q(\alpha;T_{\bW}) +\zeta_1 )}{\mathbb{P}(T_{\bZ} > q(\alpha;T_{\bZ}) )} - 1 \right|,\quad
\mathrm{II}_{2} := \left|
\frac{\mathbb{P}(T_{\bZ} > q(\alpha;T_{\bW}) -\zeta_1 )}{\mathbb{P}(T_{\bZ} > q(\alpha;T_{\bZ}) )} - 1 \right|.
\end{equation}
The above two terms can be bounded similarly. Take $\mathrm{II}_{1}$ as an example, we use similar strategy as in Proposition \ref{prop:cramer_gmb}. Consider the event $S:=\{\maxdiff \le \delta \}$ where $\delta$ satisfies $\frac{1}{\delta}\sqrt{\frac{s \log d_2}{n }} =O(1)$, we apply Lemma \ref{lem:compqt} and bound $\mathrm{II}_{1}$ by
%further by $\mathrm{II}_{12}+\mathrm{II}_{22}$, which are defined as below,
 $$
%\underbrace{\frac{1}{\alpha} \Big|
%\PPP \Big( T_{\bZ} > q(\frac{\alpha}{1 - \pi(\delta_n)};T_{\bZ}) + \zeta_1  \Big) - \PPP(T_{\bZ} > q(\alpha;T_{\bZ}) )
%\Big|}_{\mathrm{II}_{11}} 
\frac{1}{1 - \pi(\delta)}\cdot \mathrm{II}_{11} + \mathrm{II}_{12}+ \frac{ \PP{\maxdiff > \delta}}{\alpha},
$$
where $\mathrm{II}_{11}$ and $\mathrm{II}_{11}$ are defined as
\begin{eqnarray}\nonumber
\mathrm{II}_{11} &:=&  \frac{1 - \pi(\delta)}{\alpha} \Big|
\PPP \Big( T_{\bZ} > q(\frac{\alpha}{1 - \pi(\delta)};T_{\bZ}) + \zeta_1  \Big) - \PPP(T_{\bZ} > q(\frac{\alpha}{1 - \pi(\delta)};T_{\bZ}) )
\Big|,
\\
\mathrm{II}_{12} &:=& \frac{1}{\alpha} \Big|
\PPP \Big( T_{\bZ} > q(\frac{\alpha}{1 - \pi(\delta)};T_{\bZ})   \Big) - \PPP(T_{\bZ} > q(\alpha;T_{\bZ}) )  \Big|,
\end{eqnarray}
where $\pi(\maxdiff) = [A(\maxdiff)+1]e^{M_1(\log d)^{3/2} A(\maxdiff)}-1$.
By applying the part 3 of Theorem 2.1 in \cite{arun2018cram} (with $r+ \epsilon = q(\frac{\alpha}{1 - \pi(\delta)};T_{\bZ}) + \zeta_1, r- \epsilon = q(\frac{\alpha}{1 - \pi(\delta)};T_{\bZ})$) to the Gaussian random vector $\bZ$, we have
\begin{equation}
  \mathrm{II}_{11} \le K_4 \zeta_1 \big( q(\frac{\alpha}{1 - \pi(\delta)};T_{\bZ}) + {\zeta_1}/{2}
  \big) \le C \zeta_1 \log d_2.
\end{equation}
where the second inequality holds due to the similar reason stated in the proof of Proposition \ref{prop:approx_cramer_gmb}. And the term $\mathrm{II}_{11}$ can be simply derived as
\begin{equation}
  \mathrm{II}_{12} = \frac{1}{\alpha} \Big|\frac{\alpha}{1 - \pi(\delta)} - \alpha \Big|  = \frac{ \pi(\delta)}{1- \pi(\delta)}.
\end{equation}
Combing the results above, we have
\begin{equation}\nonumber
\left|
\frac{\mathbb{P}(T > q(\alpha;T_{\bW}) )}{\mathbb{P}(T_{\bZ} > q(\alpha;T_{\bZ}) )} - 1 
\right| 
 \le C'\zeta_1 \log d_2 + \frac{ \pi(\delta)}{1- \pi(\delta)} + \frac{ \pi(\delta)}{1+ \pi(\delta)} + \frac{ 2\PP{\maxdiff > \delta}}{\alpha}  + \frac{2 \zeta_2}{\alpha}.
\end{equation}
%Choosing $\delta_n =$, we finally have
%\begin{equation}\nonumber
%\left|
%\frac{\mathbb{P}(T > q(\alpha;T_{\bW}) )}{\mathbb{P}(T_{\bZ} > q(\alpha;T_{\bZ}) )} - 1 
%\right| 
% \le 2C\zeta_1 \log d + \frac{ \pi(\delta_n)}{1- \pi(\delta_n)} + \frac{ \pi(\delta_n)}{1+ \pi(\delta_n)} + \frac{ 2\PP{\maxdiff > \delta_n}}{\alpha}  + \frac{2 \zeta_2}{\alpha}.
%\end{equation}
%\begin{equation}\nonumber
%\mathrm{II}_{11} := \left|
%\frac{\mathbb{P}(T_{\bZ} > q(\alpha;T_{\bW}) +\zeta_1 )}{\mathbb{P}(T_{\bZ} > q(\alpha;T_{\bW}) )} - 1 \right|,\quad
%\mathrm{II}_{12} := \left|
%\frac{\mathbb{P}(T_{\bZ} > q(\alpha;T_{\bW}) )}{\mathbb{P}(T_{\bW} > q(\alpha;T_{\bW}) )} - 1 \right|.
%\end{equation}
%By the non-uniform anti-concentration bound in \cite{arun2018cram}, we have
%$$
%\mathrm{II}_{11} \le K_4 \zeta_1 (q(\alpha;T_{\bW}) + \zeta_1/2) \le  K_4 \zeta_1 \sqrt{\log d_2} %(q(\alpha_L/2;T_{\bZ})).
%$$
%Combining the above with the bound on $\mathrm{II}_{12}$ from Lemma \ref{lem:linear_cab} 
Applying the bound in Lemma \ref{lem:linear_maxdiff_bound}, we finally establish \eqref{eq:linear_approx_cmb} i.e., $\eta_0(d_1,d_2, n, \zeta_1, \zeta_2, \alpha_L):= \zeta_1 \log d_2 + (\log d_2)^{5/2}\delta^{1/2}+ \frac{\eta + \zeta_2}{\alpha_L}$ up to some constant factor, where $\eta =  e^{-c_1 n} + \frac{1}{d_2} + \frac{1}{n \delta^2}$.
% $\eta_2(d_1, d_2, n, \delta_n,\eta_n, \zeta_{1}, \zeta_{2})= \eta_1(d_1, n, \delta_{n}, \eta_n)+\zeta_1 \sqrt{\log d_2}+{\zeta_2}/\alpha_L$ up to some constant factor.
\end{proof}

\subsection{Proof of Theorem \ref{thm:fdr_gen}}
\label{app:pf:thm:fdr_gen}
In this section, we present the proof of our FDR control result for general graphical models, i.e., Theorem \ref{thm:fdr_gen}.
%state our general FDR control result and present its proof. 
Recall that the generic estimator $\tTheta$ is approximated by general mean zero random vectors $\bY_i(e), i\in [n]$, i.e.,
$$
\tilde \bTheta_e - \bTheta_e - \frac{1}{n} \sum_{i=1}^n \bY_i(e)  =o_P(1/\sqrt{n}).
$$
We estimate the quantile of $T_{E}$ by some Gaussian multiplier bootstrap statistic, i.e.,
\begin{equation}\nonumber
%\label{eq:ggm_gen}
\hat{c} (\alpha,E) = \inf \left\{ t\in \RR : \PPP_\xi \left( T^{\cB}_{E} \le t  \right) \ge 1-\alpha    \right\},
\end{equation}
where $T^{\cB}_{E}$ can be approximated by $T_{0E}^{\cB} = \max_{e\in E} |\frac{1}{\sqrt{n}}\sum_{i=1}^{n}\bY_{i} (e) \xi_i |$. 

\begin{proof}[Proof of Theorem \ref{thm:fdr_gen}]
\label{pf:thm:fdr_gen}
%We can similarly establish Lemma \ref{lem:single_gen} as Lemma A.2 (see its proof in). 
First note we have a general version of Lemma \ref{lem:single_test} i.e., Lemma \ref{lem:single_gen}. Then we follow the same derivations as in the proof of Theorem \ref{thm:fdr_hub} (specifically from the beginning to \eqref{eq:seq_III_3}) and reduce the FDR control problem to proving: for any $\epsilon>0$, 
\begin{align}\label{eq:key_3terms}
    \begin{split}
    & ~ \underbrace{\sum_{m=1}^{\lambda_{d}} \frac{\sum_{j\in \cH_{0}}\Var{I_j(\alpha_{m}) - d_{0}\alpha_{m}}}{\epsilon^{2} d_{0}^{2}\alpha_{m}^{2}}}_{\mathrm{III}_1} + \underbrace{ 
    \sum_{m=1}^{\lambda_{d}} \frac{\rbr{\EE{\sum_{j\in \cH_{0}} I_j(\alpha_{m}) - d_{0}\alpha_{m}}}^2}{\epsilon^{2} d_{0}^{2}\alpha_{m}^{2}} }_{\mathrm{III}_2} \\
    &~~~~~ +\underbrace{ \sum_{m=1}^{\lambda_{d}}\frac{ 
    \sum_{j_1,j_2 \in \cH_0, j_1 \ne j_2} \Cov{I_{j_1}(\alpha_m)}{ I_{j_2}(\alpha_m)}
    }{\epsilon^{2} d_{0}^{2}\alpha_{m}^{2}}}_{\mathrm{III}_3}   \rightarrow 0.
    \end{split}
\end{align}
By Lemma \ref{lem:moments_Ij_gen} and Lemma \ref{lem:crossterm_Ij_gen}, we have 
\begin{eqnarray}
\nonumber
\mathrm{III}_{1} + \mathrm{III}_{2}  +\mathrm{III}_{3}&\le& \frac{C't_{\lambda_d} }{\epsilon^2h_d } 
\left( \frac{d}{d_0|\cB|}+
\eta^2 \right) 
+ \frac{C'''t_{\lambda_d} }{\epsilon^2 \rho  h_d }   \cdot
\left( \eta +\frac{|S|\log d}{d_0^2 }
\right)\\ \nonumber
&\le& \frac{C_1 t_{\lambda_d} \eta^2 }{\epsilon^2h_d } + \frac{C_2}{\epsilon^2 \rho d_0}  \cdot\frac{t_{\lambda_d}}{h_d  } \cdot \left(1 + \eta d_0+\frac{|S|\log d}{d_0 }
\right),
%\le \frac{C_1 t_{\lambda_d} \eta^2 }{\epsilon^2h_d } + \frac{C_2}{\epsilon^2 \rho d_0}  \cdot\frac{t_{\lambda_d}}{h_d  } \cdot \left(1 + \eta d_0+\frac{|S|\log d}{d_0 p}
%\right),
%\label{eq:III123_bound}
\end{eqnarray}
%$
%\le \frac{C'''t_{\lambda_d}}{\rho\epsilon^2  h_d}
%\left(\eta(d,n,\zeta_1,\zeta_2, \alpha_L) +\frac{|S|\log d}{d^2_0}
%\right)
%$
where $\eta :=\eta(d,n,\zeta_1,\zeta_2, \alpha_L)= O \big(\frac{(\log d)^{19/6}}{n^{1/6}} + \frac{ (\log d)^{11/6}}{n^{1/6}\alpha_L^{1/3}} + \zeta_1 \log d + \frac{\zeta_2}{\alpha_L}\big)$ as defined in Lemma \ref{lem:moments_Ij_gen}, with $\zeta_1$ and $\zeta_2$ being the general terms defined in Assumption \ref{asmp:zeta12} and $\alpha_L = q|\cB|/d = \Omega  (\rho) $. 
%Expanding 

Recall that $t_{\lambda_d} = q(\alpha_L;T^B_{N_{0j}}) = O\rbr{\sqrt{\log d}}$ with probability growing to 1 and $t_{\lambda_d}/h_d = O(\log d)$.
Under Assumption \ref{asp:tradeoff_fdp_gen},
we have 
\begin{eqnarray} \nonumber
 \frac{ \log d }{\rho}\rbr{
    \frac{(\log d)^{19/6}}{n^{1/6}} + \frac{ (\log d)^{11/6}}{\rho^{1/3} n^{1/6}} + 
   \zeta_1 \log d + \frac{\zeta_2}{\rho}
   } = o(1),\quad  \frac{\log d}{\rho d_0}  +  \frac{({\log d})^{2}|S|}{\rho d_0^2 } = o(1), 
\end{eqnarray}
and thus $\mathrm{III}_{1} + \mathrm{III}_{2}  +\mathrm{III}_{3} = o(1)$ with probability growing to 1.
Therefore, we have proved $\eqref{eq:key_3terms}$, and finally establish the FDP control result below,
$$
\text{FDP}(\hat \alpha) \le q\frac{d_0}{d} + \smallop.
$$
FDR control can be similarly established as in the proof of Theorem \ref{thm:fdr_hub}.
\end{proof}

\subsection{Ancillary lemmas for Theorem \ref{thm:fdr_gen}}
\label{app:lems:fdr_gen}

\begin{lemma}\label{lem:single_gen} 
Suppose Assumptions \ref{asmp:zeta12} and \ref{asp:tradeoff_fdp_gen} hold. Given some $ 1\le j  \le d$, we have the following results.
\begin{enumerate}[(a)]
    \item \acc{Additionally, suppose for any $|\bTheta_{jk}|>0$, we also have $|\bTheta_{jk}| \ge c\sqrt{\log d/ n}$ for some constant $c>0$.} Under the alternative hypothesis $H_{1j}:  \jdeg \ge k_{\tau}$, we then have for any $\alpha \in (0,1)$,
    % If $\|\bTheta_j\|_{0} \ge k_{\tau}$, i.e., the alternative hypothesis $H_{1j}:\text{ degree of node } j \ge k$ is true, we have
    
 $$
   \lim_{(n,d)\rightarrow \infty} \PPP(\psi_{j,\alpha} = 1)= 1.
   %\text{ and } \lim_{(n,d)\rightarrow \infty} \PPP( \alpha_j \le 1/d) = 1;
 $$
    \item 
    % \st{Additionally, suppose for any $|\bTheta_{jk}|>0$, we also have $|\bTheta_{jk}| \ge c\sqrt{\log d/ n}$ for some constant $c>0$.} 
    Under the null hypothesis $H_{0j}:  \jdeg < k_{\tau} $, we have for any $u \in (0,1)$,
% $\|\bTheta_j\|_{0} \le k_{\tau}-1$, i.e., the null hypothesis $H_{0j}: \text{ degrees of node } j \le k$ is true, we have for any $u \in (0,1)$,
 $$
   \lim_{(n,d)\rightarrow \infty} \PP{\psi_{j,\alpha} = 1}\le \alpha. 
   %\text{ and } \lim_{(n,d)\rightarrow \infty} \PP{ \alpha_j \le u} = u.
 $$

\end{enumerate}
\end{lemma}
\begin{proof}[Proof of Lemma \ref{lem:single_gen}]
\label{pf:lem:single_gen}
Following the same derivations as in the proof of Lemma \ref{lem:single_test} (specifically from the beginning to \eqref{eq:power-e2}), we can reduce the proof of (a) to 
\begin{align} \label{eq:power-e2_gen}
%   {\color{blue}\nonumber \text{my correct version}:} \\ 
 \PP{\min_{e \in N_{0j}^c   } |\bTheta_e| > \frac{\hat{c} (\alpha, E_0)}{\sqrt{n}} + C_0 
   \sqrt{\frac{\log d}{n}} \text{ and } %\max_{e \in \cV \times \cV}
   \norm{\tTheta - \bTheta}_{\max} \le C_0 \sqrt{\frac{\log d}{n}}} > 1-3/d^2.
\end{align}
for some constant $C_0 >0$, where $N_{0j}^c = \{(j,k) : |\bTheta_{jk}| >0 \}$ and $E_{0} = \{(j,k): k \neq j,  k \in [d]\}$. For any fixed $\alpha \in (0,1)$, we consider sufficiently large $d$ such that $1/d \le \alpha$. Then due to \eqref{eq:quantile_logd_gen} in Lemma \ref{lem:tail_gen} and the definition $\hat{c} (\alpha,E) = \inf \left\{ t\in \RR : \PPP_\xi \left( T^{\cB}_{E} \le t  \right) \ge 1-\alpha  \right\}$, we have $\hat{c} (\alpha, E_0) \le C_0 \sqrt{\log d}$ for some constant $C_0>0$, with probability greater than $1 - 1/d^2$. Choosing the constant in the signal strength condition of Lemma \ref{lem:single_gen} to be $2 C_0$, we have with probability greater than $1 - 1/d^2$
% \colored{%\color{blue}
\begin{align*}
   ~& \min_{e \in N_{0j}^c } |\sTheta_e| \ge 2C_0 \sqrt{\frac{\log d}{n}}\ge \frac{\hat{c} (\alpha, E_0)}{\sqrt{n}} +C_0 \sqrt{\frac{\log d}{n}} 
%   \text{~~and~}\\
%  ~&~ \PP{ \norm{\tdTheta - \sTheta}_{\max} \le C_0 \sqrt{\frac{\log d}{n}} } \ge 1- 2/d^2.
\end{align*}
Then applying \eqref{eq:max_tail_gen} in Lemma \ref{lem:tail_gen}, we have \eqref{eq:power-e2_gen} holds, thus proving (a). Regarding (b), it can be proved by following the same derivations of the second part of Lemma \ref{lem:single_test} and noting the accuracy results of quantile approximation i.e., Lemma \ref{lem:quantile-valid_gen}, thus we omit the details.
\end{proof}

\begin{lemma}\label{lem:tail_gen}
Under Assumptions \ref{asmp:zeta12} and \ref{asp:tradeoff_fdp_gen}, we have
\begin{equation}\label{eq:max_tail_gen}
   \PPP\Big( \max_{e \in E_{\cV}} \sqrt{n}|\tTheta_{e} -  \bTheta _{e}| > C_0 \sqrt{\log d }\Big) <\frac{2}{d^2},
\end{equation}
\begin{equation} \label{eq:quantile_logd_gen}
 \max_{j \in [d]} \mathbb{P}\left(  \mathbb{P}_{\xi}(T^{\cB}_{E_j}  \ge 2 C_0 \sqrt{\log d}\mid 
\{\bX_i\}_{i=1}^n
) \ge  2/d\right)  \le  2/d^2,
\end{equation}	
for some constant $C_0 >0$, where $E_{\cV} = \cV \times \cV$ and $E_j = \{(j,k): k \neq j,  k \in [d]\}$.
% and $T_{0E_j}^{\cB} = \max_{e\in E_j} \frac{1}{\sqrt{n}}\sum_{i=1}^{n}\bY_{i} (e) \xi_i$.

\end{lemma}
\begin{proof}[Proof of Lemma \ref{lem:tail_gen}]\label{pf:lem:tail_gen}
To prove \eqref{eq:max_tail_gen}, we first notice that
\begin{align} \nonumber
   ~&~ \PPP\Big( \max_{e \in E_{\cV}} \sqrt{n}|\tTheta_{e} -  \bTheta _{e}| > 2C \sqrt{\log d }\Big)   \\ \nonumber
  = ~&~ \PPP\Big( \max_{e \in E_{\cV}} \sqrt{n}|\tTheta_{e} -  \bTheta _{e} -   n^{-1} \sum_{i=1}^n \bY_i(e) +   n^{-1} \sum_{i=1}^n \bY_i(e) | > 2C \sqrt{\log d }\Big)  \\ \nonumber
  \le  ~&~ \PPP\Big( \max_{e\in E_{\cV} } \sqrt{n}|\tTheta_{e} -  \bTheta _{e} -   n^{-1} \sum_{i=1}^n \bY_i(e)  | > C \sqrt{\log d }\Big) + \PPP\Big( \max_{e \in E_{\cV}} |\frac{1}{\sqrt{n}} \sum_{i=1}^n \bY_i(e) | > C \sqrt{\log d }\Big)\\ 
  :=  ~&~ \mathrm{II}_1 +  \mathrm{II}_2
%    \le  ~&~ \sum_{e \in E} \PPP\Big( \sqrt{n}|\tTheta_{e} -  \bTheta _{e} -   n^{-1} \sum_{i=1}^n \bY_i(e)  | > C \sqrt{\log d }\Big) + \PPP\Big( \max_{j,k \in [d]} |\frac{1}{\sqrt{n}} \sum_{i=1}^n \bY_i((j, k)) | > C \sqrt{\log d }\Big) 
\label{eq:max_tail_decomp}
\end{align}
Regarding the term $\mathrm{II}_1$, we have
\begin{align}\nonumber
	 \mathrm{II}_1  
	 \le  ~&~ \sum_{e \in E_{\cV}} \PPP\Big( \sqrt{n}|\tTheta_{e} -  \bTheta _{e} -   n^{-1} \sum_{i=1}^n \bY_i(e)  | > C \sqrt{\log d }\Big) \\
	 =  ~&~  \sum_{e \in E_{\cV}}  \PPP\Big( | T_{e} -  T_{0e}|  > C \sqrt{\log d} \Big)
	 \le  \sum_{e \in E_{\cV}}  \PPP\Big( | T_{e} -  T_{0e}|  > \zeta_1 \Big)
	\le \zeta_2 d^2   \le 1/d^2, \label{eq:max_tail_II1}
	 \end{align}
where the first inequality is by the union bound, the second to last inequalities hold due to Assumptions \ref{asmp:zeta12} and \ref{asp:tradeoff_fdp_gen} and $|E_{\cV}| = d^2$.  
Regarding the term $\mathrm{II}_2$, we apply the maximal inequality (Lemma 2.2.2 in \cite{vanderVaart1996Weak})) and obtain
\begin{align}\label{eq:max_tail_II2}
	 \mathrm{II}_2 = \PPP\Big( \max_{e \in E_{\cV}} \big|\frac{1}{n} \sum_{i=1}^n \bY_i(e) \big| > C \sqrt{\frac{\log d}{n} }\Big) \le 1/d^2
\end{align}
under the tail condition ($\max_{e \in \cV \times \cV}  \norm{\bY(e)}_{\psi_1} \le C$) in Assumption \ref{asmp:zeta12}. Therefore, combining \eqref{eq:max_tail_decomp}, \eqref{eq:max_tail_II1} with \eqref{eq:max_tail_II2} establishes \eqref{eq:max_tail_gen}.

Regarding \eqref{eq:quantile_logd_gen}, we notice that
\begin{align}
\nonumber
~&~\mathbb{P}\left(  \mathbb{P}_{\xi}(T^{\cB}_{E_j}  \ge 2C_0 \sqrt{\log d}\mid 
\{\bX_i\}_{i=1}^n
) \ge  2/d\right)  \\ \nonumber
=~&~ \mathbb{P}\left(  \mathbb{P}_{\xi}(T^{\cB}_{0E_j}  + T^{\cB}_{E_j} - T^{\cB}_{0E_j} \ge 2C_0  \sqrt{\log d}\mid 
\{\bX_i\}_{i=1}^n
) \ge  2/d\right)  \\ \nonumber
\le ~&~ \mathbb{P}\left(  \mathbb{P}_{\xi}(T^{\cB}_{0E_j}  \ge C_0  \sqrt{\log d}\mid 
\{\bX_i\}_{i=1}^n 
)  + \mathbb{P}_{\xi}( |T^{\cB}_{E_j} - T^{\cB}_{0E_j} |\ge C_0  \sqrt{\log d}\mid 
\{\bX_i\}_{i=1}^n
)
\ge  2/d\right) \\ \nonumber
\le ~&~ \mathbb{P}\left(  \mathbb{P}_{\xi}(T^{\cB}_{0E_j}  \ge C_0  \sqrt{\log d}\mid 
\{\bX_i\}_{i=1}^n 
)  \ge 1/d \right) \\ \nonumber
~&~ + \mathbb{P}\left( \mathbb{P}_{\xi}( |T^{\cB}_{E_j} - T^{\cB}_{0E_j} |\ge C_0  \sqrt{\log d}\mid 
\{\bX_i\}_{i=1}^n
)
\ge 1/d \right) \\
\le ~&~ \mathbb{P}\left(  \mathbb{P}_{\xi}(T^{\cB}_{0E_j}  \ge C_0  \sqrt{\log d}\mid 
\{\bX_i\}_{i=1}^n 
)  \ge  1/d \right) + 1/d^2, 
\label{eq:quantile_logd_gen_sub1}
\end{align}
where the first and second inequalities hold due to the union bound, and the last inequality holds under Assumptions \ref{asmp:zeta12} and \ref{asp:tradeoff_fdp_gen}. Now it suffices to prove
\begin{equation} \label{eq:quantile_logd_gen_sub2}
 \max_{j \in [d]} \mathbb{P}\left(  \mathbb{P}_{\xi}(T^{\cB}_{0E_j}  \ge C_0 \sqrt{\log d}\mid 
\{\bX_i\}_{i=1}^n
) \ge  1/d\right)  \le  1/d^2,
\end{equation}	
where $T_{0E_j}^{\cB} = \max_{e\in E_j} \frac{1}{\sqrt{n}}\sum_{i=1}^{n}\bY_{i} (e) \xi_i$. By applying the union bound, we derive the following inequality
\begin{align}
\nonumber
~~& \mathbb{P}\left(  \mathbb{P}_{\xi}(T^{\cB}_{0E_j}  \ge C_0 \sqrt{\log d}\mid 
\{\bX_i\}_{i=1}^n
)   \ge   1/d\right) \\ \nonumber 
\le~&  \mathbb{P}\left( \sum_{e \in E_j}  \mathbb{P}_{\xi}( \big |\frac{1}{\sqrt{n}}\sum_{i=1}^{n}\bY_{i} (e) \xi_i  \big | \ge C_0 \sqrt{\log d}\mid 
\{\bX_i\}_{i=1}^n
) 
 1/d\right).
\end{align}
Then we follow the same techniques as in the proof of Lemma \ref{lem:quantile_logd}, i.e., applying the union bound again and reduce the problem to proving 
\begin{align}\label{eq:xi_cond_tail_bound}
	\mathbb{P}_{\xi}( \big |\frac{1}{\sqrt{n}}\sum_{i=1}^{n}\bY_{i} (e) \xi_i  \big | \ge C_0 \sqrt{\log d}\mid 
\{\bX_i\}_{i=1}^n
) \ge  1/d^2
\end{align}
with probability less than $1/d^2$ for any $e \in E_j$. Conditional on $\{\bX_i\}_{i=1}^n $, we note $\frac{1}{\sqrt{n}} \sum_{i=1}^n \bY_i(e) \xi_i $ are mean zero Gaussian random variables with the variance as $\Varc{ \frac{1}{\sqrt{n}} \sum_{i=1}^n \bY_i(e) \xi_i  }{  \{\bX_i\}_{i=1}^n } =  \frac{1}{n}  \sum_{i=1}^n  \bY^2_i(e)$, for any $e \in E_j$. Under the tail condition ($\max_{e \in \cV \times \cV}  \norm{\bY^2(e)}_{\psi_1} \le C\sqrt{n/\log d} $) in Assumption \ref{asmp:zeta12}, we have
%Then it suffices to prove
\begin{align}
\nonumber
\mathbb{P}\left( \frac{1}{n}  \sum_{i=1}^n  \bY^2_i(e) > C \right) = \mathbb{P}\left( \frac{1}{n}  \sum_{i=1}^n  \frac{\bY^2_i(e)}{\sqrt{n/\log d}} > C \sqrt{\frac{\log d}{n}} \right) \le  1/d^2. 
\end{align}
Then by the Bernstein's inequality (see e.g., Section 2.1 of \cite{wainwright2019}), we prove \eqref{eq:xi_cond_tail_bound}. Therefore, combining \eqref{eq:quantile_logd_gen_sub1} with \eqref{eq:quantile_logd_gen_sub2} yields \eqref{eq:quantile_logd_gen}.
%by the sub-Gaussian condition.
%Then we follow the same techniques as in the proof of Lemma \ref{lem:quantile_logd} and reduce the problem to proving
%Then we calculate the conditional variance $\Varc{ \frac{1}{\sqrt{n}} \sum_{i=1}^n \bY_i((j, k)) \xi_i  }{  \{\bX_i\}_{i=1}^n } =  \frac{1}{n}  \sum_{i=1}^n  \bY^2_i((j, k))$. 
%\begin{align}
%\nonumber
%\end{align}
\end{proof}

\begin{lemma}\label{lem:quantile-valid_gen}
Under Assumptions \ref{asmp:zeta12} and \ref{asp:tradeoff_fdp_gen}, we have the generic estimator $\tTheta$ satisfies, for any edge set $E \subseteq \cV \times \cV$ and any given $\alpha \in (0, 1)$,

\begin{equation}\label{eq:quantile-valid_gen}
    \lim_{(n,d)\rightarrow \infty}
    \left|\PPP
    \left( \max_{e \in E}  \sqrt{n} |\tTheta_{e}-\bTheta_e|> \hat{c}(\alpha, E) 
    \right) -  \alpha
    \right|=0.
\end{equation}
\end{lemma}
\begin{proof}[Proof of Lemma \ref{lem:quantile-valid_gen}]\label{pf:lem:quantile-valid_gen}
To prove Lemma \ref{lem:quantile-valid_gen}, we use the same strategy of proving Lemma \ref{lem:quantile}, that is to verify the three conditions in Corollary 3.1 of \cite{chernozhukov2013gaussian}:
\begin{enumerate}[(a)]
 \item $\min_{e\in \cV \times \cV }\mathbb{E}[ \bY^2(e)]>c$ and $\max_{e\in \cV \times \cV}  \norm{\bY(e)}_{\psi_1} \le C$ for some positive constants $c$ and $C$ and $(\log(dn))^7/n = o(1)$.
   \item $\PPP(|T_E - T_{0E}| > \zeta_1) < \zeta_2$ holds for some $\zeta_1 , \zeta_2 >0$;
   \item And $\PPP(\PPP_{\xi}(|T_{E}^{\cB}-T_{0E}^{\cB}| > \zeta_1\mid \{\bX_i\}_{i=1}^n)> \zeta_2) < \zeta_2$ holds for $\zeta_1 \sqrt{\log d}  + \zeta_2 = o(1)$.
\end{enumerate}
The above three condition holds under Assumptions \ref{asmp:zeta12} and \ref{asp:tradeoff_fdp_gen}. Therefore, by Corollary 3.1 of \cite{chernozhukov2013gaussian}, we establish \eqref{eq:quantile-valid_gen}.

\end{proof}

\begin{lemma}\label{lem:moments_Ij_gen}
Recalling the definitions of $\mathrm{III}_{1}, \mathrm{III}_{2}$ in \eqref{eq:key_3terms}, we have
$$
\mathrm{III}_{1} + \mathrm{III}_{2}\le  \frac{C't_{\lambda_d} }{\epsilon^2h_d } 
\left(
\frac{1}{\rho d_0}+
\eta^2(d,n,\zeta_1,\zeta_2, \alpha_L)  \right),
$$
where $\eta(d,n,\zeta_1,\zeta_2, \alpha_L)= O \big(\frac{(\log d)^{19/6}}{n^{1/6}} + \frac{ (\log d)^{11/6}}{n^{1/6}\alpha_L^{1/3}} + \zeta_1 \log d + \frac{\zeta_2}{\alpha_L}\big)$ with $\zeta_1$ and $\zeta_2$ being the general terms defined in Assumption \ref{asmp:zeta12}.
% with $\zeta_1 = {s(\log d)^2}/\sqrt{n} $, $\zeta_2 = 1/d^2$.
\end{lemma}
\begin{proof}[Proof of Lemma \ref{lem:moments_Ij_gen}]
Note that the Cram\'{e}r-type deviation results in Appendix C (i.e., Propositions C.1 and C.2) are general and do not rely on any particular distributional assumptions about the data. Hence we follow the proof of Lemma \ref{lem:moments_Ij} to establish the following bounds 
\begin{eqnarray}\nonumber
\mathrm{III}_{1}  + \mathrm{III}_{2} 
\le   \frac{C't_{\lambda_d} }{\epsilon^2h_d } \left(
\frac{1}{\rho d_0}+
\eta^2(d,n,\zeta_1,\zeta_2, \alpha_L) 
\right)
\end{eqnarray}
for some constant $C'$, where $\alpha_L =q|\cB|/d$ as defined similarly in the proof of Theorem \ref{thm:fdr_hub} and $\rho = |\cB|/d$. The only difference between this lemma and Lemma \ref{lem:moments_Ij} is that $\zeta_1$ and $\zeta_2$ are not derived for a specific model but defined as general terms in Assumption \ref{asmp:zeta12}.  
% where the second inequality holds by the definition $\alpha_L =q|\cB|/d$ in the proof of Theorem \ref{thm:fdr_hub} and the definition $\rho = |\cB|/d$ in Section \ref{sec:hub_selection}.

\end{proof}

\begin{lemma}\label{lem:crossterm_Ij_gen}
Recalling the definition of $\mathrm{III}_{3}$ in \eqref{eq:key_3terms}, we have
$$
\mathrm{III}_{3} 
%\le 
%\frac{C'''t_{\lambda_d}}{\rho\epsilon^2 d_0 h_d}
%\left(1 + \eta(d,n,\zeta_1,\zeta_2, \alpha_L)  d_0+\frac{|S|\log d}{d_0 p}
%\right), 
\le \frac{C'''t_{\lambda_d}}{\rho\epsilon^2  h_d}
\left(\eta(d,n,\zeta_1,\zeta_2, \alpha_L) +\frac{|S|\log d}{d^2_0}
\right),
$$
where $\eta(d,n,\zeta_1,\zeta_2, \alpha_L)= O \big(\frac{(\log d)^{19/6}}{n^{1/6}} + \frac{ (\log d)^{11/6}}{n^{1/6}\alpha_L^{1/3}} + \zeta_1 \log d + \frac{\zeta_2}{\alpha_L}\big)$ with $\zeta_1$ and $\zeta_2$ being the general terms defined in Assumption \ref{asmp:zeta12}.
%with $\zeta_1 = {s(\log d)^2}/\sqrt{n}$, $\zeta_2 = 1/d^2$.
\end{lemma}
\begin{proof}[Proof of Lemma \ref{lem:crossterm_Ij_gen}]
Note that $\mathrm{III}_{3}$ in \eqref{eq:key_3terms} equals
\begin{align}
\label{eq:Ijdef_gen}
\begin{split}
\mathrm{III}_{3}&=\sum_{m=1}^{\lambda_{d}}\frac{ 
    \sum_{j_1,j_2 \in \cH_0, j_1 \ne j_2} \Cov{I_{j_1}(\alpha_m)}{ I_{j_2}(\alpha_m)}
    }{\epsilon^{2} d_{0}^{2}\alpha_{m}^{2}},\\
&\quad \text{ where $I_j(\alpha) = \Indrbr{T_{N_{0j}} = \max_{e \in N_{0j} } \sqrt{n}|\tTheta_e - \bTheta_e| = \max_{e \in N_{0j} } \sqrt{n}|\tTheta_e| \ge \hat c(\alpha, N_{0j})}$ } 
\end{split}
\end{align}
for $j\in \{j_1,j_2\}$. To quantify the covariance between $I_{j_1}(\alpha_m)$ and $I_{j_2}(\alpha_m)$ for $j_1,j_2 \in \cH_0, j_1\ne j_2$, we define 
\begin{equation}\label{eq:Wjdef_gen}
    W_j(\alpha) = \Indrbr{\max_{e \in N_{0j} } |Z_{e}| \ge  c(\alpha, N_{0j})},
\end{equation}
where $(Z_e)_{e \in E}$ (with $E= N_{0j}$) is a Gaussian random vector and shares the same mean vector and covariance matrix as the term $(\frac{1}{\sqrt{n}}\sum_{i=1}^{n}\bY_{i}(e))_{e \in E}$ in $T_{0E}$. Here $T_{0E}$ (with $E= N_{0j}$) has the explicit form below
$$
 T_{0E}= \max_{e\in E} \frac{1}{\sqrt{n}}\sum_{i=1}^{n}\bY_{i}(e).
$$
Following the same derivations as in the proof of Lemma, we define $\mathrm{III}'_{3}$ as
\begin{equation}\label{eq:III3'_def_gen}
    \mathrm{III}'_{3} := \sum_{m=1}^{\lambda_{d}}\frac{ 
    \sum_{j_1,j_2 \in \cH_0, j_1 \ne j_2} \Cov{W_{j_1}(\alpha_m)}{ W_{j_2}(\alpha_m)}
    }{\epsilon^{2} d_{0}^{2}\alpha_{m}^{2}}
\end{equation}
then bound $\mathrm{III}_{3} - \mathrm{III}'_{3}$ via bounding the term $\left| \EE{I_{j_1}(\alpha)I_{j_2}(\alpha)} -  \EE{W_{j_1}(\alpha)W_{j_2}(\alpha)} \right|$ and the term $\left| \EE{I_{j_1}(\alpha)}\EE{I_{j_2}(\alpha)} -  \EE{W_{j_1}(\alpha)}\EE{W_{j_2}(\alpha)} \right|$. By Lemma \ref{lem:IW_diff_bound_gen} and applying the same strategy to the term $\EE{I_{j_1}(\alpha)}\EE{I_{j_2}(\alpha)}$, we have up to some constant factor, 
\begin{align}
\nonumber
~~&\alpha^{-2}(\left| \EE{I_{j_1}(\alpha)I_{j_2}(\alpha)} -  \EE{W_{j_1}(\alpha)W_{j_2}(\alpha)} \right|  
+ \left| \EE{I_{j_1}(\alpha)}\EE{I_{j_2}(\alpha)} -  \EE{W_{j_1}(\alpha)}\EE{W_{j_2}(\alpha)} \right| )  \\ \nonumber
  \le~&   \frac{\eta(d,n,\zeta_1,\zeta_2, \alpha_L)}{\alpha}.
\end{align}
Thus we obtain
$$
\left|
\mathrm{III}_{3} - \mathrm{III}'_{3}
\right|\le \frac{1}{\epsilon^{2}}\sum_{m=1}^{\lambda_{d}}\frac{\eta(d,n,\zeta_1,\zeta_2, \alpha_L)
    }{ \alpha_{m}}\le \frac{C' t_{\lambda_d} }{\rho \epsilon^2 h_d} \cdot \eta(d,n,\zeta_1,\zeta_2, \alpha_L).
$$
where $\alpha_L =q|\cB|/d$ as defined similarly in the proof of Theorem \ref{thm:fdr_hub} and $\rho = |\cB|/d$.

%where the second inequality holds due to the fact $\alpha_m \ge \alpha_L$ $\forall~ 1\le m \le \lambda_d$ and $t_{\lambda_d} = (\lambda_d - 1)h_d$, the definition $\alpha_L =q|\cB|/d$ in the proof of Theorem \ref{thm:fdr_hub}, and the definition $\rho = |\cB|/d$ in Section \ref{sec:hub_selection}.
Combining the above bound with Lemma \ref{lem:crossterm_Wj_gen}, immediately produces
\begin{eqnarray}\nonumber
\mathrm{III}_{3}
&\le& \frac{C' t_{\lambda_d} }{\rho \epsilon^2 h_d} \cdot \eta(d,n,\zeta_1,\zeta_2, \alpha_L)
+\frac{C''' t_{\lambda_d} |S| \log d}{\rho \epsilon^2 d^2_{0}h_d}
%\frac{C''t_{\lambda_d}}{\rho\epsilon^2  d_{0} h_d} 
% \left(1+ C_{\bTheta}\frac{|S|\log d}{d_0 p}
%\right)
 \\ \nonumber
&\le& \frac{C'''t_{\lambda_d}}{\rho\epsilon^2  h_d}
\left(\eta(d,n,\zeta_1,\zeta_2, \alpha_L) +\frac{|S|\log d}{d^2_0}
\right),
\end{eqnarray}
for some constant $C'''$.
\end{proof}

\begin{lemma}\label{lem:crossterm_Wj_gen}
Recalling the term $\mathrm{III}'_{3}$ from \eqref{eq:III3'_def_gen} in the proof of Lemma \ref{lem:crossterm_Ij_gen}, we have
\begin{equation}\label{eq:crossterm_Wj_bound_gen}
	 \mathrm{III}'_{3} = \sum_{m=1}^{\lambda_{d}}\frac{ 
    \sum_{j_1,j_2 \in \cH_0, j_1 \ne j_2} \cov(W_{j_1}(\alpha_m), W_{j_2}(\alpha_m))
    }{\epsilon^{2} d_{0}^{2}\alpha_{m}^{2}}
    \le 
    \frac{C''' t_{\lambda_d} |S| \log d}{\rho \epsilon^2 d^2_{0}h_d}.
%    -
%    \frac{C''t_{\lambda_d}}{\rho \epsilon^2 d_{0}h_d} 
%     \left(1+ C_{\bTheta}\frac{|S|\log d}{d_0 p}
%    \right).
\end{equation}
\end{lemma}
\begin{proof}[Proof of Lemma \ref{lem:crossterm_Wj_gen}]
The proof of Lemma \ref{lem:crossterm_Wj_gen} is similar to the proof of Lemma \ref{lem:crossterm_Wj}. The key strategy is to utilize the equivalent expression of $\cov(W_{j_1}(\alpha_m), W_{j_2}(\alpha_m))$ (see \eqref{eq:note_CGB}) and the Cram\'{e}r-type Gaussian comparison bound with $\ell_0$ norm. First, we define $(Z_e)_{e \in N_{0j_1} \cup N_{0j_2}}$ to be jointly Gaussian such that this $(|N_{0j_1}| + |N_{0j_2}|)$-dimensional Gaussian random vector shares the same mean vector and covariance matrix as the term $(\frac{1}{\sqrt{n}}\sum_{i=1}^{n}\bY_{i}(e))_{e \in E}$ with $E=N_{0j_1}\cup N_{0j_2}$, where $\bY_{i}(e), i\in [n]$ are the general random vectors defined in Assumption \ref{asmp:zeta12}. Then we define $(Z'_e)_{e \in N_{0j_1}}, (Z'_e)_{e \in N_{0j_2}}$ to be two Gaussian random vectors such that 
\begin{equation}\nonumber
%\label{eq:Ze12_def}
(Z'_e)_{e \in N_{0j_1}} \stackrel{d}{=} (Z_e)_{e \in N_{0j_1}},~~  (Z'_e)_{e \in N_{0j_2}} \stackrel{d}{=} (Z_e)_{e \in N_{0j_2}} ~~ \text{ and }  (Z'_e)_{e \in N_{0j_1}} \independent  (Z'_e)_{e \in N_{0j_2}}.
\end{equation}
Similarly as in the proof of Lemma \ref{lem:crossterm_Wj}, we can quantify the difference between the covariance matrices of $(Z_e)_{e \in N_{0j_1}\cup N_{0j_2}}$ and $(Z'_e)_{e \in N_{0j_1}\cup N_{0j_2}}$ as a block matrix, denoted as $\Delta$; see \eqref{eq:bSigmaZ_diff} for the expression. We have that the diagonal block entries of $\Delta$ are all zero and the each of off-diagonal block entries equals to the covariance between $Z_{e_1}$ and $Z_{e_2}$ for $ e_1 \in N_{0j_1},  e_2 \in N_{0j_2}$. Thus we have,
\begin{align} \nonumber
~&~	\frac{
\left|
\cov(W_{j_1}(\alpha_m), W_{j_2}(\alpha_m))
\right| }{\alpha^2} \\ \nonumber
=~&~ \frac{1}{\alpha^2}
\Big|
\PPP(\max_{e \in N_{0j_1}\cup N_{0j_2} } |Z_{e}| \ge  t) - \PPP(\max_{e \in N_{0j_1}\cup N_{0j_2} } |Z'_{e}| \ge  t)
\Big|\\ \nonumber
\le  ~&~ \frac{1}{\alpha } \cdot  C''  \Delta_0 \log d \\
%\sum_{e_1:=(j_1, k_1)\in N_{0j_1}, e_2:=(j_2, k_2) \in N_{0j_2}} \Indrbr{\cov( Z_{e_1}, Z_{e_2}) \ne 0 } \\
= ~&~ \frac{C'' \log d}{\alpha }  \sum_{e_1:=(j_1, k_1) \in N_{0j_1}, e_2:=(j_2, k_2) \in N_{0j_2}    }\Indrbr{\cov(\bY((j_1, k_1)), \bY((j_2, k_2)) ) \ne 0 }
\label{eq:crossterm_Wj_bound_gen_sub1}
\end{align}
where the first equality holds due to \eqref{eq:note_CGB}, the first inequality holds by applying the Cram\'{e}r-type Gaussian comparison bound with $\ell_0$ norm (i.e., Theorem \ref{thm:ccb_sparse}), and the last equality holds by the definition of $\Delta_0$ and the construction of $(Z_e)_{e \in N_{0j_1} \cup N_{0j_2}}$. Note that when applying the Cram\'{e}r-type Gaussian comparison bounds, we do not impose any assumption about the connectivity of the associated graph of $\Delta$, hence deriving a bound of $O(  \Delta_0 \log d )$ instead of $O(\frac{\Delta_0 \log d }{\discon})$ as in the proof of Lemma \ref{lem:crossterm_Wj}.
Now we obtain the following bound for $ \mathrm{III}'_{3}$
\begin{align} \nonumber
	\mathrm{III}'_{3} 
	=~& \sum_{m=1}^{\lambda_{d}}\frac{ 
    \sum_{j_1,j_2 \in \cH_0, j_1 \ne j_2} \cov(W_{j_1}(\alpha_m), W_{j_2}(\alpha_m))
    }{\epsilon^{2} d_{0}^{2}\alpha_{m}^{2}}  \\ \nonumber
    \le ~& \Big( \sum_{m=1}^{\lambda_{d}} 
    \frac{ C'' \log d}{\epsilon^{2} d_{0}^{2}\alpha_{m} }\Big ) \sum_{j_1,j_2 \in \cH_0, j_1 \ne j_2} \sum_{(j_1, k_1) \in N_{0j_1}, (j_2, k_2) \in N_{0j_2}    }\Indrbr{\cov(\bY((j_1, k_1)), \bY((j_2, k_2)) ) \ne 0 } \\ \nonumber
   = ~& \Big( \sum_{m=1}^{\lambda_{d}} 
    \frac{1}{\alpha_{m} }\Big )  
    \frac{ C'' |S| \log d}{\epsilon^2 d_0^2 } \le   \frac{C''' t_{\lambda_d}}{\rho h_d}  \cdot 
    \frac{|S| \log d}{ \epsilon^2 d^2_0 }, 
\end{align}
 where the first inequality holds due to \eqref{eq:crossterm_Wj_bound_gen_sub1}, the first equality holds by the definition of $S$, and the last inequality holds due to the same derivations for $\mathrm{III}_1$ in the proof of Lemma \ref{lem:moments_Ij}. Therefore \eqref{eq:crossterm_Wj_bound_gen} is established.
%Indeed, the proof does not rely on $\bY_{i}(e)$ taking any specific form and only requires
%$$
%\forall i \in [n], ~~
% \cov(\bY_{i}(e_1)
%, \bY_{i}(e_2)) = \bTheta_{j_1 j_2} \bTheta_{k_1 k_2} + \bTheta_{j_1 k_2}\bTheta_{j_2 k_1},
%$$
%where $e_1 = (j_1, k_1)$ and $e_2 = (j_2, k_2)$. Since $e_1, e_2 \in E=N_{0j_1}\cup N_{0j_2}$, we have $\bTheta_{j_1 k_1} \bTheta_{j_2 k_2} = 0 $ thus the above expression holds under Assumption \ref{asp:cov_form}. Therefore, we can follow the proof of Lemma \ref{lem:crossterm_Wj} to establish the bound in \eqref{eq:crossterm_Wj_bound_gen}.
\end{proof}

\begin{lemma} \label{lem:IW_diff_bound_gen}
Recall the definitions of $I_{j}(\alpha)$ and $W_{j}(\alpha)$ in \eqref{eq:Ijdef_gen} and  \eqref{eq:Wjdef_gen}, for $j_1, j_2 \in \cH_0, j_1 \ne j_2$, when $\alpha \in [\alpha_L,1]$, we have
  \begin{equation}\nonumber
%  \label{eq:lem_IW_diff_bound_gen}
{\big| \EE{I_{j_1}(\alpha)I_{j_2}(\alpha)} -  \EE{W_{j_1}(\alpha)W_{j_2}(\alpha)} \big|   } \le    {\eta(d,n,\zeta_1,\zeta_2, \alpha_L)} {\alpha}.
\end{equation}
\end{lemma}

\begin{proof}[Proof of Lemma \ref{lem:IW_diff_bound_gen}]
Lemma \ref{lem:IW_diff_bound_gen} is the same as Lemma \ref{lem:IW_diff_bound} except that the terms $I_{j}(\alpha)$ and $W_{j}(\alpha)$ in \eqref{eq:Ijdef_gen} and \eqref{eq:Wjdef_gen} are defined in terms of general $T_{E}, T_{0E}, T_E^{\cB}, T_{0E}^{\cB}$ in Assumption \ref{asmp:zeta12}. The proof of Lemma \ref{lem:IW_diff_bound} only relies on Assumption \ref{asmp:zeta12} and general results in Appendix C and \cite{arun2018cram}. Thus we can follow it to establish the above bounds for $I_{j}(\alpha)$ and $W_{j}(\alpha)$ (defined in \eqref{eq:Ijdef_gen} and \eqref{eq:Wjdef_gen}).
\end{proof}

\section{Proofs of Cram\'{e}r-type comparison bounds}
\label{app:pf:cramer_theory}
In this section, we will prove two types of Cram\'{e}r-type comparison bounds:  Theorems \ref{thm:ccb_max} and \ref{thm:ccb_sparse_unitvar}.
One of the challenges to derive the comparison bounds for Gaussian maxima is that the maximum function is non-smooth. In order to show the Cram\'{e}r-type comparison bound, we first   consider smooth approximation of the maximum. The following lemma from \cite{bentkus1990smooth} show the existence of such smooth approximation.
%consider the smooth approximation, which can be dated back to \cite{bentkus1990smooth} (presented as a lemma \ref{lem:l_inf_smooth}), also used in a most recent paper \cite{arun2018cram}. 
\begin{lemma}[Theorem 1, \cite{bentkus1990smooth}]
\label{lem:l_inf_smooth}
Consider the Euclidean space $\RR^{d}$ with $\ell_{\infty}$-norm, for any $t,\epsilon \ge 0$, there exists a smooth approximating function $\varphi_{r,\epsilon}$ satisfying the following:
\begin{enumerate}[(a)]
    \item $\varphi_{r,\epsilon}: \RR^{d} \rightarrow [0,1],  \varphi_{r,\epsilon} \in \mathbb{C}^{\infty}$, where $\mathbb{C}^{\infty}$ is the smooth function class with functions differentiable for all degrees of differentiation.
    \item $\varphi_{r,\epsilon}(x)=1$ if $ \maxnorm{x} \le r$, $\varphi_{r,\epsilon}(x)=0$ if $\maxnorm{x}\ge r+\epsilon$,
    \item $\sup_{x\in \RR^{d}}\norm {D^{j}\varphi_{r,\epsilon}(x)}_{1} \le c(j)\epsilon^{-j}\log^{j-1}(d+1)$,
\end{enumerate}
where $ \norm{D^{j}\varphi_{r,\epsilon}(x)}_{1}=\sum_{i_1=1}^{d}\cdots \sum_{i_j=1}^{d}\left|
\frac{\partial^j \varphi_{r,\epsilon}(x) }{\partial x_{i_1} \cdots \partial x_{ i_j}}  \right|$ and the constants $c(j)$ only depends on $j$.
\end{lemma}
% \jlmargin{\begin{remark}\label{rk:lem:l_inf_smooth}
% %First we note an important fact that the dependence on the dimension $d$ in property (c) is necessary but only logarithmic. 
% The following equivalent expressions of $\Indrbr{\maxnorm{x } \le t}$ is derived from property (b),
% \begin{equation*}
% \Indrbr{\maxnorm{x }\le t} = \varphi_{t,\epsilon}(x) - \Indrbr{t < \maxnorm{x } < t+\epsilon}\cdot \varphi_{t,\epsilon}(x)= \varphi_{t-\epsilon,\epsilon}(x) - \Indrbr{t-\epsilon < \maxnorm{x } < t}\cdot \varphi_{t-\epsilon,\epsilon}(x).
% \end{equation*}
% \end{remark}}{delete}
\begin{remark}\label{rk:l_inf_func} 
\cite{arun2018cram} gives a concrete example of $\varphi_{r,\epsilon}(x)$  satisfying the three properties in Lemma \ref{lem:l_inf_smooth}:
 \begin{equation}\label{eq:l_inf_func}
\varphi_{r,\epsilon}(x) = g_0\left(\frac{2(F_{\beta}(z_x - r\mathbf{1}_{2d}) - \epsilon/2)}{\epsilon}\right),
\end{equation}
where $\beta = 2 \log (2d) /\epsilon$, $g_0(t) := 30\Indrbr{0 \le t \le 1}\int_t^1 s^2(1 - s)^2 ds  + \Indrbr{t \le 0}$, $F_{\beta}(\cdot)$ is the ``softmax'' function 
\[
F_{\beta}(z) := \frac{1}{\beta}\log \Big(\sum_{m=1}^{2d} \exp\left(\beta z_m\right) \Big)\quad \text{ for }  z\in\RR^{2d},
\]
$z_x = (x^{\top},-x^{\top})^{\top}$, and $\mathbf{1}_{2d}$ is the vector of $1$'s of dimension $2d$. 

In fact, in the proof of  Theorem \ref{thm:ccb_max}, we do not need a specific form of $\varphi_{r,\epsilon}(x)$ and any function satisfying Lemma \ref{lem:l_inf_smooth} will work. While in the proof of Theorem \ref{thm:ccb_sparse_unitvar}, we need to utilize the specific form in \eqref{eq:l_inf_func}.

\end{remark}
% \lzmargin{mention it will be needed for thm 3.2}{introduce softmax}
\subsection{Proof of Theorem \ref{thm:ccb_max}} \label{app:pf:thm:ccbmax}
As mentioned in Remark \ref{rk:thm:ccb_max}, we can prove the Cram\'er-type comparison bound with max norm difference as
$
M_3(\log d)^{3/2} A(\maxdiff)e^{M_3(\log d)^{3/2} A(\maxdiff)},
$ 
without the assumption on $\maxdiff$.
% where $M_1$ is a universal constant and $A(\maxdiff)=K_1\maxdiff^{1/2} \exp{(K_2\maxdiff^{1/2})}$ with $K_1$ and $K_2$ only depending on
% %$\max_{j}(\bSigma^U)_{jj}$, $\min_{j}(\bSigma^V)_{jj}$, $\max_{j}(\bSigma^V)_{jj}$ 
% the median of Gaussian maxima $\maxnorm{U},\maxnorm{V}$, and the variance terms $\min_{1\le j\le d}\{\sigma^U_{jj},\sigma^V_{jj}\},\max_{1\le j\le d}\{\sigma^U_{jj},\sigma^V_{jj}\}$,
 Therefore we state the more general form of Theorem \ref{thm:ccb_max} below and give its proof. Note that under the assumption $(\log d)^{5}\maxdiff = O(1)$ and the discussions in Remark \ref{rk:thm:ccb_max}, the bound \eqref{eq:ccb_max} in Theorem \ref{thm:ccb_max} immediately follows from Theorem \ref{thm:ccb_max_general}.
\begin{theorem}[CCB with max norm difference]\label{thm:ccb_max_general} Let $U$ and $V$ be two Gaussian random vectors and we have
% Suppose $(\log d)^{5}\maxdiff = O(1)$, then we have
\begin{equation}\label{eq:ccb_max_general}
    \sup_{0\le t \le C_0\sqrt{\log d}}\left|\frac{\mathbb{P}(\maxnorm{U} > t)}{\mathbb{P}(\maxnorm{V} > t)}-1\right|\le M_3(\log d)^{3/2} A(\maxdiff)e^{M_1(\log d)^{3/2} A(\maxdiff)}, %:= \pi(\maxdiff),
    %C_1(\log d)^{5/2}\maxdiff^{1/2},
    %M_1(\log d)^{3/2} A(\maxdiff)e^{M_1(\log d)^{3/2} A(\maxdiff)}
    %[A(\maxdiff)+1]e^{M_1(\log d)^{3/2} A(\maxdiff)}-1 %:= \pi(\maxdiff)
\end{equation}
%\jlmargin{}{Move definition of $\pi(\maxdiff)$ to prop c.1}
where $C_0>0$ is some constant, $A(\maxdiff)=M_1 \log d \maxdiff^{1/2} \exp{(M_2 \log^2 d \maxdiff^{1/2})}$, the constants $M_1, M_2$ only depend on $\min_{1\le j\le d}\{\sigma^U_{jj},\sigma^V_{jj}\},\max_{1\le j\le d}\{\sigma^U_{jj},\sigma^V_{jj}\}$, and $M_3$ is a universal constant.
% where the constant $C_1$ only depends on
% the variance terms $\min_{1\le j\le d}\{\sigma^U_{jj},\sigma^V_{jj}\},\max_{1\le j\le d}\{\sigma^U_{jj},\sigma^V_{jj}\}$.
\end{theorem}
\begin{proof}[Proof of Theorem \ref{thm:ccb_max_general}]\label{pf:thm:ccb_max_general}
Using the smooth approximation in Lemma \ref{lem:l_inf_smooth}, we can bound the difference between the distribution functions of Gaussian maxima as
\begin{eqnarray}\nonumber \label{eq:phi_smoothing}
    &~&\big|
    \PPP(\maxnorm{U} > t) - \PPP(\maxnorm{V} > t)\big|
     \\ \nonumber
    &=& \big| \EEE[\Indrbr{\maxnorm{U} \le t}- \Indrbr{\maxnorm{V} \le t}]\big| \\
    &\le&   \PPP(t-\epsilon \le \maxnorm{V} \le t+\epsilon) + \max_{j=1,2}\abr{\EEE\varphi_{j}(U) - \EEE\varphi_{j}(V)}, 
    \label{eq:kol_dist}
\end{eqnarray}
where $\varphi_{1}(x):= \varphi_{t,\epsilon}(x), \varphi_{2}(x) := \varphi_{t-\epsilon,\epsilon}(x)$. Regarding the inequality in \eqref{eq:kol_dist}, we first notice that
%The second inequality holds since
$$
\Indrbr{\maxnorm{x } \le t} = \varphi_{t,\epsilon}(x) - \Indrbr{t < \maxnorm{x }  < t+\epsilon}\cdot \varphi_{t,\epsilon}(x)= \varphi_{t-\epsilon,\epsilon}(x) - \Indrbr{t-\epsilon < \maxnorm{x }  < t}\cdot \varphi_{t-\epsilon,\epsilon}(x),
$$
{where the first equality is due to property (b) in Lemma \ref{lem:l_inf_smooth}.} Hence we have

\begin{eqnarray*}
  \Indrbr{\maxnorm{U} \le t} &\le & \varphi_{j}(U),\quad j = 1,2 \\
  \Indrbr{\maxnorm{V} \le t} &\ge & \varphi_{1}(V) - \Indrbr{t< \nbr{V}< t + \epsilon},\\
  \Indrbr{\maxnorm{V} \le t} &\ge & \varphi_{2}(V) - \Indrbr{t- \epsilon < \nbr{V}< t },
\end{eqnarray*}
then \eqref{eq:kol_dist} immediately follows by combining the above three inequalities.
%due to the property (b) in Lemma \ref{lem:l_inf_smooth}.
% $I_{(||x||\le t)} = \varphi_{t,\epsilon}(x) - I_{(t\le||x||\le t+\epsilon)}*\varphi_{t,\epsilon}(x)$ and $I_{(||x||\le t)} = \varphi_{t-\epsilon,\epsilon}(x) - I_{(t-\epsilon\le||x||\le t)}*\varphi_{t-\epsilon,\epsilon}(x)$.

The first term in \eqref{eq:kol_dist} is  related to the anti-concentration inequalities for the Gaussian maxima. By applying Theorem 2.1 in \cite{arun2018cram}, we have
\begin{equation}\label{eq:anti_nonunif} 
    \mathbb{P}(t-\epsilon \le \maxnorm{V} \le t+\epsilon) \le K_1(t+1)\epsilon\exp(K_2(t+1)\epsilon)\mathbb{P}(\maxnorm{V} > t).
\end{equation}
% \jlmargin{}{1, K1 K2 to median; 2. median and log d directly cite gaussian maximal inequality}
The explicit forms of $K_1,K_2$ can be found in Theorem 2.1 of \cite{arun2018cram}. They only depend on $\min_{1\le j\le d}\{\sigma^U_{jj},\sigma^V_{jj}\},\max_{1\le j\le d}\{\sigma^U_{jj},\sigma^V_{jj}\}$ and the median of Gaussian maxima. 
% https://math.stackexchange.com/questions/503710/distance-between-mean-and-median
% Remark that we have $|\text{Med}(X)-\EEE(X)|\le \EEE[|X - \EEE[X]|]$, together with the result on the expected value of sub-Gaussian maxima, we can control the median of $\maxnorm{V}$ by $O(\sqrt{\log d})$.
Remark that the median of $\maxnorm{V}$ is bounded by $O(\sqrt{\log d})$ by the maximal inequalities for sub-Gaussian random variables (Lemma 5.2 in \cite{van2014probability}). Plugging this into the explicit form of $K_1,K_2$ in Theorem 2.1 of \cite{arun2018cram}, we have $K_1 = O(\log d), K_2 = O({\log^2 d})$. Then \eqref{eq:anti_nonunif} can be written as
$$
 \mathbb{P}(t-\epsilon \le \maxnorm{V} \le t+\epsilon) \le  M_1 \log d (t+1)\epsilon\exp( {M_2 \log^2d}~ (t+1)\epsilon)\mathbb{P}(\maxnorm{V} > t),
$$
for some constants $M_1, M_2$ only depending on $\min_{1\le j\le d}\{\sigma^U_{jj},\sigma^V_{jj}\},\max_{1\le j\le d}\{\sigma^U_{jj},\sigma^V_{jj}\}$.
%And $K_1, K_2$ are polynomial in the median of $\maxnorm{V}$ (with degrees of $2$ and $4$ respectively). Therefore, the worst rate for $K_1$ will be $O(\log d)$, and $O({\log^2 d})$ for $K_2$.
% the median of the Gaussian maxima $\maxnorm{U}, \maxnorm{V}$ is $O(\sqrt{\log d})$, which comes from the classical result on sub-Gaussian maxima and the fact that ; and the constants $K_1,K_2$ depend on the median of Gaussian maxima only up to $2$-nd power. 

Overall the above bound has only a logarithmic dependence on the dimension $d$, similar to the anti-concentration bounds from \cite{chernozhukov2014anti}. But it quantifies the deviation with respect to the tail probability of the Gaussian maxima, thus offers a more refined characterization, which is crucial to our proof. 

Now we deal with the second term in \eqref{eq:kol_dist}. It is not hard to check that the following proof works for both $\varphi_1$ and $\varphi_2$. Therefore, without loss of generality, we use a unified notation $\varphi$ to represent either functions. We consider the Slepian interpolation between $U$ and $V$: $ W(s):= \sqrt{s}U + \sqrt{1-s}V, ~~ s\in[0,1]$.
% \begin{equation}\nonumber
%     W(s):= \sqrt{s}U + \sqrt{1-s}V, ~~ s\in[0,1]
% \end{equation}
Let $\Psi_{t}(s) = \EEE[\varphi(W(s))]$, then we have 
\begin{equation}\label{eq:slepian_int} 
    \abr{\EEE\varphi(U) - \EEE\varphi(V)} = |\Psi_{t}(1)-\Psi_{t}(0)| = \left|\int_{0}^{1}\Psi_t '(s)ds \right |,
\end{equation}
where $ \Psi_{t}'(s)=\frac{1}{2}\sum_{j=1}^{d}\EEE[\partial_{j}\varphi(W(s))(s^{-1/2}U_{j}-(1-s)^{-1/2}V_{j})]$. Applying Stein's identity (Lemma 2 of \cite{chernozhukov2015comparison}) to $(s^{-1/2}U_{j}-(1-s)^{-1/2}V_{j},W(s)^{\top})^{\top}$ and $\partial_{j}\varphi(W(s))$, we have
\begin{eqnarray}\label{eq:Psi_prime}
    \Psi_{t}'(s)=\frac{1}{2}\sum_{j,k=1}^{d}(\sigma_{jk}^U - \sigma_{jk}^V)\EEE[\partial_j \partial_k \varphi(W(s))].
\end{eqnarray}
Hence we obtain the following bound on \eqref{eq:slepian_int},
\begin{eqnarray}\nonumber
    \left|\int_{0}^{1}\Psi_t '(s)ds \right |  
    &\le& \frac{1}{2}\sum_{j,k=1}^{d}|\sigma_{jk}^U - \sigma_{jk}^V|\cdot \left|\int_{0}^{1}\EEE[\partial_j \partial_k \varphi(W(s))]ds\right|\\ \nonumber
    &\le&  \frac{\maxdiff}{2}\int_{0}^{1}\sum_{j,k=1}^{d}\EEE[\abr{\partial_j \partial_k \varphi(W(s))}]ds\\ \nonumber
    &\le&  \frac{\maxdiff}{2}\int_{0}^{1}\sum_{j,k=1}^{d}\EEE[|\partial_j \partial_k \varphi(W(s))|\cdot \Indrbr{t-\epsilon \le ||W(s)||_{\infty}\le t +\epsilon}]ds\\ \nonumber
     &\le& \frac{\maxdiff}{2}\int_{0}^{1}\sup_{x\in\RR^{d}}\norm{D^{2}\varphi(x)}_{1} \cdot \EEE[\Indrbr{t-\epsilon \le ||W(s)||_{\infty}\le t +\epsilon)}]ds\\
    &\le& \frac{c(2)\maxdiff \log (d+1)}{2\epsilon^2}\int_{0}^{1}\mathbb{P}(t-\epsilon\le ||W(s)||_{\infty}\le t +\epsilon )ds  \label{eq:Ws_anti_term}
\end{eqnarray}
where the second inequality is by the definition of $\maxdiff$ and the third one comes from the property (b) in Lemma \ref{lem:l_inf_smooth} for $\varphi_{j}(x),j=1,2$ (recalling $\varphi_1 (x) = \varphi_{t,\epsilon}(x)$ and $\varphi_2 (x) = \varphi_{t-\epsilon,\epsilon}(x)$). Note that property (c) gives a upper bound for the partial derivative terms. Thus the fourth inequality holds.
% As for the fourth one, we make use of the property (c), which gives a upper bound for the partial derivative terms.
% \begin{eqnarray}
%     \left|\int_{0}^{1}\Psi_t '(s)ds \right | 
%         &\le& \frac{\Delta}{2}\int_{0}^{1}\sup_{x\in\RR^{d}}||D^{2}\varphi(x)||_{1}\EEE[I_{(t-\epsilon \le ||W(s)||_{\infty}\le t +\epsilon))}]ds\\
%     &\le& \frac{c(2)\Delta \log (d+1)}{2\epsilon^2}\int_{0}^{1}\mathbb{P}(t-\epsilon\le ||W(s)||_{\infty}\le t +\epsilon )ds
% \end{eqnarray}

By the definition of Slepian interpolation, we have, for any $s\in[0,1]$, $W(s)$ is a Gaussian random vector and the variances can be controlled between $\min_{1\le j\le d}\{\sigma^U_{jj},\sigma^V_{jj}\}$ and $\max_{1\le j\le d}\{\sigma^U_{jj},\sigma^V_{jj}\}$.
%$\bSigma_{jj}^{U}$ and $\bSigma_{jj}^{V}$. 
The median of $\maxnorm{W(s)}$ can also be similarly bounded by $O(\sqrt{\log d})$ as $\maxnorm{V}$. Applying the {anti-concentration inequalities} again to $W(s)$ in \eqref{eq:Ws_anti_term}, we thus obtain
% Thus the constants can be viewed as the same, when applying the anti-concentration bounds to $W(s)$. We still denote those constants as $K_1, K_2$ without causing confusion. Now we obtain, \jlmargin{anti-concentration inequalities}{ref; does it mean K1, K2 defiition is changed?}
\begin{align}\label{eq:slepian_bounds}
     \left|\int_{0}^{1}\Psi_t '(s)ds \right|  \le \frac{c(2)\maxdiff  \log (d+1)}{2\epsilon^2} \cdot M_1 \log d(t+1)\epsilon\exp( M_2 \log^2 d (t+1)\epsilon)\cdot \int_{0}^{1}\mathbb{P}(||W(s)||_{\infty}> t ) ds.
\end{align}
Let $Q_t(u)=\mathbb{P}(||W(u)||_{\infty}>t)$ and $R_t(u) = Q_t(u)/Q_t(0)-1$. Combining \eqref{eq:phi_smoothing}, \eqref{eq:anti_nonunif}, \eqref{eq:slepian_int} and \eqref{eq:slepian_bounds}, we have
\begin{eqnarray}\nonumber
    |Q_t(1)-Q_t(0)| 
    &=& \big|
    \PPP(\maxnorm{U} > t) - \PPP(\maxnorm{V} > t)\big|
    \\ \nonumber
     &\le & M_1 \log d (t+1)\epsilon\exp(M_2 \log^2 d(t+1)\epsilon)Q_{t}(0) \\
    &~~& +~ \frac{c(2)\maxdiff \log (d+1)}{2\epsilon^2}  M_1 \log d(t+1)\epsilon\exp(M_2 \log^2 d(t+1)\epsilon) \int_{0}^{1}Q_{t}(s) ds.
    \label{eq:cdf_diff_expansion}
\end{eqnarray}
%\jlmargin{}{K1 K2 same above?}
If starting with the interpolation between $W(s)$ and $V$ instead of that between $U$ and $V$, we can similarly obtain the bound on $ |Q_t(s)-Q_t(0)|$. And the integral $ \int_{0}^{1}Q_{t}(s) ds$ in \eqref{eq:cdf_diff_expansion} can be directly replaced by $ \int_{0}^{u}Q_{t}(s) d s$. Namely, we have
\begin{equation}\label{eq:integral_eq}
  \frac{ |Q_t(u)-Q_t(0)| }{|Q_t(0)|}  =|R_{t}(u)|\le A(t,\epsilon)B(\maxdiff,\epsilon) \int_{0}^{u} |R_{t}(s)| ds + A(t,\epsilon)B(\maxdiff,\epsilon)\cdot u + A(t,\epsilon),
\end{equation}
where we denote $A(t,\epsilon)=M_1 \log d (t+1)\epsilon\exp(M_2 \log^2 d (t+1)\epsilon)$ and $B(\maxdiff,\epsilon)=\frac{c(2)\maxdiff \log (d+1)}{2\epsilon^2}$.

Notice that \eqref{eq:integral_eq} is an integral inequality and we can thus bound $R_{t}(s)$ by Gr\"onwall's inequality \cite{gronwall1919note} 
\[
|R_{t}(u)| \le (A(t,\epsilon)B(\maxdiff,\epsilon)u+A(t,\epsilon))e^{A(t,\epsilon)B(\maxdiff,\epsilon)u}.
\]
% \begin{equation}\nonumber
%     % |R_{t}(s)| \le (A+1)e^{ABs} -1 
%       |R_{t}(s)| \le (ABs+A)e^{ABs}
% \end{equation}
In particular, we have $|R_{t}(1)| \le (A(t,\epsilon)B(\maxdiff,\epsilon)+A(t,\epsilon))e^{A(t,\epsilon)B(\maxdiff,\epsilon)}$. Remember that $\epsilon$ is the smoothing parameter that controls the level of approximation. Choosing $\epsilon=\maxdiff^{1/2}/(t+1)$, we then have $A(\maxdiff):=A(t,\epsilon)=M_1 \log d \maxdiff^{1/2} \exp{(M_2 \log^2 d \maxdiff^{1/2})}$ for some constants $M_1, M_2$ only depending on $\min_{1\le j\le d}\{\sigma^U_{jj},\sigma^V_{jj}\},\max_{1\le j\le d}\{\sigma^U_{jj},\sigma^V_{jj}\}$ and $B(t):=B(\maxdiff,\epsilon)=\frac{c(2) \log (d+1) (t+1)^2}{2}$. When $0\le t \le C_0\sqrt{\log d}$, we have $B(t) \le M_1(\log d)^{3/2}$ for some universal constant $M_3$. Therefore the bound in \eqref{eq:ccb_max_general} is established, i.e.,
$$
\sup_{0\le t \le C_0\sqrt{\log d}}|R_{t}(1)| \le  M_3(\log d)^{3/2} A(\maxdiff)e^{M_3(\log d)^{3/2} A(\maxdiff)}. %\pi(\maxdiff)
$$
%where $ \pi(\maxdiff) = $
%for some constant $M_1$. Since we have
%$0\le t \le C_0\sqrt{\log d}$, then $B(t) \le M_1(\log d)^{3/2}$ for some constant $M_1$, thus \eqref{eq:ccb_max} is established.
\end{proof}
\subsection{Proof of Theorem \ref{thm:ccb_sparse_unitvar}}
\label{app:pf:thm:ccb_sparse_unitvar}
Before proving Theorem \ref{thm:ccb_sparse_unitvar}, we note its assumption about the connectivity can be relaxed. Therefore, we first present Theorem \ref{thm:ccb_sparse} with a weaker connectivity assumption,
% with unit variance assumption is a direct consequence of a more general result, i.e., Theorem \ref{thm:ccb_sparse}. 
which is stated below.
% and will be the key assumption of Theorem \ref{thm:ccb_sparse}.
\begin{assumption}[$\discon$-connectivity property]
\label{asp:connect_prop}
We say two Gaussian random vectors $U$ and $V$ satisfy the $\discon$-connectivity property if for any $j$ such that $\sigma^U_{jk}\ne \sigma^V_{jk}$ for some $k$, there exists a subset $\cE_{0} \subset [d]$ satisfying the following three requirements:
\begin{enumerate}[(a)]
  \item $j\in \cE_0,|\cE_{0}| = \discon  + 1$;
  \item When $m, m' \in \cE_0$ and $m\ne m'$, $ \sigma^{U}_{m m} = \sigma^{U}_{m' m'}$ and $\sigma^{U}_{m m'} = \sigma^{V}_{m m'} = 0$ hold;
  \item $\forall~ k \in [d]$, $|\{m \in \cE_0: |\sigma^{U}_{k m}|+ |\sigma^{V}_{k m}| \ne 0 \}| \le c_0$ for some constant $c_0$.
\end{enumerate}
\end{assumption}
This assumption gives a characterization of the connectivity of the associated graphs of the Gaussian random vectors $U$ and $V$. Below we give a few sufficient conditions (SC) for it.
\begin{enumerate}[\bf SC1]
    \item  $U$ and $V$ have unit variances. There exists a disjoint $(\discon + 2)$-partition of nodes $\cup_{\ell=1}^{\discon +2}\cC_\ell = [d]$ such that  $\sigma^U_{jk}=\sigma^V_{jk}=0$ when $j \in \cC_{\ell}$ and  $k \in \cC_{\ell'}$ for some $\ell \neq \ell'$.
    \item  $U$ and $V$ have unit variances. There exist disjoint partitions of nodes $\cup_{\ell=1}^{\discon + 2}\cC^U_\ell = \cup_{\ell=1}^{\discon + 2}\cC^V_\ell= [d], $ such that  $\sigma^U_{jk}$ ($\sigma^V_{jk}$) equals $0$ when $j,k$ belong to different elements $\cC^U_\ell$ ($\cC^V_\ell$), and $\forall \ell \in[\discon + 2]$, $\cC^U_{\ell} \cap \cC^V_{\ell} \ne \emptyset$.
    \item $\forall ~s \in [0,1]$, the Gaussian random vector $W(s):=\sqrt{s}U + \sqrt{1-s}V$ always has the same variances $\sigma_s^2$ across different components. The associated graph of $W(s)$ has at least $\discon + 2$ components, i.e., there exists a disjoint partition of nodes $\cup_{\ell=1}^{\discon + 2}\cC^W_\ell = [d]$, such that each $\cC^W_\ell$ comes from a different component. And the partition $\cup_{\ell=1}^{\discon + 2}\cC^W_\ell = [d]$ works any $s \in [0,1]$.
\end{enumerate}
%\begin{condition}
%\end{condition}
%\begin{condition}
%\end{condition}
%\begin{condition}
%\end{condition}
\begin{remark}\nonumber
Note that the above first condition SC1 is the main assumption of Theorem \ref{thm:ccb_sparse_unitvar} (except that $\discon + 2$ is replaced by $\discon $). It is immediate that the condition SC1 implies SC2. We will verify SC2 is indeed a sufficient condition of Assumption \ref{asp:connect_prop} in the following paragraph. 
% Thus Theorem \ref{thm:ccb_sparse_unitvar} follows as a consequence of Theorem \ref{thm:ccb_sparse}. 
Regarding SC3, its sufficiency can be verified similarly, thus we omit the details.

%First, the variance requirement in Theorem \ref{thm:ccb_sparse_unitvar} that $\sigma^U_{jj}=\sigma^V_{jj}, j \in [d]$ is easily satisfied by the unit variance assumption in Theorem \ref{thm:ccb_sparse_unitvar}. 
Simply, we have $\sigma^U_{jj}=\sigma^V_{jj} = 1, j \in [d]$ by the unit variance assumption. For any $j$ such that $\sigma^U_{jk}\ne \sigma^V_{jk}$ for some $k$, we will construct a subset $\cE_0$ and show it satisfies the three requirements (a), (b) and (c). 
%the existence of the subset $\cE_0$ by construction. 
Note that the condition SC1 assumes the existence of disjoint partitions of nodes $\cup_{\ell=1}^{\discon + 2}\cC^U_\ell = \cup_{\ell=1}^{\discon + 2}\cC^V_\ell= [d]$. We suppose $j \in \cC^{U}_{\ell_1} \cap  \cC^{V}_{\ell_2} $ for some $\ell_1,\ell_2$, then $\cE_0$ is constructed by including $j$ and picking one element $m_{\ell}$ from $\cC^U_\ell \cap \cC^V_\ell $ for each $\ell \in [\discon + 2]\setminus \{\ell_1, \ell_2\}$. As $\cC^U_\ell \cap \cC^V_\ell \ne \emptyset, \forall \ell \in [\discon + 2]$, we have $|\cE_0| \ge 1+ \discon $, hence the requirement (a) is satisfied. Regarding the requirement (b), when $m, m' \in \cE_0, m\ne m'$, we immediately have $\sigma^{U}_{m m} = \sigma^{V}_{m 'm'} = 1$ by the unit variance assumption. Since every element in $\cE_0$ comes from a different component $\cC^U_\ell$ ($\cC^V_\ell$), we also have $\sigma^{U}_{m m'} = \sigma^{V}_{m m'} = 0$ when $m, m' \in \cE_0, m\ne m'$. Lastly, due to the same reason, we have $\forall k \in [d]$, $|\{m \in \cE_0: |\sigma^{U}_{k m}|+ |\sigma^{V}_{k m}| \ne 0 \}| \le 2$. Hence the requirement (c) is also satisfied.
%Under the unit variance assumption, $\sigma^{U}_{m m} = \sigma^{U}_{m' m'}$ simply holds. Therefore, our construction of $\cE_0$ satisfies all the requirements in Theorem \ref{thm:ccb_sparse}. 
\end{remark}
Now we prove Theorem \ref{thm:ccb_sparse}, which is stated below. Note that it requires weaker connectivity assumption compared with Theorem \ref{thm:ccb_sparse_unitvar} but needs to assume minimal eigenvalue conditions.
\begin{theorem}[CCB with elementwise $\ell_{0}$ norm difference]
\label{thm:ccb_sparse}
%\cyancom{Assume $U$ and $V$ have equal variances i.e., $\sigma^U_{jj}=\sigma^V_{jj}, j \in [d]$, and the number of connected components $\cC_\ell$ on the associated graph is at least $\discon$, denoted as $\cup_{l=1}^{+1}\cC_\ell \in [d]$, satisfying $\sigma^U_{jk}= \sigma^V_{jk},\forall j,k\in \cup_{\ell=1}^{p}\cC_\ell$. For any $j$ such that $\sigma^U_{jk}\ne \sigma^V_{jk}$ for some $k$, assume there exist}
%% among the largest principal submatrix of $\{(j,k):\sigma^U_{jk}= \sigma^V_{jk}\}$, assume the number of connected components on the associated graph is at least $p+1$ 
%$\discon$ random variables $G_m$ from $\cup_{l=1}^{p}\cC_\ell$ such that $\var(G_m) = \var(X_j)$, then we have
Consider the two Gaussian random vectors $U$ and $V$ to have equal variances  $\sigma^U_{jj}=\sigma^V_{jj} = O(1)$, for  $j \in [d]$ and we assume  $\lambda_{\min}(\bSigma^U)\ge 1/b_0>0,\lambda_{\min}(\bSigma^V) \ge 1/b_0>0$ for some constant $b_0>0$. Suppose $U$ and $V$ also satisfy Assumption \ref{asp:connect_prop}, we then have
%Suppose for any $j$ such that $\sigma^U_{jk}\ne \sigma^V_{jk}$ for some $k$, there exists a subset $\cE_{0} \subset [d]$ satisfying $j\in \cE_0,|\cE_{0}| = \discon  + 1$, $ \sigma^{U}_{m m} = \sigma^{U}_{m' m'}$, $\sigma^{U}_{m m'} = \sigma^{V}_{m m'} = 0$ when $m, m' \in \cE_0, m\ne m'$, and $\forall~ k \in [d]$, $|\{m \in \cE_0: |\sigma^{U}_{k m}|+ |\sigma^{V}_{k m}| \ne 0 \}| \le c_0$ for some constant $c_0$, then we have
%$1\le m < m' \le  \discon + 1$, then we have,
\begin{equation}\label{eq:ccb_sparse}
   \sup_{0\le t \le C_0 \sqrt{\log d}} 
   \left|
   \frac{\PPP(\maxnorm{U} > t)}{\PPP(\maxnorm{V} > t)} - 1
   \right| = 
   O\left(\frac{\zerodiff \log d}{\discon}
  \right).
\end{equation}
for some constant $C_0>0$.
% whenever $t$ satisfies $0\le t \le K\sqrt{\log d}$.
% is chosen that that $\PPP(T_{U} > t), \PPP(T_{V} > t)\ge \frac{1}{d}$
\end{theorem}
%\begin{remark}\label{rk:thm:ccb_sparse}
% The above bound can be rewritten in terms of the following explicit form:
% $$
% \sup_{0\le t \le C_0 \sqrt{\log d}} \left|\frac{\PPP(\maxnorm{U} > t)}{\PPP(\maxnorm{V}> t)}-1\right| = O\bigg(\frac{\log d }{\discon}  \sum_{j\ne k}\Indrbr{\sigma^U_{jk}\ne \sigma^V_{jk}}
% \bigg).
% $$
%\end{remark}

% \begin{lemma}
% Consider $(U_1, U_2), (V_1, V_2)\in \RR^{p},  $ be two centered Gaussian random vectors with different covariance matrices $\bSigma^{U},\bSigma^{V}$, and also write 
% $$
% \bSigma^{U} = \left(
% \begin{array}{cc}
%   \bSigma^{U}_{11}   & \bSigma^{U}_{12} \\
%   \bSigma^{U}_{21}   & \bSigma^{U}_{22}
% \end{array}
% \right)
% $$
%  and 
% $$
% \bSigma^{V} = \left(
% \begin{array}{cc}
%   \bSigma^{V}_{11}   & \bSigma^{V}_{12} \\
%   \bSigma^{V}_{21}   & \bSigma^{V}_{22}
% \end{array}
% \right)
% $$ 
% we assume $ \bSigma^{U}_{11}  = \bSigma^{V}_{11},\bSigma^{U}_{22}  = \bSigma^{V}_{22}, \bSigma^{V}_{12} = \bm{O}, || \bSigma^{U}_{12}||_{0}\le k_{\tau}^2$ and denote the entry of $\bSigma^{U}_{12}$ by $\sigma^{U}_{jk}$, also assume there exists a certain number of independent blocks, whose cardinality is $p$, then we have
% \begin{equation}
%     \left|\frac{\mathbb{P}(T_{U} > t)}{\mathbb{P}(T_{V} > t)}-1\right|\le C_{\bSigma} ???
% \end{equation}
% where $0\le t \le 2\sigma_{\max}\sqrt{\log d}$ (Note that we only need to prove the case where $\frac{\omega_d}{d} \le \alpha \le 1$. When choosing $\alpha = \frac{2}{d}$, we can find the upper bound of $t$ by considering the independence case) and $C_{\bSigma}$
% \end{lemma}

\begin{proof}[Proof of Theorem \ref{thm:ccb_sparse}]
\label{pf:thm:ccb_sparse}
%  We can utilize the strategy in Theorem \ref{thm:ccb_max}, which leads to the following bound 
Following the same derivations as in Theorem \ref{thm:ccb_max_general}, we have
\begin{align}
\nonumber
    ~~& \abr{ \PPP(\maxnorm{U} > t) - \PPP(\maxnorm{V} > t)} \\ \nonumber
    \le~ & M_1 \log d (t+1)\epsilon\exp(M_2 \log^2 d (t+1)\epsilon)\PPP(\maxnorm{V} > t) +  \max_{j=1,2}\abr{\EEE[\varphi_{j}(U)] - \EEE[\varphi_{j}(V)]} \\
    \le ~& M_1 \log d(t+1)\epsilon\exp(M_2 \log^2 d (t+1)\epsilon)\PPP(\maxnorm{V} > t) +  \left|\int_{0}^{1}\Psi_t '(s)ds \right|,
    \label{eq:ccbsparse_terms}
    \end{align}
    where $\epsilon = {c}/{\max\{(\log d)^{3/2}, \discon \log d \}}$ for some small enough constant $c>0$, and the constants $M_1, M_2$ only depend on $\min_{1\le j\le d}\{\sigma^U_{jj}\}$, $\max_{1\le j\le d}\{\sigma^U_{jj}\}$. 
     The above two inequalities hold by \eqref{eq:kol_dist}, \eqref{eq:anti_nonunif} and \eqref{eq:slepian_int}. We further bound $\left|\int_{0}^{1}\Psi_t '(s)ds \right|$ as below,
% with $ \max_{j=1,2}\abr{\EEE[\varphi_{j}(U)] - \EEE[\varphi_{j}(V)]}\le |\Psi_{t}(1)-\Psi_{t}(0)|$ where
\begin{eqnarray} \nonumber
%   |\Psi_{t}(1)-\Psi_{t}(0)|  \nonumber
    &~& \left|\int_{0}^{1}\Psi_t '(s)ds \right| \\ \nonumber
    &\le& \frac{1}{2}\sum_{j,k=1}^{d}|\sigma_{jk}^U - \sigma_{jk}^V|\left|\int_{0}^{1}\EEE[\partial_j \partial_k \varphi(W(s))]ds\right|\\ \nonumber
    &\le&  \frac{M}{2} \sum_{j\ne k,\sigma^U_{jk}\ne \sigma^V_{jk}}\int_{0}^{1}\EEE[|\partial_j \partial_k \varphi(W(s))|]ds\\ \nonumber %\label{eq:uneq_entry}
    &\le&  \frac{M}{2} \sum_{j\ne k,\sigma^U_{jk}\ne \sigma^V_{jk}}\int_{0}^{1}\EEE[|\partial_j \partial_k \varphi(W(s))|\cdot \Indrbr{t-\epsilon \le ||W(s)||_{\infty}\le t +\epsilon}]ds  \\ \label{eq:uneq_entry_max}
    &\le& \frac{M \zerodiff  }{2}\max_{j\ne k,\sigma^U_{jk}\ne \sigma^V_{jk}} 
    \int_{0}^{1} \EEE[|\partial_j \partial_k \varphi(W(s))|\cdot \Indrbr{t-\epsilon \le ||W(s)||_{\infty}\le t +\epsilon} ]ds,
\end{eqnarray}
where the first inequality holds due to \eqref{eq:Psi_prime},  the second inequality is because $\sigma_{jk}^{U} = O(1), \sigma_{jk}^{V} = O(1)$ for all $j,k$  and the constant $M$  only depends on the maximal variances of the elements of $U,V$, the third inequality holds by the property (b) in Lemma \ref{lem:l_inf_smooth} for $\varphi_j(x), j=1,2$, and the last inequality holds by the definition of $\zerodiff$. Note that $\varphi_{1}(x):= \varphi_{t,\epsilon}(x), \varphi_{2}(x) := \varphi_{t-\epsilon,\epsilon}(x)$ as defined in the proof of Theorem \ref{thm:ccb_max_general}. We use the same strategy to deal with $\varphi_{1}(x)$ and $\varphi_{2}(x)$. Below we give the derivations when $\varphi= \varphi_{1}(x)$ and it is not hard to check these derivations work for $\varphi_{2}(x)$ as well. Recall the explicit construction of $\varphi:\RR^d\to \RR$ introduced in Remark \ref{rk:l_inf_func},
\[
\varphi(x) = \varphi_{r,\epsilon}(x) = g_0\left(\frac{2(F_{\beta}(z_x - r\mathbf{1}_{2d}) - \epsilon/2)}{\epsilon}\right),
\]
where  $\beta = 2 \log (2d) /\epsilon$, $g_0(t) := 30\Indrbr{0 \le t \le 1}\int_t^1 s^2(1 - s)^2 ds  + \Indrbr{t \le 0 }$, $F_{\beta}$ is the ``softmax'' function 
\[
F_{\beta}(z) := \frac{1}{\beta} \log \Big(\sum_{m=1}^{2d}  \exp\left(\beta z_m \right) \Big) \quad \text{ for } z\in \RR^{2d},
\]
$z_x = (x^{\top}, -x^{\top})^{\top}$ and $\mathbf{1}_{2d}$ is the vector of $1$'s of dimension $2d$. 

To bound \eqref{eq:uneq_entry_max}, we consider the case where $j\ne k$ and $\sigma^U_{jk}\ne \sigma^V_{jk}$. Note that 
\begin{equation}\label{eq:partials_bounds}
|\partial_j \partial_k \varphi(W(s))| \le ||g''||_{\infty} |\tilde{\pi}_j(Z)\tilde{\pi}_k(Z)|+\beta ||g' ||_{\infty} |\tilde{\pi}_j(Z) \tilde{\pi}_k(Z)|,
\end{equation}
where $g(t):=g_0(\frac{2(t-\epsilon/2)}{\epsilon})$, $Z:=W(s)$ and 
$$
\tilde{\pi}_j(z):= \frac{e^{\beta z_j} - e^{-\beta z_j} }{\sum_{m=1}^d e^{\beta z_m} + \sum_{m=1}^d e^{-\beta z_m} }.
$$
The above result follows from a direct calculation. Due to the boundedness of $|| g_0'||_{\infty}, || g_0''||_{\infty}$ and $\beta = {2 \log (2d)}/{\epsilon}$, we obtain the following bound on \eqref{eq:partials_bounds},
\begin{eqnarray*}
|\partial_j \partial_k \varphi(W(s))| &\le & 
(||g''||_{\infty} + \beta ||g' ||_{\infty}  ) |\tilde{\pi}_j(Z) \tilde{\pi}_k(Z)| \\
&\le & 
\Big( \frac{4}{\epsilon^2}  ||g''_0||_{\infty} + \frac{2 \beta}{\epsilon}  || g_0'||_{\infty}  \Big) |\tilde{\pi}_j(Z) \tilde{\pi}_k(Z)| \\
&\le & \frac{C_1 \log(2d) }{\epsilon^2}|\tilde{\pi}_j(Z) \tilde{\pi}_k(Z)| \le  \frac{C_1 \log(2d) }{\epsilon^2}|{\pi}_j(Z) {\pi}_k(Z)|  ,
\end{eqnarray*}
 for some constant $C_1$, where $\pi_j(z)= {e^{\beta |z_j|} }/{\sum_{m=1}^d e^{\beta |z_m|} }$.
%By the symmetry of Gaussian random vectors and the bound
%Denote $
%% \pi_j(z)= \frac{e^{\beta z_j} }{\sum_{m=1}^p e^{\beta z_m} }
%\pi_j(z)= {e^{\beta |z_j|} }/{\sum_{m=1}^d e^{\beta |z_m|} }
%$, we immediately have
%$
%|\tilde{\pi}_j(z) \tilde{\pi}_k(z)| \le  \pi_j(z)   \pi_k(z)
%%|\tilde{\pi}_j(z) \tilde{\pi}_k(z)| \le (\pi_j(z) + \pi_j(-z) ) ( \pi_k(z) + \pi_k(-z) )
%$.
Recalling $Z=W(s)$, we have
\begin{eqnarray} \nonumber
    &~ &\int_{0}^{1} \EE{|\partial_j \partial_k \varphi(W(s))|\cdot \Indrbr{t-\epsilon \le ||W(s)||_{\infty}\le t +\epsilon}}ds \\ \nonumber
    &\le&
    \frac{C_1 \log(2d) }{\epsilon^2}{
    \int_{0}^1 {\EE{ \pi_j(Z)\pi_k(Z) \cdot
    \Indrbr{t-\epsilon \le \Zmax \le t +\epsilon} } } ds } \\
    &= &
    \frac{C_1 \log(2d) }{\epsilon^2} 
     \mathbb{P}(\maxnorm{V} > t) \int_{0}^1 \underbrace{\frac{\EE{ \pi_j(Z)\pi_k(Z)\cdot \Indrbr{t-\epsilon \le \Zmax
    \le t +\epsilon}} }{ \PPP( \maxnorm{V} > t)
    } }_{\rm{II}(s)} ds .\label{eq:int_inside}
    % {
    % \int_{0}^1 \sqrt{\EEE[ (\pi_j(Z)\pi_k(Z))^2\cdot I_{(t-\epsilon \le \Zmax\le t +\epsilon)}   ] }\sqrt{\PPP(t-\epsilon \le \Zmax \le t +\epsilon)  }ds } 
\end{eqnarray}
Below we focus on bounding the term $\rm{II(s)}$ for any $ s \in [0,1]$. %below, %for any $s$,
%\begin{eqnarray}\label{eq:int_inside}
%\frac{\EEE[ \pi_j(Z)\pi_k(Z)\cdot \Indrbr{t-\epsilon \le \Zmax \le t +\epsilon}] }{ \mathbb{P}(\maxnorm{V} > t)
%    }, \quad \forall s \in [0,1].
%% \EEE[ (\pi_j(Z)\pi_k(Z))^2\cdot I_{(t-\epsilon \le \Zmax \le t +\epsilon)}   ] 
%% &\le& \EEE[ \pi_j(Z)\pi_k(Z) I_{(t-\epsilon \le \Zmax \le t +\epsilon)}] \nonumber
%\end{eqnarray}
First we rewrite $\pi_j(Z)\pi_k(Z)$ and simply derive the following inequality,
\begin{eqnarray}\nonumber
\pi_j(Z)\pi_k(Z)  
&=& \frac{e^{\beta |Z_j|} }{\sum_{m=1}^d e^{\beta |Z_m|} } \cdot \frac{e^{\beta |Z_k|} }{\sum_{m=1}^d e^{\beta |Z_m|} } \\ \nonumber
&=& \frac{ e^{-\beta (\Zmax - |Z_j|)}\cdot e^{-\beta (\Zmax -|Z_k|)}   }{(1 + \sum_{|Z_m|\ne \Zmax  } e^{-\beta (\Zmax -|Z_m|)} )^2} \\
&\le& e^{-\beta (\Zmax -|Z_j|)} \cdot e^{-\beta (\Zmax -|Z_k|)},
\label{eq:piZjk_bound}
\end{eqnarray}
where the second equality comes from dividing both the numerator and denominator by $ e^{2\beta \Zmax}$ in the first line. Note that $\PPP(|Z_j|=|Z_k|)=0$ since the random vector $Z$ follows a non-degenerate $d$-dimensional multivariate Gaussian distribution. Hence we have
\begin{equation}\label{eq:as_event}
1 = \Indrbr{|Z_j| = \Zmax, |Z_k| < \Zmax} +  \Indrbr{|Z_j| < \Zmax}, \text{ almost surely}.
\end{equation}
 Plugging the equality \eqref{eq:as_event} into \eqref{eq:piZjk_bound}, we can further bound $\pi_j(Z)\pi_k(Z)$ as
$$
\pi_j(Z)\pi_k(Z) \le e^{-\beta (\Zmax - |Z_k|)}\cdot  \Indrbr{|Z_k| < \Zmax}  +
e^{-\beta (\Zmax - |Z_j|)}\cdot  \Indrbr{|Z_j| < \Zmax},  \text{ almost surely}.
$$
%Without loss of generality, the second term in the above expression is larger than the first term.
%Hence in order to bound \eqref{eq:int_inside}, it suffices to deal with the following quantity
%Combining the above result with \eqref{eq:int_inside} immediately yields %quantity.
%\begin{eqnarray}\nonumber
%       &~ &\int_{0}^{1} \EE{|\partial_j \partial_k \varphi(W(s))|\cdot \Indrbr{t-\epsilon \le ||W(s)||_{\infty}\le t +\epsilon}}ds \\ \nonumber
%    &\le&
%   \frac{ \EEE[  e^{-\beta (\Zmax - |Z_j|)}\cdot  \Indrbr{|Z_j| < \Zmax} \Indrbr{t-\epsilon \le \Zmax \le t +\epsilon} ]}{ \mathbb{P}(\maxnorm{V} > t)} \label{eq:pi_jpi_k_reduce}
%\end{eqnarray}
Then we can bound $\mathrm{II}(s)$ by
\begin{align}\nonumber
  \mathrm{II}(s) &=   \frac{\EE{ \pi_j(Z)\pi_k(Z)\cdot \Indrbr{t-\epsilon \le \Zmax
    \le t +\epsilon}} }{ \PPP( \maxnorm{V} > t)
    } \\ \label{eq:pi_jpi_k_reduce1}
  &\le    \frac{
  ~ \EEE[ e^{-\beta (\Zmax - |Z_k|)}\cdot  \Indrbr{|Z_k| < \Zmax} \Indrbr{t-\epsilon \le \Zmax \le t +\epsilon} ]}{ \mathbb{P}(\maxnorm{V} > t)} \\ 
   &\quad + \frac{
  ~ \EEE[ e^{-\beta (\Zmax - |Z_j|)}\cdot  \Indrbr{|Z_j| < \Zmax} \Indrbr{t-\epsilon \le \Zmax \le t +\epsilon} ]}{ \mathbb{P}(\maxnorm{V} > t)} . \label{eq:pi_jpi_k_reduce}
\end{align}
%up to a factor of 2, where $Z=W(s)$.
%Under the assumption that there exists a certain number of independent blocks (with cardinality $p$) among $\{(j,k):\bSigma^U_{jk}= \bSigma^U_{jk}\}$, 

%\cyancom{Under the stated assumption, 
%%(for any $j,k$ such that $\sigma^U_{jk}\ne \sigma^V_{jk}$, among the largest principal submatrix of $\{(j,k):\sigma^U_{jk}= \sigma^V_{jk}\}$, assume the number of connected components on the associated graph is at least $p+1$ and it is possible to pick up $p$ random variables $G_m$ from other components different from $j(k)$ such that $\var(G_m) = \var(X_j)$)
%for $Z = U$ or $V$, we can find a $\discon-$dimensional random vector $G$ such that $(Z_j,G)$ are independent Gaussian random variables and for $m\in [\discon]$, $\var(G_m) = \var(Z_j) = \sigma_{j}^2$. Denote the indices of the random variables in $G$ by $\cE_{G}$ i.e., $\cE_G=\{j \in [d]:X_j = G_m \text{ for some } m \in [p]\}$, then $\cE_{G} \in \cup_{l=1}^{p}\cC_\ell$. Since $\sigma^U_{jk}= \sigma^V_{jk},\forall j,k\in \cup_{l=1}^{p}\cC_\ell$, thus we can find such $G$ for $Z=W(s) , s\in [0,1]$.}
We use the same strategy to bound \eqref{eq:pi_jpi_k_reduce1} and \eqref{eq:pi_jpi_k_reduce}. Below we give the derivations for bounding \eqref{eq:pi_jpi_k_reduce} and note these also work for \eqref{eq:pi_jpi_k_reduce1}.

For any $j\ne k$ and $\sigma^U_{jk}\ne \sigma^V_{jk}$, Assumption \ref{asp:connect_prop} says that there exists a subset $\cE_{0} \subset [d]$ satisfying $j\in \cE_0,|\cE_{0}| = \discon  + 1$, and $ \sigma^{U}_{m m} = \sigma^{U}_{m' m'}$, $\sigma^{U}_{m m'} = \sigma^{V}_{m m'} = 0$ when $m, m' \in \cE_0, m\ne m'$. This implies the following: when $s=0$ or $1$ (i.e., $Z = U$ or $V$), we can find a $\discon-$dimensional random vector $G$ such that $(Z_j,G)$ are all independent and $\var(G_\ell) = \var(Z_j) = \sigma_{j}^2$ for $\ell \in [\discon]$. Note that $G$ is constructed as $(Z_m)_{m \in \cE_0, m \ne j}$ with $\cE_0$ being the same for $Z= U$ and $V$. Therefore, for any $s \in (0,1), Z=W(s)=\sqrt{s}U + \sqrt{1-s}V$, we can construct $G = (Z_m)_{m \in \cE_0, m \ne j} $ such that $(Z_j,G)$ are all independent and $\var(G_\ell) = \var(Z_j) = \sigma_{j}^2$ for $ \ell \in [\discon]$. Throughout the following proof and the lemmas in Appendix \ref{app:lemmas:thm:ccb_sparse}, we will use the notation $Z, G$ without making the dependence on $s$ explicitly. And we
%Again under the assumption, there exists such $G$ satisfying the aforementioned properties for both $Z=U$ and $Z=V$, thus it also satisfies the properties for any $Z=W(s)=\sqrt{s}U + \sqrt{1-s}V, s\in [0,1]$. 
denote the indices of the random variables in $G$ (among $Z$) by $\cE_{G}$, i.e., $\cE_G= \cE_0 \setminus \{j \} = \{m \in [d]:Z_m = G_\ell \text{ for some } \ell \in [\discon]\}$.

 We will consider two separate cases based on whether $\Gmax = \Zmax$ holds. Formally, we write $
  \Indrbr{|Z_j| < \Zmax} \le
  \Indrbr{E_1} +  \Indrbr{E_2}$
 %\Indrbr{ \exists k, k\ne j, Z_k\notin \cE_G, |Z_k| = \Zmax, |Z_k| > \Gmax} +  \Indrbr{\Gmax = \Zmax , \Gmax > |Z_j|},
% \end{eqnarray*}
 with ${E_1}$ and ${E_2}$ defined as
 \begin{eqnarray}\label{eq:E1set_def}
    E_1 &:=& \{ \Zmax >  \Gmax, \Zmax > |Z_j|
%     \exists~ k \in [d], k\ne j, Z_k\notin \cE_G, \text{ s.t. } |Z_k| = \Zmax, |Z_k| > \Gmax 
    \}, \\
   {E_2}&:=& \{\Gmax = \Zmax > |Z_j| \}.
   \label{eq:E2set_def}
 \end{eqnarray}
 Then the numerator of the fraction in \eqref{eq:pi_jpi_k_reduce} can be bounded by the summation of the following two terms:
%  $\mathrm{II}_1 $ and $\mathrm{II}_2 $, which are defined as below
\begin{equation}
\label{eq:II1II2_def}
\begin{aligned}
    \mathrm{II}_1 &:= \EE{  e^{-\beta (\Zmax - |Z_j|)}\cdot  \Indrbr{E_1}\cdot  \Indrbr{t-\epsilon \le \Zmax \le t +\epsilon} }, \\
    \mathrm{II}_2 &:= \EE{  e^{-\beta (\Zmax - |Z_j|)}\cdot  \Indrbr{E_2}\cdot \Indrbr{t-\epsilon \le \Zmax \le t +\epsilon} }.
    \end{aligned}
\end{equation}
% $$
% $$
% $$
% $$
Combining \eqref{eq:pi_jpi_k_reduce} with \eqref{eq:II1II2_def} and applying Lemmas \ref{lem:II1_bound} and \ref{lem:II2_bound}, we have
%\frac{\EEE[ \pi_j(Z)\pi_k(Z)\cdot \Indrbr{t-\epsilon \le \Zmax
%    \le t +\epsilon} ] }{ \mathbb{P}(\maxnorm{V} > t)
%    } 
\begin{equation}\label{eq:IIs_bound}
\mathrm{II}(s)    \le  \frac{2( \mathrm{II}_1 + \mathrm{II}_2) }{  \mathbb{P}(\maxnorm{V} > t) } \le  \frac{C' \epsilon \log d}{\beta \discon},\quad \forall s \in [0,1],
\end{equation}
for some constant $C'$.
% we can rigorously prove that 
% $$
% \Delta_1 = O\left( \frac{\epsilon(d-p)}{\beta p^{1+ \rho_0}} \right)
% $$
% where $\rho_0 = (\frac{\sigma_k}{\sigma_j})^2\frac{(1- \rho\sigma_j/\sigma_k^2)}{1-\rho^2}$.
% As for the term $\Delta_2$, we can utilize the same strategy and show that
% $$
% \Delta_2 = O\left( \frac{\epsilon}{\beta p} \right)
% $$
% Then we have
% $$
% \frac{\EEE[ \pi_j(Z)\pi_k(Z)\cdot I_{(t-\epsilon \le \Zmax
%     \le t +\epsilon)} ] }{ \mathbb{P}(T_V > t)
%     } \le 
% % \EEE[ \pi_j(Z)\pi_k(Z) I_{(t-\epsilon \le \Zmax \le t +\epsilon)}] \le \frac{C'\epsilon}{\beta p} \cdot (1+ \frac{d-p}{p^{\rho_0}})
% $$
By \eqref{eq:ccbsparse_terms}, \eqref{eq:uneq_entry_max}, \eqref{eq:int_inside} and \eqref{eq:IIs_bound}, we thus obtain the following inequality
\begin{eqnarray}\nonumber
   &~& \left|\mathbb{P}(\maxnorm{U} > t) - \mathbb{P}(\maxnorm{V} > t)\right|\\  \nonumber
   &\le& A(t,\epsilon) \mathbb{P}(\maxnorm{V} > t) +    \frac{C_1 M \zerodiff   \log(2d) }{2\epsilon^2} 
     \mathbb{P}(\maxnorm{V} > t) \cdot  \frac{C'\epsilon \log d}{\beta  \discon} \\
   &=& \mathbb{P}(\maxnorm{V} > t)\left(A(t,\epsilon)+ B(\zerodiff , \discon) \right),
   \label{eq:Qt1_bound}
     %\sqrt{\frac{\epsilon}{\beta p} (1 + \frac{d-p}{p^{\rho_0}}) } \int_{0}^{1} \sqrt{A}\sqrt{\mathbb{P}(||W(s)||_{\infty} > t)  }ds
\end{eqnarray}
 where $A(t,\epsilon) := M_1 \log d (t+1)\epsilon\exp(M_2 \log^2 d (t+1)\epsilon)$, $B(\zerodiff , \discon):= C''{(\log d/\discon )\zerodiff  } $ for some constants $M_1,M_2, C''$. In the last line, we also subsitute $\beta=\frac{2 \log (2d)}{\epsilon}$. 
%  Same as the proof of Theorem \ref{thm:ccb}, we can also start with the interpolation between $V$ and $W(s)$, then the upper limit of the integral can be replaced by $s$, and we have
% $$
% \sqrt{\mathbb{P}(||W(s)||_{\infty} > t) } \le 
% \frac{| \mathbb{P}(||W(s)||_{\infty} > t)  -  \mathbb{P}(T_V > t)|}{ \sqrt{\mathbb{P}(T_V > t)}} +  \sqrt{\mathbb{P}(T_V > t)}
% $$
% By re-arranging, we have 
% $$
% R_t(s) \le A + \frac{\sqrt{A}B}{\sqrt{Q_t(0)}}\int_0^s R_t(y)dy +  \frac{\sqrt{A}B}{\sqrt{Q_t(0)}}s
% $$
% where $B= \frac{C k_{\tau}^2}{2} \frac{C' \log(2p) }{\epsilon^2} \sqrt{\frac{\epsilon}{\beta p} (1 + \frac{d-p}{p^{\rho_0}}) }$.
By re-arranging \eqref{eq:Qt1_bound}, we finally have 
$$
\left|\frac{\PPP(\maxnorm{U} > t)}{\PPP(\maxnorm{V}> t)}-1\right| \le A(t,\epsilon) + B(\zerodiff, \discon).
%(A+1)e^{\sqrt{A}B/ \sqrt{Q_t(0)}} -1
$$ 
Since $0\le t \le C_0 \sqrt{\log d}$ and $\epsilon = {c}/{\max\{(\log d)^{3/2}, \discon \log d \}}$ for some small enough constant $c>0$, we have $A(t,\epsilon)= O(B(\zerodiff, \discon))$. Then \eqref{eq:ccb_sparse} can be established, i.e., %for $0\le t \le C_0 \sqrt{\log d}$,
$$
\sup_{0\le t \le C_0 \sqrt{\log d}} \left|\frac{\PPP(\maxnorm{U} > t)}{\PPP(\maxnorm{V}> t)}-1\right| \le C''' B(\zerodiff  ,\discon) =   O\left(\frac{ \zerodiff \log d  }{\discon}
\right).
$$
%As a byproduct, we can immediately establish a similar bound as above but involving the explicit term $\Indrbr{\sigma^U_{jk}\ne \sigma^V_{jk}}$; see the full expression of the bound in Remark \ref{rk:thm:ccb_sparse}. This comes from working with \eqref{eq:uneq_entry} instead of \eqref{eq:uneq_entry_max} in the previous derivations.
% we can further simplify $B$ into
% $$
% B=\frac{C(\sqrt{\log (2d)})^3 }{\epsilon}\sqrt{\frac{1}{p}}\sqrt{(1+\frac{d-p}{p^{\rho_0}})}
% $$
% Note that 
% $
% Q_t(0) = \alpha^2
% $
% As for the term $O(A)$, we can bound it easilty in the final form. As for the other term, we have it bounded as
% $$
% \frac{(\log (2d))^2}{\sqrt{\epsilon}}\frac{1}{\sqrt{p}}\sqrt{(1+\frac{d-p}{p^{\rho_0}})}\frac{1}{\alpha} 
% $$
\end{proof}
Now we prove Theorem \ref{thm:ccb_sparse_unitvar} using similar strategies as in Theorem \ref{thm:ccb_sparse}. Recall that the connectivity assumption in Theorem \ref{thm:ccb_sparse_unitvar} assumes that 
%there exist disjoint partitions of nodes $\cup_{\ell=1}^{\discon}\cC^U_\ell = \cup_{\ell=1}^{\discon}\cC^V_\ell= [d], $ such that  $\sigma^U_{jk}$ ($\sigma^V_{jk}$) equals $0$ when $j,k$ belong to different components $\cC^U_\ell$ ($\cC^V_\ell$), and $\forall \ell \in[\discon]$, $\cC^U_{\ell} \cap \cC^V_{\ell} \ne \emptyset$. 
there exists a disjoint $\discon$-partition of nodes $\cup_{\ell=1}^{\discon}\cC_\ell = [d]$ such that  $\sigma^U_{jk}=\sigma^V_{jk}=0$ when $j \in \cC_{\ell}$ and  $k \in \cC_{\ell'}$ for some $\ell \neq \ell'$. Since this connectivity assumption is stronger than that in Theorem \ref{thm:ccb_sparse}, we are able to do slightly more careful analysis in Lemma \ref{lem:II1_bound}. As a result, the minimal eigenvalue condition is no longer needed.  Also note that Theorem \ref{thm:ccb_sparse_unitvar} assumes the unit variance condition and there exists some $\sigma_0<1$ such that $|\sigma^V_{jk} |\le  \sigma_0, |\sigma^U_{jk} | \le \sigma_0$ for any $j\ne k$. Both the variance condition and the covariance condition can be relaxed.
%there exists some $\sigma_0<1$ such that $|\sigma^V_{jk} |\le  \sigma_0$ for any $j\ne k$ and $ |\{(j,k): j\ne k, |\sigma^U_{jk} |>  \sigma_0 \} | \le b_0$ for some constant $b_0$.
In the following proof, we establish the Cram\'{e}r-type comparison bound under a general variance condition. This general version is actually used in the proof of Theorem \ref{thm:fdr_hub}. Specifically, the general variance condition says that $a_0 \le\sigma^U_{jj}=\sigma^V_{jj} \le a_1,~\forall j \in [d]$. After relaxing the unit variance assumption, some balanced variance assumption on the above components $\cC_\ell$ is required. It says that given any $j \in \cC_\ell$ with some $\ell$, there exists at least one $m  \in \cC_{\ell'}$ such that $\sigma^U_{jj}= \sigma^V_{jj} = \sigma^{U}_{mm} = \sigma^{V}_{mm}$ for any $\ell' \ne \ell$. Remark this condition is mainly needed for Lemma \ref{lem:II1_bound_v2}. We will call all these assumptions about variances as general variance condition.
%after relaxing the unit variance assumption. 
%
%It also the 
%
% and $ \sigma^U_{jj} = \sigma^U_{kk}$ for $j, k $ belong to the same component. Remark the condition that $ \sigma^U_{jj} = \sigma^U_{kk}$ for $j, k $ from the same component is still needed for Lemma \ref{lem:II1_bound_v2} after relaxing the unit variance assumption. 
%This condition is needed for Lemma \ref{lem:II1_bound_v2}.
Denote $\tilde{\sigma}^{U}_{jk} = {\sigma}^{U}_{jk}/\sqrt{{\sigma}^{U}_{jj} {\sigma}^{U}_{kk}}$. Accordingly, the covariance condition on $\sigma_{jk}$ in Theorem \ref{thm:ccb_sparse_unitvar} can also be relaxed into the following: there exists some $\sigma_0<1$ such that $|\tilde{\sigma}^V_{jk}| = |{\sigma}^V_{jk}|/\sqrt{{\sigma}^V_{jj} {\sigma}^V_{kk}}\le  \sigma_0$ for any $j\ne k$ and $ |\{(j,k): j\ne k, |\tilde{\sigma}^U_{jk} | = |{\sigma}^U_{jk} |/\sqrt{\sigma^U_{jj} \sigma^U_{kk} }>  \sigma_0 \} | \le b_0$ for some constant $b_0$. We will call this condition as general covariance condition.

\begin{proof}[Proof of Theorem \ref{thm:ccb_sparse_unitvar}]
\label{pf:thm:ccb_sparse_unitvar}
Following exactly the same derivations in Theorem \ref{thm:ccb_sparse} (up to \eqref{eq:II1II2_def}), we arrive at the following 
\begin{equation}\nonumber
    \mathrm{II}(s)    \le  \frac{2( \mathrm{II}_1 + \mathrm{II}_2) }{  \mathbb{P}(\maxnorm{V} > t) },
\end{equation}
where $\mathrm{II}(s), \mathrm{II}_1,  \mathrm{II}_2$ are defined in \eqref{eq:int_inside} and \eqref{eq:II1II2_def}, except that the random vector $G$ can be constructed to satisfy more properties. Assuming the connectivity assumption of Theorem \ref{thm:ccb_sparse_unitvar} and the general variance condition, we construct $G$ by choosing one random variable $Z_m$ from each component (except the one to which $Z_j$ belongs) satisfying $\Var{Z_m} =\sigma^{U}_{mm} = \sigma^{V}_{mm} = \sigma^{U}_{jj} = \sigma^{V}_{jj} = \Var{Z_j}$. Such construction still satisfies the mentioned properties in Theorem \ref{thm:ccb_sparse}. Specifically, $(G, Z_j)$ consists of $(\discon + 1)$ i.i.d. Gaussian random variables. Moreover, for any $k\ne j, k\notin \cE_{G} = \{m \in [d]:Z_m = G_\ell \text{ for some } \ell \in [\discon]\}$, there exists at most one $m \in \{j\} \cup \cE_{G}$, such that $Z_k$ and $Z_m$ belong to the same component. Based on this property, we prove Lemma \ref{lem:II1_bound_v2} and Lemma \ref{lem:maximas_pdf_bound_v2}, which do not require minimal eigenvalue conditions compared with Lemma \ref{lem:II1_bound} and Lemma \ref{lem:maximas_pdf_bound}. We still apply Lemma \ref{lem:maximas_pdf_bound_v2} to bound the term $\mathrm{II}_2$. Regarding the term $\mathrm{II}_1$, we control it by using Lemma  \ref{lem:II2_bound}. Therefore, we obtain the following
% satisfies the following property in additional to being independent and  
% Regarding the term 
\begin{equation}
\label{eq:IIs_bound_v2}
	 \mathrm{II}(s)  \le \frac{C'\epsilon \log  d}{\beta \discon}\left( 1 + \frac{b_0}{\sqrt{1 - (s + (1-s)\sigma_0)^2 } } \right).
\end{equation}
Note a simple calculus result:
\begin{equation}\nonumber
	\int_0^1 \frac{b_0}{\sqrt{1 - (s + (1-s)\sigma_0)^2}}  \le \frac{0.5 \pi b_0 }{1- \sigma_0} <  C''
 \end{equation}
for some constant $C''$ when $\sigma_0 <1$. Combining the above bound with \eqref{eq:IIs_bound_v2}, \eqref{eq:ccbsparse_terms}, \eqref{eq:uneq_entry_max}, \eqref{eq:int_inside} and \eqref{eq:IIs_bound}, we establish the bound \eqref{eq:ccb_sparse_unitvar} thus prove Theorem \ref{thm:ccb_sparse_unitvar}.
\end{proof}
\subsection{Ancillary lemmas for Theorem \ref{thm:ccb_sparse}}
\label{app:lemmas:thm:ccb_sparse}
Throughout the lemmas in this section, we will use $Z$ and $G$ without making the dependence on $s$ explicitly, as mentioned in the proof of Theorem \ref{thm:ccb_sparse}.
\begin{lemma}\label{lem:II1_bound}
 Suppose $\lambda_{\min}(\bSigma^U)\ge 1/b_0>0,\lambda_{\min}(\bSigma^V) \ge 1/b_0>0$ for some constant $b_0>0$. For the term $\mathrm{II}_1 = \EE{  e^{-\beta (\Zmax - |Z_j|)}\cdot  \Indrbr{E_1}\cdot  \Indrbr{t-\epsilon \le \Zmax \le t +\epsilon} }$ with ${E_1}$ defined in \eqref{eq:E1set_def} and $\epsilon = {c}/{\max\{(\log d)^{3/2}, \discon \log d \}}$ for some small enough constant $c>0$,
%Denote $ \mathrm{II}_1 = \EE{  e^{-\beta (\Zmax - |Z_j|)}\cdot  \Indrbr{ \exists k, k\neq j, k \notin E_G, |Z_k| = \Zmax, |Z_k| > \Gmax} \Indrbr{t-\epsilon \le \Zmax \le t +\epsilon} }
% $, 
whenever $t$ satisfies $0\le t \le C_0\sqrt{\log d}$ for some constant $C_0>0$, we have
%  $ \frac{\phi(\frac{t}{\sigma_{\max}})}{\frac{t}{\sigma_{\max}}} \ge \frac{1}{d^2} $, 
 %under the condition that 
%$\rho_0^{-1} = 1 + \frac{1}{\log d}$, 
%we have
\begin{eqnarray}\label{eq:boundII1_denom}
\frac{\mathrm{II}_1 }{  \mathbb{P}(\maxnorm{V} > t) } \le  \frac{C'\epsilon \log  d}{\beta \discon}. %\cdot \frac{(\sqrt{\log d} )^{2}}{p}
 \end{eqnarray}
\end{lemma}
\begin{proof}[Proof of Lemma \ref{lem:II1_bound}]
% Consider the conditional distribution of $\Zmax | \{Z_j, G\}$, by Lemma \eqref{lem:maximas_pdf_bound}, its density function can be upper bounded by $C\sqrt{\log d}$, whenever $t$ satisfies $0\le t \le C_0\sqrt{\log d}$.
% $\frac{\phi(t)}{t}\ge \frac{1}{d^2}$.
%$ \frac{\phi(\frac{t}{\sigma_{\max}})}{\frac{t}{\sigma_{\max}}} \ge \frac{1}{d^2} $.
We will bound $\mathrm{II}_1$ by the law of total expectation. Specifically, we first calculate the conditional expectation given $(G,Z_j)$ then take expectation with respect to $(G,Z_j)$. Denoting the conditional density function of $\Zmax \mid  Z_j = z_j, G = g $ by $ f_{g,z_j}(u)$, we write out the integral form of $\mathrm{II}_1$ as 
%Specifically, we first calculate the conditional expectation given $(G,X_j)$ then take expectation with respect to $(G,X_j)$, i.e.,
 
%Specifically,
% Then we have the following
\begin{align}\nonumber
\mathrm{II}_1 =~&  \EE{  e^{-\beta (\Zmax - |Z_j|)}\cdot   \Indrbr{ 
\maxnorm{Z} >  \maxnorm{G}, \maxnorm{Z} > Z_j 
% \exists k, k\ne j, Z_k\notin \cE_G, |Z_k| = \Zmax, |Z_k| > \Gmax
} \cdot  \Indrbr{t-\epsilon \le \Zmax \le t +\epsilon} } \\ \nonumber
=~ &  \EE{ e^{\beta |Z_j|}\cdot \Indrbr{\maxnorm{G} \le t+\epsilon,|Z_j| \le t+\epsilon} \left(
\int_{t-\epsilon}^{t+\epsilon} f_{G,Z_j}(u) e^{-\beta u}\Indrbr{u>\maxnorm{G}, u >|Z_j|}du
 \right)}\\    \nonumber
\le ~ &  \EE{ e^{\beta |Z_j|}\cdot \Indrbr{\maxnorm{G} \le t+\epsilon,|Z_j| \le t+\epsilon} \left(
\int_{t-\epsilon}^{t+\epsilon} C\sqrt{\log d}\cdot e^{-\beta u}\Indrbr{u>|Z_j|}du
 \right)}\\   \nonumber  
 \le ~ & C\sqrt{\log d}~ \mathbb{P}(\Gmax \le t + \epsilon)  \EE{ \int_{|z_j|\le t+\epsilon} \phi\rbr{\frac{z_j}{\sigma_j}}e^{\beta|z_j|} \left(
\int_{t-\epsilon}^{t+\epsilon} e^{-\beta u}\Indrbr{u>|z_j|}du
 \right)dz_{j} }\\  
\le ~ & C\sqrt{\log d} ~\mathbb{P}(\Gmax \le t + \epsilon) \underbrace{\int_{|z_j| \le t+\epsilon} \phi\rbr{\frac{z_j}{\sigma_j}}e^{\beta|z_j|} \left(
\int_{t-\epsilon}^{t+\epsilon} e^{-\beta u}\Indrbr{u>|z_j|}du
 \right)dz_{j}}_{\mathrm{III}}, \label{eq:lem:II1_bound}
\end{align}
where the first inequality holds since $\Indrbr{u>\maxnorm{G},|Z_j|} \le \Indrbr{u>|Z_j|}$ and the conditional density function $ f_{g,z_j}(u)$ is bounded by $C\sqrt{\log d}$ when $ \maxnorm{g},|z_j| < u \le t+\epsilon$ and $0\le t \le C_0\sqrt{\log d}$, as a result of Lemma \ref{lem:maximas_pdf_bound}. Recall that $\phi(\cdot)$ denotes the standard Gaussian PDF. We use the fact that $ Z_j \independent G$, $Z_j \sim \cN(0, \sigma_j^2)$ and write out the integral form of the expectation with respect to $Z_j$, thus the second inequality follows. Then the integral $\mathrm{III}$ can be further rewritten as
\begin{eqnarray} \nonumber
% &~&  \int_{|z_j|\le t+\epsilon} \phi\rbr{\frac{z_j}{\sigma_j}}e^{\beta|z_j|} \left(
%\int_{t-\epsilon}^{t+\epsilon} e^{-\beta u}\Indrbr{u>|z_j|}du
% \right)dz_{j}\\ \nonumber
 \mathrm{III} &=& 2 \int_{t-\epsilon}^{t+\epsilon}   e^{-\beta u}  \left( \int_{0}^{u} \phi\rbr{\frac{x}{\sigma_j}}e^{\beta x} dx \right)du\\ \nonumber
 &=& 2 \int_{t-\epsilon}^{t+\epsilon}   e^{-\beta u} 
 \left(
e^{\frac{\beta^2 \sigma_{j}^2}{2}} \int_{0}^{u} \phi\rbr{\frac{x}{\sigma_j} - \beta\sigma_j} dx
\right)du\\ \nonumber
 &=& 2 \int_{t-\epsilon}^{t+\epsilon}   e^{-\beta u} 
 \left(
e^{\frac{\beta^2 \sigma_{j}^2}{2}} \int_{-\beta \sigma_j}^{u/\sigma_j - \beta \sigma_j} \phi\rbr{x} dx
\right)du\\ \nonumber
 &\le& 
2 \int_{t-\epsilon}^{t+\epsilon}   e^{-\beta u} 
 \left(
e^{\frac{\beta^2 \sigma_{j}^2}{2}} 
\bar{\Phi}\big( \beta\sigma_j - {u}/{\sigma_j} \big)
\right)du\\ \nonumber
&\le& 
2 \int_{t-\epsilon}^{t+\epsilon}   e^{-\beta u} 
 \left(
e^{\frac{\beta^2 \sigma_{j}^2}{2}} 
\frac{e^{-\frac{ (\beta\sigma_j - {u}/{\sigma_j} )^2 }{2}}}{ 
\beta\sigma_j - {u}/{\sigma_j}
}
\right)du\\ \label{eq:lem:III_bound}
&\le& 
\frac{4}{\beta \sigma_j} \int_{t-\epsilon}^{t+\epsilon}   e^{-\beta u} 
 \left(
e^{\beta u} 
{e^{-\frac{u  }{2 \sigma_j}}}
\right)du \le
\frac{8\epsilon}{\beta\sigma_j} 
\exp\rbr{-\frac{(t-\epsilon)^2}{2 \sigma^2_{j}}},  %~~ (Remark: \sigma_{\max} = \sigma_{\min} = 1)
 \end{eqnarray}
where the first equality holds by Fubini's theorem, and the second equality holds by the definition of $\phi(\cdot)$. Regarding the first inequality, we use the fact that $u/\sigma_j - \beta \sigma_j < 2u/\sigma_j - \beta \sigma_j <0$ for $u \le t+\epsilon$ and $t\le C_0\sqrt{\log d}$. This is because
  $\beta = \frac{2 \log(2d)}{\epsilon}$ and $\epsilon = {c}/{\max\{(\log d)^{3/2}, \discon \log d \}}$ for some small enough constant $c>0$.
% and 
%note $u \le t+\epsilon$ where $t\le C_0\sqrt{\log d}$. Since $\beta = \frac{2 \log(2d)}{\epsilon}$ and $\epsilon = {c}/{\max\{(\log d)^{3/2}, \discon \log d \}}$, we have $u/\sigma_j - \beta \sigma_j <0$ for $u \le t+\epsilon$.
% and choose $\beta$ large enough such that $u/\sigma_j - \beta \sigma_j <0$. 
 Then $\int_{-\beta \sigma_j}^{u/\sigma_j - \beta \sigma_j} \phi\rbr{x} dx \le \bar{\Phi}( \beta\sigma_j - {u}/{\sigma_j} )$, recalling  $\bar{\Phi} = 1- {\Phi}$, where ${\Phi}$ is the standard Gaussian CDF. The second inequality holds as a result of Lemma \ref{lem:mills}. The third inequality holds due to $\beta\sigma_j > 2 u/\sigma_j$ for $u  \le t+ \epsilon$.

By \eqref{eq:lem:II1_bound} and \eqref{eq:lem:III_bound}, we arrive at the following bound
\begin{eqnarray} \nonumber
\frac{\mathrm{II}_1 }{  \mathbb{P}(\maxnorm{V} > t) } 
&\le& C\sqrt{\log d}  \cdot 
\frac{\mathbb{P}(\Gmax \le t + \epsilon)}{  \mathbb{P}(\maxnorm{V}> t)}  \cdot  \frac{8\epsilon}{\beta\sigma_j} 
\exp\rbr{-\frac{(t-\epsilon)^2}{2 \sigma^2_{j}}}  \\ \nonumber
&\le&  C\sqrt{\log d} \cdot \frac{C_1\epsilon}{\beta} \cdot 
\frac{\mathbb{P}(\Gmax \le t + \epsilon)}{  \mathbb{P}(\Gmax > t)} 
\cdot \phi\big( \frac{t - \epsilon}{\sigma_{j}}\big)/\sigma_j\\ \nonumber
&=&  C\sqrt{\log d} \cdot \frac{C_1\epsilon}{\beta} \cdot 
\underbrace{
\frac{(1 - 2 \bar{\Phi}(\frac{t + \epsilon}{\sigma_{j}}))^{\discon}}{ 1- (1 - 2 \bar{\Phi}\big(\frac{t }{\sigma_{j}})\big)^{\discon}}
\cdot \phi\big( \frac{t - \epsilon}{\sigma_{j}}\big)/\sigma_j}_{\Lambda(t,\epsilon,\discon)},
 \end{eqnarray}
 for some constants $C,C_1$, where the second inequality holds due to the definition of $\phi(z)$ and $ \mathbb{P}(\maxnorm{V}> t) \ge  \mathbb{P}(\maxnorm{G}> t)$. This is because 
 \begin{equation}\label{eq:VG_ineq}
  \mathbb{P}(\maxnorm{V}> t) \ge  \mathbb{P}(\max_{k \in \cE_G}|V_k|> t) = \mathbb{P}(\maxnorm{G_{V}}> t) = \mathbb{P}(\maxnorm{G}> t) , 
 \end{equation}
 where $ G_V = (Z_m)_{m \in \cE_0, m \ne j}$ with $Z = V$ has the same distribution as $G$. 
% where the second inequality holds due to $\maxnorm{V}\ge \maxnorm{G}$ and the definition of $\phi(z)$. 
 Regarding the last line, by the construction of $ G =(G_\ell )_{\ell \in [\discon]} = (Z_m)_{m \in \cE_0, m \ne j}$ in the proof of Theorem \ref{thm:ccb_sparse}, we have $\{G_\ell \}_{\ell \in [\discon]}$ are $\discon$ i.i.d. Gaussian random variables with $\Var{G_\ell}= \Var{Z_j} = \sigma_j^2$. 
% Because
% \[
%  \mathbb{P}(\maxnorm{V}> t) \ge  \mathbb{P}(\max_{k \in \cE_G}|V_k|> t) = \mathbb{P}(\maxnorm{G}> t). 
% \]
%  
 %Remark that in the second line, the random vector $G$ appearing in the numerator is chosen w.r.t $Z= W(s), s\in [0,1]$ and the random vector $G$ in the denominator is chosen w.r.t. $V=W(0)$. But we do not distinguish their notations since only their PDFs (CDFs) matter in our derivations and they are identically distributed.
 By applying Lemma \ref{lem:Lambda_bound} to the term ${\Lambda(t,\epsilon,\discon)}$ in the last line, we further obtain, %and assuming the stated conditions, we have
\begin{eqnarray} \nonumber
\frac{ \mathrm{II}_1 }{  \mathbb{P}(\maxnorm{V} > t) } 
&\le&  C'\sqrt{\log d} \cdot \frac{\epsilon}{\beta} \frac{\sqrt{\log d} }{\discon} = \frac{C'\epsilon \log  d}{\beta \discon},
\end{eqnarray}
for some constant $C'$, therefore \eqref{eq:boundII1_denom} is established.
\end{proof}
\begin{lemma}\label{lem:II2_bound}

For the term $\mathrm{II}_2 = \EE{  e^{-\beta (\Zmax - |Z_j|)}\cdot \Indrbr{E_2} \cdot \Indrbr{t-\epsilon \le \Zmax \le t +\epsilon}}$ with ${E_2}$ defined in \eqref{eq:E2set_def} and $\epsilon = {c}/{\max\{(\log d)^{3/2}, \discon \log d \}}$ for some small enough constant $c>0$,
%Denote $\mathrm{II}_2 = \EE{  e^{-\beta (\Zmax - |Z_j|)}\cdot  \Indrbr{\Gmax = \Zmax,\Gmax > |Z_j| } \Indrbr{t-\epsilon \le \Zmax \le t +\epsilon}}$, 
whenever $t$ satisfies $0\le t \le C_0\sqrt{\log d}$ for some constant $C_0>0$,
% $ \frac{\phi(\frac{t}{\sigma_{\max}})}{\frac{t}{\sigma_{\max}}} \ge \frac{1}{d^2} $, 
 %under the condition that 
%$\rho_0^{-1} = 1 + \frac{1}{\log d}$, 
we have
\begin{eqnarray}\label{eq:boundII2_denom}
\frac{ \mathrm{II}_2 }{  \mathbb{P}(\maxnorm{V} > t) } \le \frac{C''\epsilon\sqrt{\log d}}{\beta \discon}.  %\cdot \frac{1}{p}
 \end{eqnarray}
\end{lemma}
% \begin{lemma} Consider 
%  $$
% \Delta_1 = \EEE[  e^{-\beta (\Zmax - |Z_j|)}\cdot  I_{( \exists k, Z_k\notin E_G, |Z_k| = \Zmax, |Z_k| > \Gmax)} I_{(t-\epsilon \le \Zmax \le t +\epsilon)} ]
% $$
% and 
% $$
% \Delta_2 = \EEE[  e^{-\beta (\Zmax - |Z_j|)}\cdot  I_{(\Gmax = \Zmax,\Gmax > |Z_j| )} I_{(t-\epsilon \le \Zmax \le t +\epsilon)} ].
% $$
% we have

% $$
% \Delta_1 = O\left( \frac{\epsilon(d-p)}{\beta p^{1+ \rho_0}} \right)
% $$
% where $\rho_0 = (\frac{\sigma_k}{\sigma_j})^2\frac{(1- \rho\sigma_j/\sigma_k)^2}{1-\rho^2}$.
% As for the term $\Delta_2$, we can utilize the same strategy and show that
% $$
% \Delta_2 = O\left( \frac{\epsilon}{\beta p} \right)
% $$
% \end{lemma}
\begin{proof}[Proof of Lemma \ref{lem:II2_bound}]
By the definition of $E_2$ in \eqref{eq:E2set_def} and the tower property, we have
%To deal with the term,
\begin{align}\nonumber
\mathrm{II}_2
  &~= \EE{  e^{-\beta (\Zmax - |Z_j|)}\cdot  \Indrbr{\Gmax = \Zmax > |Z_j| } \cdot \Indrbr{t-\epsilon \le \Zmax \le t +\epsilon}}\\ \nonumber
&~= \EE{  e^{-\beta (\Gmax - |Z_j|)}\cdot  \Indrbr{\Gmax = \Zmax,\Gmax > |Z_j| } \cdot \Indrbr{t-\epsilon \le \Gmax \le t +\epsilon}}\\ \nonumber
&~ \le \EE{  e^{-\beta (\Gmax - |Z_j|)}\cdot  \Indrbr{\Gmax > |Z_j| } \cdot \Indrbr{t-\epsilon \le \Gmax \le t +\epsilon}}\\
  &~= \EE{ \Ec{  e^{ \beta |Z_j|} \Indrbr{|Z_j| < \Gmax}}{G} e^{-\beta \Gmax}\cdot  \Indrbr{ t-\epsilon \le \Gmax \le t +\epsilon} }. \label{eq:II2_cond_bound}
\end{align}
First we bound $\mathrm{III}(g):=\Ec{  e^{ \beta |Z_j|} \Indrbr{|Z_j| < \Gmax}}{G=g}$ when $\maxnorm{g} \in [t-\epsilon, t + \epsilon]$. Specifically,
\begin{eqnarray} \nonumber
\mathrm{III}(g) = \frac{2}{\sigma_j}\int_{0}^{\maxnorm{g}} e^{\beta x}\phi\Big( \frac{x}{\sigma_j} \Big) dx
&=& \frac{2 e^{{\beta^2\sigma_j^2}/{2}}}{\sigma_j} \int_{0}^{\maxnorm{g}} \phi\Big(\frac{x -\beta\sigma_j^2}{\sigma_j}\Big)dx
\\ \nonumber
&\le& {2  e^{\beta^2\sigma_j^2/2}} \int_{-\infty}^{\maxnorm{g}/\sigma_j-\beta\sigma_j} \phi(y)dy\\ \nonumber
&= &  {2  e^{\beta^2\sigma_j^2/2}} ~\bar{\Phi}(\beta\sigma_j- {\maxnorm{g}}/{\sigma_j})\\ \nonumber
&\le & {2  e^{\beta^2\sigma_j^2/2}} \frac{\phi\rbr{ \beta\sigma_j- {\maxnorm{g}}/{\sigma_j}}}{\beta\sigma_j- {\maxnorm{g}}/{\sigma_j} }\\
&\le & \frac{4}{\beta \sigma_j} \phi\rbr{\frac{\maxnorm{g}}{\sigma_j}} e^{\beta\maxnorm{g}}, 
\label{eq:bound1condonG}
\end{eqnarray}
% for some constant $C_1,C_2,C_3>0$, 
where the first equality holds due to $Z_j \independent G$, and the second equality comes from rearranging. The first inequality holds by the change of variable $y=(x-\beta\sigma_j^2)/\sigma_j$ and setting the lower limit of the integral as $-\infty$. Because $\beta = \frac{2 \log(2d)}{\epsilon}$ and $\epsilon = {c}/{\max\{(\log d)^{3/2}, \discon \log d \}}$ for some small enough constant $c>0$, we have $\maxnorm{g}/\sigma_j < \beta \sigma_j   $ for $\maxnorm{g} \le t+\epsilon$ and $t\le C_0\sqrt{\log d}$.
%we follow the same strategy as in the proof of Lemma \ref{lem:II1_bound}, i.e.,
%%and the second inequality holds by change of variable $x = |z_j|$, 
% choosing $\beta$ large enough such that $ \maxnorm{g}/\sigma_j-\beta\sigma_j<0$ (note that $\maxnorm{g}\le t+\epsilon \le (C_0+1)\sqrt{\log d}$). 
% 
Then the second inequality holds as a result of Lemma \ref{lem:mills} and the fact that $\beta \sigma_j  - \maxnorm{g}/\sigma_j >0$. The last inequality comes from rearranging and the fact that $  \beta \sigma_j  >  2\maxnorm{g}/\sigma_j$ for $\maxnorm{g} \le t+\epsilon$ and $t\le C_0\sqrt{\log d}$.
%choosing $\beta$ large enough such that $\beta > 2 \maxnorm{g}/\sigma_j^2$ when $\maxnorm{g} \in [t-\epsilon, t + \epsilon], 0\le t \le C_0\sqrt{\log d}$. 
Combining \eqref{eq:bound1condonG} with \eqref{eq:II2_cond_bound}, we have
% $\mathrm{II}_2$ as below:
\begin{eqnarray} \nonumber
\mathrm{II}_2
&\le &  \EE{ \mathrm{III}(G)\cdot  e^{-\beta \Gmax}\cdot  \Indrbr{ t-\epsilon \le \Gmax \le t +\epsilon} } \\ \nonumber
&\le & \frac{4}{\beta \sigma_j}
 \EE{  \phi\Big( \frac{\maxnorm{G}}{\sigma_j}\Big)  e^{\beta (\Gmax)}  \cdot e^{-\beta\maxnorm{G}}
 \cdot   \Indrbr{t-\epsilon \le \Gmax \le t +\epsilon} }\\
 %\cdot \frac{1}{p}
 &\le& 
\frac{4}{\beta \sigma_j}\int_{t-\epsilon  }^{t+\epsilon} \phi
\big( \frac{y}{\sigma_j}\big) f(y)dy,
%  &=& 
% \frac{C_3}{\beta}\int_{t-\epsilon  }^{t+\epsilon} \phi(\frac{\maxnorm{g}}{\sig})f(\maxnorm{g})d\maxnorm{g}
\label{eq:II2_1stbound}
%%%%%%%%%%%%%%%%%%%%%%%%%%%%%%% Since \Indrbr{\Gmax = \Zmax} \le 1
\end{eqnarray}
where $f(y)$ denotes the PDF of $\Gmax$. As $\{G_\ell \}_{\ell \in [\discon]}$ are i.i.d. Gaussian random variables satisfying $\forall~\ell \in [\discon]$, $\EE{G_\ell}=0$ and $\Var{G_\ell}=\sigma_j^2$, we have for $y>0$,
\begin{equation}\label{eq:Gmax_CDF}
 \PPP(\Gmax \le y)= \PPP(\bigcup_{\ell \in [\discon]} |G_\ell| \le y  )= (1 - 2 \PPP(G_\ell/\sigma_j > y/\sigma_j) )^\discon = ( 1 - 2\bar{\Phi}({y}/{\sigma_j}) )^{\discon}.
\end{equation}
%  since $|E_G|=p$, 
Thus we have the PDF of $\Gmax$ equals $f(y) = \frac{2\discon}{\sigma_j}  \left( 1 - 2\bar{\Phi}(\frac{y}{\sigma_j}) )^{\discon}\right)^{\frac{\discon-1}{\discon}} \phi(\frac{y}{\sigma_j})$. Plugging the expression of $f(y)$ into \eqref{eq:II2_1stbound}, we further derive the following bound
 \begin{eqnarray}\nonumber
\frac{ \mathrm{II}_2 }{  \mathbb{P}(\maxnorm{V} > t) } 
&\le&  
\frac{8 \discon}{\beta \sigma_j^2 } 
\int_{t-\epsilon  }^{t+\epsilon}
\frac{\left( 1 - 2\bar{\Phi}(\frac{y}{\sigma_j}) )^{\discon}\right)^{\frac{\discon-1}{\discon}} \phi^2(\frac{y}{\sigma_j}) }{\mathbb{P}(\maxnorm{V} > t) 
}dy\\ \nonumber
&\le&  
\frac{8 \discon }{\beta \sigma_j^2} 
\int_{t-\epsilon  }^{t+\epsilon}
\frac{ \left( 1 - 2\bar{\Phi}(\frac{y}{\sigma_j}) )^{\discon}\right)^{\frac{\discon-1}{\discon}} \phi^2(\frac{y}{\sigma_j}) }{1-\mathbb{P}(\Gmax \le t) 
}dy\\ \nonumber
&=&  
\frac{8 \discon }{\beta \sigma_j^2} 
\int_{t-\epsilon  }^{t+\epsilon}
\frac{\left( (1 -  2\bar{\Phi}(\frac{y}{\sigma_j}))^{\discon}\right)^{\frac{\discon-1}{\discon}} \phi^2(\frac{y}{\sigma_j}) }{1-
(1 -  2\bar{\Phi}(\frac{t}{\sigma_j}))^{\discon}
}dy\\ \nonumber
%&\le&
%\frac{16 \epsilon}{\beta \sigma_j \discon} 
%\frac{\left( (1 -  2\bar{\Phi}(\frac{t+\epsilon}{\sigma_j}) )^{\discon}\right)^{\frac{\discon-1}{\discon}} (\discon\phi( \frac{t-\epsilon}{\sigma_j}))^2 }{1-
%(1 - 2\bar{\Phi}(\frac{t}{\sigma_j}))^{\discon}
%}\\ \nonumber
&\le&
\frac{16 \epsilon}{\beta \sigma_j^2 \discon} 
\frac{\left( (1 -  2\bar{\Phi}(\frac{t+\epsilon}{\sigma_j}) )^{\discon}\right)^{\frac{\discon-1}{\discon}} (\discon\phi( \frac{t-\epsilon}{\sigma_j}))^2 }{1-
(1 - 2\bar{\Phi}(\frac{t+\epsilon}{\sigma_j}))^{\discon}
}
\le
\frac{C'' \epsilon \sqrt{\log d}}{\beta \discon},
%\cdot \frac{1}{p}
\end{eqnarray}
for some constant $C'$, where the second inequality holds due to \eqref{eq:VG_ineq}, as mentioned in the proof of Lemma \ref{lem:II1_bound}.
%since $\maxnorm{V}\ge \maxnorm{G}$. 
%Remark that in the second line, the Gaussian random vector $G$ appearing in the denominator is chosen w.r.t. $V=W(0)$. But we do not distinguish its notation from previous $\{G_m, m\in [\discon]\}$ since only their marginal $\discon$-dimensional PDFs (CDFs) matter in our derivations and these two actually have the same marginal distributions, due to the explanation in the proof of Theorem \ref{thm:ccb_sparse}.
% Remark that in the second line, the Gaussian random vector $G$ appearing in the denominator is chosen w.r.t. $V=W(0)$. But we do not distinguish this notation from previous $(G_\ell) _{\ell \in [\discon]}$ due to the same reason explained in the proof of Lemma \ref{lem:II1_bound}.
 %  since only their PDFs (CDFs) matter in our derivations and they are identically distributed.
 The equality holds as a result of substituting the expression of $\PPP(\Gmax \le t)$ by \eqref{eq:Gmax_CDF}. The third inequality holds since $1 - 2\bar{\Phi}(z)$ is monotonically increasing and ${\phi}(z)$ is monotonically decreasing when $z\ge 0$. As for the last line, we apply Lemma \ref{lem:Lambda_bound2}.
%apply a similar result as in Lemma \ref{lem:Lambda_bound} (except that $(\discon \phi( \frac{t-\epsilon}{\sigma_j}))^2$ replaces $\discon \phi( \frac{t-\epsilon}{\sigma_j})$). The explicit form \eqref{eq:phi_bound_similar} can be found in Remark \ref{rk:lem:Lambda_bound}. 
 Finally, \eqref{eq:boundII2_denom} is established.
\end{proof}

\begin{lemma}
\label{lem:maximas_pdf_bound}
%\cyancom{Let $G$ be a $\discon$-dimensional subvector of $Z$ with each entry coming from each different component (also different from the component to which $Z_j$ belongs)} and
 Suppose $\lambda_{\min}(\bSigma^U)\ge 1/b_0>0,\lambda_{\min}(\bSigma^V) \ge 1/b_0>0$ for some constant $b_0>0$. Recall that the density function of the conditional distribution of $\Zmax \mid  \{Z_j = z_j, G = g\}$ 
%(also $ \Zmax \mid Z_j$ since $G$ is independent from the other random variables in $Z$) 
is denoted by $f_{ g,z_j}(z)$. Suppose $\epsilon>0$, when $0\le t \le C_0\sqrt{\log d}$ for some constant $C_0>0$ and $|z_j|,\maxnorm{g} \le t+\epsilon$, we have
% $\frac{\phi(\frac{t}{\sigma_{\max}})}{\frac{t}{\sigma_{\max}}} \ge \frac{1}{d^2} $, 
$$
f_{g, z_j}(z)\le C \sqrt{\log d} ,\quad \forall ~ z \in ( \max\{|z_j|,\maxnorm{g}\}, t+\epsilon]. % <  z\le t+\epsilon.
%\quad z\le t+\epsilon,~~ |z_j|,\maxnorm{g} \le t+\epsilon
$$
\end{lemma}
%\begin{remark}
%Through out this lemma, we denote the conditional distribution of $\Zmax \mid  \{Z_j = z_j, G = g\}$ by $f_{(z_j, g)}(z)$ instead of $f_{(z_j, g)}(\maxnorm{z})$ for notation simplicity.  
%\end{remark}

\begin{proof}[Proof of Lemma \ref{lem:maximas_pdf_bound}]
%First notice that $\Pc{\Zmax \ge \max{(|z_j|, \maxnorm{g}})}{Z_j = z_j, G =g} = 1$. In the following, we only need to consider bounding $f_{g, z_j}(z)$ when $|z| > \max{(|z_j|, \maxnorm{g})} $. 
%Conditioning on $\{Z_j, G\}$, we denote the $(d-\discon - 1)$-dimensional random vector (after excluding $\{Z_j, G\}$) by $X$ without causing confusion, and accordingly the maximum by $\Xmax$.
% 
First we introduce some new notations. Let $(\sigma_{jk})_{1\le j, k \le d} \in \RR^{d \times d}$ be the covariance matrix of $Z$. For given $j$, we denote
\begin{equation}  \label{eq:sigmakkj_def}
  \sigma_{kk\cdot j}:= \sigma_{kk} - \sigma_{kj}^2 \sigma_{jj}^{-1} - \sum_{m\in \cE_G }\sigma_{km}^2 \sigma_{mm}^{-1}. 
%  \quad \underline{\sigma_{\cdot j}}:=\min_{k\ne j}  \sqrt{\sigma_{kk\cdot j}}, 
%  \quad \bar{\rho}_{j} := \max_{k\ne j}\frac{|\sigma_{jk}|}{\sigma_{jj}}.
\end{equation}
%Denote $ \sigma_{kk\cdot j}= \sigma_{kk} - \sigma_{kj}^2 \sigma_{jj}^{-1} - \sum_{m\in \cE_G }\sigma_{km}^2 \sigma_{mm}^{-1}$, and $\underline{\sigma_{\cdot j}}=\min_{k\ne j}  \sigma_{kk\cdot j}$.   
As $z \in ( \max\{|z_j|,\maxnorm{g}\}, t+\epsilon]$, we can choose $\delta$ such that $0< \delta <  z- \max\{|z_j|,\maxnorm{g}\}$. Throughout the following proof, we will work with such $\delta$. Since $\max\{|z_j|,\maxnorm{g}\} - z < -\delta$, we have
%When $\forall~ |z_j|,\maxnorm{g} <  z\le t+\epsilon$, choosing $\delta$ small enough, we have
\begin{equation}\label{eq:density_eps}
\Pc{ \big| \Zmax - z \big| \le \delta}{Z_j = z_j,G=g} = \Pc{ \big| \Xmax - z \big| \le \delta}{Z_j = z_j,G=g}, 
\end{equation}
where $X$ denotes the $(d-\discon - 1)$-dimensional random vector by excluding $Z_j, G$ from $Z$ and therefore $\Zmax = \max\{\Xmax,|Z_j|,\Gmax\}$. %The equality holds since we can choose $\delta$ such that $\delta <  z- \max\{|z_j|,\maxnorm{g}\}$.
% notice $G\independent \{Z\setminus G\}$ and 
%Let $\tilde{Z} = Z+\bar{\rho}_j (|z_j| + \maxnorm{g})\bm{1}_{d}$, $\tilde{z}= z + \bar{\rho}_j(|z_j| + \maxnorm{g})$, $\tilde{X} = X+\bar{\rho}_j (|z_j| + \maxnorm{g})\bm{1}_{d-\discon - 1}$, where $\bar{\rho}_{j} := \max_{k\ne j}\frac{|\sigma_{jk}|}{\sigma_{jj}}$, 
%%and denote $\maxnorm{\tilde X}$ accordingly, 
%then \eqref{eq:density_eps} can be rewritten as,
%\begin{equation}
%\label{eq:density_eps_rewri}
%\Pc{ \big| \Xmax - z \big| \le \delta}{Z_j = z_j,G=g} = \Pc{ \Big| \maxnorm{\tilde{X}}  - \tilde{z}  \Big| \le \delta}{Z_j = z_j, G=g}.
%\end{equation}

Recalling $G = (G_\ell)_{\ell \in [\discon]} = (Z_m)_{m\in \cE_G}$, where $\cE_{G}$ denotes the indices of the random variables in $G$ (among $Z$), i.e., $\cE_G= \{m \in [d] :  Z_m = G_\ell \text{ for some } \ell \in [\discon]\}$, we have 
$$
\Xmax = \max_{k \in [d],~ k\notin \{j,\cE_G \} }\{\max\{Z_k,-Z_k\}\}.
$$
%Note that $\Xmax = \max_{1\le k \le d,~ k\notin \{j,\cE_G \} }\{\max\{Z_k,-Z_k\}\}$ where $\cE_{G}$ denote the indices of the random variables in $G$ (among $Z$), i.e., $\cE_G= \{m \in [d]:Z_m = G_\ell \text{ for some } \ell \in [\discon]\}$ and $G = (G_\ell)_{\ell \in [\discon]} = (Z_m)_{m\in \cE_G}$. 
Given $j$ and the choice of $G$, we also denote
\begin{equation} \label{eq:sigma_rho_def} 
  \quad \underline{\sigma_{\cdot j}}:=\min_{k \in [d], k\notin \{j,\cE_G \} }  \sqrt{\sigma_{kk\cdot j}}, 
  \quad \bar{\rho}_{j} := \max_{k \in [d], k\notin \{j,\cE_G \}}\frac{|\sigma_{jk}|}{\sigma_{jj}}.
\end{equation}
For each $k \in [d], k\notin \{j,\cE_G \}$, the conditional expectation $\Ec{Z_k}{Z_j,G}$ has the following expression, 
\begin{eqnarray}\label{eq:CE_expression}
\Ec{Z_k}{Z_j,G} = \sigma_{kj} \sigma_{jj}^{-1}Z_j + \sum_{m\in \cE_G}(\sigma_{km} \sigma_{mm}^{-1}Z_m),
\end{eqnarray}
since $(Z_j, G)$ are all independent. Note that the requirement $(c)$ in Assumption \ref{asp:connect_prop} says $\forall~ k \in [d]$, $|\{m \in \cE_0: |\sigma^{U}_{k m}|+ |\sigma^{V}_{k m}| \ne 0 \}| \le c_0$, we thus have 
\begin{eqnarray} \nonumber
\sum_{m \in \cE_G }\Indrbr{\sigma_{km}\ne 0} 
&=& \sum_{m \in \cE_G }\Indrbr{ (s \sigma^U_{km} + {(1-s)}\sigma^V_{km})\ne 0}\\
&\le &  \sum_{m \in \cE_G }\Indrbr{\sigma^{U}_{km} \ne 0\text{ or } \sigma^{V}_{km}\ne 0} \le c_0, \label{eq:indic_bound}
\end{eqnarray}
where the first equality holds by the definition of $\sigma_{km}$ and $Z= W(s)=\sqrt{s}U + \sqrt{1-s}V$.    
%$\sum_{m \in \cE_G }\Indrbr{\sigma_{km}\ne 0} =\sum_{m \in \cE_G }\Indrbr{\big((\sqrt{s}\sigma^U_{km})^2 + (\sqrt{1-s}\sigma^V_{km})^2\big)\ne 0} = \sum_{m \in \cE_G }\Indrbr{\sigma^{U}_{km}\text{ or } \sigma^{V}_{km}\ne 0} \le c_0$. 
Combining \eqref{eq:indic_bound} with \eqref{eq:CE_expression}, it yields the following bound on $|\Ec{Z_k}{Z_j=z_j,G=g}|$,
\begin{eqnarray}\nonumber
|\Ec{Z_k}{Z_j=z_j,G=g}| &=& \big|  \sigma_{kj} \sigma_{jj}^{-1}z_j + \sum_{m\in \cE_G}(\sigma_{km} \sigma_{mm}^{-1}z_m) \big| \\
&\le & \bar{\rho}_j (|z_j| + c_0\maxnorm{g}), \label{eq:CE_bound}
%\sigma_{kj} \sigma_{jj}^{-1}z_j + c_0\min_{m\in \cE_G}(\sigma_{km} \sigma_{mm}^{-1}z_m).
%\le \bar{\rho}_j (|z_j| + \max_{m\in E_G}{|z_m|})=\bar{\rho}_j (|z_j| + \maxnorm{g})
\end{eqnarray}
where $\bar{\rho}_{j} = \max_{k\in \cE_X}\frac{|\sigma_{jk}|}{\sigma_{jj}}$ as defined.
Denoting $\cE_X := \{k: k \in [d],~ k\notin \{j,\cE_G \}\}$, we define the following random variables,
\begin{equation}\label{eq:Wm_def}
\tilde{W}_{2k - 1} = \frac{ {Z_k} - {z} }{\sqrt{\sigma_{kk\cdot j}}}+ \frac{\tilde{z}}{\underline{\sigma_{\cdot j}}},\quad
\tilde{W}_{2k} = \frac{-{Z_k} - {z}}{\sqrt{\sigma_{kk\cdot j}}} + \frac{\tilde{z}}{\underline{\sigma_{\cdot j}}},\quad  k \in \cE_X, %^1\le k \le d,~ k\notin \{j,\cE_G \},
\end{equation}
where $\tilde{z} = z + \bar{\rho}_j (|z_j| + c_0\maxnorm{g})$.
%Notice that $\Xmax = \max_{1\le k \le d,~ k\notin \{j,\cE_G \} }\{\max\{Z_k,-Z_k\}\}$ where $\cE_{G}$ is the set of corresponding indexes of the sub-vector $G$ in $Z$, we define the following random variables% $\{W_m\}$,
%\begin{equation}%\label{eq:Wm_def}
%W_{2k - 1} = \frac{\tilde{Z_k} - \tilde{z} }{\sigma_{kk\cdot j}}+ \frac{\tilde{z}}{\underline{\sigma_{\cdot j}}},\quad
%W_{2k} = \frac{-\tilde{Z_k} - \tilde{z}}{\sigma_{kk\cdot j}} + \frac{\tilde{z}}{\underline{\sigma_{\cdot j}}},\quad  1\le k \le d,~ k\notin \{j,\cE_G \},
%\end{equation}
%where $ \sigma_{kk\cdot j}= \sigma_{kk} - \sigma_{kj}^2 \sigma_{jj}^{-1} - \sum_{m\in \cE_G }\sigma_{km}^2 \sigma_{mm}^{-1}$, and $\underline{\sigma_{\cdot j}}=\min_{k\ne j}  \sigma_{kk\cdot j}$.
%$\forall~ k \in [d]$, $|\{m \in \cE_0: \sigma^{U}_{k m}$ or $\sigma^{V}_{k m} \ne 0 \}| \le 2$
Then by the definitions of $\sigma_{kk\cdot j}, \underline{\sigma_{\cdot j}}$ and $\bar{\rho}_{j} $ in \eqref{eq:sigmakkj_def} and \eqref{eq:sigma_rho_def}, we have the above random variables satisfy the following properties,
%where $1\le k \le d,~ k\notin \{j,\cE_G \}$. Then by the definition of $ \bar{\rho}_{j} = \max_{k\ne j}\frac{|\sigma_{jk}|}{\sigma_{jj}}$ and $\tilde{Z}_k = Z_k  +\bar{\rho}_j (|z_j| + \maxnorm{g}) $, we have
%$$ 
%\Ec{\tilde{Z}_k}{Z_j = z_j, G=g}\ge 0,~~\Varc{\tilde{Z}_k}{Z_j = z_j, G=g}=\sigma_{kk\cdot j}^2,
%$$
%thus the random variables defined in \eqref{eq:Wm_def} satisfy
 $$
 \Ec{\tilde{W}_m}{Z_j = z_j, G=g}\ge 0 , \quad \Varc{\tilde{W}_m}{Z_j = z_j, G=g} =1,
 $$
 where $m= 2k-1$ or $2k$ and $k\in \cE_X$. Denote those random variables defined in \eqref{eq:Wm_def} by $\{\tilde{W}_m\}$ for notation simplicity. We let $q_{z_j, g}(w)$ be the PDF of the conditional distribution of $\max_m\{\tilde{W}_m\}\mid Z_j = z_j, G=g$. Then we will apply the derivation of Step $2$ in Theorem $3$ of \cite{chernozhukov2015comparison} to bound $q_{z_j, g}(w)$. Note that for the following derivations, we always conditional on the event $Z_j = z_j, G=g$. First, we verify the condition on $\{\tilde{W}_m\}$. Since $|\mathrm{Corr}{(U_j, U_k)}| \ne 1$, $|\mathrm{Corr}{(V_j, V_k)}| \ne 1$ for distinct $j,k \in [d]$, we then have the correlation between $\tilde{W}_{m_1}$ and $\tilde{W}_{m_2}$ for $m_1 \ne m_2$ is less than $1$.
% $\{\tilde{W}_m\}$ satisfy $\mathrm{Corr}{(\tilde{W}_{m_1}, \tilde{W}_{m_2})} < 1$ for $m_1 \ne m_2$. 
Therefore, by applying the derivation of Step $2$ in Theorem $3$ of \cite{chernozhukov2015comparison} to $\{\tilde{W}_m\}$, we have 
$$
q_{z_j, g}(w) \le h(w):= 2(w\vee 1)\exp\left\{- \frac{(w-\bar{w}- a_d)^2_{+}}{2}\right\},
$$
where $\bar{w}=\max_{m}\Ec{\tilde{W}_m}{Z_j = z_j, G=g} $ and 
$$
a_d = \max_{m}\Ec{\Big(\tilde{W}_m - \Ec{\tilde{W}_m}{Z_j = z_j, G=g}\Big)}{Z_j = z_j, G=g}. 
$$ 
When $w\le \bar{w}+a_d$, we have $h(w)\le 2( \bar{w}+a_d)$. To deal with the case where $w> \bar{w}+a_d$, we consider
\begin{eqnarray} \nonumber
\log (h(w)) &=& \log(2w) -\frac{(w-\bar{w}- a_d)^2}{2}, \\ \nonumber
\frac{d \log (h(w))}{d w} &=& \frac{1}{w} - (w-\bar{w}- a_d), \\ \nonumber
\frac{d^2 \log (h(w))}{d w^2} &=& -\frac{1}{w^2} - 1 <0.
\end{eqnarray}
Solving $\frac{d}{dw}\log (h(w)) = 0$ yields $
w^{\star}=\frac{ \bar{w}+ a_d + \sqrt{(\bar{w}+a_d)^2+4}}{2} 
$. Therefore, the PDF of the conditional distribution of $\max_m\{\tilde{W}_m\}\mid Z_j = z_j, G=g$ can be bounded by
\begin{equation}\label{eq:pdfg_bound}
h(w) \le h(w^{\star}) \le 3 (\bar{w}+a_d).  
\end{equation}
%$g(w) \le 3 (\bar{w}+a_d)$. 
%can be upper bounded by
%$
%3 (\bar{w}+a_d).
%$
%Denote $\cE_X := \{k:1\le k \le d,~ k\notin \{j,\cE_G \}\}$, then by combining
%\eqref{eq:density_eps}, \eqref{eq:density_eps_rewri} and \eqref{eq:Wm_def} we have the following bound,
%
%\begin{eqnarray}\nonumber
% &~& \Pc{ \big| \Zmax - z \big| \le \delta}{Z_j = z_j,G=g} \\
%     &=&\Pc{ \left| \maxnorm{\tilde{X}}  - \tilde{z} \right|\le \delta}{Z_j = z_j, G=g}\\ \nonumber
%     &=& \Pc{ \left| \max_{k\in \cE_X}\left\{\tilde{Z}_k,-\tilde{Z}_k\right\}  - \tilde{z} \right|\le \delta}{Z_j = z_j, G=g}\\ \nonumber
%     &\le& \Pc{ \left| \max_{k\in \cE_X}\left\{\frac{\tilde{Z}_k - \tilde{z}}{ \sigma_{kk\cdot j}},\frac{-\tilde{Z}_k - \tilde{z}}{ \sigma_{kk\cdot j}}\right\}   \right|\le \frac{\delta}{ \underline{\sigma_{\cdot j}}}}{Z_j = z_j, G=g}\\ \nonumber
%     &\le& \sup_{y\in \RR}\Pc{ \left| \max_{k\in \cE_X}\left\{\frac{\tilde{Z}_k - \tilde{z}}{ \sigma_{kk\cdot j}},\frac{-\tilde{Z}_k - \tilde{z}}{ \sigma_{kk\cdot j}}\right\} 
%     + \frac{\tilde{z}}{ \underline{\sigma_{\cdot j}}}- y
%     \right|\le \frac{\delta}{ \underline{\sigma_{\cdot j}}}}{Z_j = z_j, G=g}\\ \nonumber
%     &=& \sup_{y\in \RR}\Pc{ \left| \max_{m}\left\{ W_m\right\}  - y
%     \right|\le \frac{\delta}{ \underline{\sigma_{\cdot j}}}}{Z_j = z_j, G=g} \le \frac{6\delta}{\underline{\sigma_{\cdot j}}} (\bar{w}+a_d)
%\end{eqnarray}
Now we have
\begin{align}\nonumber
  ~~& \Pc{ \big| \Zmax - z \big| \le \delta}{Z_j = z_j,G=g} \\ \nonumber
     =~&\Pc{ \left| \maxnorm{{X}}  - {z} \right|\le \delta}{Z_j = z_j, G=g}\\ \nonumber
     = ~& \Pc{ \left| \max_{k\in \cE_X}\left\{{Z}_k,-{Z}_k\right\}  - {z} \right|\le \delta}{Z_j = z_j, G=g}\\ \nonumber
     \le ~& \Pc{ \left| \max_{k\in \cE_X}\left\{\frac{{Z}_k - {z}}{ \sqrt{\sigma_{kk\cdot j}}},\frac{-{Z}_k - {z}}{ \sqrt{\sigma_{kk\cdot j}}}\right\}   \right|\le \frac{\delta}{ \underline{\sigma_{\cdot j}}}}{Z_j = z_j, G=g}\\ \nonumber
     \le ~& \sup_{y\in \RR}\Pc{ \left| \max_{k\in \cE_X}\left\{\frac{{Z}_k - {z}}{\sqrt{ \sigma_{kk\cdot j}}} + \frac{\tilde{z}}{\underline{\sigma_{\cdot j}}} ,\frac{-{Z}_k - {z}}{ \sqrt{ \sigma_{kk\cdot j}} } + \frac{\tilde{z}}{\underline{\sigma_{\cdot j}}}\right\} 
     - y
     \right|\le \frac{\delta}{ \underline{\sigma_{\cdot j}}}}{Z_j = z_j, G=g}\\ 
     =~& \sup_{y\in \RR}\Pc{ \left| \max_{m}\{ \tilde{W}_m\}  - y
     \right|\le \frac{\delta}{ \underline{\sigma_{\cdot j}}}}{Z_j = z_j, G=g} \le \frac{6\delta}{\underline{\sigma_{\cdot j}}} (\bar{w}+a_d),
     \label{eq:Zpdf_final}
\end{align}
where the first equality holds by \eqref{eq:density_eps}, the second equality holds by the definition of $X$ and $\cE_X$, the first inequality holds since $ \underline{\sigma_{\cdot j}}=\min_{k\ne j}  \sqrt{\sigma_{kk\cdot j}}$, the third equality holds by the definition of $\{\tilde{W}_m\}$ in \eqref{eq:Wm_def}, and the last inequality holds by the  bound on $h(w)$ in \eqref{eq:pdfg_bound}. 
Regarding the quantity $\bar{w}=\max_{m}\Ec{\tilde{W}_m}{Z_j = z_j, G=g}$, we have
\begin{align}\nonumber
 \bar{w} %&=& \max_{m}\Ec{W_m}{Z_j = z_j, G=g}\\ \nonumber
=~&\max_{k\in \cE_X}\left\{ \frac{\pm \Ec{{Z_k}}{Z_j = z_j, G=g} - {z}}{\sqrt{ \sigma_{kk\cdot j}} } + \frac{\tilde{z}}{\underline{\sigma_{\cdot j}}}\right\}\\ \nonumber
\le ~& \max_{k\in \cE_X}\left\{ 
\frac{\pm \Ec{{Z_k}}{Z_j = z_j, G=g} }{ \sqrt{\sigma_{kk\cdot j}}}
\right\}  + \max_{ k\in \cE_X} \left\{ \frac{1}{\underline{\sigma_{\cdot j}}}  - \frac{1}{ \sqrt{\sigma_{kk\cdot j}} }
\right\}z+ \frac{\bar{\rho}_j (|z_j| + c_0\maxnorm{g})}{\underline{\sigma_{\cdot j}} }\\ \nonumber
\le ~& \max_{k\in \cE_X}\left\{ 
\frac{\pm  \big(\sigma_{kj} \sigma_{jj}^{-1}z_j + \sum_{m\in \cE_G}(\sigma_{km} \sigma_{mm}^{-1}z_m)\big) }{\sqrt{ \sigma_{kk\cdot j}}}
\right\} +  \frac{z}{ \underline{\sigma_{\cdot j}}}  %+ ( 1/\underline{\sigma_{\cdot j}} ~-  1/\overline{\sigma_{\cdot j}})z 
+ \frac{\bar{\rho}_j (|z_j| + c_0\maxnorm{g})}{\underline{\sigma_{\cdot j}} } \\ 
\le ~& \frac{2 \bar{\rho}_j(|z_j| +c _0 \maxnorm{g} ) }{ \underline{\sigma_{\cdot j}}} + \frac{z}{ \underline{\sigma_{\cdot j}}} \le \frac{2 \bar{\rho}_j(1+c_0) +1  }{ \underline{\sigma_{\cdot j}}}(t +\epsilon),
\label{eq:wbar_bound}
\end{align}
where $\max\{\pm A\} := \max\{A, -A\}$, the first inequality holds by the definition of $\tilde{z}$, the second inequality holds by \eqref{eq:CE_expression}, and the last inequality holds by the definitions of $\bar{\rho}_j$ and $\underline{\sigma_{\cdot j}}$ and the fact $\sum_{m \in \cE_G }\Indrbr{\sigma_{km}\ne 0}  \le c_0$. 

Let $\delta$ in \eqref{eq:Zpdf_final} go to $0$, we get the following bound on the density function of the conditional distribution of $\Zmax \mid  \{Z_j = z_j, G = g\}$, i.e., when $0\le t \le C_0\sqrt{\log d}$ and $|z_j|,\maxnorm{g} \le t+\epsilon$,
\begin{equation} \label{eq:f_bound_v1}
  f_{g, z_j}(z)\le \frac{6}{\underline{\sigma_{\cdot j}}} (\bar{w}+a_d)\le \frac{6}{\underline{\sigma_{\cdot j}}} \left( 
\frac{   2 \bar{\rho}_j(1+c_0) +1}{ \underline{\sigma_{\cdot j}}}C_1\sqrt{\log d}
+C_2\sqrt{\log d}
\right), % f_{g, z_j}(z)\le C \sqrt{\log d},
\end{equation}
for any $ z \in ( \max\{|z_j|,\maxnorm{g}\}, t+\epsilon]$. The first inequality holds by \eqref{eq:Zpdf_final}. Regarding the second inequality, we apply the result in $\eqref{eq:wbar_bound}$ and bound $(t+\epsilon)$ and $a_d$ by $C_1\sqrt{\log d}$ for some constant $C_1$. Note $a_d\le C_1\sqrt{\log d}$ is because of  the maximal inequalities for sub-Gaussian random variables (Lemma 5.2 in \cite{van2014probability}).
%Note can simply bound $(t+\epsilon)$ by $C_1\sqrt{\log d}$ and $a_{d}$ by $C_2\sqrt{\log d}$
%for some constants $C_1,C_2>0$, let $\delta$ go to $0$, we finally obtain, when $0\le t \le C_0\sqrt{\log d}$ and $|z_j|,\maxnorm{g} \le t+\epsilon$,
%$$
%f_{g, z_j}(z)\le  \frac{6}{\underline{\sigma_{\cdot j}}} \left( 
%\frac{   2 \bar{\rho}_j(1+c_0) +1}{ \underline{\sigma_{\cdot j}}}C_1\sqrt{\log d}
%+C_2\sqrt{\log d}
%\right)\le C \sqrt{\log d},\quad  \forall~ |z_j|,\maxnorm{g} < z \le t+ \epsilon,
%$$
As for $\bar{\rho}_{j} = \max_{k\in  \cE_X}\frac{|\sigma_{jk}|}{\sigma_{jj}}$, we have
$$
\bar{\rho}_{j}^2 \le \max_{k\ne j} {\frac{\sigma_{kk}}{\sigma_{jj}}} \le \frac{\max_j \sigma_{jj}^U }{ \min_j \sigma_{jj}^U } \le \frac{\max_j \sigma_{jj}^U }{ \lambda_{\min}(\bSigma^{U}) }   = O(1), %\le \frac{a_1}{a_0},
$$
where the first inequality holds by the Cauchy-Schwarz inequality, the second inequality holds by the definition of $Z$ and $\sigma_{jj}^U = \sigma_{jj}^V$, the third inequality holds by the fact that $ \min_j \sigma_{jj}^U  \ge \lambda_{\min}(\bSigma^{U})$, and the last step holds under the stated assumption of Theorem \ref{thm:ccb_sparse}. As for $\underline{\sigma_{\cdot j}} = \min_{k \in  \cE_X}  \sqrt{\sigma_{kk\cdot j}}$ where $\sigma_{kk \cdot j}= \sigma_{kk} - \sigma_{kj}^2 \sigma_{jj}^{-1} - \sum_{m\in \cE_G }\sigma_{km}^2 \sigma_{mm}^{-1} = \Varc{Z_k}{Z_j, G}$, we have
\begin{eqnarray*} 
\frac{1}{\underline{\sigma^2_{\cdot j}}} &=& \frac{1}{\min_{k \in \cE_X}\Varc{Z_k}{Z_j, G}}  \\
&\le & \frac{1}{\min_{k}\Varc{Z_k}{Z_{\text{-}k}}}  \\
 &=&  \max_{k} ((\bSigma^{Z})^{-1})_{kk} \\
 &\le &  \lambda_{\max}((\bSigma^{Z})^{-1}) \\
& = &1/\lambda_{\min}(\bSigma^{Z}) \\
 &\le &  (\min\{ \lambda_{\min}(\bSigma^{U}), \lambda_{\min}(\bSigma^{V}\})^{-1} \le b_0,
\end{eqnarray*}
under the stated assumption that $\lambda_{\min}(\bSigma^{U}) \ge  1/b_0, \lambda_{\min}(\bSigma^{V}) \ge 1/b_0$, where the first inequality holds since $(Z_j, G)$ is  a sub-vector of  $  Z_{\text{-}k}:= Z_{(1:d)\setminus k}$, the second equality holds by the relationship between the partial variances and the inverse covariance matrix, and the last three hold by the definitions of $\lambda_{\min}(\cdot), \lambda_{\max}(\cdot)$. Thus we have $f_{g, z_j}(z)\le C \sqrt{\log d}$ for some constant $C$, i.e., Lemma \ref{lem:maximas_pdf_bound} is proved. 
%depending on $a_0, a_1, b_0, c_0$, i.e., Lemma \ref{lem:maximas_pdf_bound} is proved.

%  when $|Z_j| = |z_j| \le t+\epsilon$ where $t$ satisfies $ \frac{\phi(\frac{t}{\sigma_{\max}})}{\frac{t}{\sigma_{\max}}} \ge \frac{1}{d^2} $. Note that it suffices to derive a upper bound on the PDF function of the conditional distribution of $\Zmax + \rho_{\max}|Z_j|\mid Z_j$. It only differs from $\max\{W\}\mid Z_j$ by the normalizing factor $\sigma_{kk\cdot j}$. 
\end{proof}
\begin{lemma}\label{lem:mills}
For $z>0$, we have
$$
\frac{\phi(z)}{2(z \vee 1)}  \le \bar{\Phi}(z)  = 1- {\Phi}(z)  \le \frac{\phi(z)}{z},
$$
where $\phi(z), \Phi(z)$ is the PDF and CDF of the standard Gaussian distribution respectively. 
\end{lemma}
\begin{proof}[Proof of Lemma \ref{lem:mills}]
This is a simple fact derived from Mill's inequality; see the derivations in the proof of Theorem 3 in \cite{chernozhukov2015comparison}.
\end{proof}
\begin{lemma}
\label{lem:Lambda_bound}
% Denote $\rho_{0} =(\frac{\sigma_{\min}}{\sigma_{\max}})^2$, we have
% \begin{equation}\label{eq:phi_bound}
% \Lambda(t,\epsilon,p):=
% \frac{(1 - 2 \bar{\Phi}(\frac{t + \epsilon}{\sigma_{\min}}))^p}{ 1- (1 - 2 \bar{\Phi}(\frac{t }{\sigma_{\min}}))^p} \cdot \phi( \frac{t - \epsilon}{\sigma_{\max}}) \stackrel{?}{\le} C  \frac{(d^2)^{\rho_0^{-1}-1}}{p}(\sqrt{\log d})^{\rho_0}
% \end{equation}
% whenever $t$ satisfies $ \frac{\phi(\frac{t}{\sigma_{\max}})}{\frac{t}{\sigma_{\max}}} \ge \frac{1}{d^2} $ and $\epsilon$ chosen to be at smaller enough. Further, if we impose the condition that 
% $\rho_0^{-1} = 1 + \frac{1}{\log d}$, we have
% $$
% \frac{(1 - 2 \bar{\Phi}(\frac{t + \epsilon}{\sigma_{\min}}))^p}{ 1- (1 - 2 \bar{\Phi}(\frac{t }{\sigma_{\min}}))^p} \cdot \phi( \frac{t - \epsilon}{\sigma_{\max}})  \stackrel{?}{\le} C'  \frac{\sqrt{\log d}}{p} 
% $$
%Suppose $\discon > 1$, 
Whenever $0\le t \le C_0\sqrt{\log d}$ for some constant $C_0>0$, and $\epsilon = {c}/{\max\{(\log d)^{3/2}, \discon \log d \}}$ for some small enough constant $c>0$, we have %for some constant $C>0$, the following holds,
\begin{equation}
\label{eq:phi_bound}
\Lambda(t,\epsilon,\discon):=
\frac{(1 - 2 \bar{\Phi}(\frac{t + \epsilon}{\sigma_{j}}))^\discon}{ 1- (1 - 2 \bar{\Phi}(\frac{t }{\sigma_{j}}))^\discon} \cdot \phi\big( \frac{t - \epsilon}{\sigma_{j}}\big) = O\rbr{ \frac{\sqrt{\log d}}{\discon}}.
\end{equation}
%$t$ satisfies $ \frac{\phi(\frac{t}{\sigma_{\max}})}{\frac{t}{\sigma_{\max}}} \ge \frac{1}{d^2} $ 
% Further, if we impose the condition that 
% $\rho_0^{-1} = 1 + \frac{1}{\log d}$, we have
% $$
% \frac{(1 - 2 \bar{\Phi}(\frac{t + \epsilon}{\sigma_{\min}}))^p}{ 1- (1 - 2 \bar{\Phi}(\frac{t }{\sigma_{\min}}))^p} \cdot \phi( \frac{t - \epsilon}{\sigma_{\max}})  \stackrel{?}{\le} C'  \frac{\sqrt{\log d}}{p} 
% $$
\end{lemma}
%\begin{remark}\label{rk:lem:Lambda_bound}
%
%\end{remark}
\begin{proof}[Proof of Lemma \ref{lem:Lambda_bound}]
% First notice a simple fact that
% $$
% \frac{\phi(z)}{2(z \vee 1)}  \le \bar{\Phi}(z)  = 1- {\Phi}(z)  \le \frac{\phi(z)}{z}
% $$
% where $\phi(z), \Phi(z)$ is the PDF and CDF function of the standard Gaussian distribution. 
By Lemma \ref{lem:mills}, we can simplify $\Lambda(t,\epsilon,\discon)$ into the following
$$
\Lambda(t,\epsilon,\discon)\le
\frac{\left(1 -  \frac{{\phi}(\frac{t + \epsilon}{\sigma_{j}})}{\frac{t + \epsilon}{\sigma_{j}}\vee 1} \right)^\discon}{ 1- \left(1 -  \frac{\phi(\frac{t }{\sigma_{j}})}{\frac{t }{\sigma_{}}\vee 1 }\right)^\discon} \cdot \phi\big( \frac{t - \epsilon}{\sigma_{j}}\big).
$$
When $\frac{t}{\sigma_{j}}\le 1$, we have $\frac{t+\epsilon}{\sigma_j}\le 2 $ due to the choice of $\epsilon$. Because $\frac{t}{\sigma_{j}}>0,\frac{t+\epsilon}{\sigma_j}> 0$ and $\phi(z)$ is monotonically decreasing when $z>0$, we then have the the bound below, 
$$
\Lambda(t,\epsilon,p) \le  \frac{(1-  \phi(2)/2)^\discon}{1 - (1- \phi(1))^\discon}= O\rbr{ \frac{\sqrt{\log d}}{\discon}},
$$
where the second inequality holds due to $0< \phi(2)<\phi(1)<0.5$ and $\discon > 1$. Now it suffices to consider the case where $ \frac{t+\epsilon}{\sigma_{j}}> \frac{t}{\sigma_{j}} > 1$ and deal with the following
$$
\Lambda(t,\epsilon, \discon) \le  \frac{\left(1 -  \frac{{\phi}(\frac{t + \epsilon}{\sigma_{j}})}{\frac{t + \epsilon}{\sigma_{j}}} \right)^\discon}{ 1- \left(1 -  \frac{\phi(\frac{t +\epsilon }{\sigma_{j}})}{\frac{t + \epsilon }{\sigma_{j}} }\right)^\discon} \cdot \phi\rbr{ \frac{t - \epsilon}{\sigma_{j}}}.
$$
We further bound $\Lambda(t,\epsilon, \discon)$ as
\begin{eqnarray} \nonumber
 \Lambda(t,\epsilon,\discon) 
 &\le & \frac{\left(1 -  \frac{{\phi}(\frac{t + \epsilon}{\sigma_j})}{\frac{t + \epsilon}{\sigma_j}} \right)^\discon}{ 1- \left(1 -  \frac{\phi(\frac{t +\epsilon }{\sigma_j})}{\frac{t + \epsilon }{\sigma_j} }\right)^\discon} \cdot \phi\Big( \frac{t + \epsilon}{\sigma_j}\Big) \cdot  e^{\frac{t\epsilon}{2\sigma^2_{j}}} \\ \nonumber
 &\le &  2~\frac{\left(1 -  \frac{{\phi}(\frac{t + \epsilon}{\sigma_j})}{\frac{t + \epsilon}{\sigma_j}} \right)^\discon}{ 1- \Bigg(1 -  \frac{\phi(\frac{t +\epsilon }{\sigma_j})}{\frac{t + \epsilon }{\sigma_j} }\Bigg)^\discon}  \cdot 
  \phi\Big( \frac{t + \epsilon}{\sigma_j}\Big)  \\ \nonumber
    &:= & 2 e^{H(\lambda)} \cdot \frac{t+\epsilon}{\discon\sigma_j}
    % \le   \frac{C\sqrt{\log d}}{\discon},
\end{eqnarray}
where the first inequality comes from rearranging, the second inequality holds since $ \exp{(\frac{t\epsilon}{2\sigma^2_{j}})} < 2$ for $t\le C_0\sqrt{\log d}$. This is because $\epsilon = {c}/{\max\{(\log d)^{3/2}, \discon \log d \}}$ for some small enough constant $c>0$. The last line holds by rewriting using some new notations: $\lambda:= \discon~\frac{{\phi}(\frac{t + \epsilon}{\sigma_{j}})}{\frac{t + \epsilon}{\sigma_{j}}}$ and 
% $H(\lambda)$
\begin{align}\label{eq:Hlambda_def}
    H(\lambda) 
    := \log \left(\frac{(1 - \frac{\lambda}{\discon})^\discon}{1- (1 - \frac{\lambda}{\discon})^\discon} 
    \cdot \lambda\right)
    = \discon \log\Big( 1 - \frac{\lambda}{\discon}\Big) - \log\Big({ 1- (1 - \frac{\lambda}{\discon})^\discon }\Big) +  \log \lambda.
\end{align}
Since $\frac{t+\epsilon}{\sigma_{j}}>1$, we have $0< \lambda <\discon$. Below we will first deal with $H(\lambda) $ then obtain the bound on $ \Lambda(t,\epsilon,\discon)$. To bound $H(\lambda) $, consider taking the derivative of $H(\lambda)$ with respect to $\lambda$, then we have
\begin{eqnarray} \nonumber
    H'(\lambda) 
    &=& \frac{\discon}{\lambda - \discon} - \frac{(1 - \frac{\lambda}{\discon})^{(\discon-1)}}{ 1- (1 - \frac{\lambda}{\discon})^\discon  } + \frac{1}{\lambda}
    \\ \nonumber
    &=& \frac{\discon}{\lambda - \discon} -  \frac{1}{1 - \frac{\lambda}{\discon}  }\cdot \frac{(1 - \frac{\lambda}{\discon})^{\discon}}{ 1- (1 - \frac{\lambda}{\discon})^\discon  } + \frac{1}{\lambda}\\ \nonumber
    &\le & \frac{\discon}{\lambda - \discon} -  \frac{1}{1 - \frac{\lambda}{\discon}  }\cdot \frac{1-\lambda }{1- (1-\lambda)} + \frac{1}{\lambda}\\ \nonumber
    &=& \frac{\discon}{\lambda - \discon} + \frac{1}{ 1 - \frac{\lambda}{\discon}} - \frac{1}{\lambda}\cdot \frac{1}{ 1 - \frac{\lambda}{\discon}} +  \frac{1}{\lambda} \\ \nonumber
    &\le& \frac{1}{\lambda}\Big(1 - \frac{\discon}{\discon-\lambda}\Big) < 0 ,
\end{eqnarray}
 where the first inequality holds by the Bernoulli's inequality: ${(1+x)^{r}\geq 1+rx}$ when $r \in \NN,1+x\ge 0$, and the last inequality holds since $0<\lambda< \discon$. Now we have $H(\lambda)$ is monotone decreasing. When $0\le t \le C_0\sqrt{\log d}$, we will first find the lower bound on $\lambda= \discon~\frac{{\phi}(\frac{t + \epsilon}{\sigma_{j}})}{\frac{t + \epsilon}{\sigma_{j}}}$, denoted by  $\underline{\lambda}$. Then we have $H(\lambda)$ is bounded by $H(\underline{\lambda})$ due to its monotonicity.
%  For $t$ satisfies $0\le t \le C_0\sqrt{\log d}$,
%  %$ \frac{\phi(\frac{t}{\sigma_{\max}})}{\frac{t}{\sigma_{\max}}} \ge \frac{1}{d^2} $, 
% %  which will give a upper bound on $t$, 
% if we find the lower bound  $\underline{\lambda}$ on $\lambda= \discon~\frac{{\phi}(\frac{t + \epsilon}{\sigma_{j}})}{\frac{t + \epsilon}{\sigma_{j}}}$, then $H(\lambda)$ can be upper bounded by $H(\underline{\lambda})$ due to its monotonicity.
%First we derive $\underline{\lambda}$ then deal with $H(\underline{\lambda})$.
% Below we first derive the bound on $H(\underline{\lambda})$ and
% Thus $H(\lambda)$ can be upper bounded by $H(\underline{\lambda})$.
%  where 
%  $$
%  \underline{\lambda} \ge \frac{p \phi(\bar{t}) }{ \bar{t}}.
%  $$
%  To derive $\underline{\lambda}$, we first notice a very coarse upper bound on $t$, i.e.
%  $$
% e^{-\frac{t^2}{2\sigma_{\max}^2}}\ge   \frac{\sqrt{2\pi}t}{\sigma_{\max}} \cdot \frac{1}{d^2} \ge \frac{1}{d^2}
%  $$
%  As $t\le K\sqrt{\log d}$, 
Regarding $\underline{\lambda}$, we denote $ \bar{x}:=  2 C_0 \sqrt{\log d}/\sigma_j$ and note $\frac{\phi(x)}{x}$ is monotone decreasing when $x\ge 0$. Then we have, when  $0\le t \le C_0\sqrt{\log d}$,
\begin{eqnarray*}
 \discon ~ \frac{\phi(\frac{t + \epsilon}{\sigma_{j}})}{\frac{t + \epsilon}{\sigma_{j}}} 
    \ge  \discon ~ \frac{\phi( \bar{x} )}{\bar{x}  }  \ge  \frac{\discon}{d^{a_1}} := \underline{\lambda},
       \end{eqnarray*}
where $a_1 >2$. 
% \begin{eqnarray} \nonumber
%    \lambda =
%    \discon ~ \frac{\phi(\frac{t + \epsilon}{\sigma_{j}})}{\frac{t + \epsilon}{\sigma_{j}}} 
%    = \discon  \cdot \frac{ \phi(\frac{t }{\sigma_{j}})}{\frac{t}{\sigma_{j}}} \cdot \frac{t}{t+\epsilon}\cdot 
%    e^{-\frac{\epsilon^2 + 2t\epsilon}{\sigma^2_{j}}}
%    \ge a_1 \discon  \cdot \frac{ \phi(\frac{t }{\sigma_{j}})}{\frac{t}{\sigma_{j}}}
%    % &=&  c_1 p \cdot \left( \frac{ \phi(\frac{t }{\sigma_{j}})}{\frac{t}{\sigma_{j}}} \right)^{\frac{1}{\rho_0}} \cdot \left(\frac{\sqrt{2\pi} t }{ \sigma_{j}}\right)^{(\frac{1}{\rho_0} - 1)}\\
%    \ge a_1  \frac{\discon }{d^{a_2}}:= \underline{\lambda} > 0,
% \end{eqnarray}
% for some constants $0< a_1 < 1$, $ a_2 > 2$. The first inequality holds since $\epsilon$ can be chosen small enough such that $\frac{t}{t+\epsilon}\exp\big(-\frac{\epsilon^2 + 2t\epsilon}{\sigma^2_{j}}\big) > a_1$. Regarding the second inequality, we first have $ { \phi(\frac{t}{\sigma_{j}})}/{\frac{t}{\sigma_{j}}}$ is monotonically decreasing over the range of $t$, then for $t\le C_0\sqrt{\log d}$, we can lower bound ${ \phi(\frac{t}{\sigma_{j}})}/{\frac{t}{\sigma_{j}}}$ by $1/d^{a_2}$ for some $2 < a_2 < M_1$. 
 Therefore we obtain 
 %$\underline{\lambda} = C\frac{p}{d^{2/\rho_0}}$ for some constant $C>0$. Then we have
 \begin{align}
     H(\lambda) \le H( {\underline{\lambda} })
     =
     \left .
     \log \left(\frac{(1 - \frac{\lambda}{\discon})^\discon}{1- (1 - \frac{\lambda}{\discon})^\discon} 
    \cdot \lambda\right)  \right \vert _{\lambda = \underline{\lambda} }
    \le 
     \left .
     \log \left(
     \frac{{\lambda}}{1 - (1- \discon \frac{{\lambda}}{\discon } + \frac{(\discon -1)\discon }{2} \frac{\lambda^2}{\discon ^2} )} 
     \right) \right \vert _{\lambda = \underline{\lambda} }
      \le C',
      \label{eq:Hlambda_bound}
 \end{align}
 where the second inequality holds due to the fact that $(1 - \frac{\lambda}{\discon})^\discon \le 1$, $\frac{\lambda}{\discon} \in [0,1]$ and Lemma \ref{lem:bern_result}. The third inequality holds since  $\underline{\lambda} =   \frac{\discon }{d^{a_1}} \le  \frac{1}{d^{a_1 -1}} < \frac{1}{d}$, then we have
$$
 \left .
    \left(
     \frac{{\lambda}}{1 - (1- \discon \frac{{\lambda}}{\discon } + \frac{(\discon -1)\discon }{2} \frac{\lambda^2}{\discon ^2} )} 
     \right) \right \vert _{\lambda = \underline{\lambda} } = \frac{\underline{\lambda}}{\underline{\lambda} - \frac{2(\discon - 1)}{\discon}\underline{\lambda}^2} 
     \le \frac{\underline{\lambda}}{\underline{\lambda} - 2\underline{\lambda}^2}  
    = \frac{1}{1 -2\underline{\lambda}} \le C'_1,
$$
for some constant $C'_1$. Now we figure out the bound on $\Lambda(t,\epsilon,\discon) $, %and establish Lemma \ref{lem:Lambda_bound}.
 \begin{eqnarray} \nonumber
 \Lambda(t,\epsilon,\discon) 
 &\le&  \frac{\left(1 -  \frac{{\phi}(\frac{t + \epsilon}{\sigma_j})}{\frac{t + \epsilon}{\sigma_j}} \right)^\discon}{ 1- \left(1 -  \frac{\phi(\frac{t +\epsilon }{\sigma_j})}{\frac{t + \epsilon }{\sigma_j} }\right)^\discon} \cdot \phi\Big( \frac{t - \epsilon}{\sigma_j}\Big) \\ \nonumber 
  &\le& \frac{\left(1 -  \frac{{\phi}(\frac{t + \epsilon}{\sigma_j})}{\frac{t + \epsilon}{\sigma_j}} \right)^\discon}{ 1- \left(1 -  \frac{\phi(\frac{t +\epsilon }{\sigma_j})}{\frac{t + \epsilon }{\sigma_j} }\right)^\discon} \cdot \phi\Big( \frac{t + \epsilon}{\sigma_j}\Big) \cdot  e^{\frac{t\epsilon}{2\sigma^2_{j}}} \\ \nonumber
   &\le&  2~\frac{\left(1 -  \frac{{\phi}(\frac{t + \epsilon}{\sigma_j})}{\frac{t + \epsilon}{\sigma_j}} \right)^\discon}{ 1- \Bigg(1 -  \frac{\phi(\frac{t +\epsilon }{\sigma_j})}{\frac{t + \epsilon }{\sigma_j} }\Bigg)^\discon}  \cdot 
  \phi\Big( \frac{t + \epsilon}{\sigma_j}\Big)  \\ \nonumber
    &\le& 2 e^{H(\lambda)} \cdot \frac{t+\epsilon}{\discon\sigma_j}
     \le   \frac{C\sqrt{\log d}}{\discon},
\end{eqnarray}
where the second inequality comes from rearranging, the third inequality holds since $ \exp{(\frac{t\epsilon}{2\sigma^2_{j}})} < 2$ for $t\le C_0\sqrt{\log d}$. This is because $\epsilon = {c}/{\max\{(\log d)^{3/2}, \discon \log d \}}$ for some small enough constant $c>0$. And the last line holds by \eqref{eq:Hlambda_def} and  \eqref{eq:Hlambda_bound}. Therefore Lemma \ref{lem:Lambda_bound} is established.
\end{proof}

\begin{lemma} \label{lem:Lambda_bound2}
Under the same conditions as Lemma \ref{lem:Lambda_bound}, we have
\begin{equation}\label{eq:phi_bound_similar}
%\discon~ \Lambda(t,\epsilon,\discon) \le 
\frac{  (1 -  2\bar{\Phi}(\frac{t+\epsilon}{\sigma_j}) )^{\discon - 1}  }{ 1- (1 - 2 \bar{\Phi}(\frac{t }{\sigma_{j}}))^\discon} \cdot \Big(\discon \phi\big( \frac{t - \epsilon}{\sigma_{j}}\big)\Big)^2 = O\rbr{\sqrt{\log d}}. 
\end{equation}  
\end{lemma}
\begin{proof}
Note that the result \eqref{eq:phi_bound} in Lemma \ref{lem:Lambda_bound} can be rewritten as
$$
\discon~ \Lambda(t,\epsilon,\discon) = \frac{(1 - 2 \bar{\Phi}(\frac{t + \epsilon}{\sigma_{j}}))^\discon}{ 1- (1 - 2 \bar{\Phi}(\frac{t }{\sigma_{j}}))^\discon} \cdot \Big(\discon \phi\big( \frac{t - \epsilon}{\sigma_{j}}\big)\Big) =  O\rbr{\sqrt{\log d}}.
$$
By similar derivations as in the proof of Lemma \ref{lem:Lambda_bound}, we can establish
\begin{equation}\nonumber
%\label{eq:phi_bound_similar}
%\discon~ \Lambda(t,\epsilon,\discon) \le 
\frac{(1 - 2 \bar{\Phi}(\frac{t + \epsilon}{\sigma_{j}}))^{\discon - 1}}{ 1- (1 - 2 \bar{\Phi}(\frac{t }{\sigma_{j}}))^\discon} \cdot \Big(\discon \phi\big( \frac{t - \epsilon}{\sigma_{j}}\big)\Big)^2 = O\rbr{\sqrt{\log d}}.  
\end{equation}
\end{proof}
\begin{lemma}\label{lem:bern_result}
% Suppose $\discon >1$, 
  For $x \in [0, 1]$, we have 
  $(1-x)^{\discon} \le 1- \discon x + 0.5 \discon(\discon - 1) x^2 
  $.
\end{lemma}
\begin{proof}
When $\discon = 1$, the above simply holds. Now we consider the case where $\discon >1$. Let $Q(x) = (1-x)^{\discon} - (1- \discon x + 0.5 \discon(\discon - 1) x^2)$, we have $Q(0) = 0$ and   
\begin{align}
Q'(x) = -\discon(1-x)^{(\discon-1)} + \discon - \discon(\discon-1)x \le -\discon (1 - (\discon-1)x) + \discon - \discon(\discon-1)x = 0,
\end{align}
where the inequality holds by applying Bernoulli's inequality to $(1-x)^{(\discon-1)}$ for $\discon >1,  x \in [0,1]$. Therefore, $Q(x)$ is monotonically decreasing, and the statement is proved. 
\end{proof}

\subsection{Ancillary lemmas for Theorem \ref{thm:ccb_sparse_unitvar}}
\begin{remark}\label{rk:general_var_cond}
Recall that the connectivity assumption of Theorem \ref{eq:ccb_sparse_unitvar} assumes 
%that there exist disjoint partitions of nodes $\cup_{\ell=1}^{\discon}\cC^U_\ell = \cup_{\ell=1}^{\discon}\cC^V_\ell= [d], $ such that  $\sigma^U_{jk}$ ($\sigma^V_{jk}$) equals $0$ when $j,k$ belong to different components $\cC^U_\ell$ ($\cC^V_\ell$), and $\forall \ell \in[\discon]$, $\cC^U_{\ell} \cap \cC^V_{\ell} \ne \emptyset$. 
that there exists a disjoint $\discon$-partition of nodes $\cup_{\ell=1}^{\discon}\cC_\ell = [d]$ such that  $\sigma^U_{jk}=\sigma^V_{jk}=0$ when $j \in \cC_{\ell}$ and  $k \in \cC_{\ell'}$ for some $\ell \neq \ell'$. The more general version of the variance condition  assumes: $a_0 \le\sigma^U_{jj}=\sigma^V_{jj} \le a_1,~\forall j \in [d]$; given any $j \in \cC^U_\ell$ with some $\ell$, there exists at least one $m  \in \cC^U_{\ell'}$ such that $\sigma^U_{jj} = \sigma^{U}_{mm}$ for any $\ell' \ne \ell$.
% and $ \sigma^U_{jj} = \sigma^U_{kk}$ for $j, k $ belong to the same component. 
Denote $\tilde{\sigma}^{U}_{jk} = {\sigma}^{U}_{jk}/\sqrt{{\sigma}^{U}_{jj} {\sigma}^{U}_{kk}}$. And the general covariance condition says that there exists some $\sigma_0<1$ such that $|\tilde{\sigma}^V_{jk}| = |{\sigma}^V_{jk}|/\sqrt{{\sigma}^V_{jj} {\sigma}^V_{kk}}\le  \sigma_0$ for any $j\ne k$ and $ |\{(j,k): j\ne k, |\tilde{\sigma}^U_{jk} | = |{\sigma}^U_{jk} |/\sqrt{\sigma^U_{jj} \sigma^U_{kk} }>  \sigma_0 \} | \le b_0$ for some constant $b_0$. 
%The connectivity assumption assumes that there exist disjoint partitions of nodes $\cup_{\ell=1}^{\discon}\cC^U_\ell = \cup_{\ell=1}^{\discon}\cC^V_\ell= [d], $ such that  $\sigma^U_{jk}$ ($\sigma^V_{jk}$) equals $0$ when $j,k$ belong to different components $\cC^U_\ell$ ($\cC^V_\ell$), and $\forall \ell \in[\discon]$, $\cC^U_{\ell} \cap \cC^V_{\ell} \ne \emptyset$.	
\end{remark}
\begin{lemma}\label{lem:II1_bound_v2} 
For the term $\mathrm{II}_1 = \EE{  e^{-\beta (\Zmax - |Z_j|)}\cdot  \Indrbr{E_1}\cdot  \Indrbr{t-\epsilon \le \Zmax \le t +\epsilon} }$ with ${E_1}$ defined in \eqref{eq:E1set_def} and $\epsilon = {c}/{\max\{(\log d)^{3/2}, \discon \log d \}}$ for some small enough constant $c>0$,
whenever $t$ satisfies $0\le t \le C_0\sqrt{\log d}$ for some constant $C_0>0$, we have
\begin{eqnarray}\label{eq:boundII1_denom_rate}
\frac{\mathrm{II}_1 }{  \mathbb{P}(\maxnorm{V} > t) } \le  \frac{C'\epsilon \log  d}{\beta \discon}\left( 1 + \frac{b_0}{\sqrt{1 - (s + (1-s)\sigma_0)^2 } } \right). 
 \end{eqnarray}
 for any $s \in (0,1)$, where $\sigma_0 < 1$ and $b_0$ are the constants in the assumption of Theorem \ref{eq:ccb_sparse_unitvar}.
\end{lemma}
\begin{remark}
Recall the definition of $Z = W(s)$. Hence the term $\mathrm{II}_1$  depends on $s$. In Lemma \ref{lem:II1_bound}, we are able to derive a uniform upper bound when assuming the minimal eigenvalue condition as in Theorem \ref{thm:ccb_sparse}. Since Theorem \ref{thm:ccb_sparse_unitvar} does not make assumptions about the minimal eigenvalue condition, we will bound the term $\mathrm{II}_1$ differently and the upper bound depend on $s$, as showed in the following proof. 
\end{remark}

\begin{proof}[Proof of Lemma \ref{lem:II1_bound_v2}]
We basically use the same proof strategy as Lemma but will separately deal with two cases. First recall that
\begin{align}\nonumber
\mathrm{II}_1 =~&  \EE{  e^{-\beta (\Zmax - |Z_j|)}\cdot   \Indrbr{ 
\maxnorm{Z} >  \maxnorm{G}, \maxnorm{Z} > Z_j 
% \exists k, k\ne j, Z_k\notin \cE_G, |Z_k| = \Zmax, |Z_k| > \Gmax
} \cdot  \Indrbr{t-\epsilon \le \Zmax \le t +\epsilon} }.  
\end{align}	
We define $Z^\dag = (Z_k)_{k \in \cE^\dag}$ where 
\begin{equation}\label{eq:cEdeg_def}
\cE^\dag = \{j\}\cup \cE_G \cup \{k \in [d]:  |\tilde{\sigma}^U_{jk}|  \le \sigma_0,
 \max_{m \in \cE_G}\{ |\tilde{\sigma}^U_{mk}| \} \le \sigma_0
  \}.
 \end{equation}
 Under the condition of Theorem \ref{thm:ccb_sparse_unitvar}, we have
 $|[d] \setminus \cE^\dag| \le  |\{(j,k): j\ne k, |\tilde{\sigma}^U_{jk}|\} >  \sigma_0 \} | \le b_0$ for some constant $b_0$. Note we can write $1 = \Indrbr{ \maxnorm{Z^\dag} = \maxnorm{Z}} +\sum_{k \in [d] \setminus \cE^\dag }\Indrbr{ |Z_k|= \maxnorm{Z} } $. Then we have
 \begin{eqnarray}\label{eq:boundII1_denom_v2}
\frac{\mathrm{II}_1 }{  \mathbb{P}(\maxnorm{V} > t) } \le  \frac{\mathrm{II}^\dag_1 }{  \mathbb{P}(\maxnorm{V} > t) }  + b_0 \cdot \max_{k \in [d] \setminus \cE^\dag  }\frac{\mathrm{II}^{(k)}_1 }{  \mathbb{P}(\maxnorm{V} > t) }, 
 \end{eqnarray}
 where $\mathrm{II}^\dag_1$ and $\mathrm{II}^{(k)}_1$ are defined as
 \begin{align}\label{eq:IIdag_k_def}
  \begin{split} 
 \mathrm{II}^\dag_1 :=& \EE{  e^{-\beta (\maxnorm{Z^\dag} - |Z_j|)}\cdot   \Indrbr{ 
\maxnorm{Z^\dag} >  \maxnorm{G}, \maxnorm{Z^\dag} > Z_j } \cdot  \Indrbr{t-\epsilon \le \maxnorm{Z^\dag} \le t +\epsilon} },\\
 \mathrm{II}^{(k)}_1 :=& \EE{  e^{-\beta (|Z_k| - |Z_j|)}\cdot   \Indrbr{ 
|Z_k| >  \maxnorm{G}, |Z_k| > Z_j } \cdot  \Indrbr{t-\epsilon \le |Z_k| \le t +\epsilon} }.
 \end{split}
  \end{align}
Denote the conditional density function of $\maxnorm{Z^\dag} \mid  Z_j = z_j, G = g $ by $ f^\dag_{g,z_j}(u)$. Then we apply exactly the same derivations as in Lemma \ref{lem:II1_bound} (except that $f^\dag_{g,z_j}(u)$ is bounded using Lemma \ref{lem:maximas_pdf_bound_v2} instead of Lemma \ref{lem:maximas_pdf_bound}) and obtain the following bound
\begin{equation}\label{eq:IIdag_bound}
	\frac{\mathrm{II}^\dag_1 }{  \mathbb{P}(\maxnorm{V} > t) }\le \frac{C'\epsilon \log  d}{\beta \discon}.
\end{equation}
Regarding the term $ \mathrm{II}^{(k)}_1$, we follow the same derivations as in the beginning of the proof of Lemma \ref{lem:II1_bound}. Specifically, we have
% except that we do not need to deal with the bound on the conditional density function of Gaussian maxima. Directly, we have
\begin{align}\nonumber
\mathrm{II}^{(k)}_1 =~&  \EE{  e^{-\beta (|Z_k| - |Z_j|)}\cdot   \Indrbr{ 
|Z_k| >  \maxnorm{G}, |Z_k| > Z_j } \cdot  \Indrbr{t-\epsilon \le |Z_k| \le t +\epsilon} } \\ 
=~ &  \EE{ e^{\beta |Z_j|}\cdot \Indrbr{\maxnorm{G} \le t+\epsilon,|Z_j| \le t+\epsilon} \left(
\int_{t-\epsilon}^{t+\epsilon} f_{Z_j}(u) e^{-\beta u}\Indrbr{u>\maxnorm{G}, u >|Z_j|}du
 \right)}, \label{eq:IIk_bound}
\end{align} 
where $f_{Z_j,G}(u)$ denotes the conditional density of $Z_k$ given $Z_j, G$. Recall the construction of $G$ described in the proof of Theorem \ref{thm:ccb_sparse_unitvar}, we have for any $k\ne j, k\notin \cE_{G} = \{m \in [d]:Z_m = G_\ell \text{ for some } \ell \in [\discon]\}$, there exists at most one $m \in \{j\} \cup \cE_{G}$, such that $Z_k$ and $Z_m$ belong to the same component. Denote that random variable by $Z_{m_0}$, then $f_{Z_j,G}(u)$ is just the conditional density function of $Z_k$ given $Z_{m_0}$. Since $Z$ follows a multivariate Gaussian distribution, we can immediately figure out the expression of the conditional density $f_{Z_{m_0}}(u)$ and simply derive a bound
\begin{align}\nonumber
	f_{Z_j,G}(u) = f_{Z_{m_0}}(u) &\le  \frac{1}{\sqrt{2\pi \Varc{Z_k}{Z_{m_0}}}}  \\ \nonumber
	&= \frac{1}{\sqrt{2\pi (\sigma_{kk} - \sigma^2_{k m_0} /\sigma_{m_0 m_0} ) )} } \\ \nonumber
	&= \frac{1}{\sqrt{2\pi \sigma_{kk} }} \cdot \frac{1}{1 - \sigma^2_{k m_0} /(\sigma_{k k} \sigma_{m_0 m_0} ) }\\
	& \le \frac{1}{\sqrt{2\pi a_0 }} \cdot \frac{1}{1 - \sigma^2_{k m_0} /(\sigma_{k k} \sigma_{m_0 m_0} ) }  ,  \label{eq:fZm0_bound}
\end{align} 
where $\sigma_{kk} = \Var{Z_k} , \sigma_{m_0 m_0} = \Var{Z_{m_0}}, \sigma_{k m_0} = \Cov{Z_k}{Z_{m_0}} $ and we use the fact that 
$\sigma_{kk} = \Var{Z_k} = \sigma^U_{kk} \ge a_0$ (under the general variance assumption). Note $Z = \sqrt{s} U + \sqrt{1-s} V$, then we have $ \sigma^2_{k m_0}= (\Cov{Z_k}{Z_{m_0}})^2 = (s {\sigma}_{km_0}^U + (1 - s) {\sigma}_{km_0}^V)^2$ where $m_0 \in \{j\} \cup \cE_{G}$. 
%Recall that $k \in [d]\setminus \cE^\dag $ where $\cE^\dag$ is define in \eqref{eq:cEdeg_def}, we have
Since $|\tilde{\sigma}_{km_0}^U| \le 1 $ by definition and $|\tilde{\sigma}_{km_0}^V| \le \sigma_0 $ under the assumption of Theorem \ref{thm:ccb_sparse_unitvar}, we have
\begin{equation}\label{eq:cond_cov_bound}
 (s {\sigma}_{km_0}^U + (1 - s) \sigma_{km_0}^V)^2/(\sigma_{kk}\sigma_{m_0m_0})  = (s \tilde{\sigma}_{km_0}^U + (1 - s) \tilde{\sigma}_{km_0}^V)^2  \le  (s + (1-s)\sigma_0)^2.
\end{equation}
Now we obtain a upper bound on the conditional density function $f_{Z_j,G}(u)$ based on \eqref{eq:fZm0_bound} and \eqref{eq:cond_cov_bound}. Combining this bound and following the same derivations as in Lemma \ref{lem:II1_bound} to deal with the term in \eqref{eq:IIk_bound}, we establish the upper bound on the term $\mathrm{II}^{(k)}_1 /\mathbb{P}(\maxnorm{V} > t) $ for any $k \in [d] \setminus \cE^\dag$,
\begin{equation}\label{eq:IIk_bound_final}
	\frac{\mathrm{II}^{(k)}_1}{\mathbb{P}(\maxnorm{V} > t)} \le \frac{C'\epsilon \log  d}{\beta \discon}\cdot \frac{1}{\sqrt{1 - (s + (1-s)\sigma_0)^2 } }.
\end{equation}
Combining \eqref{eq:boundII1_denom_v2}, \eqref{eq:IIdag_k_def},\eqref{eq:IIdag_bound} with \eqref{eq:IIk_bound_final}, we derive the bound in \eqref{eq:boundII1_denom_rate}.
%A simple calculation gives the following result on $f_{Z_j,G}(u)$:
%\begin{equation}
%	f_{Z_j,G}(u) =f_{Z_{m_0}}(u) \le 
%\end{equation}
 \end{proof}

\begin{lemma}\label{lem:maximas_pdf_bound_v2} 
Recall that the density function of the conditional distribution of $\maxnorm{Z^\dag} \mid  \{Z_j = z_j, G = g\}$ is denoted by $f^{\dag}_{ g,z_j}(z)$ where $Z^\dag$ is defined in .... Suppose $\epsilon>0$, when $0\le t \le C_0\sqrt{\log d}$ for some constant $C_0>0$ and $|z_j|,\maxnorm{g} \le t+\epsilon$, we have
\begin{equation}\label{eq:f_bound_v2_state}
f^{\dag}_{g, z_j}(z)\le C \sqrt{\log d} ,\quad \forall ~ z \in ( \max\{|z_j|,\maxnorm{g}\}, t+\epsilon]. % <  z\le t+\epsilon.
\end{equation}
where the finite constant $C$ depends on $a_0$ and $\sigma_0 < 1$. %which is the assumed bound on $\sigma^U_{jk}, \sigma^V_{jk}, j\ne k$ in the assumption of Theorem \ref{eq:ccb_sparse_unitvar}.
\end{lemma}
\begin{proof}[Proof of Lemma \ref{lem:maximas_pdf_bound_v2}]
Following exactly the same derivations as in Lemma \ref{lem:maximas_pdf_bound} (up to \eqref{eq:f_bound_v1}), we have
\begin{equation}\label{eq:f_bound_v2}
  f^{\dag}_{g, z_j}(z)\le  \frac{6}{\underline{\sigma_{\cdot j}}} \left( 
\frac{   2 \bar{\rho}_j(1+c_0) +1}{ \underline{\sigma_{\cdot j}}}C_1\sqrt{\log d}
+C_2\sqrt{\log d}
\right), % f_{g, z_j}(z)\le C \sqrt{\log d},
\end{equation}
for any $ z \in ( \max\{|z_j|,\maxnorm{g}\}, t+\epsilon]$, where $C_1,C_2$ are some constants. First, $\bar{\rho}_j$ is defined in \eqref{eq:sigma_rho_def}. Simply, we have 
$$
\bar{\rho}_j \le  \max_{k\ne j} \frac{|\sigma_{jk}|}{\sigma_{jj}} \le \frac{\max_{j}\sigma_{jj}^U}{\min_{j}\sigma_{jj}^U} \le \frac{a_1}{a_0}
$$
under the general variance assumption. Recall the construction of $G$ described in the proof of Theorem \ref{thm:ccb_sparse_unitvar}, we have for any $k\ne j, k\notin \cE_{G} = \{m \in [d]:Z_m = G_\ell \text{ for some } \ell \in [\discon]\}$, there is at most one $m \in \cE_{G}$, such that $Z_k$ and $Z_m$ belong to the same component. Then we have
$$
\sum_{m \in \cE_{G}} \Indrbr{\sigma_{km} \ne 0} \le 1, 
$$
hence $c_0 = 1$ by definition. Also note by the definition of $Z^\dag$ and $\cE^\dag$ in \eqref{eq:cEdeg_def}, for any $k\in \cE^\dag, k\ne  j, k \ne \cE_G$, we have 
  \begin{equation}\label{eq:cor_cstrt}
 \max\{|\tilde{\sigma}^U_{jk}|, |\tilde{\sigma}^V_{jk}|\} \}  \le \sigma_0, \quad 
 \max_{m \in \cE_G}\{ |\tilde{\sigma}^U_{mk}|, |\tilde{\sigma}^V_{mk}|\} \le \sigma_0
  \end{equation} 
under the assumption of Theorem \ref{thm:ccb_sparse_unitvar}.
% says that there exists some $\sigma_0<1$ such that $ |\{(j,k): \max\{|\sigma^U_{jk}|, |\sigma^V_{jk}|\} \ge  \sigma_0 \} | \le b_0$ for some constant $b_0$. 
 We will take advantage of this together with the above property of $G$ to derive a bound on $\underline{\sigma_{\cdot j}}$. Similarly as in \eqref{eq:sigma_rho_def}, we have $\underline{\sigma^2_{\cdot j}} :=  \min_{k\in \cE_X}\Varc{Z_k}{Z_j, G}$ with $\cE_X := \{k \in \cE^\dag: k \ne j, k \notin \cE_G\}$. For each $k \in \cE_X$, we have it can at most belong to the same component as one of $\{j\} \cup \cE_G$,
%$$
%\Indrbr{\sigma_{jk} \ne 0} + \sum_{m \in \cE_{G}} \Indrbr{\sigma_{km} \ne 0} \le 1
%$$
due to the property of $G$. Then we have
\begin{eqnarray}\nonumber
	\Varc{Z_k}{Z_j, G} 
	&\ge & \min\{ \Varc{Z_k}{Z_j}, \min_{m \in \cE_G}\{ \Varc{Z_k}{Z_m}\}\}\\ \nonumber
	&=& {\sigma_{kk}} \cdot  \min\{ 1 -\sigma_{jk}^2/(\sigma_{jj}\sigma_{kk}), \min_{m \in \cE_G}\{ 1 -\sigma^2_{mk}/(\sigma_{mm}\sigma_{kk})\}\} \\ 
	& \ge & a_0 \cdot  \min\{ 1 -\sigma_{jk}^2/(\sigma_{jj}\sigma_{kk}), \min_{m \in \cE_G}\{ 1 -\sigma^2_{mk}/(\sigma_{mm}\sigma_{kk})\}\} \label{eq:condvar_lbd}.
\end{eqnarray}
since $(Z_j, G)$ are all independent and $\sigma_{kk} = \Var{Z_k} = \sigma^U_{kk} \ge a_0$ (under the general variance assumption). Recall the definition of $Z = \sqrt{s}U + \sqrt{1-s}V$, we have
\begin{equation}\label{eq:cor_bound}
	|\sigma_{mk}|/\sqrt{ \sigma_{mm} \sigma_{kk}} = |\Cov{Z_k}{Z_m}|/\sqrt{ \sigma_{mm} \sigma_{kk}}= |s \tilde{\sigma}^U_{mk} + (1-s) \tilde{\sigma}^V_{mk}| \le \sigma_0,\quad \forall~ s \in [0,1],
\end{equation}
when $m \in \{j\} \cup \cE_G$. This is due to \eqref{eq:cor_cstrt}. Then we can derive a bound on ${1}/{\underline{\sigma^2_{\cdot j}}}$, i.e.,
 %where $\cE_X$ is similarly defined in the proof of Lemma \ref{lem:maximas_pdf_bound}.
\begin{eqnarray*} 
\frac{1}{\underline{\sigma^2_{\cdot j}}} &=& \frac{1}{\min_{k\in \cE_X}\Varc{Z_k}{Z_j, G}}  \\
&\le & \frac{1}{a_0 \min_{k \in \cE_X} \min\{ 1 -\sigma_{jk}^2/(\sigma_{jj}\sigma_{kk}), \min_{m \in \cE_G}\{ 1 -\sigma_{mk}^2/(\sigma_{mm}\sigma_{kk})\}\} }  \\
 &\le &  \frac{1}{a_0(1 -\sigma^2_0)},
 \end{eqnarray*}
where the first inequality holds by \eqref{eq:condvar_lbd} and the second equality holds by \eqref{eq:cor_bound}. 
%where $\cE_X = \{k \in \cE^\dag: k \ne j, k \notin \cE_G\}$. Here $\cE_X$ is similarly defined in the proof of Lemma \ref{lem:maximas_pdf_bound}.
% except that the $[d]$ is replaced by $\cE^\dag$, due to the definition of $Z^\dag$.
Combining the above bound with \eqref{eq:f_bound_v2}, we finally establish \eqref{eq:f_bound_v2_state} for some finite constant $C$. 
\end{proof}
\section{Ancillary propositions for FDR control}
\label{sec:gmb_theory}

Throughout this section, we introduce some new notations. For a given mean zero random vector $\bY \in \RR^{d}$ with positive semi-definite covariance matrix $\bSigma^Y:= \EE{\bY \bY^{\top}} \in \RR^{d \times d}$, we denote its Gaussian counterpart by $\bZ \in \RR^{d}$ (i.e., $\EE{\bZ}=\mathbf{0}$ and its covariance matrix $\EE{\bZ \bZ^{\top}}:=\bSigma^Z$ equals $\bSigma^Y= (\sigma_{jk}^Y)_{1\le j,k\le d}$ ). Consider $n$ i.i.d. copies of $\bY$, denoted by $\bY_1,\cdots,\bY_n \in \RR^{d}$. We define the maximum $T_{\bY}$ and $T_{\bZ}$ as below,
%with positive semi-definite covariance matrix $\bSigma^Y$ and their Gaussian counterparts $\bZ_1,\cdots,\bZ_n \in \RR^{d}$ (i.e., $\forall~ i \in [n]$, $\EE{\bZ_i}=0$, $\EE{\bZ_i \bZ^{\top}_i}:=\bSigma^Z = \bSigma^Y= (\sigma_{jk}^Y)_{1\le j,k\le d}$), and define the maximum $T_{\bX}$ and $T_{\bY}$ as below:
\begin{equation}\label{eq:maxima_sum}
T_{\bY}:= \left \Vert \frac{1}{\sqrt{n}}\sum_{i=1}^{n}\bY_{i}\right \Vert_{\infty},
% =\max_{1\le j \le p}\left|\frac{1}{\sqrt{n}}\sum_{i=1}^{n}\bY_{i}(j)\right|, 
\quad T_{\bZ}:= \maxnorm{\bZ},
%\quad T_{\bZ}:=\left\Vert\frac{1}{\sqrt{n}}\sum_{i=1}^{n}\bZ_{i}\right\Vert_{\infty},
\end{equation}
where $q(\alpha;T_{\bY})$ and $q(\alpha;T_{\bZ})$ ($\alpha\in[0,1]$) are the corresponding upper quantile functions. Define the Gaussian multiplier bootstrap counterpart as
\begin{equation}\label{eq:gmb_def}
T_{\bW}:= \left \Vert \frac{1}{\sqrt{n}}\sum_{i=1}^{n}\bY_{i}\xi_{i}\right \Vert_{\infty},
\end{equation}
where $\xi_{i}\stackrel{i.i.d.}{\sim}\cN(0,1)$ and are independent from $\bY_1, \cdots, \bY_n$. Let $q_{\cqt}(\alpha;T_{\bW})$ be the conditional quantile of $T_{\bW}$, then we have $\Pp{\cqt}{ T_{\bW} \ge  q_{\cqt}(\alpha;T_{\bW})} = \alpha$. Note that we use the $\cqt$ subscript to remind ourselves that the probability measure is induced by the multiplier random variables $\{\xi_{i}\}_{i=1}^n$ conditional on $\{\bY_i\}_{i=1}^n$. And we have the covariance matrix of $\frac{1}{\sqrt{n}}\sum_{i=1}^{n}\bY_{i}\xi_{i}$ (conditional on $\{\bY_i\}_{i=1}^n$) equals $\bSigma^W :=\frac{1}{n}\sum_{i=1}^{n}\bY_{i}\bY_{i}^{\top}$. Denote $\maxdiff = \maxnorm{ \bSigma^Z - \bSigma^W} $, which measures the maximal differences between the true covariance matrix $\bSigma^Z$ and the sample version $\bSigma^W$.

\subsection{Cram\'{e}r-type deviation bounds for the Gaussian multiplier bootstrap}
%\label{sec:gmb_theory}
Based on the Cram\'{e}r-type Gaussian comparison bound in Theorem \ref{thm:ccb_max}, the Cram\'{e}r-type approximation bound \cite{arun2018cram}, the maximal inequalities and a careful treatment to the comparison of quantiles, we will establish the Cram\'{e}r-type deviation bounds for the Gaussian multiplier bootstrap (CGMB) in this section.
\begin{proposition}[CGMB]\label{prop:cramer_gmb}
Assuming  the covariance matrix $\bSigma^{Y}$ satisfies $0 <c_1\le \sigma^{Y}_{jj} \le c_2<\infty$, for any $j\in [d]$ and $\bY$ satisfies the tail condition that $\max_{1\le i\le n}\max_{1\le j \le p}\norm{\bY_{ij}}_{\psi_{1}}\le K_{3}$ %($\alpha =1$)
%the condition on the diagonal term of the covariance matrix $\bSigma^{Y}$ that 
for some constants $c_1,c_2,K_{3}$, under the scaling condition $(\log ed)^3 (\log (ed+n))^{56/3}/n = o(1)$, we have the following bound,
\begin{equation}
\label{eq:cramer_mb}
\sup_{\alpha \in [\alpha_{L},1]}
\left|
\frac{\mathbb{P}(T_{\bY} > q_{\cqt}(\alpha;T_{\bW}) )}{\mathbb{P}(T_{\bZ} > q(\alpha;T_{\bZ}) )} - 1 
\right| 
= O\rbr{  \frac{ (\log d)^{11/6}}{n^{1/6}\alpha_L^{1/3}}
+\frac{   (\log d)^{19/6}}{n^{1/6}} },
%O\rbr{\frac{(\log d)^{19/6}}{n^{1/6}} },
% O\left(\eta_1(d,n,\alpha_L) \right)
%, \bluecom{~?~ O\left(\frac{(\log d)^{19/6}}{n^{1/6}}\right)}
%c_{T_{\bY}}(\alpha)\\
%c_{\alpha}(T_{\bY})\\
%q(\alpha;T_{\bY}) \\
%c_{\alpha}(||Y||_{\infty})\\
%q(\alpha;||Y||_{\infty}) \\
%c_{||Y||_{\infty}}(\alpha)\\
%c_{T_0}(\alpha)\\
%    \left|\frac{\mathbb{P}(T^{\bY} > t)}{\mathbb{P}(T^{V} > t)}-1\right|\le [A(\Delta)+1]e^{A(\Delta)B(t)}-1
\end{equation}
where
% $\eta_1(d,n,\alpha_L)=
% \frac{(\log d)^{8/3}}{n^{1/6}}\cdot q({\alpha_L}/{2};T_{\bZ}) \bigvee
% \frac{(\log d)^{{5}/{8}}}{n^{1/8}}\cdot q^2(\alpha_{L};T_{\bZ})$ and 
$\alpha_{L}$ satisfies $q(\alpha_{L};T_{\bZ})=O\left(\sqrt{\log d}\right)$ and $\frac{\log^{11} d}{n\alpha_L} =O(1)$.
\end{proposition}
% \cyancom{Note the rate of $o(1)$ actually has a polynomial decaying in $n$ (up to a polynomial factor of $\log d$). This is crucial in the proof of FDR control.
% where $\alpha_{L}=\Omega(\frac{1}{d})$, remark that here we don't emphasize on the exact form of $\alpha_{L}$ with explicit constants. And we are not claiming this lower bound is tight, which is not our main focus. But $\Omega(\frac{1}{d})$ is enough for establishing our core result about asymptotic FDR control in this paper. }

The proof can be found in Appendix \ref{app:pf:cramer_gmb}. In practice, there are many situations where the relevant statistics come from the maxima of approximated averages. In particular, the test statistics in our node selection problem can not be directly expressed as maxima of scaled averages, but can be approximated by a $T_{\bY}$-like term with the approximation error suitably controlled. Therefore, we also prove an extended version of Proposition \ref{prop:cramer_gmb}. Suppose the statistics of interest and its Gaussian multiplier bootstrap counterpart, denoted by $T$ and $T^{\cB}$ respectively, can be approximated by $T_{\bY}$ (defined in \eqref{eq:maxima_sum}) and $T_{\bW}$ (defined in \eqref{eq:gmb_def}). The quantile functions $q(\alpha;T)$ and $q_{\cqt}(\alpha;T^{\cB})$ are defined correspondingly.
\begin{proposition}[CGMB with approximation]\label{prop:approx_cramer_gmb}
Under the same conditions as in Proposition \ref{prop:cramer_gmb} and the additional assumption about the differences between the maximum statistics:
\begin{eqnarray}\label{eq:approx_cond1}
    \mathbb{P}(|T-T_{\bY}|>\zeta_{1})<\zeta_{2},\\
    \mathbb{P}(\mathbb{P}_{\cqt}(|T^{\cB}-T_{\bW}|>\zeta_{1})>\zeta_{2})<\zeta_{2},
    \label{eq:approx_cond2}
\end{eqnarray}
where $\zeta_{1},\zeta_{2}\ge 0 $ characterize the approximation error and satisfy $\zeta_{1} \log d  = O(1), \zeta_{2} = O(\alpha_L)$, % and satisfy the condition that
% \begin{eqnarray}\label{eq:approx_rate}
%   \zeta_{2} = o\rbr{\alpha_{L}}, ~~~~ \zeta_{1}=o\rbr{\frac{1}{q(\alpha_{L};T_{\bZ}){\log d}}}
% \end{eqnarray}
we have the following Cram\'{e}r-type deviation bound
\begin{equation}
\label{eq:approx_cramer_mb}
\sup_{\alpha \in [\alpha_{L},1]}
\left|
\frac{\mathbb{P}(T > q_{\cqt}(\alpha;T^{\cB}) )}{\mathbb{P}(T_{\bZ} > q(\alpha;T_{\bZ}) )} - 1 
\right| 
=  \eta(d,n,\zeta_1,\zeta_2, \alpha_L),
% O\left(\eta_2(d,n,\alpha_L)
% \right) 
% o(1), \bluecom{~?~ \max\left\{
% O\left(\frac{(\log d)^{19/6}}{n^{1/6}}\right)
% ,
% O\left(\frac{s(\log d)^{3/2}}{n^{1/2}}\right)
% \right\}}
\end{equation}
where $\eta(d,n,\zeta_1,\zeta_2, \alpha_L)= O\rbr{\frac{(\log d)^{19/6}}{n^{1/6}} + \frac{ (\log d)^{11/6}}{n^{1/6}\alpha_L^{1/3}} + \zeta_1 \log d + \frac{\zeta_2}{\alpha_L} }$. 
 %O\rbr{\frac{(\log d)^{19/6}}{n^{1/6}} + \zeta_1 \log d + \frac{\zeta_2}{\alpha_L}}$.
% +\zeta_1 q(\alpha_{L}/2;T_{\bZ})
% +\frac{\zeta_2}{\alpha_L}$.
\end{proposition}

\subsection{Proof of Proposition \ref{prop:cramer_gmb}}
\label{app:pf:cramer_gmb}
Before proving Proposition \ref{prop:cramer_gmb}, we present Lemma \ref{lem:compqt}. It bounds the conditional quantile $q_{\cqt}(\alpha;T_{\bW}) $ in terms of the quantile $q(\alpha;T_{\bZ})$ of the Gaussian maxima $T_{\bZ}$ when the maximal covariance matrix differences are controlled. In the proof of Lemma \ref{lem:compqt}, we apply the 
Cram\'{e}r-type comparison bound \eqref{eq:ccb_max}, which is establised in Theorem \ref{thm:ccb_max}. To simplify the notation, we denote the bound $C_1(\log d)^{5/2}\maxdiff^{1/2}$ in \eqref{eq:ccb_max} by $\pi (\maxdiff)$, where the constant $C_1$ only depends on $\min_{1\le j\le d}\{\sigma^U_{jj},\sigma^V_{jj}\}$, $\max_{1\le j\le d}\{\sigma^U_{jj},\sigma^V_{jj}\}$.
\begin{lemma}\label{lem:compqt}
Suppose $\delta$ satisfies $(\log d)^{5}\delta = O(1)$. On the event $ \{\maxdiff \le \delta\}$, we have 
\begin{eqnarray}\label{eq:compqt1}
{q_{\cqt}(\alpha;T_{\bW}) \ge q\Big(\frac{\alpha}{1-\pi(\delta)} ;T_{\bZ} \Big) },\\  %\ge 1- \PP{\maxdiff > \delta },\\
{q_{\cqt}(\alpha;T_{\bW}) \le q\Big(\frac{\alpha}{1+\pi(\delta)} ;T_{\bZ} \Big) }. %  \ge 1- \PP{\maxdiff > \delta }.
\label{eq:compqt2}
\end{eqnarray}
\end{lemma}  
\begin{proof}[Proof of Lemma \ref{lem:compqt}] \label{pf:lem:compqt} 
On the event $ \{\maxdiff \le \delta\}$, we have $(\log d)^{5}\maxdiff \le (\log d)^{5}\delta= O(1)$, then by applying Theorem \ref{thm:ccb_max} to $\bZ$ and $\bW$, we obtain the following,
\begin{equation}\nonumber
    \sup_{0\le t \le C_0\sqrt{\log d}}\left|\frac{\mathbb{P}_{\cqt}(T_{\bW} > t)}{\mathbb{P}(T_{\bZ} > t)}-1\right|\le \pi(\delta).
    %M_1(\log d)^{3/2} A(\maxdiff)e^{M_1(\log d)^{3/2} A(\maxdiff)}
    %[A(\maxdiff)+1]e^{M_1(\log d)^{3/2} A(\maxdiff)}-1 %:= \pi(\maxdiff)
\end{equation}
%where $\pi (\delta)$ denotes $C_1(\log d)^{5/2}\delta^{1/2}$ with the constant $C_1$ only depends on $\min_{1\le j\le d}\{\sigma^U_{jj},\sigma^V_{jj}\}$, $\max_{1\le j\le d}\{\sigma^U_{jj},\sigma^V_{jj}\}$. 
Therefore we have
$$
\Pp{\cqt}{T_{\bW} \ge q\Big(\frac{\alpha}{1-\pi(\delta)} ;T_{\bZ} \Big)} \ge \PP{T_{\bZ} \ge q\Big(\frac{\alpha}{1-\pi(\delta)} ;T_{\bZ} \Big) }\cdot (1 - \pi(\delta)) = \alpha,
$$
when $t$ satisfies $0\le t \le C_0\sqrt{\log d}$. Then $q_{\cqt}(\alpha;T_{\bW}) \ge q\Big(\frac{\alpha}{1-\pi(\delta)} ;T_{\bZ} \Big)$ immediately follows, i.e., \eqref{eq:compqt1} holds. Similarly, on the event $\{\maxdiff \le \delta\}$, we have
$$
\Pp{\cqt}{T_{\bW} \ge q\Big(\frac{\alpha}{1+\pi(\delta)} ;T_{\bZ} \Big)} \le \PP{T_{\bZ} \ge q\Big(\frac{\alpha}{1+\pi(\delta)} ;T_{\bZ} \Big) }\cdot (1 + \pi(\delta)) = \alpha.
$$
Thus $q_{\cqt}(\alpha;T_{\bW}) \le q\Big(\frac{\alpha}{1+\pi(\delta)} ;T_{\bZ} \Big)$, i.e., \eqref{eq:compqt2} holds.
\end{proof}

\begin{proof}[Proof of Proposition \ref{prop:cramer_gmb}]
By the triangle inequality, we have
\begin{equation} \label{eq:I_II_def}
 \left|
\frac{\mathbb{P}(T_{\bY} > q_{\cqt}(\alpha;T_{\bW}) )}{\mathbb{P}(T_{\bZ} > q(\alpha;T_{\bZ}) )} - 1 
\right| \le 
    \underbrace{\left|
    \frac{\mathbb{P}(T_{\bY} > q(\alpha;T_{\bZ}) )}{\mathbb{P}(T_{\bZ} > q(\alpha;T_{\bZ}) )} - 1 
    \right|}_{\mathrm{I}} +  
    \underbrace{\frac{\left|
    {\mathbb{P}(T_{\bY} > q_{\cqt}(\alpha;T_{\bW}) )} - \mathbb{P}(T_{\bY} > q(\alpha;T_{\bZ}) )
    \right|}{\mathbb{P}(T_{\bZ} > q(\alpha;T_{\bZ}) )}}_{\mathrm{II}}.
 \end{equation}
Regarding the first term $\mathrm{I}$, we will directly apply Corollary 5.1 in \cite{arun2018cram}. Specifically, we verify %the scaling condition on $(n,d)$, 
the tail assumption on $\bY$ and the condition on the quantile that ${q(\alpha;T_{\bZ})} \le q(\alpha_{L};T_{\bZ})=O\left(\sqrt{\log d}\right)$ when $\alpha \in [\alpha_L, 1]$. Then we obtain the following bound 
\begin{equation}\label{eq:I_rate} 
  \mathrm{I} = \left|\frac
    {\mathbb{P}(T_{\bY} > q(\alpha;T_{\bZ}) )}{\PP{T_{\bZ} >  q(\alpha;T_{\bZ}) } } - 1 
    \right| = O\rbr{\frac{(\log d)^{19/6}}{n^{1/6}} }.
\end{equation}
Regarding the second term $\mathrm{II}$, we write it as
\begin{eqnarray}\nonumber
  \mathrm{II}
   & = & \frac{1}{\alpha}\left|
    {\mathbb{P}(T_{\bY} > q_{\cqt}(\alpha;T_{\bW}) )} - \mathbb{P}(T_{\bY} > q(\alpha;T_{\bZ}) )
    \right|  \\  \nonumber
   &\le  & \frac{1}{\alpha}
    \mathbb{P}(  \{T_{\bY} > q_{\cqt}(\alpha;T_{\bW})\} 
    \ominus
   \{T_{\bY} > q(\alpha;T_{\bZ})\}) \\ \nonumber
    &=  & \frac{1}{\alpha} \Big(
    \mathbb{P}(  T_{\bY} > q_{\cqt}(\alpha;T_{\bW}),
   T_{\bY} \le  q(\alpha;T_{\bZ})) + \mathbb{P}(  T_{\bY} \le q_{\cqt}(\alpha;T_{\bW}),
   T_{\bY} >  q(\alpha;T_{\bZ}))
    \Big)  \\  \nonumber
  &\le  & \frac{1}{\alpha} 
    \mathbb{P}(  T_{\bY} > q_{\cqt}(\alpha;T_{\bW}),
   T_{\bY} \le  q(\alpha;T_{\bZ}), \maxdiff \le \delta) 
     \\ \nonumber
   &~~& + ~  \frac{1}{\alpha}\mathbb{P}(  T_{\bY} \le q_{\cqt}(\alpha;T_{\bW}),
   T_{\bY} >  q(\alpha;T_{\bZ}),  \maxdiff \le \delta)+   \frac{2 \PP{\maxdiff > \delta } }{\alpha },
\end{eqnarray}
where the first inequality holds by the definition of the symmetric difference; recall the symmetric difference between $A$ and $B$ is defined as $A \ominus B = (A\setminus B) \cup (B\setminus A)$). Remark that we will give the explicit choice of $\delta$ later in the proof. Now we apply Lemma \ref{lem:compqt} (whose condition will be verified in \eqref{eq:check_delta}) and further bound $\mathrm{II}$ as,
\begin{eqnarray}\nonumber
   \mathrm{II} &\le  & \frac{1}{\alpha} \Big(
    \mathbb{P}\big(  T_{\bY} \ge  q\big(\frac{\alpha}{1-\pi(\delta)} ;T_{\bZ} \big),
   T_{\bY} \le  q(\alpha;T_{\bZ})\big) \\ \nonumber
  &~&\quad +~ \mathbb{P}(  T_{\bY} \le q\big(\frac{\alpha}{1+\pi(\delta)} ;T_{\bZ} \big),
   T_{\bY} >  q(\alpha;T_{\bZ}))\Big) +   \frac{2 \PP{\maxdiff > \delta } }{\alpha } \\            
    \label{eq:cond_key_step}
    &= & \frac{1}{\alpha}~{
    \PP{ q\big(\frac{\alpha}{1-\pi(\delta)} ;T_{\bZ} \big) \le T_{\bY}  \le q\big(\frac{\alpha}{1+\pi(\delta)} ;T_{\bZ} \big)}
    } + \frac{2 \PP{\maxdiff > \delta } }{\alpha }
    \\ \nonumber
        &\le & \frac{1}{\alpha}~{
    \PP{ q\big(\frac{\alpha}{1-\pi(\delta)} ;T_{\bZ} \big) \le T_{\bZ}  \le q\big(\frac{\alpha}{1+\pi(\delta)} ;T_{\bZ} \big)}
    } + \frac{2 \PP{\maxdiff > \delta } }{\alpha } + \mathrm{III}
    \\ \label{eq:II_rate}
    & = & \frac{2\pi(\delta) }{1- \pi^2(\delta)}+ \frac{2 \PP{\maxdiff > \delta } }{\alpha } + \mathrm{III},
%    \\
%   & \le & \left|
%    \frac{\PP{ T_{\bY} > q\big(\frac{\alpha}{1+\pi(\delta)} ;T_{\bZ} \big) } }{\mathbb{P}(T_{\bY} > q(\alpha;T_{\bZ}) )} - 1 
%    \right| + \left|
%    \frac{\PP{T_{\bY} > q\big(\frac{\alpha}{1-\pi(\delta)} ;T_{\bZ} \big) }}{\mathbb{P}(T_{\bY} > q(\alpha;T_{\bZ}) )} - 1 
%    \right|  + \frac{2 \PP{\maxdiff > \delta } }{\alpha } \\ \label{eq:I_123_def}
%    &:=& \mathrm{I}_1 + \mathrm{I}_2 + \mathrm{I}_3.
\end{eqnarray} 
where the term $\mathrm{III}$ in the second inequality is defined as,
\begin{eqnarray*}
\mathrm{III}:=\frac{1}{\alpha}\left|
    \PP{ q\big(\frac{\alpha}{1-\pi(\delta)} ;T_{\bZ} \big) \le T_{\bY}  \le q\big(\frac{\alpha}{1+\pi(\delta)} ;T_{\bZ} \big)} - 
    \PP{ q\big(\frac{\alpha}{1-\pi(\delta)} ;T_{\bZ} \big) \le T_{\bZ}  \le q\big(\frac{\alpha}{1+\pi(\delta)} ;T_{\bZ} \big)}
\right|.  
\end{eqnarray*}
%and derive the following bound on $\mathrm{II}$,
Below we further rewrite $\mathrm{III}$ as
\begin{eqnarray}\nonumber
\mathrm{III} =  \frac{1}{\alpha}\abr{ \frac{\alpha}{1-\pi(\delta)} \cdot \mathrm{III}_1 
 -  \frac{\alpha}{1+\pi(\delta)} \cdot \mathrm{III}_2 
},  
\end{eqnarray}
with $\mathrm{III}_1, \mathrm{III}_2$ defined as
\begin{eqnarray*}
\mathrm{III}_1 &=&  \frac{\PP{ T_{\bY}  > q\big(\frac{\alpha}{1-\pi(\delta)} ;T_{\bZ} \big)} - \PP{ T_{\bZ}  > q\big(\frac{\alpha}{1-\pi(\delta)} ;T_{\bZ} \big)}}{\PP{ T_{\bZ}  > q\big(\frac{\alpha}{1-\pi(\delta)} ;T_{\bZ} \big)} },\\
\mathrm{III}_2 &=&  \frac{\PP{ T_{\bY}  > q\big(\frac{\alpha}{1+\pi(\delta)} ;T_{\bZ} \big)} - \PP{ T_{\bZ}  > q\big(\frac{\alpha}{1+\pi(\delta)} ;T_{\bZ} \big)}}{\PP{ T_{\bZ}  > q\big(\frac{\alpha}{1+\pi(\delta)} ;T_{\bZ} \big)} }.
\end{eqnarray*}
Thus by applying Corollary 5.1 of \cite{arun2018cram} to $\mathrm{III}_1, \mathrm{III}_2$ similarly as in \eqref{eq:I_rate}, we have the following bound on $\mathrm{III}$,
\begin{eqnarray}\label{eq:III_rate}
\mathrm{III} =  O\rbr{\frac{(\log d)^{19/6}}{n^{1/6}} }.
\end{eqnarray}
%Combining \eqref{eq:I_II_def}, \eqref{eq:I_rate}, \eqref{eq:II_rate} and \eqref{eq:III_rate} yields the following
Combining \eqref{eq:II_rate} and \eqref{eq:III_rate} yields the following bound,
\begin{eqnarray}\nonumber
%\left|
%\frac{\mathbb{P}(T_{\bY} > q_{\cqt}(\alpha;T_{\bW}) )}{\mathbb{P}(T_{\bZ} > q(\alpha;T_{\bZ}) )} - 1  
%\right| 
 \mathrm{II}
 &\le  & \frac{1}{\alpha}
    \mathbb{P}(  \{T_{\bY} > q_{\cqt}(\alpha;T_{\bW})\} 
    \ominus
   \{T_{\bY} > q(\alpha;T_{\bZ})\}) \\
\nonumber&\le &   {\frac{C(\log d)^{19/6}}{n^{1/6}} }  + C_0' \pi(\delta) +  {\frac{C''~\PP{\maxdiff > \delta }}{\alpha} }  \\ \nonumber
&\le &  {\frac{C(\log d)^{19/6}}{n^{1/6}} }  + C'(\log d)^{5/2} \delta^{1/2} +  {\frac{C''~\EE{\maxdiff}}{\delta \alpha} }  \\
& =&
O\rbr{  \rbr{\frac{\EE{\maxdiff} \log^5 d}{\alpha}}^{1/3}
+\frac{   (\log d)^{19/6}}{n^{1/6}} },
\label{eq:prop1_rate}
\end{eqnarray}
where the second inequality holds due to the definition of $\pi(\delta)$ and Markov's inequality, the last line holds by choosing $\delta$ to be $ (\EE{\maxdiff})^{2/3}/(\alpha^{1/3}(\log d)^{5/3}) $. We will bound the term $\EE{\maxdiff}$ using Lemma C.1 in \cite{chernozhukov2013gaussian}. Specifically, under the stated tail assumption on $\bY$, the condition (E.1) of Lemma C.1 in \cite{chernozhukov2013gaussian} is satisfied; see Comment 2.2 in \cite{chernozhukov2013gaussian}. Thus we have 
\begin{equation}\label{eq:EE_maxdiff_rate} 
\EE{\maxdiff} \le \sqrt{\frac{B_{n}^{2}\log d}{n}} \vee \frac{B_{n}^{2}(\log (d n))^{2}(\log d)}{n},
\end{equation}
where $B_n$ equals some constant $C$ which does not depend on $n$. As promised previously, we 
verify the assumption of Lemma \ref{lem:compqt} for our choice of $\delta$. Specifically, for $\delta= (\EE{\maxdiff})^{2/3}/(\alpha^{1/3}(\log d)^{5/3}) $, we have $(\log d)^5  \delta $ satisfies the following 
\begin{equation}\label{eq:check_delta}
(\log d)^5  \delta  \le \frac{(\log d)^{5}(\EE{\maxdiff})^{2/3}}{\alpha_L^{1/3}(\log d)^{5/3}} = \rbr{\frac{\log^{11} d}{n\alpha_L}}^{1/3} = O(1),
\end{equation}
under the stated condition on $\alpha_L$. 
%Therefore, the assumption in Lemma \ref{lem:compqt} is verified.
Finally, when $\alpha \in [\alpha_{L},1]$, we combine \eqref{eq:I_II_def}, \eqref{eq:I_rate}, \eqref{eq:prop1_rate} with \eqref{eq:EE_maxdiff_rate}, then establish \eqref{eq:cramer_mb}, i.e.,
\begin{equation}\nonumber
  \sup_{\alpha \in [\alpha_{L},1]}
\left|
\frac{\mathbb{P}(T_{\bY} > q_{\cqt}(\alpha;T_{\bW}) )}{\mathbb{P}(T_{\bZ} > q(\alpha;T_{\bZ}) )} - 1 
\right|  = O\rbr{  \frac{ (\log d)^{11/6}}{n^{1/6}\alpha_L^{1/3}}
+\frac{   (\log d)^{19/6}}{n^{1/6}} }.
\end{equation}
\end{proof}

\subsection{Proof of Proposition \ref{prop:approx_cramer_gmb}}
\label{app:pf:approx_cramer_gmb}
Before proving Proposition \ref{prop:approx_cramer_gmb}, we need to present a simple lemma. It translates the approximation error $\zeta_1, \zeta_2$ into the bounds on the quantiles. And its proof is quite straightforward thus omitted. 
%bounds the conditional quantile $q_{\cqt}(\alpha;T_{\bW}) $ in terms of the quantile $q(\alpha;T_{\bZ})$ of the Gaussian maxima $T_{\bZ}$ when the maximal covariance matrix differences are controlled.

%Firstly we present a lemma on the comparison of quantiles.
\begin{lemma}\label{lem:approx_quantile}
Under the assumption in \eqref{eq:approx_cond2}, we have, for $\alpha\in (0,1)$,
\begin{eqnarray} \nonumber
    \mathbb{P}(q_{\cqt}(\alpha;T^{\cB})\le q_{\cqt}(\alpha+\zeta_{2};T_{\bW})+\zeta_{1})\ge 1-\zeta_{2},\\ \nonumber
    \mathbb{P}(q_{\cqt}(\alpha;T^{\cB})\ge q_{\cqt}(\alpha-\zeta_{2};T_{\bW})-\zeta_{1})\ge 1-\zeta_{2}.
\end{eqnarray}
\end{lemma}
\begin{proof}[Proof of Proposition \ref{prop:approx_cramer_gmb}]
By the triangle inequality, we have
\begin{equation}\label{eq:prop2_I_II_def}
  \left|
\frac{\mathbb{P}(T> q_{\cqt}(\alpha;T^{\cB}) )}{\mathbb{P}(T_{\bZ} > q(\alpha;T_{\bZ}) )} - 1 
 \right| \le \underbrace{ \left|
\frac{\mathbb{P}(T_{\bY}>  q(\alpha;T_{\bZ}) )}{\mathbb{P}(T_{\bZ} > q(\alpha;T_{\bZ}) )} - 1 
 \right|}_{\mathrm{I}} + \underbrace{
\frac{|\mathbb{P}(T> q_{\cqt}(\alpha;T^{\cB}) ) - \mathbb{P}(T_{\bY}>  q(\alpha;T_{\bZ}) )|}{\mathbb{P}(T_{\bZ} > q(\alpha;T_{\bZ}) )}  
}_{\mathrm{II}}.
\end{equation}
Note that \eqref{eq:I_rate} in the proof of Proposition \ref{prop:cramer_gmb} immediately gives the bound on $\mathrm{I}$, i.e.,
\begin{equation}\label{eq:prop2_I_rate}
  \mathrm{I} = O\Big(\frac{   (\log d)^{19/6}}{n^{1/6}}\Big).
\end{equation}
Regarding the term $\mathrm{II}$, we have
\begin{eqnarray}\nonumber
  \mathrm{II}
   & = & \frac{1}{\alpha}\left|
    \mathbb{P}(T> q_{\cqt}(\alpha;T^{\cB}) ) - \mathbb{P}(T_{\bY}>  q(\alpha;T_{\bZ}) )
    \right|  \\  \nonumber
  &\le & \frac{1}{\alpha}\left|
    \mathbb{P}(  \{ T> q_{\cqt}(\alpha;T^{\cB}) \} 
    \ominus
   \{T_{\bY} > q(\alpha;T_{\bZ})\})
    \right|  \\ 
  &=& \frac{1}{\alpha}
    \mathbb{P}( T> q_{\cqt}(\alpha;T^{\cB}),
   T_{\bY} \le  q(\alpha;T_{\bZ})) +  \frac{1}{\alpha} \mathbb{P}( T \le q_{\cqt}(\alpha;T^{\cB}),
   T_{\bY} >  q(\alpha;T_{\bZ})) \label{eq:II_twoterms}.
\end{eqnarray} 
To bound the two terms in \eqref{eq:II_twoterms}, first notice that on the event $|T-T_{\bY}|>\zeta_{1}$, we have 
$$ \{ T> q_{\cqt}(\alpha;T^{\cB}),
   T_{\bY} \le  q(\alpha;T_{\bZ}) \} \subset \{T_{\bY}> q_{\cqt}(\alpha;T^{\cB}) - \zeta_1,
   T_{\bY} \le  q(\alpha;T_{\bZ})\}.
$$ 
Then under the assumption in \eqref{eq:approx_cond1}, i.e., $ \mathbb{P}(|T-T_{\bY}|>\zeta_{1})<\zeta_{2}$, we obtain 
$$ \mathbb{P}( T> q_{\cqt}(\alpha;T^{\cB}),
   T_{\bY} \le  q(\alpha;T_{\bZ})) \le  \mathbb{P}( T_{\bY}> q_{\cqt}(\alpha;T^{\cB}) - \zeta_1,
   T_{\bY} \le  q(\alpha;T_{\bZ})) + \zeta_2.
 $$ 
 Applying such strategies to the second term in 
    $\eqref{eq:II_twoterms}$ similarly, we get the following,
  \begin{eqnarray}\label{eq:prop2_II_bound}
     \mathrm{II} &\le & \mathrm{II}_1 + \mathrm{II}_2   +  \frac{2 \zeta_2 }{\alpha },\quad \text{where}\\ \nonumber
     \mathrm{II}_1 &:=&   \frac{1}{\alpha}
    \mathbb{P}( T_{\bY}> q_{\cqt}(\alpha;T^{\cB}) - \zeta_1,
   T_{\bY} \le  q(\alpha;T_{\bZ})), \\ \nonumber
  \mathrm{II}_2 &:=&   \frac{1}{\alpha} \mathbb{P}( T_{\bY} \le q_{\cqt}(\alpha;T^{\cB}) + \zeta_2,
   T_{\bY} >  q(\alpha;T_{\bZ})). 
\end{eqnarray}  
Under the assumption \eqref{eq:approx_cond2}, by Lemma \ref{lem:approx_quantile}, we have   
\begin{eqnarray} \nonumber
%\label{eq:approx_quantile}
    \mathbb{P}(q_{\cqt}(\alpha;T^{\cB})\le q_{\cqt}(\alpha+\zeta_{2};T_{\bW})+\zeta_{1}) &\ge & 1-\zeta_{2},\\ \nonumber
    \mathbb{P}(q_{\cqt}(\alpha;T^{\cB})\ge q_{\cqt}(\alpha-\zeta_{2};T_{\bW})-\zeta_{1})
    &\ge & 1-\zeta_{2}.
\end{eqnarray}    
Hence we can bound $\mathrm{II}_1, \mathrm{II}_2$ as below,
\begin{eqnarray*}
   \mathrm{II}_1 & \le  & \frac{1}{\alpha}  \PP{ T_{\bY}> q_{\cqt}(\alpha - \zeta_{2};T_{\bW})-2\zeta_{1},  T_{\bY} \le  q(\alpha;T_{\bZ}) } + \frac{\zeta_2}{\alpha}, \\ 
   \mathrm{II}_2 & \le & \frac{1}{\alpha}  \PP{ T_{\bY} \le q_{\cqt}(\alpha + \zeta_{2};T_{\bW})+ 2\zeta_{1},  T_{\bY} >  q(\alpha;T_{\bZ}) }  + \frac{\zeta_2}{\alpha}.
\end{eqnarray*} 
Now we will use the strategy of deriving \eqref{eq:cond_key_step} in the proof of Proposition \ref{prop:cramer_gmb}, i.e., apply Lemma \ref{lem:compqt}, then we have,
\begin{eqnarray*}
   \mathrm{II}_1 & \le  & \frac{1}{\alpha}  \PP{ T_{\bY}> q\Big( \frac{\alpha - \zeta_{2}}{1 -\pi(\delta)} ;T_{\bZ}\Big)-2\zeta_{1},  T_{\bY} \le  q(\alpha;T_{\bZ}) } + \frac{\PP{\maxdiff > \delta }}{\alpha}  +  \frac{\zeta_2}{\alpha}, 
   \\ 
 \mathrm{II}_2 & \le & \frac{1}{\alpha}  \PP{ T_{\bY} \le q\Big( \frac{\alpha + \zeta_{2}}{1 + \pi(\delta)} ;T_{\bZ} \Big)+ 2\zeta_{1},  T_{\bY} >  q(\alpha;T_{\bZ}) } + \frac{\PP{\maxdiff > \delta }}{\alpha} + \frac{\zeta_2}{\alpha}.
\end{eqnarray*} 
Combining the above two inequalities with \eqref{eq:prop2_II_bound}, we have
\begin{equation}\label{eq:prop2_II_rate}
   \mathrm{II} \le \mathrm{III} +  \frac{2 \PP{\maxdiff > \delta }}{\alpha} + 
   \frac{4 \zeta_2}{\alpha},
\end{equation}
where $\mathrm{III}$ is defined as below,
\begin{align*}
   \mathrm{III} 
   :=~& \frac{1}{\alpha}\abr{ \PP{ T_{\bY}> q\Big( \frac{\alpha - \zeta_{2}}{1 -\pi(\delta)} ;T_{\bZ}\Big)-2\zeta_{1} } - \PP{ T_{\bY}> q\Big( \frac{\alpha + \zeta_{2}}{1 + \pi(\delta)} ;T_{\bZ}\Big) + 2\zeta_{1} } } \\ 
   =~& \frac{1}{\alpha}\Big |  \PP{ T_{\bY}> q\Big( \frac{\alpha - \zeta_{2}}{1 -\pi(\delta)} ;T_{\bZ}\Big)-2\zeta_{1} } - \PP{ T_{\bZ}> q\Big( \frac{\alpha - \zeta_{2}}{1 -\pi(\delta)} ;T_{\bZ}\Big)-2\zeta_{1} } \\
   ~& - \PP{ T_{\bY}> q\Big( \frac{\alpha + \zeta_{2}}{1 + \pi(\delta)} ;T_{\bZ}\Big) + 2\zeta_{1} }  + \PP{ T_{\bZ}> q\Big( \frac{\alpha + \zeta_{2}}{1 + \pi(\delta)} ;T_{\bZ}\Big) + 2\zeta_{1} } \\
   ~& + \PP{ T_{\bZ}> q\Big( \frac{\alpha - \zeta_{2}}{1 -\pi(\delta)} ;T_{\bZ}\Big)-2\zeta_{1} } -  \PP{ T_{\bZ}> q\Big( \frac{\alpha + \zeta_{2}}{1 + \pi(\delta)} ;T_{\bZ}\Big) + 2\zeta_{1} }
   \Big|   \\
   \le ~& \mathrm{III}_1 + \mathrm{III}_2 + \mathrm{III}_3.  
\end{align*}
The last line comes from the triangle inequality, with $\mathrm{III}_1, \mathrm{III}_2, \mathrm{III}_3 $ defined as, 
\begin{eqnarray*}
   \mathrm{III}_1 &:=& \frac{1}{\alpha}\Big|
   \PPP \Big( T_{\bY}> q\Big( \frac{\alpha + \zeta_{2}}{1 +\pi(\delta)} ;T_{\bZ}\Big) + 2\zeta_{1} \Big) - \PPP \Big( T_{\bZ}> q\Big( \frac{\alpha + \zeta_{2}}{1 +\pi(\delta)} ;T_{\bZ}\Big) + 2\zeta_{1} \Big)
    \Big|,\\
   \mathrm{III}_2 &:=& \frac{1}{\alpha}\Big|
   \PPP \Big( T_{\bY}> q\Big( \frac{\alpha - \zeta_{2}}{1 - \pi(\delta)} ;T_{\bZ}\Big) -2\zeta_{1} \Big) -  \PPP \Big( T_{\bZ}> q\Big( \frac{\alpha - \zeta_{2}}{1 -\pi(\delta)} ;T_{\bZ}\Big) - 2\zeta_{1} \Big) 
    \Big|,\\ 
     \mathrm{III}_3 &:=& \frac{1}{\alpha}\Big|   \PPP \Big( T_{\bZ}> q\Big( \frac{\alpha - \zeta_{2}}{1 -\pi(\delta)} ;T_{\bZ}\Big)-2\zeta_{1} \Big) -   \PPP \Big( T_{\bZ}> q\Big( \frac{\alpha + \zeta_{2}}{1 + \pi(\delta)} ;T_{\bZ}\Big) + 2\zeta_{1} \Big) \Big|.
\end{eqnarray*}
We first bound $  \mathrm{III}_3$ by the triangle inequality,
\begin{eqnarray*}
    \mathrm{III}_3 
    &=& \frac{1}{\alpha}\Big|   \PPP \Big( T_{\bZ}> q\Big( \frac{\alpha - \zeta_{2}}{1 -\pi(\delta)} ;T_{\bZ}\Big)-2\zeta_{1} \Big) -   \PPP \Big( T_{\bZ}> q\Big( \frac{\alpha + \zeta_{2}}{1 + \pi(\delta)} ;T_{\bZ}\Big) + 2\zeta_{1} \Big) \Big|\\
  &\le &  \underbrace{\frac{1}{\alpha} \Big|   \PPP \Big( T_{\bZ}> q\Big( \frac{\alpha + \zeta_{2}}{1 + \pi(\delta)} ;T_{\bZ}\Big) + 2\zeta_{1} \Big)  - \frac{\alpha + \zeta_{2}}{1 + \pi(\delta)} 
  \Big|}_{ \mathrm{III}_{31}} \\
  &+&   \underbrace{\frac{1}{\alpha} \Big|   \PPP \Big( T_{\bZ}> q\Big( \frac{\alpha - \zeta_{2}}{1 - \pi(\delta)} ;T_{\bZ}\Big)  - 2\zeta_{1} \Big)  
  - \frac{\alpha - \zeta_{2}}{1 - \pi(\delta)} \Big|}_{\mathrm{III}_{32}} 
   + \underbrace{ \frac{1}{\alpha} \Big|  \frac{\alpha -\zeta_{2}}{1  - \pi(\delta)} - \frac{\alpha + \zeta_{2}}{1 +\pi(\delta)} \Big| }_{ \mathrm{III}_{33}}.
   \end{eqnarray*} 
Note that $\mathrm{III}_{31}$ can be rewritten as 
\begin{align} \label{eq:use_antibound1}
  \mathrm{III}_{31} = ~& \frac{\alpha  +  \zeta_{2}}{\alpha(1  + \pi(\delta))}\cdot  \frac{\Big| \PPP \Big( T_{\bZ}> q\big( \frac{\alpha + \zeta_{2}}{1 +\pi(\delta)} ;T_{\bZ}\big) + 2\zeta_{1} \Big)  - \PP{ T_{\bZ}> q\big( \frac{\alpha + \zeta_{2}}{1 + \pi(\delta)} ;T_{\bZ}\big) }   \Big|}{  \PP{ T_{\bZ}>  q\big( \frac{\alpha + \zeta_{2}}{1 + \pi(\delta)} ;T_{\bZ}\big) }} \\
  \le ~& \frac{\alpha  +  \zeta_{2}}{\alpha(1  + \pi(\delta))}\cdot K_{4}\zeta_{1} \Big( q\big( \frac{\alpha + \zeta_{2}}{1 + \pi(\delta)} ;T_{\bZ}\big) +\zeta_{1} \Big) \le  C \zeta_{1} \log d, \label{eq:use_antibound2}
\end{align}
where the first inequality holds by applying a non-uniform anti-concentration bound. Specifically, we apply the part $3$ of Theorem 2.1 in \cite{arun2018cram} (with $r-\epsilon = q\big( \frac{\alpha + \zeta_{2}}{1 + \pi(\delta)} ;T_{\bZ}\big), r + \epsilon = q\big( \frac{\alpha + \zeta_{2}}{1 +\pi(\delta)} ;T_{\bZ}\big) + 2\zeta_{1} $  ) to the Gaussian random vector $\bZ$. Remark that the term $K_3$ is a constant only depending on $\min_{1\le j\le d}\{\sigma^Y_{jj}\},\max_{1\le j\le d}\{\sigma^Y_{jj}\}$ and the median of Gaussian maxima (up to 2-nd power, hence at most of rate $O({\log d})$). As for the second inequality, under the assumption $\zeta_2 = O(\alpha_L)$, we have $\frac{\zeta_2}{\alpha} \le \frac{\zeta_2}{\alpha_L} = O(1)$ when $\alpha \in [\alpha_L, 1]$; we also use the fact that $\zeta_1 =  O(\sqrt{\log d})$ (which holds under the stated assumption), and $q\big( \frac{\alpha + \zeta_{2}}{1 + \pi(\delta)} ;T_{\bZ}\big) = O(\sqrt{\log d})$ (which will be verified later in \eqref{eq:quantile_shift_bound}). Thus we show $\mathrm{III}_{31} = O( \zeta_{1} \log d)$. Similarly, $\mathrm{III}_{32}$ can be bounded as $O( \zeta_{1} \log d)$. As for $\mathrm{III}_{33}$, we have
$$
\mathrm{III}_{33} = \frac{1}{\alpha} \Big|  \frac{\alpha -\zeta_{2}}{1  - \pi(\delta)} - \frac{\alpha + \zeta_{2}}{1 +\pi(\delta)} \Big| \le \frac{2 \pi(\delta) }{1 - \pi^2(\delta)} +  \frac{2 \zeta_2 }{\alpha (1 - \pi^2(\delta))}.
$$ Thus by combining the bounds on $\mathrm{III}_{31}, \mathrm{III}_{32}, \mathrm{III}_{33}$, we obtain
\begin{equation}\label{eq:prop2_III3_bound}
\mathrm{III}_{3} \le \mathrm{III}_{31} + \mathrm{III}_{32} + \mathrm{III}_{33} \le C' \zeta_{1} \log d + 
\frac{2 \pi(\delta) }{1 - \pi^2(\delta)} +  \frac{2 \zeta_2 }{\alpha (1 - \pi^2(\delta))}. 
\end{equation}

Regarding the term $\mathrm{III}_1$, we first consider the following,
\begin{eqnarray*}
 \mathrm{III}_{11}
  &:= &  \frac{1}{\alpha}\PPP \Big( T_{\bZ}> q\big( \frac{\alpha + \zeta_{2}}{1 +\pi(\delta)} ;T_{\bZ}\big) + 2\zeta_{1} \Big)
  \\ 
 &\le & \frac{1}{\alpha} \PPP \Big( T_{\bZ}> q\big( \frac{\alpha + \zeta_{2}}{1 +\pi(\delta)} ;T_{\bZ}\big) \Big)\cdot \rbr{1 + K_{4}\zeta_{1} \Big( q\big( \frac{\alpha + \zeta_{2}}{1 + \pi(\delta)} ;T_{\bZ}\big) +\zeta_{1} \Big) }\\
   &= &  \frac{\alpha  +  \zeta_{2}}{\alpha(1  + \pi(\delta))} \cdot \rbr{1 + K_{4}\zeta_{1} \Big( q\big( \frac{\alpha + \zeta_{2}}{1 + \pi(\delta)} ;T_{\bZ}\big) +\zeta_{1} \Big) } \\
   &\le&  C'' + C \zeta_{1} \log d = O(1),
   \end{eqnarray*}
   where the first inequality holds due to the derivations from \eqref{eq:use_antibound1} to \eqref{eq:use_antibound2}, the second inequality holds due to the last inequality in \eqref{eq:use_antibound2} and the stated assumption $\zeta_2 = O(\alpha_L)$. 
Then we bound $\mathrm{III}_1$ in terms of $\mathrm{III}_{11}$ and write
\begin{eqnarray*}
   \mathrm{III}_1 
   &=& \frac{1}{\alpha}\Big|
   \PPP \Big( T_{\bY}> q\Big( \frac{\alpha + \zeta_{2}}{1 +\pi(\delta)} ;T_{\bZ}\Big) + 2\zeta_{1} \Big) - \PPP \Big( T_{\bZ}> q\Big( \frac{\alpha + \zeta_{2}}{1 +\pi(\delta)} ;T_{\bZ}\Big) + 2\zeta_{1} \Big)
    \Big| \\
   &=& \mathrm{III}_{11} \cdot \Bigg| \frac{
   \PPP \Big( T_{\bY}> q\Big( \frac{\alpha + \zeta_{2}}{1 +\pi(\delta)} ;T_{\bZ}\Big) + 2\zeta_{1} \Big) - \PPP \Big( T_{\bZ}> q\Big( \frac{\alpha + \zeta_{2}}{1 +\pi(\delta)} ;T_{\bZ}\Big) + 2\zeta_{1} \Big)}{\PPP \Big( T_{\bZ}> q\Big( \frac{\alpha + \zeta_{2}}{1 +\pi(\delta)} ;T_{\bZ}\Big) + 2\zeta_{1} \Big) }
    \Bigg| \\
    &\le &  \mathrm{III}_{11} \cdot \frac{   (\log d)^{19/6}}{n^{1/6}} = O\Big(\frac{   (\log d)^{19/6}}{n^{1/6}}\Big),
\end{eqnarray*}
where the inequality holds by applying Corollary 5.1 in \cite{arun2018cram} again to $T_{\bY}$ as the derivations of \eqref{eq:I_rate} in the proof of Proposition \ref{prop:cramer_gmb}. The term $  \mathrm{III}_2$ can be similarly bounded as $ \mathrm{III}_1$. Combining the above bounds on $  \mathrm{III}_1, \mathrm{III}_2$ and  \eqref{eq:prop2_III3_bound} yields the following bound on $\mathrm{III}$,
\begin{equation} \label{eq:prop2_III_rate}
  \mathrm{III} \le  \frac{  C (\log d)^{19/6}}{n^{1/6}} +  C' \zeta_{1} \log d + 
\frac{2 \pi(\delta) }{1 - \pi^2(\delta)} +  \frac{2 \zeta_2 }{\alpha (1 - \pi^2(\delta))}.
\end{equation}
By \eqref{eq:prop2_I_II_def}, \eqref{eq:prop2_I_rate}, \eqref{eq:prop2_II_rate} and \eqref{eq:prop2_III_rate}, we have, when $\alpha \in [\alpha_L,1]$,
\begin{eqnarray} \nonumber
  \left|
\frac{\mathbb{P}(T> q_{\cqt}(\alpha;T^{\cB}) )}{\mathbb{P}(T_{\bZ} > q(\alpha;T_{\bZ}) )} - 1 
 \right| &\le &\mathrm{I}+ \mathrm{II} 
 \le  \mathrm{I}+  \mathrm{III} +  \frac{2 \PP{\maxdiff > \delta }}{\alpha} + 
   \frac{4 \zeta_2}{\alpha}\\ \nonumber
  &\le  &  \frac{  C (\log d)^{19/6}}{n^{1/6}} +  C' \zeta_{1} \log d  + \frac{C'' \zeta_2}{\alpha}  + \frac{2 \pi(\delta) }{1 - \pi^2(\delta)} + \frac{2 \PP{\maxdiff > \delta }}{\alpha} \\ \nonumber
  &\le  &  \frac{  C (\log d)^{19/6}}{n^{1/6}} +  C' \zeta_{1} \log d  + \frac{C'' \zeta_2}{\alpha}  + \frac{ C (\log d)^{11/6}}{n^{1/6}\alpha_L^{1/3}} \\ 
  & = & O\rbr{\frac{(\log d)^{19/6}}{n^{1/6}} + \frac{ (\log d)^{11/6}}{n^{1/6}\alpha_L^{1/3}} + \zeta_1 \log d + \frac{\zeta_2}{\alpha_L} }.
  \label{eq:prop2_final_bound}
\end{eqnarray}
where the third line holds due to the derivations between \eqref{eq:III_rate} and \eqref{eq:check_delta} in the proof of Proposition \ref{prop:cramer_gmb}. Remark by the choice of $\delta$ and \eqref{eq:check_delta}, we have $\pi(\delta) = O(1)$. Also note that $\zeta_2 = O(\alpha_L)$, hence we can show
\begin{equation} \label{eq:quantile_shift_bound}
  q\big( \frac{\alpha + \zeta_{2}}{1 + \pi(\delta)} ;T_{\bZ}\big) = O(\sqrt{\log d}). 
\end{equation}
when $\alpha \in [\alpha_L, 1]$. Hence we are able to verify $q\big( \frac{\alpha + \zeta_{2}}{1 + \pi(\delta)} ;T_{\bZ}\big) = O(\sqrt{\log d})$, as promised when deriving \eqref{eq:use_antibound2}. Denoting the bound in \eqref{eq:prop2_final_bound} by $\eta(d,n,\zeta_1,\zeta_2, \alpha_L)$, we finally establish \eqref{eq:approx_cramer_mb}, i.e.,
\begin{equation}\nonumber
\sup_{\alpha \in [\alpha_{L},1]}
\left|
\frac{\mathbb{P}(T > q_{\cqt}(\alpha;T^{\cB}) )}{\mathbb{P}(T_{\bZ} > q(\alpha;T_{\bZ}) )} - 1 
\right| 
=  \eta(d,n,\zeta_1,\zeta_2, \alpha_L).
\end{equation}

\end{proof}

\section{Validity and power analysis of single node testing}
\label{app:node_test_validity} 
In this section, we focus on Lemma \ref{lem:single_test} and Lemma \ref{lem:approx_quantile}. Note that these results are established using the same strategies as Theorem 4.1, Lemma S.1 and Theorem S.7 in \cite{lu2017adaptive}. We still present their proofs for completeness. 

\subsection{Proof of Lemma \ref{lem:single_test}}\label{sec:proof:lem:single_test}
\label{app:pf:lem:single_test}
\begin{proof}
For given node $j$, we denote $N_{0j} = \{(j,k) : \bTheta_{jk} =0 \}$, then 
$N_{0j}^c = \{(j,k) : |\bTheta_{jk}| >0 \}$. First we
% we denote the set $N_{\star} = \{(j,k): |\bTheta_{jk}|>0\}$ and
consider the following event,
$$
 \cE = \Big\{\min_{e \in N_{0j}^c } \sqrt{n}|\tdTheta_e| >  \hat{c} (\alpha, E_0)\Big\}, \quad \text{ where } E_{0} = \{(j,k): k \neq j,  k \in [d]\}.
$$
By the definition of Algorithm \ref{algo:skipdown}, we immediately have the rejected edge set in the first iteration can be written as
$$
   E_1 = \{(j,k) \in E_0 : \sqrt{n} |\tdTheta_{jk}| >  \hat{c} (\alpha, E_0) \}.
$$
Regarding \myrom{1} i.e., under the alternative hypothesis $H_{1j}:  \jdeg \ge k_{\tau}$, we first note $\psi_{j,\alpha} = 1$ on the event $\cE$. Also notice that $N_{0j}^c \subseteq E_1$ given $\cE$. Then the following bound immediately follows:
\begin{equation}\label{eq:power-e1}
    \PP{\psi_{j,\alpha} = 1} \ge \PP{\cE}.
\end{equation}
% which is due to the fact that $E_{\star} \subseteq E_1$ under $\cE_1$. 
% Next we give a lower bound for $\PP{\cE_1}$,
We further derive a lower bound for $\PP{\cE}$ by the triangle inequality:
% We will bound $\PPP(\cE_1)$ next. We have $\PPP ( T_{V\times V}> c(\alpha, V\times V) ) \le \alpha$
\begin{align} \label{eq:power-e2}
%   {\color{blue}\nonumber \text{my correct version}:} \\ 
   \PP{\cE} & \ge \PP{\min_{e \in N_{0j}^c   } |\sTheta_e| > \frac{\hat{c} (\alpha, E_0)}{\sqrt{n}} + C_0 
   \sqrt{\frac{\log d}{n}} \text{ and } %\max_{e \in \cV \times \cV}
   \norm{\tdTheta - \sTheta}_{\max} \le C_0 \sqrt{\frac{\log d}{n}}}. 
\end{align}

% Applying  \eqref{eq:quantile-app} with $E = V\times V$,
% For any fixed $\alpha \in (0,1)$, we consider sufficiently large $n$ and $d$ such that $2/d^2 \le \alpha/2$ and $\PPP ( T_{E_0
% }> \hat{c}(\alpha, E_0) ) >\alpha/2 $.
% % $\PPP ( T_{\cV\times \cV}> \hat{c}(\alpha, \cV\times \cV) ) >\alpha/2$.
% \colored{%\color{blue}
% % then recall the definition $\hat{c} (\alpha,E) = \inf \left\{ t\in \RR : \PPP_\xi \left( |T^{\cB}_E| \le t  \right) \ge 1-\alpha   \right\}$, we have $\PPP ( T_{V\times V}> \hat{c}(\alpha, V\times V) ) \le \alpha$ but
% This is possible due to the convergence result of the quantile approximation \eqref{eq:quantile-app}. 
% % we have $\PPP ( T_{V\times V}> \hat{c}(\alpha, V\times V) ) >\alpha/2$, 
% Then, the above result together with \eqref{eq:max-debias-scaled} yields the following:
% }
% $$
% \hat{c} (\alpha, E_0) \le C_0 \sqrt{\frac{\log d}{n}}\cdot\sqrt{n}=C_0 \sqrt{\log d}.
% $$
For any fixed $\alpha \in (0,1)$, we consider sufficiently large $d$ such that $1/d \le \alpha$. By applying Lemma \ref{lem:quantile_logd}, we have
\begin{equation}\nonumber
% \label{eq:quantile_logd}
  \mathbb{P}\left(  \mathbb{P}_{\xi}(T^{\cB}_{E_{0}}  \ge C_0 \sqrt{\log d}\mid 
\{\bX_i\}_{i=1}^n
) \le  1/d\right)  \ge 1 - 1/d^2,
\end{equation}
where $E_{0} = \{(j,k): k \neq j,  k \in [d]\}$. Recall the definition of $\hat{c} (\alpha,E)$
$$
\hat{c} (\alpha,E) = \inf \left\{ t\in \RR : \PPP_\xi \left( T^{\cB}_{E} \le t  \right) \ge 1-\alpha    \right\}.
$$
We then have 
$
\hat{c} (\alpha, E_0) \le C_0 \sqrt{\log d}
$
for some constant $C_0>0$, with probability greater than $1 - 1/d^2$.
Choosing the constant in the signal strength condition of Lemma \ref{lem:single_test} to be $2 C_0$ 
% in $\cG_1(C; \cP)$ 
(i.e., for any $(j,k) \in N_{0j}^c $, $|\bTheta_{jk}| \ge 2C_0\sqrt{\log d/ n}$)
% here $c=2 C_0$, $=\sqrt{\log d/ n}$)} 
% and it follows from the signal strength condition, 
%\eqref{eq:power-e0} 
and applying \eqref{eq:max-debias-scaled},
we have \acc{with probability greater than $1 - 1/d^2$}
% \colored{%\color{blue}
\begin{align*}
   ~& \min_{e \in N_{0j}^c } |\sTheta_e| \ge 2C_0 \sqrt{\frac{\log d}{n}}\ge \frac{\hat{c} (\alpha, E_0)}{\sqrt{n}} +C_0 \sqrt{\frac{\log d}{n}} 
   \text{~~and~}\\
  ~&~ \PP{ \norm{\tdTheta - \sTheta}_{\max} \le C_0 \sqrt{\frac{\log d}{n}} } \ge 1- 2/d^2.
\end{align*}
% } %By the definition of $\hat{c} (\alpha,E)$ we have $PP\Big( \max_{e \in V \times V} |\hat\bTheta^{\text{d}}_e - \bTheta_e| \le \hat{c} (p, E_0) \Big) \ge 1- \alpha 2 \ge /d$.
Combining the above two inequalities with \eqref{eq:power-e1} and \eqref{eq:power-e2}, %(\col{\color{blue} my version}),
we have $\acc{\PP{\psi_{j,\alpha} = 1} \ge \PP{\cE} > 1- 3/d^2}$.
% \[
%   \PP{\psi_{\alpha} = 1} \ge \PP{(\cE_1} > 1- 2/d.
% \]
Therefore, we establish
$$
 \lim_{(n,d)\rightarrow \infty} \PP{\psi_{j,\alpha} = 1}= 1.
$$ 
Now we consider \myrom{2}, i.e., the case when $  \jdeg < k_{\tau}$. Since $ \jdeg \le k_{\tau}-1$, $\psi_{j,\alpha} = 1$ implies at least one edge in $N_{0j}$ is rejected in Algorithm \ref{algo:skipdown}. Suppose the first rejected edge in $N_{0j}$ is $(j,k_{\ast})$ and it is rejected at the $t_{\ast}$-th iteration. Then we have $N_{0j} \subseteq  E_{t_{\ast} -1}$ and
\begin{equation}\label{eq:monotone}
 \max_{e \in N_{0j}} \sqrt{n}|\tdTheta_e - \sTheta_{e}| \ge \sqrt{n}|\tdTheta_{jk_\ast}- \sTheta_{jk_\ast}| \ge \hat c( \alpha, E_{t_{\ast}-1}) \ge \hat c(\alpha, N_{0j}),
\end{equation}
%\color{blue}
where the first inequality holds since $(j,k_{\ast}) \subset N_{0j}$, the
second inequality holds since $\sTheta_{jk_\ast} = 0$ and the edge $(j,k_{\ast})$ is rejected at the $t_{\ast}$-th iteration. The last inequality holds simply because $N_{0j} \subseteq  E_{t_{\ast} -1}$.
% Note that in the above inequalities, we have $\sTheta_{e} =0$ for $e\in N_{0j}$.
Therefore by applying Lemma \ref{lem:quantile} with $E$ chosen to be $N_{0j}$, we have
$$
\lim_{(n,d)\rightarrow \infty} 
\PP{\psi_{j,\alpha} =1} \le \alpha.
$$
\end{proof}

\begin{lemma}
\label{lem:quantile_logd}
Under the same conditions as Lemma \ref{lem:quantile}, we have for any $j \in [d]$,
\begin{equation}
\label{eq:quantile_logd}
  \mathbb{P}\left(  \mathbb{P}_{\xi}(T^{\cB}_{E_{j}}  \ge 2C_0 \sqrt{\log d}\mid 
\{\bX_i\}_{i=1}^n
) \le  2/d\right)  \ge 1 - 1/d^2
\end{equation}
holds for some constant $C_0>0$, where $E_{j} := \{(j,k): k \ne j, k \in [d]\}$.
\end{lemma}

\begin{proof}[Proof of Lemma \ref{lem:quantile_logd}]
Recall the definitions of $   T_{E}^{\cB} $ and $ \breve{T}_{E}^{\cB}$ in \eqref{eq:TEcB} and \eqref{eq:brTEcB} respectively. 
% \begin{equation}\label{eq:TE}
%   T_E := \max_{(j,k) \in E}  \sqrt{n}\left|(  \dTheta_{jk}/\sqrt{\dTheta_{jj} \dTheta_{kk}} - {\bTheta_{jk}}/{\sqrt{\bTheta_{jj}\bTheta_{jk}}})\right|
% \end{equation}
% by the multiplier bootstrap process
% \begin{equation}\label{eq:TEcB}
%   T_{E}^{\cB} :=\max_{(j,k) \in E}  \frac{1}{\sqrt{n~ \hat{\bTheta}_{jj}\hat{\bTheta}_{kk} }} \bigg| \sum_{i=1}^n \hat{\bTheta}_{j}^{\top} (\bX_i \bX_i^{\top} \hat{\bTheta}_{k}- \eb_k)\xi_i \bigg|,
% \end{equation}
%we define two intermediate processes:
%\begin{align}\label{eq:brTE}
%  \breve{T}_E &:=  \max_{(j,k) \in E} \bigg| \frac{1}{\sqrt{n~ {\bTheta}_{jj}{\bTheta}_{kk}}}  \sum_{i=1}^n  {\bTheta}_{j}^{\top} (\bX_i \bX_i^{\top} {\bTheta}_{k}- \eb_k) \bigg|,
% \\ \label{eq:brTEcB}
% \breve{T}_{E}^{\cB} &:=  \max_{(j,k) \in E}\bigg| \frac{1}{\sqrt{n~ \bTheta_{jj}\bTheta_{kk}}}  \sum_{i=1}^n  {\bTheta}_{j}^{\top} (\bX_i \bX_i^{\top} {\bTheta}_{k}- \eb_k)\xi_i \bigg|.
%\end{align}
First, we have
\begin{align}
\nonumber
~&~\mathbb{P}\left(  \mathbb{P}_{\xi}(T^{\cB}_{E_j}  \ge 2C_0 \sqrt{\log d}\mid 
\{\bX_i\}_{i=1}^n
) \ge  2/d\right)  \\ \nonumber
=~&~ \mathbb{P}\left(  \mathbb{P}_{\xi}(\breve{T}_{E_j}^{\cB} + T^{\cB}_{E_j} - \breve{T}_{E_j}^{\cB} \ge 2C_0  \sqrt{\log d}\mid 
\{\bX_i\}_{i=1}^n
) \ge  2/d\right)  \\ \nonumber
\le ~&~ \mathbb{P}\left(  \mathbb{P}_{\xi}(\breve{T}_{E_j}^{\cB}  \ge C_0  \sqrt{\log d}\mid 
\{\bX_i\}_{i=1}^n 
)  + \mathbb{P}_{\xi}( |T^{\cB}_{E_j} - \breve{T}_{E_j}^{\cB} |\ge C_0  \sqrt{\log d}\mid 
\{\bX_i\}_{i=1}^n
)
\ge  2/d\right) \\ \nonumber
\le ~&~ \mathbb{P}\left(  \mathbb{P}_{\xi}(\breve{T}_{E_j}^{\cB} \ge C_0  \sqrt{\log d}\mid 
\{\bX_i\}_{i=1}^n 
)  \ge 1/d \right) \\ \nonumber
~&~ + \mathbb{P}\left( \mathbb{P}_{\xi}( |T^{\cB}_{E_j} - \breve{T}_{E_j}^{\cB} |\ge C_0  \sqrt{\log d}\mid 
\{\bX_i\}_{i=1}^n
)
\ge 1/d \right) \\
\le ~&~ \mathbb{P}\left(  \mathbb{P}_{\xi}(\breve{T}_{E_j}^{\cB} \ge C_0  \sqrt{\log d}\mid 
\{\bX_i\}_{i=1}^n 
)  \ge  1/d \right) + 1/d^2, 
\label{eq:quantile_logd_sub1}
\end{align}
where the first and second inequalities hold due to the union bound, and the last inequality holds due to \eqref{eq:zeta12_rate2} in the proof of Lemma \ref{lem:quantile}.
Now it suffices to prove
\begin{equation} \label{eq:quantile_logd_sub2}
 \max_{j \in [d]} \mathbb{P}\left(  \mathbb{P}_{\xi}(\breve{T}_{E_j}^{\cB} \ge C_0 \sqrt{\log d}\mid 
\{\bX_i\}_{i=1}^n
) \ge  1/d\right)  \le  1/d^2.
\end{equation}	
%where $\breve{T}_{E_j}^{\cB}$ is defined as 
%\begin{equation}
%	 \breve{T}_{E_j}^{\cB} :=  \max_{(j,k) \in E_j}\bigg| \frac{1}{\sqrt{n~ \bTheta_{jj}\bTheta_{kk}}}  \sum_{i=1}^n  {\bTheta}_{j}^{\top} (\bX_i \bX_i^{\top} {\bTheta}_{k}- \eb_k)\xi_i \bigg|.
%\end{equation}
Then we notice the following,
\begin{align} \nonumber
% ~ & ~ \mathbb{P}\big(\hat c(\alpha, E_{0j}) \ge C_0 \sqrt{\log d}\big) 
%  \le ~ 
 &~ \mathbb{P}_{\xi}\big( \breve{T}_{E_j}^{\cB}  \ge C_0 \sqrt{\log d} \mid \{\bX_i\}_{i=1}^n \big)\\ \nonumber
  = ~ &~  \mathbb{P}_{\xi}\Big(  \underset{(j,k)\in E_{j}}{\max} \; \frac{1}{\sqrt{n~ {{\bTheta}_{jj}{\bTheta}_{kk}}}}   \bigg|\sum_{i=1}^n  {{\bTheta}^{\top}_j \left( \bX_i \bX_i^{\top} {\bTheta}_k - \eb_k  \right)} \xi_i \bigg|  \ge C_0 \sqrt{\log d} \mid \{\bX_i\}_{i=1}^n\Big)\\
  \le ~ & ~\sum_{(j,k) \in E_{j}}  \mathbb{P}_{\xi}\Big(  \frac{1}{\sqrt{n~ {{\bTheta}_{jj}{\bTheta}_{kk}}}}   \bigg|\sum_{i=1}^n  {{\bTheta}^{\top}_j \left( \bX_i \bX_i^{\top} {\bTheta}_k - \eb_k  \right)} \xi_i \bigg|  \ge C_0 \sqrt{\log d} \mid \{\bX_i\}_{i=1}^n\Big).
  \label{eq:eachtailbound}
\end{align}
where the equality holds by the definition of $\breve{T}^{\cB}_{E_j}$ in \eqref{eq:TEcB} and the inequality holds by the union bound. In the following, we will bound \eqref{eq:eachtailbound} for each $(j,k) \in E_{j}$. Note that conditioning on $\{\bX_i\}_{i=1}^n$, the following random variable is a mean zero Gaussian random variable
\begin{equation}\nonumber
    G_{jk}:= \frac{1}{\sqrt{n~ {{\bTheta}_{jj}{\bTheta}_{kk}}}}  \sum_{i=1}^n  {{\bTheta}^{\top}_j \left( \bX_i \bX_i^{\top} {\bTheta}_k - \eb_k  \right)} \xi_i. 
\end{equation}
Hence we will bound its conditional variance then apply the sub-Gaussian tail probability bound (in Section 2.1.2 of \cite{wainwright2019}). 
Specifically, we have 
%with probability greater than $1 - 1/d^2$,
\begin{align}\nonumber
 \mathbb{P} \big( \mathrm{Var}(G_{jk}\mid \{\bX_i\}_{i=1}^n)    >2 C_0 \big) 
 =& ~  \mathbb{P} \big(   (\bTheta_{jj} \bTheta_{kk})^{-1}\cdot \frac{1}{n} \sum_{i=1}^n  \big[{{\bTheta}^{\top}_j ( \bX_i \bX_i^{\top} {\bTheta}_k - \eb_k  )}\big]^2   > 2C_0 \big)  \\ \nonumber
 \le & ~ \mathbb{P} \big(   (\bTheta_{jj} \bTheta_{kk})^{-1}\cdot \frac{1}{n} \sum_{i=1}^n  \big[{{\bTheta}^{\top}_j ( \bX_i \bX_i^{\top} {\bTheta}_k - \eb_k  )}\big]^2   > C_0 + C_0 \sqrt{\frac{\log d}{n}} \big) \\ \nonumber
 \le & ~ 
  1/d^2,
\end{align}
for some constant $C_0>0$, where the equality holds by the definition of $G_{jk}$, the first inequality holds under the scaling condition of Lemma \ref{lem:quantile}, and the last inequality holds due to the Bernstein's inequality (in Section 2.2.2 of \cite{vanderVaart1996Weak}) and the assumption that $\bTheta \in \cU(M,s, r_0)$ and $\bX_1,\ldots, \bX_n  \stackrel{\mathrm{i.i.d.}}{\sim} N_d (0,\bSigma)$.
Therefore, we have with probability greater than $1 - 1/d^2$,
\begin{equation} \nonumber
    \mathbb{P}_{\xi}\left( G_{jk}  \ge C_0 
    \sqrt{\log d} \mid \{\bX_i\}_{i=1}^n \right) \le \frac{1}{d^2} 
\end{equation}
for some constant $C_0>0$ by the sub-Gaussian tail probability bound (in Section 2.1.2 of \cite{wainwright2019}). Combining the above bound with \eqref{eq:eachtailbound}, we have
\begin{equation} \nonumber
   \mathbb{P}\left( \mathbb{P}_{\xi}\big(
    \breve{T}^{\cB}_{E_{j}}  \ge C_0 \sqrt{\log d} \mid \{\bX_i\}_{i=1}^n \big)  \le d\cdot \frac{1}{d^2}\right) \ge 1- 1/d^2
\end{equation}
since $|E_{j}|\le d$. The above derivations hold for any $j \in [d]$, thus \eqref{eq:quantile_logd_sub2} is established. Finally, combining \eqref{eq:quantile_logd_sub2} with \eqref{eq:quantile_logd_sub1} yields \eqref{eq:quantile_logd}.
\end{proof}

\subsection{Proof of Lemma \ref{lem:quantile}}
\label{app:pf:lem:quantile}
We first recall the definition of $ \cU(M,s, r_0)$ and write down the statement of Lemma \ref{lem:quantile} below.
\begin{equation}\label{eq:UM}
\begin{aligned}
 \cU(M,s, r_0) &= \Big\{\bTheta \in \RR^{d \times d} \,\big|\,  \lambda_{\min}(\bTheta) \ge 1/\lammin,
 \lambda_{\max}(\bTheta) \le \lammin,
 \max_{j \in [d]} \|\bTheta_{j}\|_{0} \le s, \|\bTheta \|_1 \le M \Big\}.
\end{aligned}
\end{equation}

\begin{lemma}
\label{pf:lem:quantile}
Suppose that $\bTheta \in \cU(M,s, r_0)$. If $(\log (dn))^7/n + s^2 (\log dn)^{4}/ {n} = o(1)$,  for any edge set $E \subseteq \cV \times \cV$, we have for any $\alpha \in [0,1]$,
\begin{equation}\label{eq:quantile-cvg}
    \lim_{(n,d)\rightarrow \infty} \sup_{\bTheta \in \cU(M,s, r_0)}  \sup_{\alpha \in (0,1)} 
    \left|\PPP
    \left( \max_{e \in E}  \sqrt{n} |\tdTheta_e -\sTheta_e|> \hat{c}(\alpha, E) 
    \right) - \alpha
    \right|=0.
\end{equation}
\end{lemma}
Throughout the following parts, we will write the standardized one-step estimator explicitly:
$$
\dTheta_{jk}/\sqrt{\dTheta_{jj} \dTheta_{kk}},\quad \text{ where } \dTheta_{jk} := \hat{\bTheta}_{jk} - \frac{\hat{\bTheta}_{j}^{\top} \left( \hat{\bSigma} \hat{\bTheta}_k - \eb_k\right)}{\hat{\bTheta}_j^{\top} \hat{\bSigma}_j}.
$$
%Remark that from now on, $\dTheta_{jk}$ refers to the one-step estimator directly, and should not be mistaken as the standardized version (which uses the same notation) in \eqref{eq:quantile-cvg}, Section \ref{sec:method} and Appendix \ref{app:pf:fdr}.

In order to prove \eqref{eq:quantile-cvg}, we need preliminary results on the estimation rates of CLIME estimator.  \cite{cai2011constrained} gives the following theorem. We can also prove the same result for the GLasso estimator \cite{jankova2018inference}. Therefore, Lemma \ref{pf:lem:quantile} applies for both the CLIME estimator and the GLasso estimator. This also implies that the results in our paper apply to both the CLIME estimator and the GLasso estimator.
\begin{lemma}\label{lem:clime_rate}
 Suppose $\bTheta \in \cU(M,s, r_0)$ and we choose the tuning parameter $\lambda  \ge C M\sqrt{\log d/n}$ in the CLIME estimator. With probability greater than $1-c/d^2$, we have the following bounds:
\begin{align}\label{eq:clime-rates}
\begin{split}
&~\norm{\hat \bSigma - \bSigma}_ {\max}  \le C \sqrt{\frac{\log d}{n}}, \norm{\hat \bTheta \hat \bSigma - \Ib}_{\max} \le CM\sqrt{\frac{\log d}{n}},\text{ and } \\
&~\norm{\hat \bTheta - \bTheta}_{\max} \le CM\sqrt{\frac{\log d}{n}},
\norm{\hat \bTheta - \bTheta}_1 \le CM\sqrt{\frac{s^2\log d}{n}},
% ~~~~~
\end{split}
\end{align}
where $C$ is a universal constant only depending on $\lammin$ in \eqref{eq:UM}.
\end{lemma}
\begin{remark}
    Note the first inequality in \eqref{eq:clime-rates} directly follows from Equation (26) in \cite{cai2011constrained}, the second inequality follows from the constraint in the CLIME estimator and the third inequality holds due to Theorem 6 in \cite{cai2011constrained}.
\end{remark}

 Given a random variable $Z$, we define its $\psi_{\ell}$-norm for $\ell \ge 1$ as $\|Z\|_{\psi_{\ell}} = \sup_{p \ge 1} p^{-1/\ell} (\EEE|Z|^p)^{1/p}$. The following lemma controls the $\psi_{\ell}$-norm of $\bX$ and gives the lower bound of the variance of the debiased estimator.
 \begin{lemma}\label{lem:asp_moment}
    There exist universal constants $c$ and $C$ only depending on $\lammin$ in \eqref{eq:UM} such that
\begin{equation}\label{eq:moment}
 \sup_{\|\vb\|_2 = 1}\|\vb^{\top} \bSigma^{-1/2}\bX\|_{\psi_2}\le C  \text{ and } \min_{j,k \in [d]}  \EEE [(\bTheta_j^{\top} (\bX\bX^{\top} - \bSigma) \bTheta_k)^2 ] \ge c.
\end{equation}
\end{lemma}

\begin{proof}
The first inequality in \eqref{eq:moment} immediately follows since
  $\vb^{\top}\bSigma^{-1/2}\bX \sim N(0,1)$ for any $\|\vb\|_2 = 1$.
%   the first inequality in \eqref{eq:moment} is straightforward.
Regarding the second inequality, note that $\EEE [(\bTheta_j^{\top} (\bX\bX^{\top} - \bSigma) \bTheta_k)^2 ] = \var(\bTheta_j^{\top} \bX\bX^{\top} \bTheta_k)$. Below we calculate the expression of the general form $\Var{\ub^\top \bX \bX^\top \vb}$. Specifically, we apply Isserlis' theorem \cite{isserlis1918formula} to deal with the moments of Gaussian random variables. For any deterministic vectors $\ub, \vb \in \RR^d$, Isserlis' theorem says
\begin{align*}
% \nonumber 
\var(\ub^{\top} \bX \bX^{\top} \vb) 
&= \EEE[(\ub^{\top} \bX)^2(\vb^{\top}\bX)^2] - (\EEE[\ub^{\top} \bX \vb^{\top}\bX])^2 \\
% \nonumber 
&= \EEE[(\ub^{\top} \bX)^2]\EEE[(\vb^{\top}\bX)^2] + (\EEE[\ub^{\top} \bX \vb^{\top}\bX])^2 \\
&= (\ub^{\top} \bSigma \ub^{\top})(\vb^{\top} \bSigma \vb^{\top}) + (\ub^{\top} \bSigma \vb^{\top})^2. 
% \nonumber
%\label{eq:low-iss}
\end{align*}
Therefore, we obtain the following,
\[
 \EEE [(\bTheta_j^{\top} (\bX\bX^{\top} - \bSigma) \bTheta_k)^2 ] = (\bTheta_j^{\top} \bSigma \bTheta_j^{\top})(\bTheta_k^{\top} \bSigma \bTheta_k^{\top}) + (\bTheta_j^{\top} \bSigma \bTheta_k^{\top})^2 = \bTheta_{jj} \bTheta_{kk} + \bTheta_{jk}^2 \ge 1/\lammin^2,
\]
where the last inequality holds since $\lambda_{\min}(\bTheta) \ge 1/\lammin$ when $\bTheta \in \cU(M,s, r_0)$.
\end{proof}

Now we are ready to prove Lemma \ref{lem:quantile}. Note the proof of this lemma follows a similar idea as the one used in Proposition 3.1 of \cite{neykov2019combinatorial}. Since Lemma \ref{lem:quantile} involves
the standardized version of the one-step estimator in \cite{neykov2019combinatorial}, we still present the detailed proof for completeness.
% Lemma \ref{lem:quantile} in order to make the proof self-consistent.
\begin{proof}[Proof of Lemma \ref{lem:quantile}]
% \jlmargin{}{check the reference}
To approximate 
 \begin{equation}\label{eq:lem:TE}
   T_E := \max_{(j,k) \in E}  \sqrt{n}\left|(  \dTheta_{jk}/\sqrt{\dTheta_{jj} \dTheta_{kk}} - {\bTheta_{jk}}/{\sqrt{\bTheta_{jj}\bTheta_{jk}}})\right|,
 \end{equation}
 by the multiplier bootstrap process
 \begin{equation}
 \label{eq:lem:TEcB}
   T_{E}^{\cB} :=\max_{(j,k) \in E}  \frac{1}{\sqrt{n~ \hat{\bTheta}_{jj}\hat{\bTheta}_{kk} }} \bigg| \sum_{i=1}^n \hat{\bTheta}_{j}^{\top} (\bX_i \bX_i^{\top} \hat{\bTheta}_{k}- \eb_k)\xi_i \bigg|,
 \end{equation}
we define two intermediate processes
\begin{align}
\label{eq:lem:brTE}
  \breve{T}_E &:=  \max_{(j,k) \in E} \bigg| \frac{1}{\sqrt{n~ {\bTheta}_{jj}{\bTheta}_{kk}}}  \sum_{i=1}^n  {\bTheta}_{j}^{\top} (\bX_i \bX_i^{\top} {\bTheta}_{k}- \eb_k) \bigg|,
 \\ 
 \label{eq:lem:brTEcB}
 \breve{T}_{E}^{\cB} &:=  \max_{(j,k) \in E}\bigg| \frac{1}{\sqrt{n~ \bTheta_{jj}\bTheta_{kk}}}  \sum_{i=1}^n  {\bTheta}_{j}^{\top} (\bX_i \bX_i^{\top} {\bTheta}_{k}- \eb_k)\xi_i \bigg|.
\end{align}
The strategy of proving this lemma is to verify the three conditions in Corollary 3.1 of \cite{chernozhukov2013gaussian}:
\begin{enumerate}[(a)]
 \item $\min_{j,k}\mathbb{E}[ ({\bTheta}_{j}^{\top} (\bX \bX^{\top} {\bTheta}_{k}- \eb_k))^2]>c$ and $\max_{j,k \in [d]}  \norm{\bTheta_{j}^{\top} (\bX \bX^{\top} {\bTheta}_{k}- \eb_k)}_{\psi_1} \le C$ for some positive constants $c$ and $C$;
   \item $\PPP(|T_E - \breve{T}_E| > \zeta_1) < \zeta_2$ holds for some $\zeta_1 , \zeta_2 >0$;
   \item And $\PPP(\PPP_{\xi}(|T_{E}^{\cB}-\breve{T}_E^{\cB}| > \zeta_1\mid \{\bX_i\}_{i=1}^n)> \zeta_2) < \zeta_2$ holds for $\zeta_1 \sqrt{\log d}  + \zeta_2 = o(1)$.

\end{enumerate}

 Notice that in \cite{chernozhukov2013gaussian}, the original conditions require  the last scaling to be $\zeta_1 \sqrt{\log d}  + \zeta_2 = o(n^{-c_1})$ for some $c_1$. This is because
they pursue a stronger result that $|\PPP ( T_E> \hat c(\alpha, E) ) - \alpha| = O(n^{-c_1})$.
 Since we do not emphasize on the polynomial decaying in our result, we only require $\zeta_1 \sqrt{\log d}  + \zeta_2 = o(1)$.
%  and same for the scaling condition $(\log(dn))^7/n = o(1)$.
%%%%%%%%%%%%%%%%%%%%%%%%%%%%%%
% Start First Part
%%%%%%%%%%%%%%%%%%%%%%%%%%%%%

We start by checking the first condition (a). Lemma \ref{lem:asp_moment} immediately implies the first part. By the second condition in  \eqref{eq:moment}, we have $\|\bX_j \bX_k - \EEE[\bX_j \bX_k]\|_{\psi_1} \le C$. By the definition of the $\psi$-norms, we have
\begin{align}
\nonumber  \max_{j,k \in [d]}  \|{\bTheta}_{j}^{\top} (\bX_i \bX_i^{\top} {\bTheta}_{k}- \eb_k)\|_{\psi_1} &\le r_0^{2} \|(\bX_j \bX_k - \EEE[\bX_j \bX_k]) \|_{\psi_1} \\
  &\le r_0^{2}  \sup_{\|\vb\|_2 = 1}\| \vb^{\top}\bX \bX^{\top}\vb- \EEE[ \vb^{\top}\bX \bX^{\top}\vb]\|_{\psi_1} = O(1). 
\nonumber
\end{align}

% For any matrix $\bTheta = (\bTheta_{1}, \ldots, \bTheta_{d})  \in \RR^{d \times d}$, we define $\bTheta_{-j} := (\bTheta_{1}, \ldots, \bTheta_{j-1},\bTheta_{j+1}, \ldots,\bTheta_{d})\in \RR^{d \times (d-1)}$ being the submatrix of $\bTheta$.

Regarding the condition (b), we check by bounding the difference $|T_E - \breve{T}_E|$. Recall the one-step estimator 
$$
\dTheta_{jk} = \hat{\bTheta}_{jk} - \frac{\hat{\bTheta}_{j}^{\top} \left( \hat{\bSigma} \hat{\bTheta}_k - \eb_k\right)}{\hat{\bTheta}_j^{\top} \hat{\bSigma}_j},
$$
and plug it into $T_E$. Then we have the following bound,
\begin{eqnarray}\label{eq:TET0_expand} \nonumber
    |T_E - \breve{T}_E|
    &=& \left| \underset{(j,k)\in E}{\max} \; \sqrt{n} \cdot \acc{\bigg|
    \frac{\dTheta_{jk}}{\sqrt{\dTheta_{jj} \dTheta_{kk}}}
    -\frac{\bTheta_{jk}}{\sqrt{\bTheta_{jj}\bTheta_{kk}} } \bigg|}
    -  \max_{(j,k) \in E}  \frac{\sqrt{n}}{\sqrt{{\bTheta}_{jj}{\bTheta}_{kk}}} 
    \acc{\bigg| {\bTheta}_{j}^{\top} (\hat{\bSigma} {\bTheta}_{k}- \eb_k) 
    \bigg|} \right |\\
    &\le&   \frac{\mathrm{I}_1\mathrm{I}_2 }{\underset{(j,k)\in E}{\min}\sqrt{{\bTheta}_{jj}{\bTheta}_{kk} } } + \frac{\mathrm{I}_3}{\underset{(j,k)\in E}{\min}\sqrt{\acc{\dTheta_{jj}\dTheta_{kk} }}},
\end{eqnarray}
where $ \mathrm{I}_1 =  \underset{(j,k)\in E}{\max}|\dTheta_{jj}\dTheta_{kk} - \bTheta_{jj}\bTheta_{kk} |$, $ \mathrm{I}_2 = \underset{(j,k)\in E}{\max}| \sqrt{n} \cdot  {\bTheta}^{\top}_{j}\big(\hat{\bSigma}    {\bTheta}_k - \eb_k\acc{\big)}| $ and 
$$
\mathrm{I}_3 = \max_{(j,k) \in E} \Big| \sqrt{n}(\dTheta_{jk} - \bTheta_{jk})- \sqrt{n} \cdot
{\bTheta}^{\top}_{j}\big(\hat{\bSigma}    {\bTheta}_k - \eb_k\big)\Big|.
$$
Note $ \mathrm{I}_1 $ can be bounded using Lemma \ref{lem:max-debias}, i.e.,
\begin{eqnarray}\label{eq:TET0_II1}
    \mathrm{I}_1 
    &=& \underset{(j,k)\in E}{\max}|\dTheta_{jj}\dTheta_{kk} - \bTheta_{jj}\bTheta_{kk}|
    \le 2M  \big\|\dTheta - {\bTheta} \big\|_{\max} \le  CM^2\sqrt{\frac{\log d}{n}},
\end{eqnarray}
with probability $1-1/d^2$. As for the term $\mathrm{I}_2$, we have
\begin{eqnarray}\nonumber
    \mathrm{I}_2 
    = \underset{(j,k)\in E}{\max}\left| \sqrt{n}  {\bTheta}^{\top}_{j}\big(\hat{\bSigma}    {\bTheta}_k - \eb_k \big)\right|
    &=& \underset{(j,k)\in E}{\max} \sqrt{n}  \left| {\bTheta}^{\top}_{j}\big(\hat{\bSigma}   -\bSigma \big) {\bTheta}_k\right| \\
    &\le&  \sqrt{n}  M^2  \big\|\hat{\bSigma} - {\bSigma} \big\|_{\max}\le  CM^2 \sqrt{{\log d}}.
    \label{eq:TET0_II2} 
\end{eqnarray}
Denote ${\cTheta}_k=  (\hat\bTheta_{k1}, \ldots,\hat\bTheta_{k(j-1)}, \bTheta_{kj}, \hat\bTheta_{k(j+1)}, \ldots, \hat\bTheta_{kd})^{\top} \in \RR^{d}$.
% \begin{align*}
% % \bgamma^* &= (\bTheta_{k1}, \ldots,\bTheta_{k(j-1)},\bTheta_{k(j+1)} \ldots \bTheta_{kd})^{\top} \in \RR^{d-1}\text{ and } \\
%  \tilde{\bTheta}_k&=  (\hat\bTheta_{k1}, \ldots,\hat\bTheta_{k(j-1)}, \bTheta_{kj}, \hat\bTheta_{k(j+1)}, \ldots, \hat\bTheta_{kd})^{\top} \in \RR^{d}.
% \end{align*}
To deal with the term $\mathrm{I}_3$,
%$\sqrt{n} ( \tilde{\bTheta}^d_{jk} - {\bTheta_{jk}}/{\sqrt{ \bTheta_{jj}\bTheta_{kk}}})$, 
we first rewrite the following
\begin{equation}
\label{eq:rewrite_bias}
   \sqrt{n} (\dTheta_{jk} - {\bTheta}_{jk}) = - \sqrt{n} \cdot \frac{ \hat{\bTheta}^{\top}_{j}\big(\hat{\bSigma}  {\cTheta}_k - \eb_k\big) }{\hat{\bTheta}_{j}^{\top} \hat{\bSigma}_j},
\end{equation}
then quantify $\sqrt{n}\hat{\bTheta}^{\top}_{j}\big(\hat{\bSigma}  {\cTheta}_k - \eb_k^{\top}\big)$. Notice that
\begin{equation}\label{eq:omega-e1}
 \sqrt{n} \cdot \hat{\bTheta}^{\top}_{j}\big(\hat{\bSigma}  {\cTheta}_k - \eb_k^{\top}\big) =  \underbrace{\sqrt{n} \cdot  \hat{\bTheta}^{\top}_{j}\big(\hat{\bSigma}    {\bTheta}_k - \eb_k^{\top}\big)}_{\mathrm{II}_1} +  \underbrace{\sqrt{n} \cdot  \hat{\bTheta}^{\top}_{j}\hat{\bSigma} \big(  {\cTheta}_k -  {\bTheta}_k \big)}_{\mathrm{II}_2}.
\end{equation}
Further we expand $\mathrm{II}_1$ as
\begin{equation}\label{eq:omega-e2}
    \mathrm{II}_1 = \underbrace{\sqrt{n} \cdot  {\bTheta}^{\top}_{j}\big(\hat{\bSigma}    {\bTheta}_k - \eb_k\big)}_{\mathrm{II}_{11}} +  \underbrace{\sqrt{n} \cdot \big(\hat{\bTheta}^{\top}_{j} - {\bTheta}^{\top}_{j}\big)\big(\hat{\bSigma}  {\bTheta}_k - \eb_k\big)}_{\mathrm{II}_{12}},
\end{equation}
where $\mathrm{II}_{11}$ can be rewritten as $\mathrm{II}_{11} = \frac{1}{\sqrt{n}}  \sum_{i=1}^n  {\bTheta}_{j}^{\top} (\bX_i \bX_i^{\top} {\bTheta}_{k}- \eb_k) $. We bound $|\mathrm{II}_{12}|$ as
% which is same as what is inside of the maximum of $T_0$.We next bound $I_{12}$ by H\"{o}lder inequality,
\begin{align} \label{eq:expan-I12-rate}
 |\mathrm{II}_{12}| &= \sqrt{n} \cdot \big(\hat{\bTheta}_{j} - {\bTheta}_{j}\big)^{\top}\big(\hat{\bSigma}   - {\bSigma}  \big) {\bTheta}_k\le  \sqrt{n} \cdot \big\|\hat{\bTheta}_{j} - {\bTheta}_{j}\big\|_1\big\|\hat{\bSigma} - {\bSigma} \big\|_{\max}\|\bTheta_k\|_{1}.
\end{align}
According to Lemma \ref{lem:clime_rate}, \eqref{eq:expan-I12-rate}  yields that
\begin{equation}\label{eq:r2-e1}
\max_{ j,k \in [d]}|\mathrm{II}_{12}| \lesssim M^2\frac{s\log d}{\sqrt{n}},
\end{equation}
with probability $1-1/d^2$. By H\"{o}lder's inequality and Lemma \ref{lem:clime_rate}, we finally obtain the bound on $\mathrm{II}_{2}$:
\begin{equation}\label{eq:r2-e2}
 \max_{ j,k \in [d]}|\mathrm{II}_{2}| \le\sqrt{n} \cdot   \max_{ j,k \in [d]}\|\hat{\bTheta}^{\top}_{j}\hat{\bSigma}_{-j}  \|_{\infty}  \big\|\hat{\bTheta}_{k} - {\bTheta}_{k}\big\|_1  \lesssim M^2\frac{s\log d}{\sqrt{n}},
\end{equation}
 with probability $1-1/d^2$.
Therefore, we conclude that by \eqref{eq:r2-e1} and \eqref{eq:r2-e2}, with probability $1-1/d^2$, the following holds:
\begin{align}\label{eq:r2-e3}
  &\max_{j,k \in [d]}\sqrt{n}\cdot \Big| \hat{\bTheta}^{\top}_{j}\big(\hat{\bSigma}  {\cTheta}_k - \eb_k^{\top}\big) -  \bTheta_{j}^{\top} \Big(\hat{\bSigma}\bTheta_{k} - \eb_k^{\top} \Big)\Big| \lesssim M^2\frac{s\log d}{\sqrt{n}}.
\end{align}
Lemma \ref{lem:clime_rate} also implies
\begin{align}\label{eq:r2-e4}
 \max_{j \in [d]}|\hat{\bTheta}_{j}^{\top} \hat{\bSigma}_j - 1 | & \le \max_{j \in [d]} \norm{\hat \bTheta_j^{\top}\hat \bSigma - \eb_j}_{\infty} \lesssim M\sqrt{\frac{\log d}{n}}.
\end{align}
Combining \eqref{eq:omega-e1}, \eqref{eq:omega-e2} with \eqref{eq:r2-e3} and \eqref{eq:r2-e4}, for sufficiently large $d, n$, we have, with probability $1-1/d^2$, the following holds:
\begin{align}
\nonumber \mathrm{I}_{3} & \le \max_{(j,k) \in E} \sqrt{n}\Big|   \frac{ \hat{\bTheta}^{\top}_{j}\big(\hat{\bSigma}  {\cTheta}_k - \eb_k\big) }{\hat{\bTheta}_{j}^{\top} \hat{\bSigma}_j} - {\bTheta}^{\top}_{j}\big(\hat{\bSigma}    {\bTheta}_k - \eb_k\big)\Big|\\
\nonumber &    
\le \max_{(j,k) \in E} \big (2 \sqrt{n} |\hat{\bTheta}_{j}^{\top} \hat{\bSigma}_j - 1 |\cdot | {\bTheta}^{\top}_{j}\big(\hat{\bSigma}  - \bSigma\big)   {\bTheta}_k | \big)
+ 2 \max_{(j,k) \in E} |\hat{\bTheta}^{\top}_{j}\big(\hat{\bSigma}  {\cTheta}_k - \eb_k\big) - {\bTheta}^{\top}_{j}\big(\hat{\bSigma}    {\bTheta}_k - \eb_k\big)|\\
 & \le2M \sqrt{n}  \max_{j \in [d]}|\hat{\bTheta}_{j}^{\top} \hat{\bSigma}_j - 1 | \cdot \norm{\hat \bSigma - \bSigma}_{\max} +2\max_{j, k \in [d]} (|I_{12}| +|I_2|)
 \lesssim M^2\frac{s\log d}{\sqrt{n}} , \label{eq:TET0_II3}
 \end{align}
 where the second inequality uses $|x/(1+\delta) - y| \le 2|y\delta| + 2 |x-y|$ for any $|\delta| < 1/2$.
Therefore, combining \eqref{eq:TET0_expand}, \eqref{eq:TET0_II1},\eqref{eq:TET0_II2} with \eqref{eq:TET0_II3} and the fact $\underset{(j,k)\in E}{\min}\sqrt{{\bTheta}_{jj}{\bTheta}_{kk} }\ge \lambda_{\min}(\bTheta)\ge 1/\lammin$ (as $\bTheta \in \cU(M,s, r_0)$), we obtain the following:
\begin{equation} \label{eq:zeta12_rate1}
\PPP(|T_E - \breve{T}_E| > \zeta_1) < \zeta_2,
\end{equation}
where $\zeta_1 = {s\log d}/\sqrt{n}$ and $\zeta_2 = 1/d^2$; thus the condition (b) is verified. Also note that $\zeta_1 \sqrt{\log d}  + \zeta_2 = {s(\log d)^{3/2}}/\sqrt{n} + 1/d^2  = o(1)$ holds under the stated scaling condition of Lemma \ref{lem:quantile}. 
%%%%%%%%%%%%%%%%%
% End
%%%%%%%%%%%%%%%%%%

Regarding the third condition (c), we bound the difference between $T_{E}^{\cB}$ and $\breve{T}_E^{\cB}$ as
% \begin{eqnarray}
%      |T^B - T_0^B| \le \frac{\mathrm{III}_1\mathrm{III}_2 }{\underset{(j,k)\in E}{\min}\sqrt{{\bTheta}_{jj}{\bTheta}_{kk} } } + \frac{\mathrm{III}_3}{\underset{(j,k)\in E}{\min}\sqrt{\hat{\bTheta}_{jj}\hat{\bTheta}_{kk} }}
% \end{eqnarray}
% where $ \mathrm{III}_1 =  \underset{(j,k)\in E}{\max}|\hat{\bTheta}_{jj}\hat{\bTheta}_{kk} - \bTheta_{jj}\bTheta_{kk} |$ and $ \mathrm{III}_2 = \frac{1}{\sqrt{n}} \sum_{i=1}^n  {\bTheta}^{\top}_{j}\big(\bX_i \bX_i^{\top}  {\bTheta}_k - \eb_k\big)\xi_i $ and 
$$
     |T_E^{\cB} - \breve{T}_E^{\cB}| \le \max_{(j,k) \in E}\Big| \frac{1}{\sqrt{n}} \sum_{i=1}^n
     \Big( \frac{   \hat{\bTheta}^{\top}_{j} }{\sqrt{ \hat{\bTheta}_{jj}\hat{\bTheta}_{kk} }} 
     \big(\bX_i \bX_i^{\top}  \hat{\bTheta}_k - \eb_k\big) - 
     \frac{\bTheta^{\top}_{j}}{ \sqrt{{\bTheta}_{jj}{\bTheta}_{kk} } }\big(\bX_i \bX_i^{\top}  {\bTheta}_k - \eb_k\big)\Big)\xi_i  \Big|
$$

Conditioning on the data $\{\bX_i \}_{i=1}^n$, the right hand side of the above inequality is a suprema of a Gaussian process. Therefore, we need to bound the following conditional variance
\begin{align*}
  {\max_{(j,k) \in E} \frac{1}{n}  \sum_{i=1}^n \Big[\frac{   \hat{\bTheta}^{\top}_{j} }{\sqrt{ \hat{\bTheta}_{jj}\hat{\bTheta}_{kk} }} \big(\bX_i \bX_i^{\top}  \hat{\bTheta}_k - \eb_k\big) - \frac{\bTheta^{\top}_{j}}{ \sqrt{{\bTheta}_{jj}{\bTheta}_{kk} } }\big(\bX_i \bX_i^{\top}   {\bTheta}_k - \eb_k \big)\Big]^2}
%   & \le 2\frac{\mathrm{III}_1\mathrm{III}_2 }{\underset{(j,k)\in E}{\min}{{\bTheta}_{jj}{\bTheta}_{kk} }} + 2\frac{\mathrm{III}_3}{\underset{(j,k)\in E}{\min}{\hat{\bTheta}_{jj}\hat{\bTheta}_{kk} }} 
\end{align*}
Note the summand (for each $i$) can be bounded by 
$$
2\frac{\mathrm{III}_1\mathrm{III}_2 }{\underset{(j,k)\in E}{\min}{{\bTheta}_{jj}{\bTheta}_{kk} }} + 2\frac{\mathrm{III}_3}{\underset{(j,k)\in E}{\min}{\hat{\bTheta}_{jj}\hat{\bTheta}_{kk} }} 
$$
where $\mathrm{III}_1,\mathrm{III}_2$ and $\mathrm{III}_3$ are defined and bounded as below:
\begin{eqnarray} \label{eq:snode_III1}
    \mathrm{III}_1 &:=&  \underset{(j,k)\in E}{\max}|\hat{\bTheta}_{jj}\hat{\bTheta}_{kk} - \bTheta_{jj}\bTheta_{kk} |^2  \le \left( CM^2\sqrt{\frac{\log d}{n}}\right)^2\\ \nonumber
    \mathrm{III}_2 
    &:=& \underset{(j,k)\in E}{\max}[{\bTheta}^{\top}_{j}\big(\bX_i \bX_i^{\top}  {\bTheta}_k - \eb_k\big)]^2 
    = \underset{(j,k)\in E}{\max}[{\bTheta}^{\top}_{j}\big(\bX_i \bX_i^{\top}  - \bSigma \big){\bTheta}_k]^2 \\ \label{eq:snode_III2}
    &\le& \Big[M^2 \max_i\|\bX_i\bX_i^{\top}- \bSigma\|_{\max} \Big]^2\\ \nonumber
    \mathrm{III}_3
    & =& \max_{(j,k) \in E}  \Big|\hat{\bTheta}^{\top}_{j}\big(\bX_i \bX_i^{\top}  \hat{\bTheta}_k - \eb_k\big) - {\bTheta}^{\top}_{j}\big(\bX_i \bX_i^{\top}   {\bTheta}_k - \eb_k \big)\Big|^2\\ \label{eq:snode_III3}
   &\lesssim& \Big[2M\|\hat\bTheta- \bTheta\|_1 \max_i\|\bX_i\bX_i^{\top}- \bSigma\|_{\max} \Big]^2.
\end{eqnarray}
% $$ \mathrm{III}_1 =  \underset{(j,k)\in E}{\max}|\hat{\bTheta}_{jj}\hat{\bTheta}_{kk} - \bTheta_{jj}\bTheta_{kk} |^2 = \mathrm{II}^2_1 \le \left( CM^2\sqrt{\frac{\log d}{n}}\right)^2,
% $$
% \begin{eqnarray}
% \mathrm{III}_2 
% &=& \underset{(j,k)\in E}{\max}[{\bTheta}^{\top}_{j}\big(\bX_i \bX_i^{\top}  {\bTheta}_k - \eb_k\big)]^2 \\
% &=& \underset{(j,k)\in E}{\max}[{\bTheta}^{\top}_{j}\big(\bX_i \bX_i^{\top}  - \bSigma \big){\bTheta}_k]^2 \\
% &\le& \Big[M^2 \max_i\|\bX_i\bX_i^{\top}- \bSigma\|_{\max} \Big]^2
% \end{eqnarray}
% and 
% \begin{eqnarray}
%     \mathrm{III}_3
%     & =& \max_{(j,k) \in E}  \Big|\hat{\bTheta}^{\top}_{j}\big(\bX_i \bX_i^{\top}  \hat{\bTheta}_k - \eb_k\big) - {\bTheta}^{\top}_{j}\big(\bX_i \bX_i^{\top}   {\bTheta}_k - \eb_k \big)\Big|^2\\
%   &\lesssim& \Big[2M\|\hat\bTheta- \bTheta\|_1 \max_i\|\bX_i\bX_i^{\top}- \bSigma\|_{\max} \Big]^2.
% \end{eqnarray}
According to Lemma \ref{lem:asp_moment}, we have with probability $1-1/d^2$,
\begin{equation}\label{eq:covmat_max_bound}
    \max_i\|\bX_i\bX_i^{\top}- \bSigma\|_{\max} \le C\sqrt{\log (dn)}.
\end{equation}
Therefore, the event
$$
  \cE = \Big\{ \max_{(j,k) \in E} \frac{1}{n}  \sum_{i=1}^n \Big[\frac{   \hat{\bTheta}^{\top}_{j} }{\sqrt{ \hat{\bTheta}_{jj}\hat{\bTheta}_{kk} }} \big(\bX_i \bX_i^{\top}  \hat{\bTheta}_k - \eb_k\big) - \frac{\bTheta^{\top}_{j}}{ \sqrt{{\bTheta}_{jj}{\bTheta}_{kk} } }\big(\bX_i \bX_i^{\top}   {\bTheta}_k - \eb_k \big)\Big]^2 \le CM^2\frac{(s\log (dn))^2}{n} \Big\}
$$
satisfies $\PPP(\cE^c)< 1/d^2$.
Therefore, by the maximal inequality, under the event $\cE$,  we have
\begin{eqnarray} \nonumber
&~~&\Ec{ \max_{(j,k) \in E} \frac{1}{\sqrt{n}} \sum_{i=1}^n \Big( \frac{   \hat{\bTheta}^{\top}_{j} }{\sqrt{ \hat{\bTheta}_{jj}\hat{\bTheta}_{kk} }} 
     \big(\bX_i \bX_i^{\top}  \hat{\bTheta}_k - \eb_k\big) - 
     \frac{\bTheta^{\top}_{j}}{ \sqrt{{\bTheta}_{jj}{\bTheta}_{kk} } }\big(\bX_i \bX_i^{\top}  {\bTheta}_k - \eb_k\big)\Big)\xi_i}{ \{\bX_i\}_{i=1}^n} \\ \nonumber
&\lesssim& M^2 \frac{(s \log dn)\sqrt{\log d}}{\sqrt{n}}.
\end{eqnarray}
Applying Borell's inequality, we  have with probability $1-1/d^2$,
\begin{eqnarray} \nonumber
\Pc{ \max_{(j,k) \in E} \frac{1}{\sqrt{n}} \sum_{i=1}^n \left( \frac{   \hat{\bTheta}^{\top}_{j} \big(\bX_i \bX_i^{\top}  \hat{\bTheta}_k - \eb_k\big)}{\sqrt{ \hat{\bTheta}_{jj}\hat{\bTheta}_{kk} }} 
      - \frac{\bTheta^{\top}_{j} \big(\bX_i \bX_i^{\top}  {\bTheta}_k - \eb_k\big)}{ \sqrt{{\bTheta}_{jj}{\bTheta}_{kk} } }\right)\xi_i  >  C \sqrt{\frac{s^2 \log^{4} dn}{n}}}{\{\bX_i\}_{i=1}^n} 
      %\\ \nonumber Given a random variable $Z$, we define its $\psi_{\ell}$-norm for $\ell \ge 1$ as $\|Z\|_{\psi_{\ell}} = \sup_{p \ge 1} p^{-1/\ell} (\EEE|Z|^p)^{1/p}$.
    \le 1/d^2.
\end{eqnarray}
This implies that
\[
 \PP{\PPP_{\xi}\big(|T^{\cB}_E - \breve{T}_E^{\cB}|> \sqrt{(s^2 \log^{4} dn)/n} \mid  \{\bX_i\}_{i=1}^n\big) > 1/d^2}<1/d^2.
\]
Now we can verify the condition (c) by showing
\begin{equation}
	\PPP(\PPP_{\xi}(|T_{E}^{\cB}-\breve{T}_E^{\cB}| > \zeta_1\mid \{\bX_i\}_{i=1}^n)> \zeta_2) < \zeta_2,
	\label{eq:zeta12_rate2}
\end{equation}
where $\zeta_1 = {s(\log d)^2}/\sqrt{n}$, $\zeta_2 = 1/d^2$ and the condition $\zeta_1 \sqrt{\log d}  + \zeta_2 = {s(\log d)^{3/2}}/\sqrt{n} + 1/d^2  = o(1)$ holds under the stated scaling condition of Lemma \ref{lem:quantile}. Therefore, by Corollary 3.1 of \cite{chernozhukov2013gaussian}, we have
\begin{equation}\label{eq:quantile-app}
\lim_{(n,d)\rightarrow \infty} |\PPP ( T_E> \hat c(\alpha, E) ) - \alpha| = 0.
\end{equation}
And it holds for any edge set $E$, thus the proof is complete.
% Since the tail probability is independent to the edge set $E$,
% the proof is complete.
\end{proof}
\begin{lemma}
\label{lem:max-debias}
Under the same conditions as Lemma \ref{lem:quantile}, we have
\begin{equation}\label{eq:max-debias}
   \PPP\Big( \max_{j,k \in [d]} |\dTheta_{jk} -  \bTheta _{jk}| > C_0 \sqrt{\frac{\log d}{n}}\Big) <\frac{2}{d^2},
\end{equation}
for some constant $C_0 >0$.
\end{lemma}
\begin{proof}
By \eqref{eq:rewrite_bias} and \eqref{eq:TET0_II3}, we have with probability $1-1/d^2$,
\[
\max_{j,k \in [d]} |\dTheta_{jk} - \bTheta _{jk} +{\bTheta}^{\top}_{j}\big(\hat{\bSigma}    {\bTheta}_k - \eb_k\big)| \le C_1 \frac{s\log d}{n}.
\]
By Lemma \ref{lem:asp_moment} and $\|\bTheta\|_{2} \le r_0$, we have $\|{\bTheta}^{\top}_{j} \bX  \bX ^{\top} {\bTheta}_{k}\|_{\psi_1} \le C_2 r_0^2$. Applying the maximal inequality (Lemma 2.2.2 in \cite{vanderVaart1996Weak}), we have for some constant $C_3>0$
\begin{align*}
%  \lefteqn{
~&~~~~ \PPP\Big(\max_{j,k \in [d]} |{\bTheta}^{\top}_{j}\big(\hat{\bSigma}    {\bTheta}_k - \eb_k\big)| >C_3 r_0^2\sqrt{\frac{\log d}{n}}\Big)
%  }
\\
~&~ \le \PPP\Big(\max_{j,k \in [d]} \Big|\frac{1}{n}\sum_{i=1}^n ({\bTheta}^{\top}_{j} \bX_i  \bX_i^{\top} {\bTheta}_{k} - \EEE[{\bTheta}^{\top}_{j} \bX_i  \bX_i^{\top} {\bTheta}_{k} )]\Big| > C_3 r_0^2\sqrt{\frac{\log d}{n}}\Big) \le 1/d^2.
\end{align*}
With $C_0 = C_1 + C_3$, \eqref{eq:max-debias} is proved. And it is not hard to show a similar result for the standardized one-step estimator also holds, i.e.,
\begin{equation}
\label{eq:max-debias-scaled}
   \PPP\Big( \max_{j,k \in [d]} |\tdTheta_{jk} -  \sTheta _{jk}| > C_0' \sqrt{\frac{\log d}{n}}\Big) <\frac{2}{d^2}
\end{equation}
for some constant $C_0'>0$.
\end{proof}
\section{Tables and plots deferred from the main paper}
\label{app:plots_tables}
\subsection{Graph pattern demonstration}
\label{app:graph_pattern}
\begin{figure}[htbp]
    \centering
    \includegraphics[width = 0.75\linewidth]{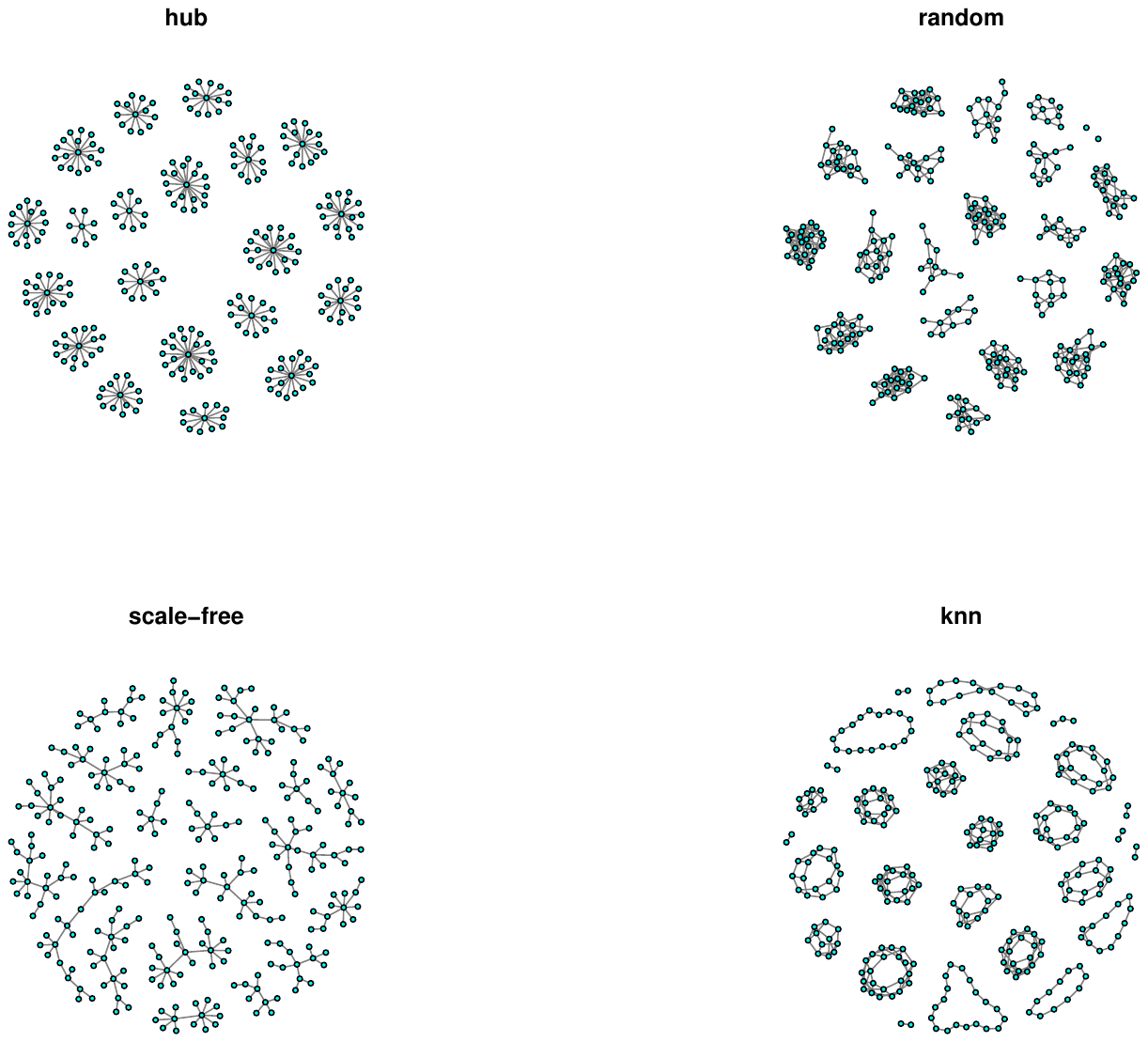}
\end{figure}
\subsection{Tables of $q\frac{d_0}{d}$}
\begin{table}[hptb]
% \small
\addtolength{\tabcolsep}{-4pt}
\begin{center}
    \caption{{\bf \rev{$q \frac{d_0}{d}$}}}
    \begin{tabular}{
    @{\hspace{1.2em}}c 
    @{\hspace{1.2em}}c
    @{\hspace{1em}}c
    @{\hspace{1em}}c|
    @{\hspace{1em}}c
    @{\hspace{1em}}c
    @{\hspace{1em}} c
    @{\hspace{1em}} c
    @{\hspace{1em}} c 
    @{\hspace{1.2em}} }
    \toprule
    $d= 300~$ &\multicolumn{3}{c}{$q=0.1$}&\multicolumn{3}{c}{$q=0.2$} \\
    \hline
    $n$ &200 &300 &400
        &200 &300 &400\\
    \hline & \multicolumn{5}{c}{$\quad\quad p=20$}&\multicolumn{1}{c}{} \\
    \hline
% hub & 0.0930 & 0.0930 & 0.0930 & 0.1870 & 0.1870 & 0.1870 \\ 
%   random & 0.0620 & 0.0610 & 0.0600 & 0.1230 & 0.1220 & 0.1200 \\ 
%   scale-free & 0.0810 & 0.0810 & 0.0810 & 0.1620 & 0.1630 & 0.1620 \\ 
%   knn & 0.0680 & 0.0700 & 0.0690 & 0.1360 & 0.1390 & 0.1390 \\ 
hub & 0.0930 & 0.0930 & 0.0930 & 0.1870 & 0.1870 & 0.1870 \\ 
  random & 0.0610 & 0.0610 & 0.0610 & 0.1220 & 0.1220 & 0.1220 \\ 
  scale-free & 0.0810 & 0.0810 & 0.0810 & 0.1620 & 0.1620 & 0.1620 \\ 
  knn & 0.0670 & 0.0670 & 0.0670 & 0.1340 & 0.1340 & 0.1340 \\     
    \hline & \multicolumn{5}{c}{$\quad\quad p=30$}&\multicolumn{1}{c}{} \\
    \hline
% hub & 0.0900 & 0.0900 & 0.0900 & 0.1800 & 0.1800 & 0.1800 \\ 
%   random & 0.0810 & 0.0810 & 0.0810 & 0.1620 & 0.1620 & 0.1620 \\ 
%   scale-free & 0.0810 & 0.0810 & 0.0810 & 0.1620 & 0.1620 & 0.1610 \\ 
%   knn & 0.0730 & 0.0750 & 0.0740 & 0.1460 & 0.1510 & 0.1480 \\ 
hub & 0.0900 & 0.0900 & 0.0900 & 0.1800 & 0.1800 & 0.1800 \\ 
  random & 0.0810 & 0.0810 & 0.0810 & 0.1620 & 0.1620 & 0.1620 \\ 
  scale-free & 0.0800 & 0.0800 & 0.0800 & 0.1610 & 0.1610 & 0.1610 \\ 
  knn & 0.0720 & 0.0720 & 0.0720 & 0.1430 & 0.1430 & 0.1430 \\ 
    \toprule
    \end{tabular}
    \label{tb:d0/dqlevel}
  \end{center}
\end{table}
\subsection{Supplementary FDP and Power plots}
\label{app:fdp_power_plots}
\begin{figure}[htbp]
    \centering
    \includegraphics[width = 0.95\linewidth]{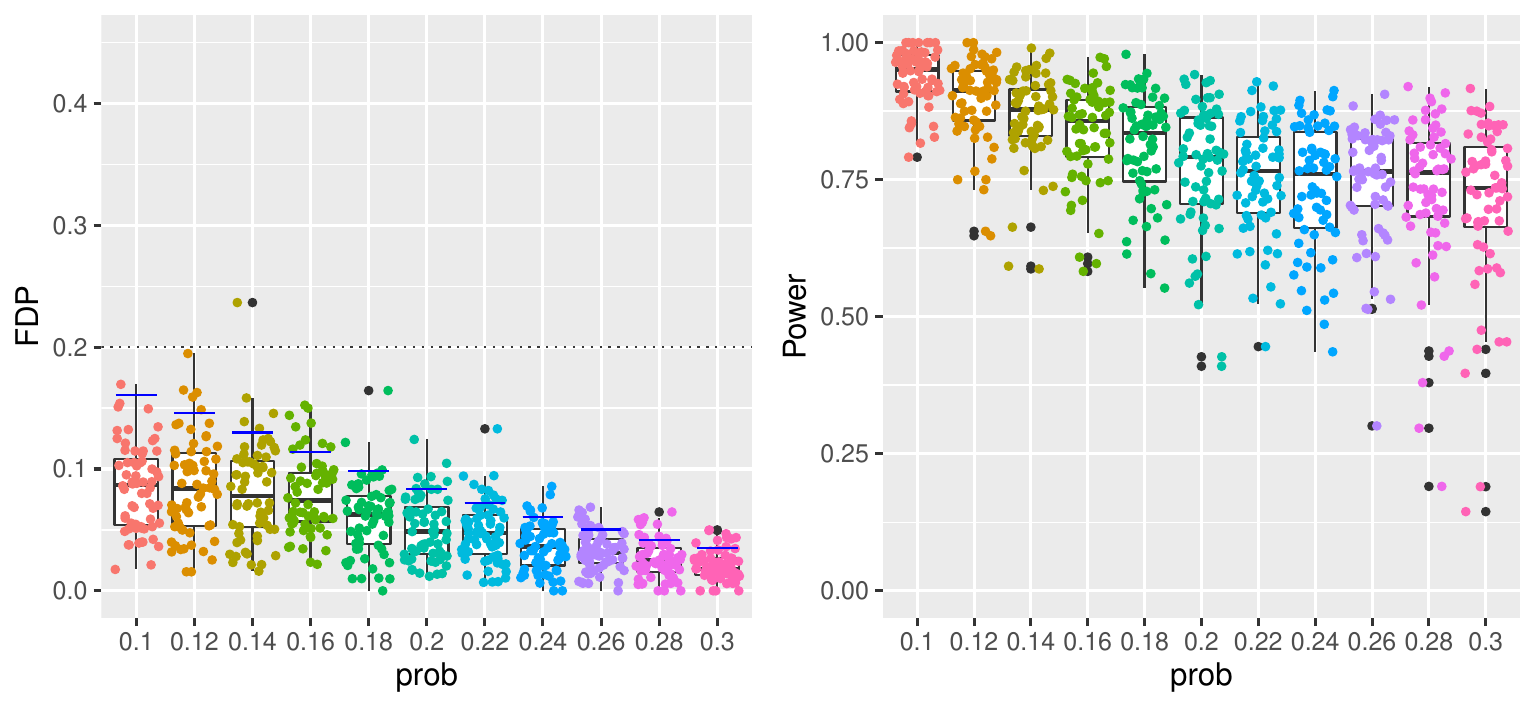}
    \caption{{FDP and power plots with $p=20$ and the nominal FDR level $q=0.2$. The other setups are the same as Figure \ref{fig:p20q0.1}.}}
    \label{fig:p20q0.2}
\end{figure}
\begin{figure}[htbp]
    \centering
    \includegraphics[width = 0.95\linewidth]{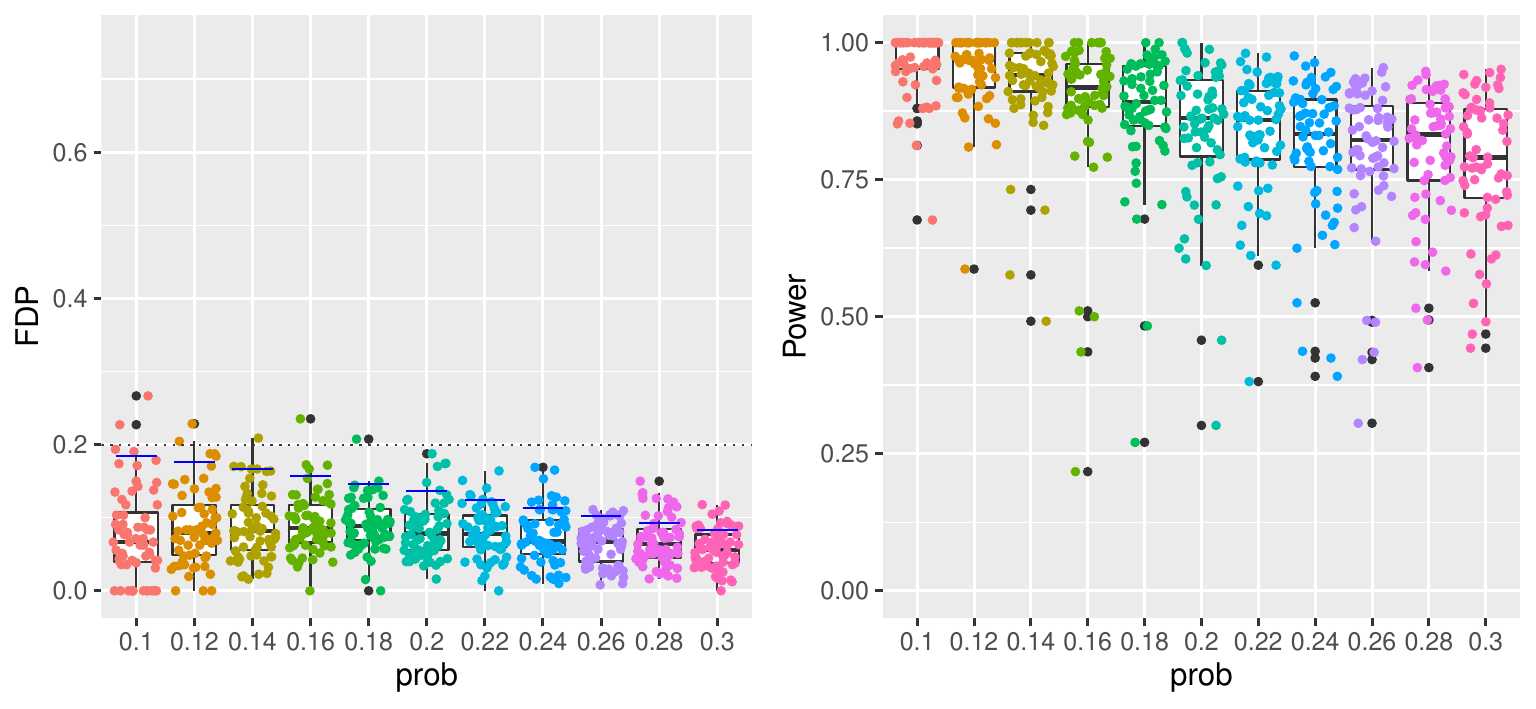}
    \caption{{FDP and power plots with $p=30$ and the nominal FDR level $q=0.2$. The other setups are the same as Figure \ref{fig:p20q0.1}.}}
    \label{fig:p30q0.2}
\end{figure}
 \end{supplement}
\end{document}